\documentclass{article}

\usepackage{latexsym}
\usepackage{amsmath}
\usepackage{amssymb}
\usepackage{amsthm}
\usepackage{amscd}

\usepackage{graphicx}
\usepackage{subfigure}
\usepackage{psfrag}

\usepackage{bm}
\usepackage{cancel}

\usepackage[mathscr]{eucal}

\usepackage{color}

\usepackage{multirow}

\newcommand{\Aplus}{{\fontfamily{phv}\fontseries{b}\selectfont A}${}_{\boldsymbol{+}}$}
\newcommand{\Apm}{{\fontfamily{phv}\fontseries{b}\selectfont A}${}_{\boldsymbol{\pm}}$}
\newcommand{\Azero}{{\fontfamily{phv}\fontseries{b}\selectfont A}${}_{\boldsymbol{0}}$}
\newcommand{\Azeroplus}{{\fontfamily{phv}\fontseries{b}\selectfont A}${}_{\boldsymbol{0}\boldsymbol{+}}$}
\newcommand{\Azerominus}{{\fontfamily{phv}\fontseries{b}\selectfont A}${}_{\boldsymbol{0}\boldsymbol{-}}$}
\newcommand{\Azeropm}{{\fontfamily{phv}\fontseries{b}\selectfont A}${}_{\boldsymbol{0}\boldsymbol{\pm}}$}
\newcommand{\Bplus}{{\fontfamily{phv}\fontseries{b}\selectfont B}${}_{\boldsymbol{+}}$}
\newcommand{\Cplus}{{\fontfamily{phv}\fontseries{b}\selectfont C}${}_{\boldsymbol{+}}$}
\newcommand{\Dplus}{{\fontfamily{phv}\fontseries{b}\selectfont D}${}_{\boldsymbol{+}}$}
\newcommand{\Aminus}{{\fontfamily{phv}\fontseries{b}\selectfont A}${}_{\boldsymbol{-}}$}
\newcommand{\Bminus}{{\fontfamily{phv}\fontseries{b}\selectfont B}${}_{\boldsymbol{-}}$}
\newcommand{\Cminus}{{\fontfamily{phv}\fontseries{b}\selectfont C}${}_{\boldsymbol{-}}$}
\newcommand{\Dminus}{{\fontfamily{phv}\fontseries{b}\selectfont D}${}_{\boldsymbol{-}}$}
\newcommand{\AminusP}{\mbox{{\fontfamily{phv}\fontseries{b}\selectfont A}${}_{\boldsymbol{-}}$\protect\hspace{-0.8em}\protect\raisebox{1ex}{\tiny IX}}}

\newcommand{\A}{{\fontfamily{phv}\fontseries{b}\selectfont A}}
\newcommand{\B}{{\fontfamily{phv}\fontseries{b}\selectfont B}}
\newcommand{\C}{{\fontfamily{phv}\fontseries{b}\selectfont C}}

\newcommand{\hx}{\hat{\xi}}
\newcommand{\heta}{\hat{\eta}}
\newcommand{\hn}{\hat{n}}

\newcommand{\textfrac}[2]{{\textstyle \frac{#1}{#2}}}

\newcommand{\R}{\mathbb{R}}

\newcommand{\weg}{\:\,}

\DeclareMathOperator{\tr}{tr}
\DeclareMathOperator{\diag}{diag}

\newtheorem{Theorem}{Theorem}

\newtheorem{Lemma}{Lemma}

\newtheorem{Assumption}{Assumption}
\theoremstyle{remark}\newtheorem*{Remark}{Remark}

\title{Bianchi Cosmologies with Anisotropic Matter: \\ Locally Rotationally Symmetric Models}
\author{Simone Calogero\footnote{E-Mail: calogero@ugr.es}\\[0.2cm]
Departamento de Matem\'atica Aplicada\\ 
Facultad de ciencias, Universidad de Granada\\
18071 Granada, Spain\\[0.5cm]
J.~Mark Heinzle\thanks{E-Mail: Mark.Heinzle@univie.ac.at}\\[0.2cm]
Gravitational Physics\\ Faculty of Physics, University of Vienna\\
1090 Vienna, Austria}
\date { }

\begin{document}

\maketitle

\begin{abstract}
The dynamics of cosmological models with isotropic matter 
sources (perfect fluids) is extensively studied in the literature; in comparison, the dynamics
of cosmological models with anisotropic matter sources is not.
In this paper we consider spatially homogeneous
locally rotationally symmetric solutions
of the Einstein equations with a large class of 
anisotropic matter models including collisionless matter (Vlasov),
elastic matter, and magnetic fields.
The dynamics of models of Bianchi types~I,~II, and~IX
are completely described; the two most striking results are the following:
(i) There exist matter models, compatible with the standard energy conditions, 
such that solutions of Bianchi type~IX (closed cosmologies) 
need not necessarily recollapse; 
there is an open set of forever expanding solutions. 
(ii) Generic type~IX solutions 
associated with a matter model like 
Vlasov matter 
exhibit oscillatory behavior toward the initial singularity.
This 
behavior differs significantly from that of 
vacuum/perfect fluid cosmologies; hence ``matter matters''.
Finally, we indicate that our methods can probably be extended to
treat a number of open problems, in particular, the dynamics
of Bianchi type~VIII and Kantowski-Sachs solutions.
\end{abstract}

%%%%%%%%%%%%%%%%%%%%%%%
\section{Introduction}\label{intro}

Spatially homogeneous (SH) cosmologies are one of the main cornerstones 
of General Relativity. 
On the one hand, they provide a rich supply
of case studies for the effects that solutions of the Einstein equations of General Relativity 
can (and will) produce; 
in particular, the behavior of SH cosmologies 
gives the necessary input to model the behavior of 
our actual universe. 
On the other hand, there are convincing arguments that
the dynamics of generic solutions of the Einstein equations
close to spacelike singularities are built on the dynamics
of SH cosmological models~\cite{BKL}.

An SH cosmological model is a 
solution of the Einstein equations that
is independent of the spatial variables in a suitable frame  
(see Sec.~\ref{SH} for a precise definition). 
Therefore, for these models, the Einstein equations reduce to an autonomous system of 
non-linear ordinary differential equations.
As first recognized by Collins~\cite{Col},
provided the system possesses the appropriate regularity,
the equations can be analyzed using the potent
methods of dynamical systems theory.  
Since then, our 
knowledge of the (qualitative) dynamics of SH cosmological 
models has increased substantially, see~\cite{Coley, HHW, WE} for reviews. 

Until recent years, the analysis of SH cosmologies has been largely restricted to isotropic
matter sources like perfect fluids (including dust) and scalar fields. 
However, the results available for anisotropic matter sources, 
such as electromagnetic fields~\cite{L}, collisionless (Vlasov) 
matter~\cite{HU2}, elastic materials~\cite{CH}, and viscous fluids~\cite{SR}, 
suggest that the dynamics of 
SH cosmological models are in fact sensitive to the choice of matter model,
not merely quantitatively, but also qualitatively.
A first step towards a systematic analysis of the influence of the matter model 
on the (asymptotic) dynamics of solutions 
was taken in~\cite{CH2}: 
By analyzing Bianchi type~I solutions 
(see Sec.~\ref{SH} for the definition) of the Einstein-matter equations
for a large class of (anisotropic) matter 
models that includes perfect fluids as a special case,
it was shown that the qualitative dynamics of solutions
(and the asymptotics towards the initial singularity, in particular) 
strongly depends on specific properties of the matter source.
Differences between perfect fluid solutions and anisotropic
matter solutions even exist when the matter model resembles a perfect fluid very closely~\cite{CH_GRG},
which suggests that, at least 
within the family of matter models considered,
the perfect fluid model is `structurally unstable'.

The purpose of this paper is to extend the analysis of~\cite{CH2} to the 
higher Bianchi types of class~A.
SH cosmologies of these types, especially the type~IX solutions,
are of particular interest. The important question
of whether a cosmological model will expand forever or will recontract (and eventually recollapse)
arises within this class of models.
Moreover, the understanding of
type~IX solutions with their oscillatory 
behavior (Mixmaster behavior) toward the initial singularity
is regarded as one of the keys towards a better understanding
of singularities in General Relativity in general.

In this work we restrict to SH cosmologies that are locally rotationally symmetric (LRS); 
in particular, the isometry group of the spacetime is 
four-dimensional, see Appendix~\ref{lrsexplained}.
The main result of the present paper can be summarized as follows:
The folklore statement ``matter does not matter'' has clear limitations,
which can be specified rigorously.
The qualitative dynamics of (SH LRS) cosmological models 
is in fact matter-dependent;
cosmological models with anisotropic matter  satisfying 
specific properties exhibit a completely different qualitative
behavior than the corresponding perfect fluid solutions.
The two most striking facts that we derive are the following.
\begin{itemize}
\item[(i)] There exist physically viable matter models such that 
the behavior of generic solutions
toward the initial singularity is oscillatory;
this is in stark contrast to the behavior of vacuum and perfect fluid solutions.
These models are LRS Bianchi type~II and~IX models.
\item[(ii)] There exist physically viable matter models such that
``closed-universe-recollapse''~\cite{Barrow/Galloway/Tipler:1986,Lin/Wald:1989} does not hold;
LRS Bianchi type~IX solutions (closed cosmologies) 
need not necessarily recollapse; there exists an open set of initial data
such that the associated solutions expand forever.
\end{itemize}

By `physically viable' matter models we mean matter models that 
satisfy the standard energy conditions, the weak energy condition,
the dominant energy condition, and the strong energy condition in particular. 
Note that perfect fluid solutions of Bianchi type~IX necessarily recollapse
provided the strong energy condition holds.

It is important to note that there exist important explicit examples of matter models that are 
directly covered by our analysis, e.g., elastic matter and collisionless (Vlasov) matter 
with massless particles, and matter models to which our techniques extend in a straightforward 
way, e.g., magnetic fields and Vlasov matter with massive particles; we refer to Sec.~\ref{mattermodelssec}.

This paper is organized as follows. In Sec.~\ref{SH} we 
define the class of spatially homogeneous spacetimes and the various Bianchi types;
the Einstein equations for these models are given.
In Sec.~\ref{Sec:matter} we specify the class of matter models we want to consider
by mild and physically motivated assumptions.
The family of matter models thus defined is sufficiently general to include 
important explicit examples and provides a non-trivial generalization of the perfect fluid 
matter model.
(In Sec.~\ref{discussion} we elaborate on possible generalizations of our assumptions on the matter model.) 
In Sec.~\ref{LRSsols} we impose the condition of locally rotational symmetry (LRS)
and discuss the Einstein equations and the matter variables for this case.
In Sec.~\ref{perfectfluid} we review the results on the dynamics of vacuum and perfect 
fluid models; the behavior of solutions with anisotropic matter we derive in the subsequent
sections can be gauged against these results.
In Sec.~\ref{reducedsystem} 
we introduce a new set of dynamical variables, which recast the Einstein equations into an autonomous 
dynamical system---the {\it reduced dynamical system}---over a relatively compact state space; we 
remark that our variables are different from the standard Hubble normalized variables~\cite{WE}. 
The following sections are devoted to the qualitative analysis
of the reduced dynamical system for the different Bianchi types employing methods and concepts 
from the theory of dynamical systems. The lower Bianchi types, i.e.,~I and~II, 
are studied in Secs.~\ref{bianchi1section} and~\ref{sec:II}, respectively.
The analysis of the Bianchi type~IX case is
divided into three sections: 
In Sec.~\ref{B89sec} we slightly reformulate the equations for later purposes; 
Sec.~\ref{B9sec} contains the analysis of the system, where we demonstrate
the need for an additional variable transformation; 
finally, in Sec.~\ref{B9res} we collect the results and prove
the main theorems on the dynamics of LRS Bianchi type~IX solutions
with anisotropic matter; in addition we put the theorems into
a broader context.
In Sec.~\ref{mattermodelssec} we show that our results directly apply 
to a specific matter model, namely collisionless (Vlasov) matter for massless 
particles; we thereby extend the results of~\cite{RT}. 
Moreover, we show that our analysis extends 
straightforwardly to the magnetic field matter model. 
Finally, in Sec.~\ref{discussion} we discuss 
possible generalizations of our results and give a list of interesting open problems. 
The paper contains three appendixes. 
In Appendix~\ref{exact} we present the exact solutions of 
the Einstein-matter equations that have been discovered
in this work, whose role is that of `attractors' of typical solutions.
Appendix~\ref{lrsexplained} contains a thorough discussion of local rotational symmetry.
Finally, Appendix~\ref{dynsysapp} reviews a few facts 
about dynamical system theory.

%%%%%%%%%%%%%%%%

%%%%%%%%%%%%%%%%%%%%%%%%%%%%%%%%%%%%%%%%%%%%%%%%%%%%%%%%%%%%%%%%%%%%%%%%%%%%%%%%%%%%%%%%%%%%%%%%%%%%%%%%%%
\section{Spatially Homogeneous spacetimes}
\label{SH}
%%%%%%%%%%%%%%%%%%%%%%%%%%%%%%%%%%%%%%%%%%%%%%%%%%%%%%%%%%%%%%%%%%%%%%%%%%%%%%%%%%%%%%%%%%%%%%%%%%%%%%%%%%

A spacetime 
$(M, {}^4\mathbf{g})$
is \textit{spatially homogeneous} if it admits an isometry 
group
whose orbits are spacelike 
hypersurfaces that foliate $M$. 
The spatially homogeneous spacetimes divide into
the Kantowski-Sachs models, see Appendix~\ref{KSmodels}, and the Bianchi models.

A \textit{Bianchi model} is defined to be a spatially homogeneous spacetime whose
isometry group possesses a three-dimensional subgroup $\mathcal{G}$, $\dim \mathcal{G}=3$, 
that acts simply transitively on the spacelike orbits.

\begin{Remark}
There are three subcases.
First, the Friedmann-Lema\^{\i}tre(-Robertson-Walker) models, see, e.g.,~\cite[Chapter~2]{WE},
which admit a six-dimensional isometry group. 
Second, the locally rotationally symmetric (LRS) Bianchi models, see Appendix~\ref{LRSclassA}, whose 
isometry group is four-dimensional; by definition, the group orbits coincide with the spacelike orbits of
the three-dimensional subgroup $\mathcal{G}$.
Finally, there are the Bianchi models that do not admit additional continuous 
symmetries beyond the three-parameter symmetry group $\mathcal{G}$.
\end{Remark}

By definition, a Bianchi spacetime $(M, {}^4\mathbf{g})$ 
has the manifold structure $M=I\times \mathcal{G}$, 
where $I\subset \R$ is an interval. Let $\mathcal{G}_t$, $t\in I$, 
denote the Lie group $\mathcal{G}_t=\{t\}\times \mathcal{G}$,
and let $\xi_i = \xi_i(t)$
be a left-invariant frame on $\mathcal{G}_t$; 
Latin indices run from $1$ to $3$. 
Expressed in terms of the coframe $\omega^i = \omega^i(t)$ that is
dual to $\xi_i$, the metric on $\mathcal{G}_t$ (the `spatial metric') takes 
the form $g_{i j}(t) \,\omega^i \omega^j$.
We temporarily adopt the Einstein summation convention, i.e.,
summation over repeated indices is understood.
In the following we assume that the spacetime coordinate $t$ 
is the standard `cosmological time' (`Gaussian time'), i.e., 
the proper time along the geodesic congruence orthogonal to the group orbits.
Then the metric on the Bianchi spacetime $(M, {}^4\mathbf{g})$ takes the form
\begin{equation}\label{metric}
{}^4\mathbf{g} = -dt^2+g_{ij}(t)\: \omega^i\omega^j\:.
\end{equation}

The Einstein equations (in units such that $8\pi G=c=1$) split into the constraint 
equations
\begin{subequations}\label{einsteinsystem}
\begin{equation}\label{constraint}
R-k^{ij} k_{ij}+(g^{ij}k_{ij})^2 = 2T_{00}\:,\qquad\nabla_ik^i_{\weg j}-\nabla_j(g^{lm} k_{lm})=-T_{0j}\:,
\end{equation}
which are called the Hamiltonian and the momentum constraint, respectively,
and the evolution equations
\begin{equation}\label{evolution}
\partial_t g_{i j} = -2 k_{i j} \:,\quad
\partial_t k^i_{\weg j} = R^i_{\weg j} + \left(g^{lm}k_{lm}\right) k^i_{\weg j} - T^i_{\weg j} +
\textfrac{1}{2} \delta^i_{\weg j} (g^{lm}T_{lm} -T_{00}) \:.
\end{equation}
Here, $g_{i j} = g_{i j}(t)$ is the (Riemannian) spatial metric, $g^{i j}$ its inverse; 
$R_{i j}$ is the associated Ricci tensor; $R=g^{lm}R_{lm}$ is the Ricci scalar;
$k_{i j} = k_{i j}(t)$ is the second fundamental form of the hypersurface
of spatial homogeneity $\mathcal{G}_t$; 
finally, $T_{00}$, $T_{0i}$, and $T_{ij}$
are the energy density, momentum density, and stress tensor, which are
derived from the energy-momentum tensor $T_{\mu\nu}$ of the matter; Greek indices run from $0$ to $3$.
Above and in the following, (Latin) indices are raised and lowered with the spatial metric $g_{ij}$ and its inverse $g^{i j}$.

We denote by
$C^k_{\weg ij} = C^k_{\weg ij}(t)$ 
the structure constants of the Lie algebra associated with $\mathcal{G}_t$.
In this paper we consider \textit{Bianchi class~A models}, 
which are characterized by $C^k_{\weg j k}=0$.
(However, we also briefly discuss LRS Bianchi type~III models which
are of class~B, see Appendix~\ref{typeIIIsubsec}.) Then
the Ricci tensor $R_{i j}$ is given
by
\begin{equation}\label{ricci}
R_{ij}=-\textfrac{1}{2}C^l_{\weg ki}\left(C^k_{\weg lj}+g_{lm}g^{kn}C^m_{\weg nj}\right)-
\textfrac{1}{4}C^m_{\weg nk}C^p_{\weg ql}g_{jm}g_{ip}g^{kq}g^{ln}\:.
\end{equation}
\end{subequations}

Let $\varepsilon_{ijk}$ denote the standard permutation symbol
and define, 
as usual, $n^{ij}=\textfrac{1}{2}\varepsilon^{ikl}C^j_{\weg kl}$.
Since \mbox{$\varepsilon_{i j l} \, n^{k l} = C^k_{\weg i j} + 2 \delta^k_{\weg (i} C^m_{\weg j)m}$}, 
it follows that, for class~A models, 
\begin{equation}\label{Cinn}
C^k_{\weg ij}=\varepsilon_{ijl} \,n^{kl}\:,
\end{equation}
where $n^{k l}$ is symmetric because $\varepsilon_{m i j} n^{[i j]} = C^l_{\weg m l} = 0$.
Using~\eqref{Cinn}, the momentum constraint in~\eqref{constraint} can be rewritten as
\begin{equation}\label{commute}
\varepsilon_{kij}n^{kl}k^i_{\weg l}=-T_{0j}\:.
\end{equation}

The above equations are valid in a general left-invariant frame $\xi_i$
and its dual $\omega^i$. 
The freedom of redefining such a frame by a \mbox{$t$-dependent} linear transformation 
can be used to simplify the equations, e.g., by 
introducing an orthonormal group-invariant frame (which is the approach 
used, e.g., in~\cite{WE}). For our purposes it is preferable to 
exploit this freedom to make the frame time-independent. 
If the frame $\xi_i(t) \equiv \hat{\xi}_i$ (and $\omega^i(t) \equiv \hat{\omega}^i$) 
is time-independent, so are the structure constants $C^k_{\ ij}(t) \equiv \hat{C}^k_{\ ij}$ 
and the matrix $n^{ij}(t)\equiv\hat{n}^{ij}$;
we use the convention that hatted quantities are time-independent.
Moreover, we are still allowed 
to use a time-independent (unimodular) linear 
transformation 
$\hat{n}^{i j} \mapsto \hat{L}^i_{\weg k} \hat{L}^j_{\weg l} \hat{n}^{k l}$
to diagonalize the constant symmetric matrix $\hat{n}^{ij}$, i.e., without loss
of generality,
\begin{equation}\label{L}
\hat{n}^{i j} = \diag(\hat{n}_1,\hat{n}_2,\hat{n}_3)\:.
\end{equation}
After a further rescaling of the frame 
and up to a permutation or an overall change of sign, 
the (`structure') constants $\hat{n}_1$, $\hat{n}_2$, $\hat{n}_3$ 
are given as in Table~\ref{tab1}. 
The values of these constants define the different \textit{Bianchi types} of Bianchi class A.

\begin{Remark}
Using $[\hat{\xi}_i,\hat{\xi}_j] =\hat{C}^k_{\weg i j} \hat{\xi}_k$ it is straightforward to see that the one-forms $\hat{\omega}^i$ 
satisfy 
\begin{equation*}
d\hat{\omega}^1=-\hat{n}_1\:\hat{\omega}^2\wedge\hat{\omega}^3\:,\quad 
d\hat{\omega}^2=-\hat{n}_2\:\hat{\omega}^3\wedge\hat{\omega}^1\:,\quad 
d\hat{\omega}^3=-\hat{n}_3\:\hat{\omega}^1\wedge\hat{\omega}^2\:.
\end{equation*} 
For each Bianchi type, the one-forms $\hat{\omega}^i$ admit different expressions
in terms of local coordinates; we refer to~\cite{Mac}.
\end{Remark}

\begin{table}
\begin{center}
\begin{tabular}{|c|rrr|}
\hline  & & & \\[-2.5ex]
Bianchi  & \multirow{2}{*}{$\hat{n}_1$} & \multirow{2}{*}{$\hat{n}_2$} & \multirow{2}{*}{$\hat{n}_3$}  \\[-0.5ex]
type & & & \\
\hline & & & \\[-2ex] 
I & 0 & 0 & 0 \\
II & 1 & 0 & 0 \\
VI$_0$ &  0 & 1 & ${-1}$ \\
VII$_0$ &  0 & 1 & 1 \\
VIII &  ${-1}$ & 1 & 1 \\
IX &  1
& 1 & 1  \\ 
\hline 
\end{tabular}
\caption{This table gives the classification of Bianchi class A models into the different \textit{Bianchi types}. 
Each Bianchi type is defined by the values of the triple $(\hat{n}_1,\hat{n}_2,\hat{n}_3)$; this is up to permutations 
and overall changes of sign of the constants (which correspond to discrete symmetries of the frame).}
\label{tab1}
\end{center}
\end{table}

We conclude this section by noting that the system of equations~\eqref{einsteinsystem} implies a number 
of basic identities.
Of particular relevance is the relation
\begin{equation}\label{equationR}
\partial_t R=2R^{ij}k_{ij}\:.
\end{equation}
Since
\begin{equation*}
\partial_tR=2k^{ij}R_{ij}+g^{ij}\partial_t R_{ij}\:,
\end{equation*}
to prove~\eqref{equationR} we need to prove that $g^{ij}\partial_t R_{ij}=0$.
Because we use a time-independent frame, the structure constants are time-independent; hence,
differentiation of~\eqref{ricci}, where we use the antisymmetry $\hat{C}^i_{jk}=-\hat{C}^i_{kj}$,
leads to 
\begin{align*}
g^{i j} \partial_tR_{ij} & = g^{ij} \, \hat{C}^l_{\ ki}\hat{C}^m_{\ nj}\, (k_{lm}g^{kn}-k^{kn}g_{lm}) \\
 & \quad - \textfrac{1}{2}\, g^{i j} \,\hat{C}^m_{\ nk}\hat{C}^p_{\ ql}\,\big(-k_{jm}g_{ip}g^{kq}g^{ln}-k_{ip}g_{jm}g^{kq}g^{ln}
+k^{kq}g_{jm}g_{ip}g^{ln}+k^{ln}g_{jm}g_{ip}g^{kq}\big) \\[0.5ex]
& = \hat{C}^l_{\ ki}\hat{C}^m_{\ nj}k_{lm} g^{ij} g^{kn} -
\hat{C}^l_{\ ki}\hat{C}^m_{\ nj} k^{kn} g^{ij} g_{lm} + \hat{C}^m_{\ nk}\hat{C}^p_{\ ql} k_{pm} g^{kq}g^{ln}
-\hat{C}^m_{\ nk}\hat{C}^p_{\ ql} k^{kq}g_{mp}g^{ln} \\[0.5ex]
& = 0 \:,
\end{align*}
which establishes the claim.

%%%%%%%%%%%%%%%%%%%%%%%%%%%%%%%%%%%%%%%%%%%%%%%%%%%%%%%%%%%%%%%%%%%%%%%%%%%%%%%%%%%%%%%%%%%%%%%%%%%%%%%%%%
\section{Anisotropic matter models}
\label{Sec:matter}
%%%%%%%%%%%%%%%%%%%%%%%%%%%%%%%%%%%%%%%%%%%%%%%%%%%%%%%%%%%%%%%%%%%%%%%%%%%%%%%%%%%%%%%%%%%%%%%%%%%%%%%%%%

The system of equations~\eqref{einsteinsystem} does not form a complete system in general 
unless some information on the energy-momentum tensor is given. For instance, one may specify
the dependence of the energy-momentum tensor on some matter fields and add to the system~\eqref{einsteinsystem} 
the evolution equations for these matter fields.
These evolution equations are required to be compatible with the Bianchi identities $\nabla_\mu T^{\mu\nu}=0$;
for certain matter models, such as perfect fluids, the Bianchi identities already contain the entire 
information on the dynamics of the matter fields. 

In this paper we do not explicitly specify matter fields or an energy-momentum 
tensor. 
Rather, the latter will be restricted by a set of mild and physically motivated assumptions 
which are known to hold for important examples of matter models. 
A particular consequence of our assumptions will be  
that the dynamics of the matter fields is contained in the Einstein equations~\eqref{einsteinsystem} 
through the Bianchi identities.

%%%%%%%%%%%%%%%%%%%%%%%%%%%%%%%%%%%%%%%%%%%%%%%%%%%%%%%%%%%%%%%%%%%%%%%%%%%%%%%%%%%%%%%%%%%%%%%%%%%%%%%%%%%%
\subsection{Basic assumptions on the matter model}
\label{assT}
%%%%%%%%%%%%%%%%%%%%%%%%%%%%%%%%%%%%%%%%%%%%%%%%%%%%%%%%%%%%%%%%%%%%%%%%%%%%%%%%%%%%%%%%%%%%%%%%%%%%%%%%%%%%  

The most common matter model in cosmology is a (non-tilted) perfect fluid. 
The energy-momentum tensor is
\begin{equation}\label{Tfluid}
T_{\mu\nu}=\rho \,( dt\otimes dt)_{\mu\nu} +p\big[{}^4\mathbf{g} + dt\otimes dt\big]_{\mu\nu}\:,
\end{equation}
where the energy density $\rho$ and the pressure $p$ are usually required to 
obey a linear equation of state, i.e.,
\begin{equation*}
p=w\rho\,, \qquad w= \mathrm{const}\:. 
\end{equation*}
The energy-momentum 
tensor~\eqref{Tfluid} is isotropic in the sense that
the eigenvalues of the spatial stress tensor $T^i_{\weg j}$ 
are all equal to the pressure $p$, since  $T^i_{\weg j} = p \,\delta^i_{\weg j}$.
Moreover, the Bianchi identities $\nabla_\mu T^{\mu\nu}=0$ 
reduce to an ordinary differential equation for the energy density $\rho$, which can be 
solved to express $\rho$, and therefore the entire tensor 
$T_{\mu\nu}$, in terms of the spatial metric:
\begin{equation}\label{rhofluid}
\rho=\rho_0 n^{1+w}\:,\qquad\rho_0=\mathrm{constant}\:,\quad n=(\det g)^{-1/2}\:.
\end{equation}

In this paper we consider \textit{anisotropic matter sources} that naturally generalize perfect fluids.
The energy-momentum tensor is not specified explicitly; instead, we introduce a set of 
assumptions on the matter models that lead to a class of \textit{anisotropic energy-momentum 
tensors} naturally encompassing~\eqref{Tfluid}. 
We note that there exist standard matter models that satisfy our assumptions; among
these are collisionless matter (Vlasov matter) for massless particles and
elastic materials (for a wide variety of constitutive equations); 
we refer to the subsequent remarks and to~\cite{CH2, letter}.

Our first (and main) assumption is
that, in analogy with the perfect fluid case, 
the energy-momentum tensor is given as a function of the 
spatial metric; 
this might be possible only under certain restrictions on the initial data.

\begin{Assumption}\label{assumptionT}
There exist initial data for the matter such that 
the components of the energy-momentum tensor 
are represented by smooth (at least $\mathcal{C}^1$)
functions of the metric $g_{ij}$. 
We assume that $\rho=\rho(g_{ij})$ is positive 
(as long as $g_{ij}$ is non-degenerate).
\end{Assumption}

\begin{Remark} 
By Assumption~\ref{assumptionT}, 
the components of the energy-momentum tensor 
are represented by functions of the spatial metric;
note that the choice of frame may affect the form of these functions.
In addition, these functions might include a number of external parameters or external
functions (which describe the properties of the matter) and the initial data
of the matter fields. 
Assumption~\ref{assumptionT} implies that the matter does not 
add any dynamical degrees of freedom to the problem.
We refer to~\cite{UJR} for a discussion of the Hamiltonian structure of the 
matter models satisfying Assumption~\ref{assumptionT}.
\end{Remark}

In the Bianchi type~I case we have shown in~\cite{CH2} that the 
spatial components of an energy-momentum tensor 
satisfying Assumption~\ref{assumptionT} are completely determined by the energy density 
\[
T_{00}=\rho=\rho(g_{ij})
\] 
through the formula 
\begin{equation}\label{Tij}
T^i_{\ j}=-2\frac{\partial\rho}{\partial g_{il}}\,g_{jl}-\delta^i_{\ j}\,\rho\:.
\end{equation}
We claim that eq.~\eqref{Tij} holds for every Bianchi class~A model.

The key to prove eq.~\eqref{Tij} is to use relation~\eqref{equationR}.
Employing~\eqref{equationR} in the derivative (w.r.t.\ time) of the first equation of~\eqref{constraint}
and using~\eqref{evolution}, we easily obtain
\begin{equation}\label{dtrho}
\partial_t\rho=k_i^{\ j}\, (T^i_{\ j}+\delta^i_{\ j}\rho)\:.
\end{equation}
On the other hand, since $\rho=\rho(g_{ij})$ 
we may use the first equation in~\eqref{evolution} to get
\begin{equation}\label{dtrho2}
\partial_t\rho=-2k_i^{\ j}g_{jl}\frac{\partial\rho}{\partial g_{jl}}\:.
\tag{\ref{dtrho}${}^\prime$}
\end{equation}
Equating the r.h.\ sides of~\eqref{dtrho} and~\eqref{dtrho2} gives a polynomial of $k^i_{\ j}$ 
whose coefficients depend only on $g_{ij}$. Therefore, the coefficients must be equal, 
which establishes eq.~\eqref{Tij}. 

The second basic property of the energy momentum tensor~\eqref{Tfluid} 
that we want to require is the absence of heat flux and viscosity terms.

\begin{Assumption}\label{t0j=0} 
We assume that there exist initial data for the matter such that $T_{0j}=0$\,.
\end{Assumption}

\begin{Remark}
In the Bianchi type~I case this is not an assumption but a consequence of the field equations, 
cf.~\eqref{commute} with $n^{i j} \equiv 0$.
\end{Remark}

%%%%%%%%%%%%%%%%%%%%%%%%%%%%%%%%%%%%%%%%%%%%%%%%%%%%%%%%%%%%%%%%%%%%%%%%%%%%%%%%%%%%%%%%%%%%%%%%%%%%%%%%%%%%
\subsection{Diagonal models}
\label{diagT}
%%%%%%%%%%%%%%%%%%%%%%%%%%%%%%%%%%%%%%%%%%%%%%%%%%%%%%%%%%%%%%%%%%%%%%%%%%%%%%%%%%%%%%%%%%%%%%%%%%%%%%%%%%%%  

By the previous assumptions, the set of equations~\eqref{einsteinsystem} forms 
a complete system of ordinary differential equations for the functions $g_{ij}(t)$ and $k^i_{\ j}(t)$.
Initial data at %a given time 
$t_0>0$ are given by
$g_{ij}(t_0)$ and $k^i_{\weg j}(t_0)$ and the initial values of the matter fields
(which determine the function $\rho(g_{ij})$ and thus the energy-momentum tensor).
It follows from Assumption~\ref{t0j=0} and eq.~\eqref{commute} 
that the matrices $\hat{n}^{ij}$
and $k^i_{\ j}$ commute; it is thus possible
to make a change of frame $\hat{\xi}_i \mapsto \hat{L}^k_{\weg i} \hat{\xi}_k$
to simultaneously diagonalize the structure constants matrix $\hat{n}^{ij}$ and the second 
fundamental form  $k^i_{\ j}(t_0)$; this is simply achieved by choosing 
eigenvectors of $k^i_{\ j}(t_0)$ as the frame; 
since these eigenvectors are---or can be chosen to be---orthogonal,  
we can assume that $g_{ij}(t_0)$ is diagonal as well. 
The so constructed (time-independent, left-invariant) frame will be called
the initially orthogonal frame.

To ensure that $(g_{ij}, k^i_{\weg j})$ remain diagonal for all times, the energy-momentum tensor
must be diagonal in the initially orthogonal frame.

\begin{Assumption}\label{assumptiondiagonal}
We assume that there exist initial data for the matter such that 
$T^i_{\weg j}$ is diagonal in the initially orthogonal frame.
\end{Assumption}

Assumption~\ref{assumptiondiagonal} ensures that the off-diagonal components of the metric satisfy a homogeneous ODE; 
uniqueness of solutions of the evolution equations then implies 
that $(g_{ij}, k^i_{\weg j})$ remain
diagonal for all times; hence, the initially orthogonal frame is in fact an orthogonal frame for all times. 
We refer to solutions of this type as~\textit{diagonal models} and restrict our analysis to these models.

For diagonal models, the Ricci tensor $R_{i j}$ of the spatial metric $g_{i j}$, cf.~\eqref{ricci},
is diagonal as well.
To neatly express $R_{i j}$ in terms of $g_{i j}$ we define
\begin{equation*}
m_i := \sqrt{\frac{g_{ii}}{{g_{jj}}{g_{kk}}}} = \sqrt{\frac{{g^{jj}}{g^{kk}}}{g^{ii}}} \;.
\end{equation*}
Here and 
henceforth we denote by $(ijk)$ 
a cyclic permutation of $(123)$ (i.e., $\varepsilon_{ijk}=1$). 
Moreover, the Einstein summation convention is suspended, i.e.,
there is no summation over repeated indices,
unless stated otherwise.

Based on~\eqref{ricci} we find that the components of the Ricci tensor 
are given by
\begin{subequations}
\begin{equation}\label{riccidiagonal}
R^i_{\ i}=\textfrac{1}{2}\left[\hat{n}_i^2m_i^2-\left(\hat{n}_jm_j-\hat{n}_km_k\right)^2\right]\:.
\end{equation}
The spatial curvature scalar is
\begin{equation}\label{curvscal}
R = \hat{n}_im_i\left(\hat{n}_jm_j+\hat{n}_km_k-\textfrac{1}{2}\hat{n}_im_i\right)
-\textfrac{1}{2}\left(\hat{n}_jm_j-\hat{n}_km_k\right)^2 \:.
\end{equation}
\end{subequations}

For diagonal models the Einstein equations~\eqref{einsteinsystem} simplify.
The Hamiltonian constraint, cf.~\eqref{constraint}, becomes
\begin{subequations}\label{einsteinsystemdiagonal}
\begin{equation}\label{constraintdiagonal}
(\tr k)^2 - \sum\nolimits_l \left(k^l_{\ l}\right)^2+ R = 2 \rho\:,
\end{equation}
where $\tr k = \sum_l k^l_{\weg l}$ and $R$ is given by~\eqref{curvscal}.
The first 
evolution equation, cf.~\eqref{evolution}, reads
\begin{equation}\label{evol1diag}
\partial_t g_{ii} = -2 \,g_{ii}\,k^i_{\ i}\:, \quad\text{or, equivalently,}\quad
\partial_t m_i = m_i \,\big(\tr k - 2 k^i_{\weg i} \big)\:.
\end{equation}
The second evolution equation, cf.~\eqref{evolution}, takes the form
\begin{equation}\label{evolutiondiagonal}
\partial_t k^i_{\ i}=R^i_{\ i}+ (\tr k)\, k^i_{\ i}-p_i+
\textfrac{1}{2}\Big(\sum\nolimits_l p_l-\rho\Big)\:,
\end{equation}
\end{subequations}
where $R^i_{\weg i}$ is given by~\eqref{riccidiagonal} and where we 
have set
\[
T^i_{\weg j}=\mathrm{diag}(p_1,p_2,p_3)\:.
\]
The quantities $p_i$ are called the \textit{principal pressures}. 
(Note that when $p_1=p_2=p_3=p$
we recover the orthogonal frame components of the energy-momentum tensor of a perfect fluid, see~\eqref{Tfluid}.)

%%%%%%%%%%%%%%%%%%%%%%%%%%%%%%%%%%%%%%%%%%%%%%%%%%%%%%%%%%%%%%%%%%%%%%%%%%%%%%%%%%%%%%%%%%%%%%%%%%%%%%%%%%%%
\subsection{Properties of anisotropic matter models}
\label{assT2}
%%%%%%%%%%%%%%%%%%%%%%%%%%%%%%%%%%%%%%%%%%%%%%%%%%%%%%%%%%%%%%%%%%%%%%%%%%%%%%%%%%%%%%%%%%%%%%%%%%%%%%%%%%%%  
The principal pressures $p_i$ give rise to
the average (or isotropic) pressure 
\[
p=\textfrac{1}{3}(p_1+p_2+p_3)\:.
\]
We make the following simplifying assumption:

\begin{Assumption}\label{assumptionwi}
We suppose that the isotropic pressure and the density are proportional, 
i.e., we assume $p=w\rho$, $w = \mathrm{const}$, where
\begin{equation}\label{eqasswi}
w \in (-\textfrac{1}{3},1)\:.
\end{equation}
\end{Assumption}

\begin{Remark}
The cases $w= 1$ and $w=-\textfrac{1}{3}$ lead to different dynamics, which we refrain from discussing here. 
We remark, however, that our results for Bianchi type I solutions, see Sec.~\ref{bianchi1section}, 
are valid for $w\in (-1,1)$. In Sec.~\ref{discussion} we outline how to treat the more general case when $w\neq\mathrm{const}$.
\end{Remark}

According to Assumption~\ref{assumptionwi}, the energy density and the isotropic pressure behave like
those of a perfect fluid with a linear equation of state satisfying the strong and the
dominant energy condition.
It is natural to focus the attention to anisotropic matter sources of this kind, because
they generalize the 
class of perfect fluid models commonly used in cosmology, see~\cite{WE, HHW} and the references therein.

Recall that the energy density of perfect fluid matter is described in terms of the spatial metric by~\eqref{rhofluid}. 
This relation possesses a natural generalization for anisotropic matter models satisfying
Assumption~\ref{assumptionT}--\ref{assumptionwi}.
Let again $n = (\det g)^{-1/2} = \sqrt{g^{11}g^{22}g^{33}}$ and define
\begin{equation}\label{svariables}
s_i := (g^{11}+g^{22}+g^{33})^{-1}g^{ii}\:.
\end{equation}
Note that $s_1+s_2+s_3=1$, thus only two 
of the variables $(s_1,s_2,s_3)$ are independent.
The domain of these variables is the interior of the triangle
\[
\mathscr{T}=\{0<s_i<1:s_1+s_2+s_3=1\}\:.
\]
By Assumption~\ref{assumptionwi} and~\eqref{Tij} we obtain that $\rho=\rho(g^{11},g^{22},g^{33})$ must satisfy the equation
\begin{equation}\label{gllrhogll}
\sum_l g^{ll}\frac{\partial\rho}{\partial g^{ll}}=\textfrac{3}{2}(1+w)\rho\:.
\end{equation}
In the variables $(n,s_1,s_2,s_3)$ 
(which are related to $(g^{11},g^{22},g^{33})$ by a one-to-one transformation) 
eq.~\eqref{gllrhogll} reads $n\,\partial_n\rho=(1+w)\rho$. Integration yields
\begin{equation}\label{rhons}
\rho=n^{1+w}\psi(s_1,s_2,s_3)\:,
\end{equation}
where $\psi$ is a function on $\mathscr{T}$; 
we require $\psi$ to be smooth (at least $\mathcal{C}^1$).

This relation generalizes~\eqref{rhofluid} (because it reduces to~\eqref{rhofluid} when $\psi=\mathrm{const}$).
We see that the energy density of an anisotropic matter model satisfying
Assumption~\ref{assumptionT}--\ref{assumptionwi} depends on the spatial metric in
a particular manner. Since $\psi$ is unspecified,~\eqref{rhons} covers 
a large variety of matter models. 

\begin{Remark}
Note that $\psi$ can (and in general will) depend on the initial
data of the matter fields. This is true even in the simplest of cases,
the case of a perfect fluid, where $\psi \equiv \rho_0$, cf.~\eqref{rhofluid},
where $\rho_0 = \mathrm{const}$.
\end{Remark}

We define the \textit{rescaled principal pressures} $w_i$ to be 
\begin{equation*}
w_i = \frac{p_i}{\rho} \:;
\end{equation*}
evidently, $w_1 + w_2 + w_3 = 3 w$. 
In general, the rescaled (anisotropic) pressures are not constant.
However, as a consequence of Assumption~\ref{assumptionwi}, 
we observe independence of $n$.
More specifically, by using~\eqref{rhons} in relation~\eqref{Tij}, 
we find that
\begin{equation}\label{wis}
w_i = w_i(s_1,s_2,s_3) = w + 2 \left( \frac{\partial\log \psi}{\partial\log s_i} - 
s_i\,\sum\nolimits_l \frac{\partial\log \psi}{\partial\log s_l} \right) \:.
\end{equation}
Like the function $\psi$, the rescaled pressures $w_i$ are defined on the domain $\mathscr{T}$.
For our purposes it is essential, however, that the quantities $w_i$ be
well-behaved in the limit $(s_1,s_2,s_3) \rightarrow \partial\mathscr{T}$,
see Sec.~\ref{assT3}.

The rescaled pressures permit 
a straightforward representation of the \textit{energy conditions}~\cite{HE}.
As stated in Assumption~\ref{assumptionT}, our basic assumption is $\rho > 0$.
Then the weak energy condition 
is expressed in terms of the rescaled principal pressures as
$w_i \geq {-1}$ and the dominant energy condition as
$|w_i(s_1,s_2,s_3)|\leq 1$, $\forall (s_1,s_2,s_3) \in \mathscr{T}$ and $\forall\,i=1,2,3$.
The strong energy condition is satisfied if the weak energy condition holds and if, in addition, 
$w\geq -1/3$. 
We are mainly interested in matter models that satisfy the energy conditions; however,
we shall also discuss the more general case.

%%%%%%%%%%%%%%%%%%%%%%%%%%%%%%%%%%%%%%%%%%%%%%%%%%%%%%%%%%%%%%%%%%%%%%%%%%%%%%%%%%%%%%%%%%%%%%%%%%%%%%%%%%%%
\subsection{Asymptotic properties of anisotropic matter models}
\label{assT3}
%%%%%%%%%%%%%%%%%%%%%%%%%%%%%%%%%%%%%%%%%%%%%%%%%%%%%%%%%%%%%%%%%%%%%%%%%%%%%%%%%%%%%%%%%%%%%%%%%%%%%%%%%%%%  

The matter models we want to consider are those that exhibit suitable 
asymptotic properties of the
rescaled pressures $w_i$.

\begin{Assumption}\label{asswi}
We assume that the functions $w_i(s_1,s_2,s_3)$ representing the rescaled principal pressures 
possess a regular (sufficiently smooth) extension from $\mathscr{T}$ to $\overline{\mathscr{T}}$.
Furthermore, we assume that there exists a constant $v_-$ such that
\begin{equation}
w_i(s_1,s_2,s_3) = v_- \quad \text{when } s_i = 0 \:,
\end{equation}
for arbitrary~$i$.
\end{Assumption} 

\begin{Remark}
The assumption 
implies that $ w_i(s_1,s_2,s_3) = v_+ := 3 w - 2 v_-$ when $s_i = 1$, for arbitrary~$i$.
\end{Remark}

\begin{Remark}
Some comments on Assumption~\ref{asswi} are in order.
The assumption that the quantities $w_i(s_1,s_2,s_3)$ take limits when $(s_1,s_2,s_3)$ converges
to a point on $\partial\mathscr{T}$ corresponds to the assumption
that the rescaled principal pressures take finite values under extreme conditions,
i.e., when one or more components of the spatial metric $g_{i j}$ go to zero or infinity
(which occurs close to a singularity).
The second part of Assumption~\ref{asswi} means that the rescaled 
principal pressure in the direction $i$
becomes independent of the values of $g_{jj}$ and $g_{kk}$ 
in the limit $g_{ii} \rightarrow \infty$
provided that $g_{jj}$, $g_{kk}$ remain bounded.
(The complementary statement is implicit in 
the assumption.
The fact that $w_i$ is well-defined for $s_i = 1$ 
means that the normalized principal pressure in a direction $i$
converges to a limit as $g_{ii} \rightarrow 0$, when
$g_{jj}$ and $g_{kk}$ remain bounded from below.)
\end{Remark}

\begin{Remark}
Note that a linear equation of state connecting the principal pressures and
the density is necessary for self-similar solutions to exist~\cite{MMT};
Assumption~\ref{asswi} is formulated to allow for asymptotic self-similarity. A natural generalization of 
Assumption~\ref{asswi} is discussed in Sec.~\ref{discussion}.
\end{Remark}

A quantity that will turn out to be ubiquitous in our analysis is the `\textit{anisotropy parameter}' $\beta$.
As in~\cite{CH2} we define
\begin{equation}\label{betadef1}
\beta := 2\,\frac{w-v_-}{1-w}\:.
\end{equation}
The anisotropy parameter $\beta$ is a measure of the deviation of the 
matter model from isotropy under extreme conditions (this is because $v_-$
denotes the rescaled pressure $w_i$ when $g_{ii} \rightarrow \infty$, see Assumption~\ref{asswi}).

The anisotropy parameter $\beta$ gives rise to a classification of the 
anisotropic matter models into different types, see Table~\ref{tab2}
and Fig.~\ref{mattermodelsfig}.
It is shown in~\cite{CH2} 
that, in the Bianchi type~I case, this \textit{anisotropy classification} 
is intimately connected with the qualitative asymptotic dynamics of the associated
cosmological models. In this paper we will see that this connection
carries over to higher Bianchi types to a large extent (although we shall need to slightly refine the anisotropy 
classification).

We make a final simplifying assumption on the anisotropic matter model:

\begin{Assumption}\label{asspsi}
We assume that either there exists a unique isotropic state of the matter
or the matter model is a perfect fluid.
\end{Assumption}

Assumption~\ref{asspsi} can be formalized straightforwardly.
Uniqueness of an isotropic state means that 
$\psi$ has only one critical point in the interior of $\mathscr{T}$, which 
is either a maximum or a minimum. 
By~\eqref{wis} this entails the existence of a unique triple
$(\bar{s}_1,\bar{s}_2,\bar{s}_3) \in \mathscr{T}$ such that
$w_1 = w_2 = w_3 = w$ at $(\bar{s}_1,\bar{s}_2,\bar{s}_3)$.
This state is then a state of either maximal or
minimal energy (where, physically, the more important case 
is that of a minimum of the energy).
Although it is physically reasonable, Assumption~\ref{asspsi} 
can easily be relaxed to admit the existence
of more than one isotropic state; for simplicity we refrain from doing so.

\begin{table}%\label{mattermodels}
\begin{center}
\begin{tabular}{|c|c|c|}
\hline  & &  \\[-2ex]
Type  & $\beta$ values& Energy conditions require  \\ \hline  & &   \\[-2ex] 
\Dminus & $\beta\leq -2$ & $\beta=-2, w\geq \textfrac{1}{3}$ \\
\Cminus & $\beta\in (-2,-1)$ & $w>0, \beta\geq\max\{-2,-\textfrac{1+w}{1-w}\}$\\
\Bminus &  $\beta=-1$ & $w\geq 0 $ \\
\Aminus &  $\beta\in (-1,0)$ &  $w\geq -\textfrac{1}{3},\beta\geq\max\{-1,-\textfrac{1+w}{1-w}\}$  \\
\!\A$_{\,0}$ &  $\beta=0$ & $w\geq -1/3$ \\
\Aplus &  $\beta\in (0,1)$ & $w\geq -1/3$ \\ 
\Bplus & $\beta=1$ & $w\geq -1/3$ \\
\Cplus &  $\beta\in (1,2)$ & violated\\
\Dplus &  $\beta\geq 2$ & violated \\
\hline 
\end{tabular}
\caption{Anisotropy classification of matter sources according to the value of
the anisotropy parameter $\beta$. 
The value $\beta=0$ is exceptional---it corresponds to a matter source which behaves like 
a perfect fluid `at the singularity'. 
In the second column we give the range of $\beta$ and $w$ that guarantees
the validity of the strong and dominant energy condition, see~\cite{CH2} for more details.} 
\label{tab2}
\end{center}
\end{table}

\begin{figure}[Ht]
\begin{center}
\psfrag{A0}[cc][cc][1][0]{\Azero}
\psfrag{Ap}[cc][cc][1][0]{\Aplus}
\psfrag{Am}[cc][cc][1][0]{\Aminus}
\psfrag{Bm}[cc][cc][1][0]{\Bminus}
\psfrag{Cm}[cc][cc][1][0]{\Cminus}
\psfrag{Dm}[cc][cc][1][0]{\Dminus}
\psfrag{Bp}[cc][cc][1][0]{\Bplus}
\psfrag{Cp}[cc][cc][1][0]{\Cplus}
\psfrag{Dp}[cc][cc][1][0]{\Dplus}
\psfrag{c}[cc][cc][0.6][0]{\AminusP}
\psfrag{b0}[cc][cc][1.1][0]{$\beta = 0$}
\psfrag{b1}[cc][cc][0.9][0]{$\beta = 1$}
\psfrag{b2}[cc][cc][0.9][0]{$\beta = 2$}
\psfrag{bm1}[cc][cc][0.9][0]{$\beta = {-1}$}
\psfrag{bm2}[cc][cc][0.9][0]{$\beta = {-2}$}
\includegraphics[width=0.98\textwidth]{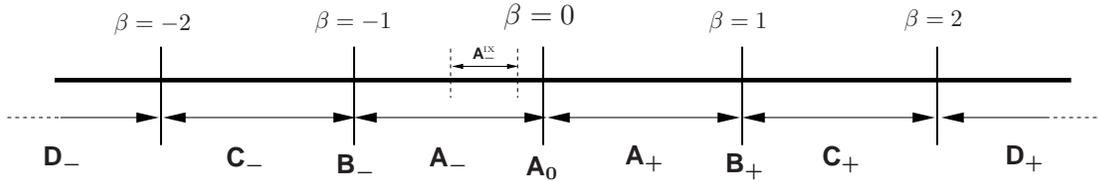}
\end{center}
\caption{Anisotropy classification according to the value of
the anisotropy parameter $\beta$. (The special case \AminusP\ is discussed
in Sec.~\ref{B9sec}; its range depends on the value of $w$.)}
\label{mattermodelsfig}
\end{figure}

To conclude this section 
we remark that there exist important examples of matter models that satisfy 
our assumptions in the Bianchi type~I case: Collisionless (Vlasov) 
matter for massless particles and elastic materials (for a wide variety of constitutive equations); we refer to~\cite{CH2} 
for details. These matter models fulfill our assumptions for the remaining
Bianchi types as well, at least within the
symmetry class that will be studied in this paper (locally rotationally 
symmetric models, see Sec.~\ref{LRSsols}). In particular, elastic matter falls into one of the $\boldsymbol{+}$ types, 
depending on the constitutive equation of the material, whereas collisionless matter (for particles with 
zero mass) is of type \Bplus, see Sec.~\ref{mattermodelssec}. If the assumption of local rotational symmetry is removed, 
it is not known how to resolve the field (Vlasov) equation for collisionless matter, 
except in the Bianchi type~I case (see~\cite{MM}). This fact precludes to prove 
that collisionless matter satisfy our assumptions in all spatially homogeneous 
models (see the remark after Assumption~\ref{assumptionT}).  
Finally there exist important matter models that do not satisfy
our assumptions  directly  but to which our techniques can be easily extended. 
Examples of such models are Vlasov matter for massive particles, see~\cite{letter}, and magnetic fields, 
see Sec~\ref{mattermodelssec}.

%%%%%%%%%%%%%%%%%%%%%%%%%%%%%%%%%%%%%%%%%%%%%%%%%%%%%%%%%%%%%%%%%%%%%%%%%%%%%%%%%%%%%%%%%%%%%%%%%%%%%%%%%%
\section{LRS models with anisotropic matter}
\label{LRSsols}
%%%%%%%%%%%%%%%%%%%%%%%%%%%%%%%%%%%%%%%%%%%%%%%%%%%%%%%%%%%%%%%%%%%%%%%%%%%%%%%%%%%%%%%%%%%%%%%%%%%%%%%%%%

In this paper we restrict our attention to locally rotationally symmetric (LRS) models. 
As discussed in detail in Appendix~\ref{lrsexplained}, there exists a one-dimensional isotropy 
group in LRS models that defines a plane of rotational symmetry. By using an adapted frame one can
choose two of the components of the diagonal metric to be 
equal (these are the components of the metric induced on the plane of rotational symmetry). W.l.o.g.\ we make a choice of
LRS frame such that
\begin{equation}
g_{22} = g_{33} \:.
\end{equation}
To analyze LRS models with anisotropic matter we require that the matter model
be compatible with the assumption of LRS symmetry;
this assumption replaces Assumption~\ref{assumptiondiagonal}.

\textbf{Assumption \ref{assumptiondiagonal}${}^\prime$.}
\textit{We assume that there exist initial data for the matter such that 
$T^i_{\weg j}$ is LRS, i.e., $T^i_{\weg j}$ is diagonal and $T^2_{\weg 2} = T^3_{\weg 3}$, in the LRS frame.}

The LRS models with anisotropic matter satisfying Assumption~\ref{assumptiondiagonal}${}^\prime$ are
a subclass of the diagonal models. These models are generated by the prescription of 
LRS initial data, i.e., $g_{22}(t_0) = g_{33}(t_0)$,
$k^2_{\weg 2}(t_0) = k^3_{\weg 3}(t_0)$, and Assumption~\ref{assumptiondiagonal}${}^\prime$.

\begin{Remark}
Assumption~\ref{assumptiondiagonal}${}^\prime$ is satisfied if
there exists initial data for the matter such that the 
function
$\psi(s_1,s_2,s_3)$ is 
symmetric under the exchange of the second and the third variable.
\end{Remark}

The Ricci tensor of the spatial metric of a Bianchi class~A LRS model 
is given by~\eqref{lrsricciA}, i.e.,
\begin{equation}\label{ricc}
R^1_{\weg 1} = \textfrac{1}{2}\: \hat{n}_1^2 \,m_1^2\:, \qquad
R^2_{\weg 2} = \hat{n}_1 \hat{n}_2 \, m_1 m_2 - \textfrac{1}{2}\: \hat{n}_1^2 \,m_1^2\:.
\end{equation}

The Einstein equations~\eqref{einsteinsystemdiagonal} simplify.
For convenience we introduce variables that are commonly used
in cosmology. We define the Hubble scalar $H$ and the shear variable $\sigma_+$
to be
\begin{equation}\label{Hubble}
H := {-\textfrac{1}{3}}\, \tr k = {-\textfrac{1}{3}}\,\big(k^1_{\weg 1} + 2 k^2_{\weg 2}\big)\:,
\qquad
\sigma_+ = \textfrac{1}{3}\,\tr k - k^2_{\weg 2} = \textfrac{1}{3} \big( k^1_{\weg 1} - k^2_{\weg 2} \big)\:.
\end{equation}
Then the Hamiltonian constraint~\eqref{constraintdiagonal} becomes
\begin{subequations}\label{EinsteinmatterLRS}
\begin{equation}\label{constraintLRS}
H^2+\textfrac{1}{3}\,\hat{n}_i\hat{n}_j m_i m_j=\textfrac{1}{3}\rho+\textfrac{1}{12}\hat{n}_i^2m_i^2+\sigma_+^2\:,
\end{equation}
where we have inserted $R = R^1_{\weg 1} + 2 R^2_{\weg 2}$ and~\eqref{ricc}.
The evolution equations $\partial_t g_{ii} = {-2} g_{ii} k^i_{\weg i}$, cf.~\eqref{evol1diag}, 
turn into
\begin{equation}\label{metricevol}
\partial_t m_1 = m_1 \,\big({-H} - 4 \sigma_+ \big)\:,\qquad
\partial_t m_2 = m_2 \,\big({-H} + 2 \sigma_+ \big)\:.
\end{equation}
Recall that $g^{11} = m_2^2$, $g^{22} = m_1 m_2$, cf.~\eqref{m1m2def}. 
Finally, the evolution equations~\eqref{evolutiondiagonal} take the form
\begin{equation}\label{evolutionLRS}
\begin{split}
\partial_t k^1_{\ 1} & = \textfrac{1}{2}\,\hat{n}_i^2{m_i}^2
+ (\tr k) k^1_{\ 1} - p_1 +\textfrac{1}{2}\left(p_1+2p_2-\rho\right) \:,\\[0.5ex]
\partial_t k^2_{\ 2} & = \textfrac{1}{2}
\big[2\hat{n}_i\hat{n}_jm_im_j-\hat{n}_i^2m_i^2\big]+ (\tr k) k^2_{\ 2}-p_2+\textfrac{1}{2}\left(p_1+2p_2-\rho\right)\:, 
\end{split}
\end{equation}
or, equivalently, 
using the constraint and expressing these equations
in terms of $H$ and $\sigma_+$,
\begin{equation}\label{evolutionLRS2}
\begin{split}
\partial_t H & = -H^2 - 2 \sigma_+^2 - \textfrac{1}{6}\, (1+ 3 w) \rho \:,\\[0.5ex]
\partial_t \sigma_+  & = -3 H \sigma_+ + \textfrac{1}{3} \,\big( \hat{n}_1^2 m_1^2 - \hat{n}_1 \hat{n}_2 m_1 m_2 \big) 
+ (w_2 - w) \rho \:.
\end{split}
\tag{\ref{evolutionLRS}${}^\prime$}
\end{equation}
\end{subequations}

The system of equations~\eqref{EinsteinmatterLRS} represents the Einstein-matter equations
for LRS models.
Due to local rotational symmetry, the functions describing the properties
of the anisotropic matter model under consideration, 
see Subsecs.~\ref{assT2} and~\ref{assT3}, simplify significantly.

Since $g_{22} = g_{33}$ implies $s_2 = s_3$, only one of the variables $(s_1,s_2,s_3)$ of~\eqref{svariables}
is independent. We set
\begin{equation}
s_2 = s_3 = s \:, \quad s_1 = 1-2s\:,
\end{equation}
which implies that $0<s<\textfrac{1}{2}$;
note that $s = g^{22}/(g^{11} + 2 g^{22}) = (2 + m_2/m_1)^{-1}$.
With slight abuse of notation, we define
\begin{equation}\label{rhoLRS}
\psi(s):=\psi(1-2s,s,s)\:.
\end{equation}
Employing~\eqref{wis} and~\eqref{rhoLRS} to compute the rescaled anisotropic
pressures we obtain
\begin{equation}\label{wisLRS}
w_1(1-2s,s,s) = w - 2 ( 1-2s) \frac{d \log \psi(s)}{d \log s}\,,\quad
w_2(1-2s,s,s) = w_3(1-2s,s,s) = w + (1-2 s)  \frac{d \log \psi(s)}{d \log s}\,,
\end{equation}
where we have used that $(\partial\psi/\partial s_2)(1-2s,s,s) = 
(\partial\psi/\partial s_3)(1-2s,s,s)$, which 
follows from
the symmetry of $\psi(s_1,s_2,s_3)$ in $(s_2,s_3)$.

Let us define the \textit{anisotropy function} $u(s)$ 
through
\begin{equation}\label{udef}
w_2(1-2s, s, s) = w_3(1-2s,s,s) = u(s)\:,\qquad w_1(1-2s,s,s) = 3 w - 2 u(s) \:.
\end{equation}
The anisotropy function $u(s)$ encodes the properties of the
rescaled anisotropic pressures; by~\eqref{wisLRS} we have
\begin{equation*}
u(s) =  w + (1-2 s)  \frac{d \log \psi(s)}{d \log s} \:.
\end{equation*}
It is this function that enters the Einstein-matter equations~\eqref{EinsteinmatterLRS};
its main properties will turn out to 
determine the dynamics of solutions of~\eqref{EinsteinmatterLRS}.

Since $0 < s < \textfrac{1}{2}$, a priori, $u(s)$ is defined on the interval
$\big(0,\textfrac{1}{2}\big)$. However, by Assumption~\ref{asswi} it admits
a regular extension to the closed interval $\big[  0,\textfrac{1}{2}\big]$.
By~\eqref{udef} we have 
\begin{alignat*}{4}
s = 0 &: \quad  &  & u(0) = v_-\,,\quad&  & w_1 = v_+ = 3 w - 2 v_-\,,\quad & & w_2 = w_3 = v_- \,, \\
s = \textfrac{1}{2} &: \quad & & u(\textfrac{1}{2}) = \textfrac{3 w - v_-}{2} = \textfrac{v_- + v_+}{2} \,,\quad
& & w_1 = v_-\,, \quad & & w_2 = w_3 = \textfrac{3 w - v_-}{2} = \textfrac{v_- + v_+}{2}\,.
\end{alignat*}

\begin{Remark}
Let us comment on the physical interpretation of $u(s)$ and $v_\pm$.
By~\eqref{udef}, $u(s)$ is the rescaled principal pressure in the directions tangential to 
the plane of local rotational symmetry;
therefore, $v_-$ is the rescaled principal pressure tangential to 
the plane of local rotational symmetry---and $v_+$ is the rescaled principal pressure
in the orthogonal direction---under the extreme conditions represented by $s \rightarrow 0$
(i.e., $g_{22} = g_{33} \rightarrow \infty$ while $g_{11}$ remains bounded).
Analogously, under the extreme conditions represented by $s \rightarrow \textfrac{1}{2}$
(which corresponds to $g_{11} \rightarrow \infty$ and $g_{22} = g_{33}$ remaining bounded)
the rescaled principal pressure tangential to 
the plane of local rotational symmetry is the mean of $v_\pm$
and the rescaled principal pressure
in the orthogonal direction is $v_-$.
\end{Remark}

The anisotropy parameter $\beta$, see~\eqref{betadef1}, takes several
equivalent forms
\begin{equation}\label{betadef}
\beta = 2 \,\frac{w-v_-}{1-w}=2\, \frac{w-u(0)}{1-w} =4\, \frac{u(\textfrac{1}{2})-w}{1-w}\:.
\end{equation}

By Assumption~\ref{asspsi}, 
$\psi(s)$ has a unique critical point $\bar{s} \in (0,\textfrac{1}{2})$; this is either a maximum or a minimum,
i.e., $\psi'(\bar{s})=0$ and $\psi''(\bar{s})\lessgtr 0$.
The point $\bar{s}$ corresponds to a unique isotropic state: $u(\bar{s}) = w$ and thus
$w_1 = w_2 = w_3 = w$ for $s = \bar{s}$.
At $s = \bar{s}$ we obtain
\begin{equation*}
\frac{d u}{d s} \Big|_{s = \bar{s}} = 
\bar{s} (1 - 2 \bar{s}) \psi^{-1}(\bar{s}) \:\frac{d^2\psi}{d s^2} \Big|_{s = \bar{s}} \gtrless 0 \:.
\end{equation*}
Therefore, we find that the anisotropic matter models we consider fall into
two main classes:
\begin{subequations}\label{inequalities}
\begin{alignat}{3}
& \boldsymbol{+} \text{ cases }  (\beta > 0) & \;:\quad &
u(0)< w < u(\textfrac{1}{2})\,, &\quad & u(\bar{s}) = w \:, \quad u^\prime(\bar{s}) > 0  \:,
\tag{\ref{inequalities}$+$}
\\
& \boldsymbol{-} \text{ cases } (\beta < 0) & \;:\quad &
u(0) > w > u(\textfrac{1}{2})\,, &\quad & u(\bar{s}) = w \:, \quad u^\prime(\bar{s}) < 0  \:.
\tag{\ref{inequalities}$-$}
\end{alignat}
\end{subequations}
Clearly, the $\boldsymbol{+}$ cases comprise the cases \Aplus, \Bplus, \Cplus, and \Dplus,
see Table~\ref{tab2} and Fig.~\ref{mattermodelsfig}; the $\boldsymbol{-}$ cases comprise the cases \Aminus, \Bminus, \Cminus, \Dminus.
For $\beta=0$, i.e., in the exceptional case \Azero, we have $u(0) = w = u(\textfrac{1}{2})$; we 
distinguish two subcases:
\[
\textnormal{\Azeroplus}\,:\quad  \beta = 0\,,\; u'(\bar{s})>0\:; 
\qquad 
\textnormal{\Azerominus}\,:\quad \beta=0\,,\; u'(\bar{s})<0\:.
\]

%%%%%%%%%%%%%%%%%%%%%%%%%%%%%%%%%%%%%%%%%%%%%%%%%%%%%%%%%%%%%%%%%%%%%%%%%%%%%%%%%%%%%%%%%%%%%%%%%%%%%%%%%%
\section{Vacuum and perfect fluid dynamics}
\label{perfectfluid}
%%%%%%%%%%%%%%%%%%%%%%%%%%%%%%%%%%%%%%%%%%%%%%%%%%%%%%%%%%%%%%%%%%%%%%%%%%%%%%%%%%%%%%%%%%%%%%%%%%%%%%%%%%

The system of equations~\eqref{EinsteinmatterLRS} constitutes the
Einstein-matter equations for LRS models with the class
of anisotropic matter sources under consideration.
Setting the energy density (and the principal pressures) 
to zero, i.e., $\rho = 0$ and $p_i = 0$ $\forall i$,
the equations reduce to the (SH LRS) Einstein vacuum equations.
Setting $p_1 = p_2 = p_3 = p = w \rho$ (or $w_1 = w_2 = w_3 = w$)
yields the (SH LRS) Einstein-perfect fluid equations. In this section 
we review the results on the asymptotic behavior of solutions of 
the system~\eqref{EinsteinmatterLRS} for
vacuum and for the perfect fluid matter models.  For a proof of these results we refer to~\cite{WE}.

There are a few known explicit solutions 
of the Einstein vacuum and perfect fluid systems that will be relevant for our analysis.
Let $a$, $b$ denote positive constants.
The following solutions are vacuum solutions, i.e., $\rho = 0$ and $p_i = 0$ $\forall i$,
and of Bianchi type~I.
\begin{itemize}
\item The flat LRS Kasner solution (\textit{Taub solution}) $\mathrm{T}$ is represented by the spatial metric
\begin{equation}\label{taub}
g_{11}=a\, t^{2}\:,\qquad g_{22}= g_{33}=b\:.
\end{equation}
\item The \textit{non-flat LRS Kasner solution} $\mathrm{Q}$ is given by
\begin{equation}\label{solQ}
g_{11}=a\, t^{-2/3}\:,\qquad g_{22} = g_{33} = b \,t^{4/3}\:.
\end{equation}
\end{itemize}
In contrast, the following solutions are perfect fluid solutions 
(where the fluid obeys a linear equation of state), i.e., 
$p_1 = p_2 = p_3 = p = w \rho$.
\begin{itemize}
\item The isotropic perfect fluid solution (Friedmann-Lema\^itre-Robertson-Walker solution) $\mathrm{F}$
is of Bianchi type~I (i.e., $\hat{n}_i = 0$ $\forall i$) and reads
\begin{equation}\label{FRWsol}
g_{11} = g_{22} =g_{33} = a\, t^{\frac{4}{3(1+w)}}\:,
\end{equation}
where $w \in (-1,1]$ is permitted. The Hamiltonian constraint implies
$\rho = \textfrac{4}{3(1+w)^2}\, t^{-2}$, which is 
in accordance with~\eqref{rhofluid}.
\item The \textit{Collins-Stewart} perfect fluid solution $\mathrm{C}$ 
is of Bianchi type~II, i.e., $\hat{n}_1=1$, $\hat{n}_2=\hat{n}_3=0$; we have
\begin{equation}\label{CSsol}
g_{11}=a \,t^{\textfrac{1-w}{1+w}}\:, \qquad g_{22} = g_{33} =b\, t^{\textfrac{3+w}{2(1+w)}}
\qquad\text{with}\quad \sqrt{a} =\textfrac{\sqrt{(1-w)(1+3w)}}{2(1+w)}\: b\:,
\end{equation}
where  $w \in (-\textfrac{1}{3},1]$ is permitted. The energy density is
$\rho=\textfrac{5-w}{4(w+1)^2}t^{-2}$.
\end{itemize}

\begin{Remark}
All solutions listed above are self-similar, see~\cite{WE}. Moreover the are all singular at time $t=0$.
\end{Remark}

It is well known that for all Bianchi types except~IX, vacuum solutions and 
perfect fluid solutions of~\eqref{EinsteinmatterLRS} with initial data at $t=t_0>0$ 
become singular in finite time in the past (w.l.o.g. we assume the singularity to be at $t=0$) 
and are defined for all times $t>t_0$ in the future.  
In the Bianchi type~IX case, solutions of~\eqref{EinsteinmatterLRS} are singular both 
in the past and in the future. After these preliminary remarks, let us review the asymptotics 
of solutions of the system~\eqref{EinsteinmatterLRS} in the vacuum and perfect fluid case. 

{\it LRS Bianchi type~I and~VII$_0$}. In the vacuum case a solution of~\eqref{EinsteinmatterLRS} is either a Taub 
solution $\mathrm{T}$, see~\eqref{taub}, 
or a non-flat LRS Kasner solution $\mathrm{Q}$, see~\eqref{solQ}. In the perfect fluid case every solution 
except the isotropic solution $\mathrm{F}$, see~\eqref{FRWsol}, is asymptotic to a vacuum solution toward the singularity (i.e., as $t\to 0$), 
while toward the future (i.e., as $t\to\infty$) every solution is asymptotic to the FLRW solution $\mathrm{F}$
and thus isotropizes.

{\it LRS Bianchi type~II}. In the vacuum case, every solution is asymptotic to $\mathrm{T}$ toward 
the singularity and to $\mathrm{Q}$ toward the future. In the perfect fluid case, every solution 
is asymptotic to the Collins-Stewart solution $\mathrm{C}$, see~\eqref{CSsol}, as $t\rightarrow \infty$, 
while toward the singularity generic solutions are asymptotic to $\mathrm{T}$; 
there is also a family of non-generic solutions that correspond to a set of zero measure of initial data, 
which behave like $\mathrm{F}$ toward the singularity and thus exhibit an isotropic singularity.

{\it LRS Bianchi type~VIII}. In the vacuum case, all solutions are asymptotic to $\mathrm{T}$ toward the singularity;
while toward the future they are asymptotic to a `plane wave' solution. (The plane wave
solution is another (SH LRS) vacuum 
solution of the Einstein equations, see
eq.~(9.7) in~\cite{WE}; it is of type~III.)
In the perfect fluid case, the asymptotics of generic
solutions are the same.

{\it LRS Bianchi type~IX}. Toward both the initial and the final singularity, vacuum and generic perfect fluid solutions 
are asymptotic to $\mathrm{T}$.
There exists also non-generic perfect fluid solutions which behave like $\mathrm{C}$ or $\mathrm{F}$ toward both singularities. 

\begin{Remark}
An important feature of the above results is that the asymptotics toward the singularity of generic Bianchi LRS perfect 
fluid solutions is the same as for vacuum solutions. This fact is referred to by saying that perfect fluid matter 
``does not matter" when the singularity of a Bianchi LRS cosmological model is approached. 
\end{Remark}

The main purpose of this paper is to extend the results on vacuum/perfect fluid solutions 
to the class of (anisotropic) matter models specified by the assumptions in Sec.~\ref{Sec:matter}. 
We shall see that matter models of type  \Azero~(which includes perfect fluids) have the same (generic) 
asymptotic behavior as perfect fluids. 
(Matter of type \Azero\ ``does not matter" when a singularity is approached.) However the other types of 
matter models may exhibit asymptotic behavior that differs considerably from that of perfect fluid 
(and vacuum) models (even for arbitrarily small (negative) values of the anisotropy parameter $\beta$.)

%%%%%%%%%%%%%%%%%%%%%%%%%%%%%%%%%%%%%%%%%%%%%%%%%%%%%%%%%%%%%%%%%%%%%%%%%%%%%%
\section{The reduced dynamical system}
\label{reducedsystem}
%%%%%%%%%%%%%%%%%%%%%%%%%%%%%%%%%%%%%%%%%%%%%%%%%%%%%%%%%%%%%%%%%%%%%%%%%%%%%%

In this section we formulate the evolution equations and
the constraint in terms of \textit{normalized variables}.
The benefit: In these variables 
some equations decouple 
and thus
the essential dynamics are described by a system
of equations that is of lower dimension (the {\it reduced dynamical system}).
Moreover, the self-similar solutions of the system~\eqref{EinsteinmatterLRS} (which govern the asymptotic 
behavior of general solutions) are represented by fixed points of the reduced dynamical system.
The local and global stability properties of the fixed points of the reduced 
dynamical system are therefore intimately connected with the asymptotic behavior of the cosmological models under study.

We define the `dominant variable' $D$ by
\begin{equation}\label{Ddef}
D = \sqrt{H^2 + \frac{\hat{n}_1 \hat{n}_2 \,m_1 m_2}{3}}\:.
\end{equation}
The constraint guarantees that the square root is real, 
cf.~\eqref{constraintLRS}.
We employ $D$ to introduce normalized variables according to
\begin{subequations}\label{domvars}
\begin{equation}\label{normvars}
H_D = \frac{H}{D}\:,\qquad
\Sigma_+ = \frac{\sigma_+}{D} \:,\qquad
M_{1} = \frac{m_1}{D} > 0 \:,\qquad
M_{2} = \frac{m_2}{D} > 0\:.
\end{equation}
In addition to~\eqref{normvars} we define a normalized energy density $\Omega$ by
\begin{equation}\label{Omegarho}
\Omega = \frac{\rho}{3 D^2} \geq 0\:,
\end{equation}
and we replace the cosmological time $t$ by a rescaled time variable $\tau$ via
\begin{equation}\label{newtime}
\frac{d}{d\tau} = \frac{1}{D} \,\frac{d}{d t}\:.
\end{equation}
\end{subequations}
Henceforth, a prime denotes differentiation w.r.t.\ $\tau$.
Expressed in these variables the Einstein evolution equations~\eqref{metricevol} and~\eqref{evolutionLRS2} 
split into a decoupled
equation for $D$,
\begin{equation}\label{Ddecoupled}
D^\prime = -D \Big( H_D (1+ q) + \Sigma_+ (1 - H_D^2)\Big)\:,
\end{equation}
and a system of coupled equations for the normalized variables~\eqref{normvars},
\begin{subequations}\label{domsys}
\begin{align}
\label{HDEq}
H_D^\prime & = -(1-H_D^2) (q - H_D \Sigma_+) \:,\\[0.5ex]
\label{Sigma+Eq}
\Sigma_+^\prime & = -(2- q) H_D\Sigma_+ + (1-H_D^2) \Sigma_+^2 - \textfrac{1}{3} 
\big({-}\hat{n}_1^2\,M_{1}^2 + \hat{n}_1 \hat{n}_2 \, M_{1} M_{2}\big) + 3\Omega \big(u(s) -w\big) \:,\\[0.5ex]
\label{M1Equ}
M_{1}^\prime & = M_{1} \big( q H_D - 4 \Sigma_+ + (1-H_D^2) \Sigma_+ \big)\:, \\[0.5ex]
\label{M2Equ}
M_{2}^\prime & = M_{2} \big( q H_D + 2 \Sigma_+ + (1-H_D^2) \Sigma_+ \big)\:, % \\[0.5ex]
\intertext{The evolution equation for $s$ is}
\label{sEq}
s^\prime & = -6 s ( 1 - 2s ) \Sigma_+\:;
\end{align}
\end{subequations}
we are free to add this equation to the system~\eqref{HDEq}--\eqref{M2Equ}.
The quantity $q$ (`deceleration parameter') is given by
\begin{equation}\label{qD}
q = 2 \Sigma_+^2 + \textfrac{1}{2} ( 1 + 3 w) \Omega\:,
\end{equation}
and the normalized energy density, $\Omega$, can be expressed in terms
of the variables~\eqref{normvars} by the Hamiltonian constraint, cf.~\eqref{constraintLRS}, which reads
\begin{equation}\label{gausscon}
\Sigma_+^2 + \textfrac{1}{12} \hat{n}_1^2 \,M_{1}^2 - 
\textfrac{1}{3} \hat{n}_1\hat{n}_2\, M_{1} M_{2} + \Omega + (1 -H_D^2) = 1 \:.
\end{equation}
It is important to emphasize that the variables in~\eqref{domsys} 
are not independent; there exist two constraints: 
The variable $s$, which appears as the argument of the function $u(s)$,
is determined by
\begin{subequations}\label{bothcons}
\begin{equation}\label{2ndcon}
s = \Big(2+ \frac{M_{2}}{M_{1}} \Big)^{-1} \:;
\end{equation}
recall that $0< s < \textfrac{1}{2}$. The second constraint is
\begin{equation}\label{1stcon}
1 - H_D^2 = (\hat{n}_1 \hat{n}_2) \frac{M_{1} M_{2}}{3}\:,
\end{equation}
\end{subequations}
which follows directly from~\eqref{Ddef}. 
Using~\eqref{1stcon} the Hamiltonian constraint~\eqref{gausscon} simplifies to
\begin{equation}\label{gaussconsimple}
\Sigma_+^2 + \textfrac{1}{12} \hat{n}_1^2 \,M_{1}^2 + \Omega  = 1 \:.
\tag{\ref{gausscon}${}^\prime$}
\end{equation}

\begin{Remark}
\textnormal{A normalization that is more common than~\eqref{domvars} is
the Hubble-normalization. Hubble-normalized variables are defined
as in~\eqref{normvars}, where $D$ is replaced by $H$.
Accordingly, $H_D^2 \equiv 1$ and the system of equations is given by~\eqref{Ddecoupled} 
and~\eqref{Sigma+Eq}--\eqref{sEq} with~\eqref{gausscon}, where $H_D^2$ is set to $1$.}
\end{Remark}

\begin{Remark}
It is sometimes useful to consider the evolution equation for $\Omega$ as
an auxiliary equation,
\begin{equation}\label{OmegaD}
\Omega^\prime = \Omega \Big( 2 H_D q + 2 \Sigma_+ (1-  H_D^2) - (1+3 w) H_D - 6 \Sigma_+ \big(u(s)-w\big) \Big).
\end{equation}
\end{Remark}

The system of equations~\eqref{domsys} describes the dynamics of 
LRS Bianchi class~A models. The equations for 
Kantowski-Sachs models and Bianchi type~III LRS models 
derive from this system as well; for a brief discussion we refer to Sec.~\ref{KSIII}.

The system~\eqref{domsys} contains redundant information; 
we will see that the essential dynamics of every model can be 
described by a subset of equations of~\eqref{domsys}, 
which differs according to the Bianchi type and
which we call the \textit{reduced dynamical system}. 

In Bianchi types~I and~II, eq.~\eqref{1stcon} enforces $H_D^2 \equiv 1$ since $\hat{n}_1 \hat{n}_2 = 0$;
therefore, the variables are Hubble-normalized. 
We restrict our attention
to expanding type~I and type~II models, which are characterized by $H_D \equiv 1$. 
(The case $H_D \equiv {-1}$ corresponds to models that are forever contracting. 
The dynamics of these models is obtained by replacing $\tau$ with $-\tau$, 
i.e., by inverting the flow in the dynamics of the expanding models.)
Bianchi type~I is characterized by $\hat{n}_1 = 0$ and $\hat{n}_2 = 0$.
Accordingly, the dependence of the r.h.s.\ of~\eqref{Sigma+Eq} on $M_1$ and $M_2$ disappears
and the equations~\eqref{M1Equ} and~\eqref{M2Equ} decouple. 
Therefore, in the Bianchi type~I case, 
the dynamical system~\eqref{domsys} reduces to a two-dimensional system
consisting of Eqs.~\eqref{Sigma+Eq} and~\eqref{sEq}.
Bianchi type~II is characterized by $\hat{n}_1^2 = 1$ and $\hat{n}_2 = 0$. Accordingly,
the equation for $M_2$ decouples and we obtain a three-dimensional
system consisting of Eqs.~\eqref{Sigma+Eq},~\eqref{M1Equ}, and~\eqref{sEq}.
Bianchi type~$\mathrm{VI}_0$ is incompatible with the assumption
of local rotational symmetry (see Appendix~\ref{lrsexplained}); LRS Bianchi type~$\mathrm{VII}_0$ reduces
to LRS type~I.
Bianchi type~VIII is characterized by $\hat{n}_1^2 =1$ and $\hat{n}_1 \hat{n}_2 = {-1}$.
From~\eqref{1stcon} and the positivity of $M_1$ and $M_2$ it follows that
$H_D > 1$. (We restrict our attention to expanding cosmological models
and thus discard the case $H_D < -1$ which corresponds to forever contracting models.)
Bianchi type~IX, on the other hand, is characterized by $\hat{n}_1^2 =1$ and $\hat{n}_1 \hat{n}_2 ={+1}$,
and~\eqref{1stcon} implies that $H_D \in (-1,1)$. Therefore, 
models of Bianchi type~IX 
that are expanding initially need not expand forever; 
an eventual recollapse is to be expected; we refer to Sec.~\ref{B9sec}.

Through the present approach, 
the analysis of the (asymptotic) dynamics of solutions of the Einstein-matter equations
has been reduced to the problem of finding the $\alpha$- and $\omega$-limit sets 
of orbits of the reduced dynamical system (for each Bianchi type).
Naturally, we are particularly interested in the asymptotic dynamics 
of typical orbits.

We say that an orbit of a dynamical system is \textit{typical} if its $\alpha$-limit (resp.\ $\omega$-limit) set
lies in the past (resp.\ future) attractor.
(Solutions of the Einstein equations that are associated with typical orbits will also be called typical.)
Atypical solutions are non-generic; the $\alpha$/$\omega$-limit set lies outside of the past/future attractor. 

The subsequent sections are devoted to an in-depth analysis of the dynamics
of LRS Bianchi class~A models.
We begin with a study of the dynamics of models of the lower Bianchi types.
The results on the lower Bianchi types 
are results in their own right; in addition, this
analysis is essential for the study of the dynamics of Bianchi types~VIII and~IX.

%%%%%%%%%%%%%%%%%%%%%%%%%%%%%%%%%%%%%%%%%%%%%%%%%%%%%%%%%%%%%%%%%%%%%%%%%%%%%%%%%%%%%%%%%%%%%%%%%%%%%%%%%%%%
\section{Bianchi type I}
\label{bianchi1section}
%%%%%%%%%%%%%%%%%%%%%%%%%%%%%%%%%%%%%%%%%%%%%%%%%%%%%%%%%%%%%%%%%%%%%%%%%%%%%%%%%%%%%%%%%%%%%%%%%%%%%%%%%%%%

In this section we study the dynamics of LRS models of Bianchi type~I. 
(Recall that LRS Bianchi type~VII$_0$ is the same as LRS Bianchi type~I, see Appendix~\ref{lrsexplained}.)
According to Table~\ref{tab1}, Bianchi type~I corresponds to setting 
$\hat{n}_1=0$ and $\hat{n}_2 = \hat{n}_3 = 0$ in~\eqref{Ddef} and in the 
reduced dynamical system~\eqref{domsys}.
We see that $H_D\equiv 1$, which reduces the 
normalization~\eqref{domvars} to the standard Hubble-normalization,
and the equations for $M_1$ and $M_2$ 
decouple from~\eqref{domsys}. Therefore, 
the essential dynamics are described by the equations for the pair 
$(\Sigma_+,s)$, which are given by 
\begin{subequations}\label{dynsysbianchiI}
\begin{align}
&\Sigma_+'=-3(1-\Sigma_+^2)\left[\textfrac{1}{2}(1-w)\Sigma_+-\left(u(s)-w\right)\right]\:,\\[0.5ex]
&s'=-6s(1-2s)\Sigma_+\:.
\end{align}
\end{subequations}
The Hamiltonian constraint~\eqref{gaussconsimple} becomes
\begin{equation}\label{Icon}
\Sigma_+^2 + \Omega =1 \:.
\end{equation}
Therefore, the reduced dynamical system~\eqref{dynsysbianchiI} is defined on the state space 
\begin{equation}\label{statespaceI}
\mathcal{X}_{\,\mathrm{I}}=\big\{(\Sigma_+,s)\in(-1,1)\times (0,\textfrac{1}{2})\big\}\:,
\end{equation}
which is a rectangle. 

The system~\eqref{dynsysbianchiI} admits a regular extension to the boundary of the state space. 
We distinguish the vacuum boundary, which is given by $\Omega = 0$ (i.e., $\rho = 0$),
or, more specifically, by $\mathcal{V}_{\,\mathrm{I}}= \mathcal{T} \cup \mathcal{Q}$, where
\[
\mathcal{T} = \big\{\Sigma_+ = {-1}\big\}\times (0,\textfrac{1}{2})\qquad\text{and}\qquad
\mathcal{Q} = \big\{\Sigma_+ =1\big\}\times (0,\textfrac{1}{2})\:,
\]
and the matter boundary, which is the boundary subset satisfying $\Omega \neq 0$;
it consists of
\begin{equation}\label{IflatIsharp}
\mathcal{I}_\flat = ({-1},1)\times \{s = 0\}  \qquad\text{and}\qquad
\mathcal{I}_\sharp = ({-1},1)\times \big\{s = \textfrac{1}{2}\big\}\:.
\end{equation}
The number of fixed points on the boundary of $\mathcal{X}_{\,\mathrm{I}}$ 
varies between four and six, depending on the anisotropy parameter $\beta$,
see~\eqref{betadef}.

\begin{list}{point}{\leftmargin1.15cm\labelwidth4cm\labelsep0cm}
\item[$\mathrm{T}_\flat$/$\mathrm{T}_\sharp$] \quad There exist two `Taub points': $\mathrm{T}_\flat: (\Sigma_+,s)=(-1,0)$ 
and $\mathrm{T}_\sharp: (\Sigma_+,s)=(-1,\textfrac{1}{2})$.
\item[$\mathrm{Q}_\flat$/$\mathrm{Q}_\sharp$] \quad There exist two `non-flat LRS points': $\mathrm{Q}_\flat: (\Sigma_+,s)=(1,0)$ 
and $\mathrm{Q}_\sharp: (\Sigma_+,s)=(1,\textfrac{1}{2})$.
\item[$\mathrm{R}_\flat$] \quad The fixed point $\mathrm{R}_\flat: (\Sigma_+,s)=(-\beta,0)$ 
exists when $\beta\in (-1,1)$.
\item[$\mathrm{R}_\sharp$] \quad The fixed point 
$\mathrm{R}_\sharp: (\Sigma_+,s)=(\textfrac{\beta}{2},\textfrac{1}{2})$ exists when $\beta\in (-2,2)$.
\end{list}
The $\mathrm{T}$-points are associated with the flat Taub solution~\eqref{taub}, 
whereas the $\mathrm{Q}$-points are associated with the non-flat LRS solution~\eqref{solQ}. 
The fixed points $\mathrm{R}_\sharp$, $\mathrm{R}_\flat$, when they exist, are associated 
with non-vacuum solutions (since $\Omega>0$ at these points), see Appendix~\ref{exact}.

\begin{Remark}
We adhere to the convention that fixed points and invariant subsets with $s = 0$ are
denoted by a subscript ${}_\flat$, while the subscript ${}_\sharp$ indicates that
$s = \textfrac{1}{2}$.
\end{Remark}

Consider, first, the vacuum boundary $\mathcal{V}_{\,\mathrm{I}}$.
The orbit $\mathcal{T}$ with $\Sigma_+ = {-1}$ connects the fixed point
$\mathrm{T}_\flat$ with the point $\mathrm{T}_\sharp$.
Like the fixed points themselves, the orbit $\mathcal{T}$ is a representation
of a Taub solution~\eqref{taub} in terms of the normalized variables.
Likewise, the orbit $\mathcal{Q}$ with $\Sigma_+ = {+1}$ connects the fixed point
$\mathrm{Q}_\sharp$ with the point $\mathrm{Q}_\flat$ and 
represents a non-flat LRS solution~\eqref{solQ}.

Second, consider the boundary components $\mathcal{I}_\flat$ and $\mathcal{I}_\sharp$.
Using~\eqref{betadef} 
we have
\begin{subequations}\label{eqcalI}
\begin{align}
\label{eqcalIflat}
& \mathcal{I}_\flat\,:\quad {\Sigma_+'}_{\,|s=0}=-\textfrac{3}{2}(1-\Sigma_+^2)(1-w)(\Sigma_++\beta)\:,
\tag{\ref{eqcalI}${}_\flat$}
\\
\label{eqcalIsharp}
& \mathcal{I}_\sharp\,:\quad {\Sigma_+'}_{\,|s=1/2}=-\textfrac{3}{2}(1-\Sigma_+^2)(1-w)(\Sigma_+-\textfrac{\beta}{2})\:;
\tag{\ref{eqcalI}${}_\sharp$}
\end{align}
\end{subequations}
in particular, $\mathrm{R}_\flat$ and $\mathrm{R}_\sharp$, when present, 
attract solutions on $\mathcal{I}_\flat$ and $\mathcal{I}_\sharp$, cf.~Fig.~\ref{BianchiIfig}.
To determine the role of the points $\mathrm{R}_\flat$ and $\mathrm{R}_\sharp$ for the flow on $\mathcal{X}_{\,\mathrm{I}}$, i.e.,
for~\eqref{dynsysbianchiI}, we use that 
\[
\left[\textfrac{d}{d\tau}\log s\right]_{\,|\mathrm{R}_\flat}=6\beta\:,\quad
\left[\textfrac{d}{d\tau}\log (1-2s)\right]_{\,|\mathrm{R}_\sharp}=3\beta\:.
\]
We infer that, in the $\boldsymbol{-}$ cases, i.e., if $\beta < 0$, these fixed points 
attract orbits from the interior of $\mathcal{X}_{\,\mathrm{I}}$ 
and thus represent sinks in $\mathcal{X}_{\,\mathrm{I}}$, 
while they are saddles in the $\boldsymbol{+}$ cases, i.e, if $\beta > 0$.
(In the cases \Azeropm\ the fixed points are non-hyperbolic 
sinks and saddles, respectively, 
which follows from a center manifold analysis.) 

In the case \Dplus, the boundary of $\mathcal{X}_{\,\mathrm{I}}$ forms a heteroclinic cycle 
connecting the fixed points at the vertexes of the rectangle, see Fig.~\ref{BianchiIfig}. 
We schematically represent this heteroclinic cycle as
\begin{equation}\label{heteroI}
\begin{CD}
\mathrm{T}_{\sharp} @>>> \mathrm{Q}_{\sharp} \\
@AAA @VVV \\
\mathrm{T}_\flat @<<< \mathrm{Q}_{\flat} 
\end{CD}
\end{equation}
We shall see that heteroclinic cycles of the form~\eqref{heteroI} appear on the boundary of the state space
in all Bianchi types.

By Assumption~\ref{asspsi} there is one and only one fixed point in the interior of $\mathcal{X}_{\,\mathrm{I}}$.  
\begin{itemize}
\item[$\mathrm{F}$] \quad  \begin{minipage}[t]{13cm} 
The fixed point F is given by $(\Sigma_+,s)=(0,\bar{s})$, where $\bar{s}\in (0,\textfrac{1}{2})$ 
is the unique solution of $u(\bar{s})=w$ ($\Leftrightarrow\psi'(\bar{s})=0$).
\end{minipage}
\end{itemize}
The fixed point $\mathrm{F}$ is associated with the isotropic perfect fluid solution~\eqref{FRWsol}.
The eigenvalues of the linearization of the dynamical system~\eqref{dynsysbianchiI} at $\mathrm{F}$ are given by  
\[
\lambda_\pm=-\textfrac{3}{4}\left(1-w\pm\sqrt{(1-w)^2-32\bar{s}(1-2\bar{s})u'(\bar{s})}\right)\:.
\]
Since $u'(\bar{s})\neq 0$ ($\Leftrightarrow\psi''(\bar{s})\neq 0$) by~\eqref{inequalities}, 
the fixed point $\mathrm{F}$ is hyperbolic. From~\eqref{inequalities} 
it follows that $\mathrm{F}$ is a sink for $\beta>0$ (i.e., in the $\boldsymbol{+}$ cases) 
and a saddle for $\beta<0$ (i.e., in the $\boldsymbol{-}$ cases).
Moreover, in the former case, if $u'(\bar{s})$ is large, $\lambda_\pm$ is complex and 
the approach of solutions to $\mathrm{F}$ is oscillatory.
(The cases \Azeroplus\ and \Azerominus\ behave like the $\boldsymbol{+}$ and
$\boldsymbol{-}$ cases, respectively.)

\begin{Theorem}\label{BianchiItheo}
The phase portraits of the dynamical system~\eqref{dynsysbianchiI} 
for the various anisotropy cases are correctly depicted
in Fig.~\ref{BianchiIfig}. The past and future attractors in the various cases are listed in Table~\ref{attractorsI}. 
\end{Theorem}

\begin{proof}
Consider the function
\begin{equation}\label{Z1def}
Z_1=\left(1-\Sigma_+^2\right)^{-1}\psi(s)\:.
\end{equation}
A straightforward calculation shows that
\[
Z_1'=-3(1-w)\Sigma_+^2 Z_1\:,\quad {Z_1'''}_{\,|\Sigma_+=0}=-54(1-w)(u(s)-w)^2\psi(s)\:.
\]
It follows that the function $Z_1$ is strictly monotonically 
decreasing along orbits in $\mathcal{X}_{\,\mathrm{I}}\backslash\{\mathrm{F}\}$. 
Applying the monotonicity principle, see Appendix~\ref{dynsysapp}, 
we infer that the $\alpha$- and $\omega$-limit sets of every orbit of the dynamical system~\eqref{dynsysbianchiI}
must be located on $\partial\mathcal{X}_{\,\mathrm{I}}\cup\{\mathrm{F}\}$. Using the local analysis of the fixed point $\mathrm{F}$ 
and the straightforward dynamics on $\partial\mathcal{X}_{\,\mathrm{I}}$ yields the claim.
\end{proof}

\textit{Interpretation of Theorem~\ref{BianchiItheo}}. 
In the $\boldsymbol{+}$ cases, the orbits of the dynamical system
converge to the fixed point $\mathrm{F}$. This means that each solution of the Einstein equations isotropizes toward the future, 
i.e., the metric is asymptotic to the isotropic perfect fluid solution~\eqref{FRWsol} 
as $t\to\infty$. This is also the future asymptotic behavior of Bianchi 
type~I cosmological models with perfect fluid matter, see Sec.~\ref{perfectfluid}. 
(Note, however, that isotropization rates may differ considerably~\cite{CH_GRG}.)
In the $\boldsymbol{-}$ cases, in general, cosmological models do not isotropize toward 
the future (which is irrespective of whether the energy conditions are satisfied or not). 
In particular, in the cases \Bminus, \Cminus, \Dminus, there exist typical solutions of the Einstein equations,
see Sec.~\ref{reducedsystem}, 
that are future asymptotic to vacuum solutions. Other typical solutions
approach the models corresponding to $\mathrm{R}_{\flat}$ and $\mathrm{R}_{\sharp}$; 
we refer to Appendix~\ref{exact} where we explicitly give the metric associated with 
these models.

Furthermore, the theorem shows that
the past asymptotic behavior of models is
intimately connected with the behavior of 
the vacuum solutions~\eqref{taub} and~\eqref{solQ},
which is analogous to the case of perfect fluid cosmologies.
However, there is an interesting exception: In 
the case of a high degree of energy condition 
violation, i.e., case \Dplus,
the approach to the singularity is oscillatory. The solutions
oscillate between
the Taub family~\eqref{taub} and the non-flat LRS family~\eqref{solQ} of solutions.
Since $\Omega\not\rightarrow 0$ for these solutions, this is an example of asymptotic
behavior where ``matter matters'', which means that the asymptotics are 
dissimilar from
the asymptotics of vacuum models.
For more details and for a generalization of the result of this section 
to the non-LRS case we refer to~\cite{CH2}.

\begin{figure}[Ht!]
\begin{center}
\psfrag{T}[cc][cc][0.7][0]{$\mathrm{T}_\flat$}
\psfrag{T*}[cc][cc][0.7][0]{$\mathrm{T}_\sharp$}
\psfrag{Q}[cc][cc][0.7][0]{$\mathrm{Q}_\flat$}
\psfrag{Q*}[cc][cc][0.7][0]{$\mathrm{Q}_\sharp$}
\psfrag{F}[cc][cc][0.7][0]{$\mathrm{F}$}
\subfigure[\Dminus]{\label{biDmin}\includegraphics[width=0.3\textwidth]{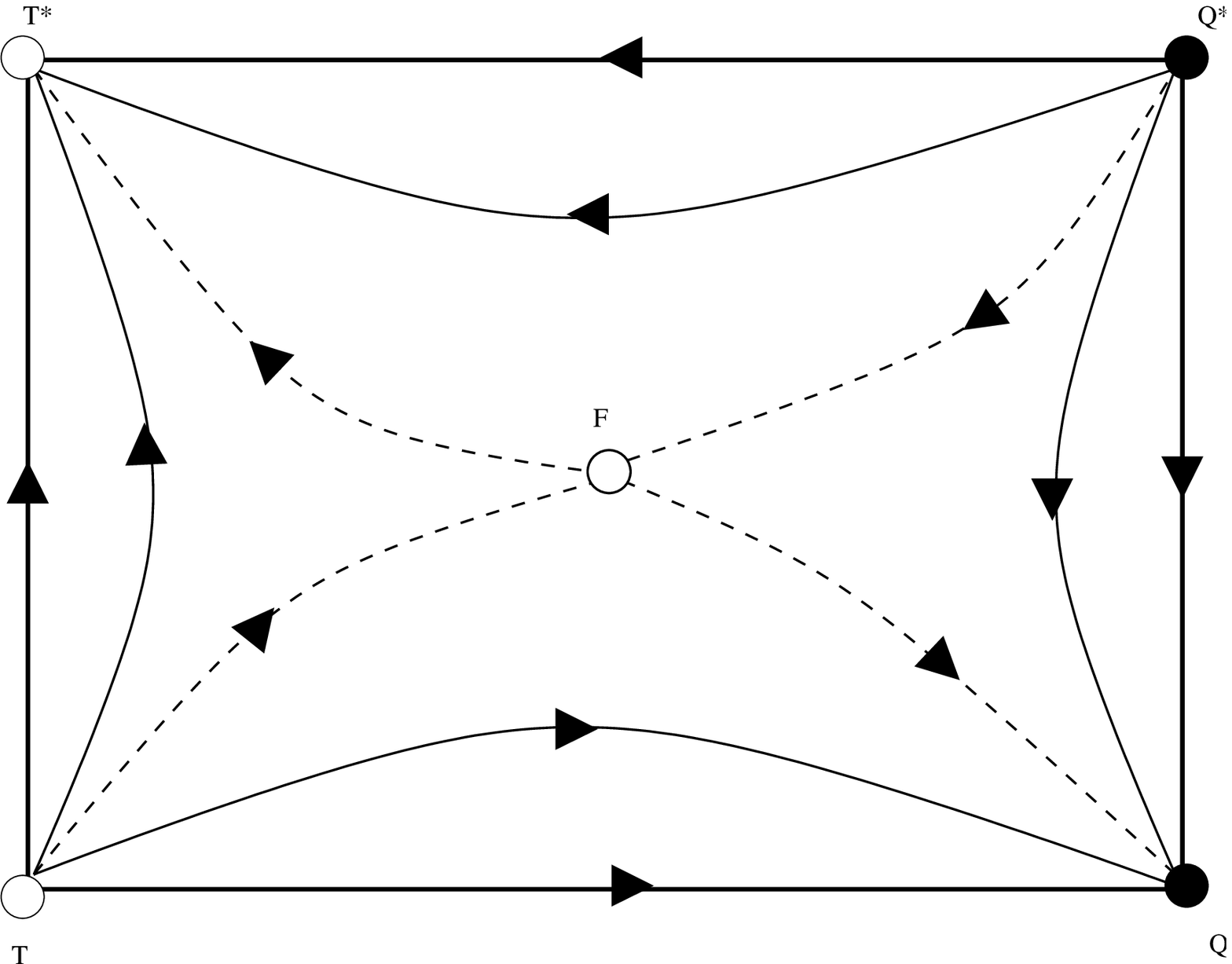}}\quad
\psfrag{T}[cc][cc][0.7][0]{$\mathrm{T}_\flat$}
\psfrag{T*}[cc][cc][0.7][0]{$\mathrm{T}_\sharp$}
\psfrag{Q}[cc][cc][0.7][0]{$\mathrm{Q}_\flat$}
\psfrag{Q*}[cc][cc][0.7][0]{$\mathrm{Q}_\sharp$}
\psfrag{D*}[cc][cr][0.7][0]{$\mathrm{R}_\sharp$}
\psfrag{F}[cc][cc][0.7][0]{$\mathrm{F}$}
\subfigure[\Bminus, \Cminus]{\includegraphics[width=0.3\textwidth]{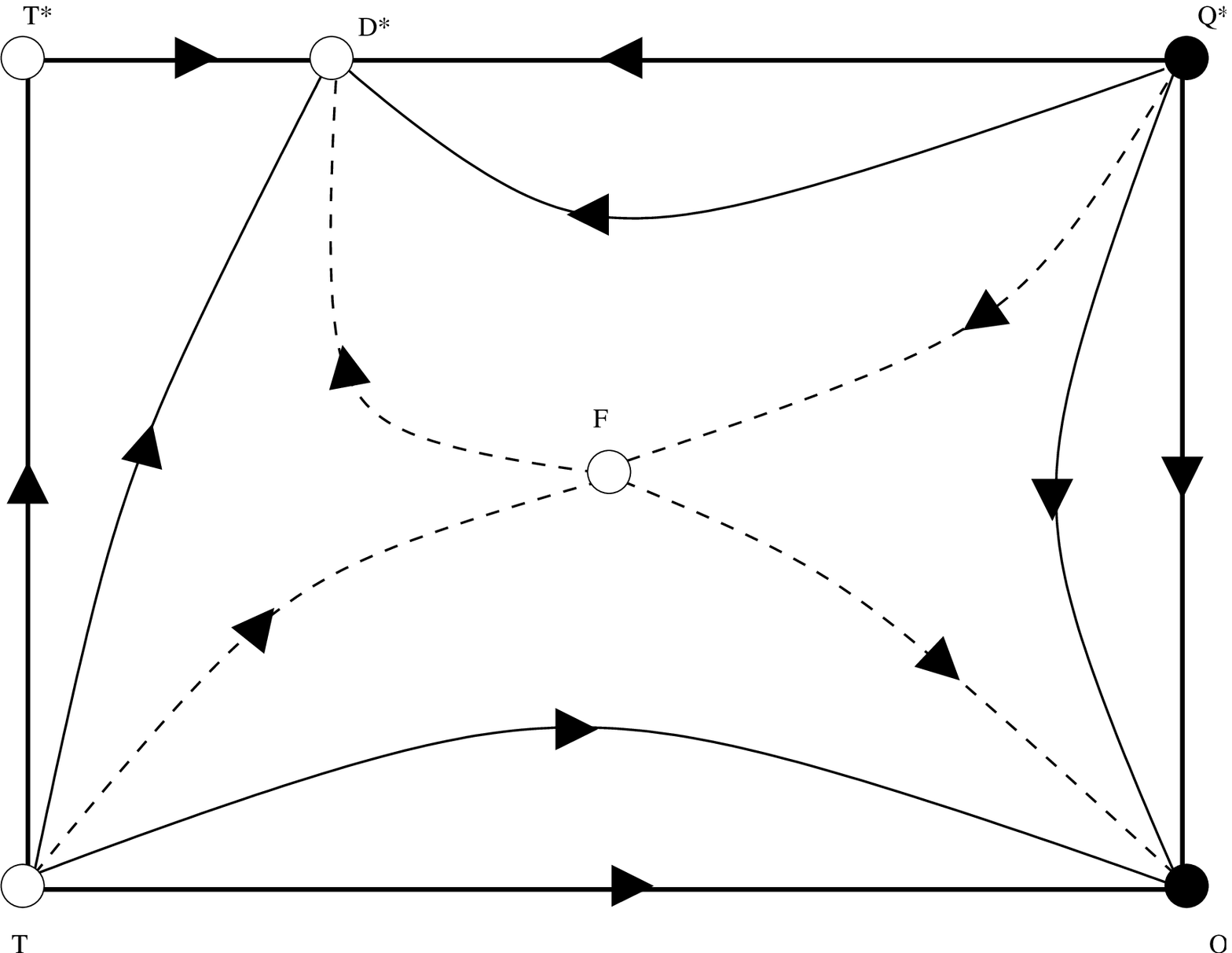}}\quad
\psfrag{T}[cc][cc][0.7][0]{$\mathrm{T}_\flat$}
\psfrag{T*}[cc][cc][0.7][0]{$\mathrm{T}_\sharp$}
\psfrag{Q}[cc][cc][0.7][0]{$\mathrm{Q}_\flat$}
\psfrag{Q*}[cc][cc][0.7][0]{$\mathrm{Q}_\sharp$}
\psfrag{R}[rl][cc][0.7][0]{$\mathrm{R}_\flat$}
\psfrag{D*}[cc][cr][0.7][0]{$\mathrm{R}_\sharp$}
\psfrag{F}[cc][cc][0.7][0]{$\mathrm{F}$}
\subfigure[\Aminus~$(-1<\beta\leq\beta_\flat)$]{\includegraphics[width=0.3\textwidth]{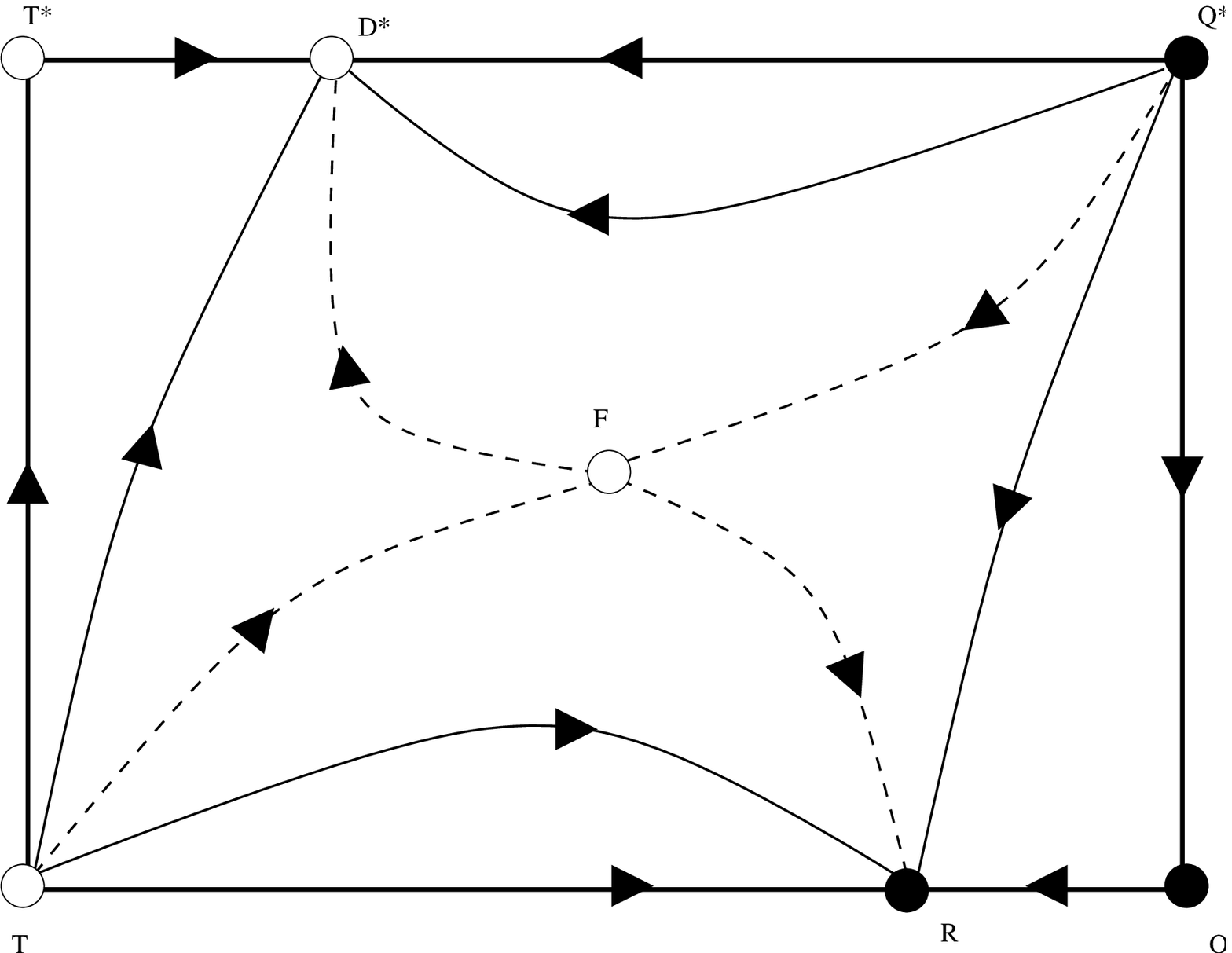}}\\
\psfrag{T}[cc][cc][0.7][0]{$\mathrm{T}_\flat$}
\psfrag{T*}[cc][cc][0.7][0]{$\mathrm{T}_\sharp$}
\psfrag{Q}[cc][cc][0.7][0]{$\mathrm{Q}_\flat$}
\psfrag{Q*}[cc][cc][0.7][0]{$\mathrm{Q}_\sharp$}
\psfrag{R}[rl][cc][0.7][0]{$\mathrm{R}_\flat$}
\psfrag{D*}[cc][cr][0.7][0]{$\mathrm{R}_\sharp$}
\psfrag{F}[cc][cc][0.7][0]{$\mathrm{F}$}
\subfigure[\Azerominus, \Aminus~$(\beta_\flat<\beta<0)$]{\includegraphics[width=0.3\textwidth]{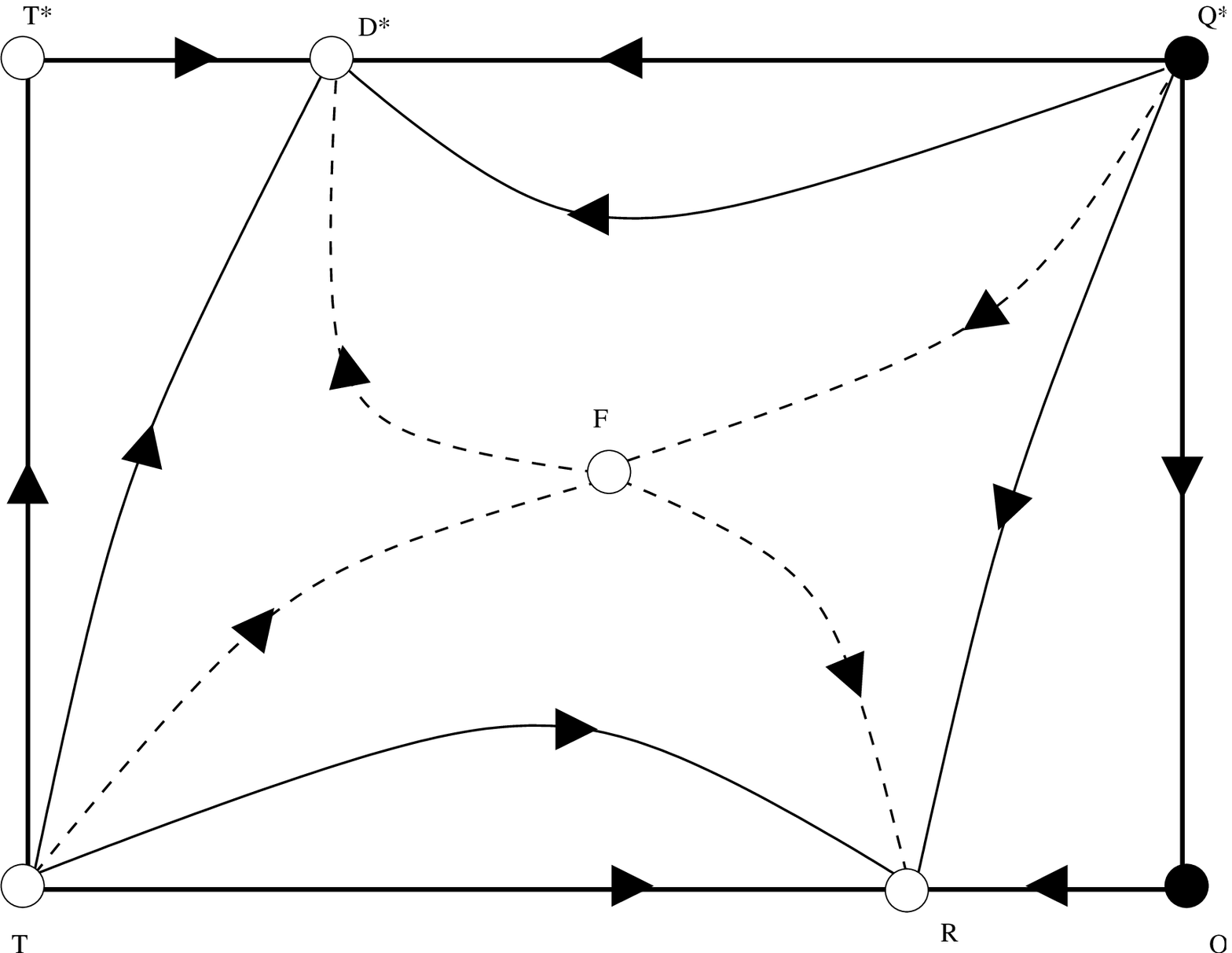}}\quad
\psfrag{T}[cc][cc][0.7][0]{$\mathrm{T}_\flat$}
\psfrag{T*}[cc][cc][0.7][0]{$\mathrm{T}_\sharp$}
\psfrag{Q}[cc][cc][0.7][0]{$\mathrm{Q}_\flat$}
\psfrag{Q*}[cc][cc][0.7][0]{$\mathrm{Q}_\sharp$}
\psfrag{R}[rl][cc][0.7][0]{$\mathrm{R}_\flat$}
\psfrag{D*}[cc][cc][0.7][0]{$\mathrm{R}_\sharp$}
\psfrag{F}[cc][cc][0.7][0]{$\mathrm{F}$}
\subfigure[\Azeroplus, \Aplus~$(0<\beta<\beta_\sharp)$]{\includegraphics[width=0.3\textwidth]{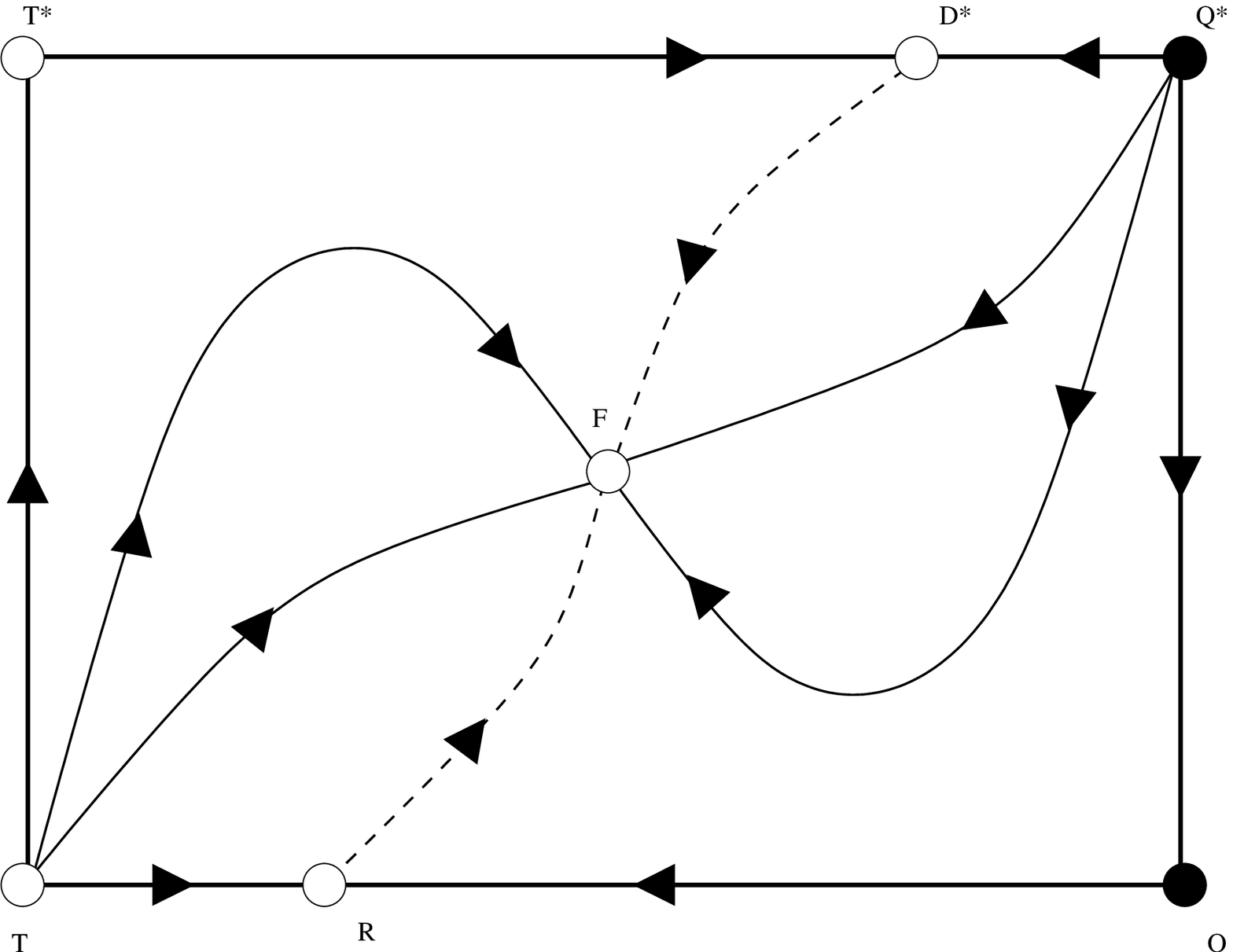}}\quad
\psfrag{T}[cc][cc][0.7][0]{$\mathrm{T}_\flat$}
\psfrag{T*}[cc][cc][0.7][0]{$\mathrm{T}_\sharp$}
\psfrag{Q}[cc][cc][0.7][0]{$\mathrm{Q}_\flat$}
\psfrag{Q*}[cc][cc][0.7][0]{$\mathrm{Q}_\sharp$}
\psfrag{R}[rl][cc][0.7][0]{$\mathrm{R}_\flat$}
\psfrag{D*}[cc][cr][0.7][0]{$\mathrm{R}_\sharp$}
\psfrag{F}[cc][cc][0.7][0]{$\mathrm{F}$}
\subfigure[\Aplus~$(\beta_\sharp\leq\beta<1)$]{\label{biAplu}\includegraphics[width=0.3\textwidth]{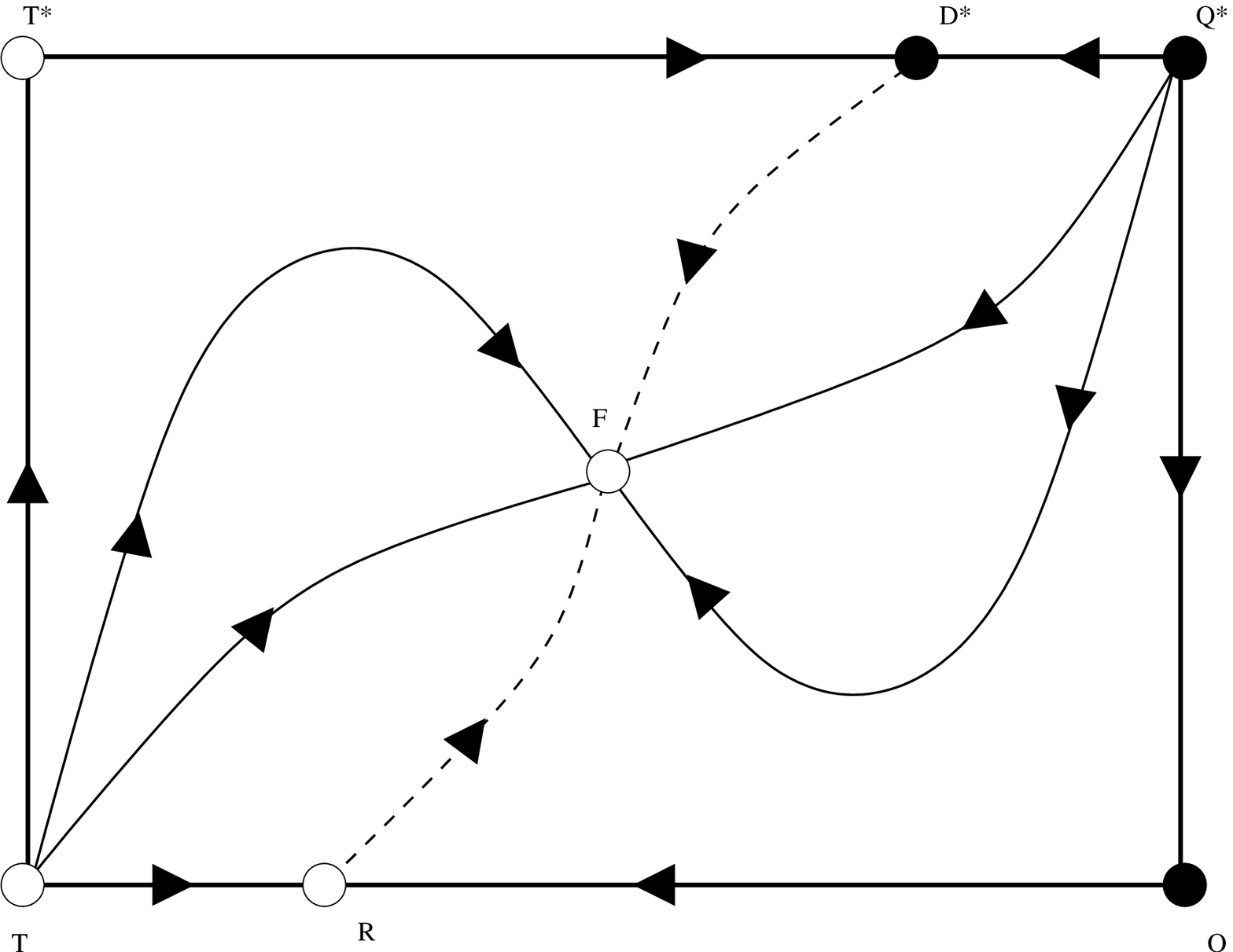}}\\
\psfrag{T}[cc][cc][0.7][0]{$\mathrm{T}_\flat$}
\psfrag{T*}[cc][cc][0.7][0]{$\mathrm{T}_\sharp$}
\psfrag{Q}[cc][cc][0.7][0]{$\mathrm{Q}_\flat$}
\psfrag{Q*}[cc][cc][0.7][0]{$\mathrm{Q}_\sharp$}
\psfrag{D*}[cc][cr][0.7][0]{$\mathrm{R}_\sharp$}
\psfrag{F}[cc][cc][0.7][0]{$\mathrm{F}$}
\subfigure[\Bplus, \Cplus]{\includegraphics[width=0.3\textwidth]{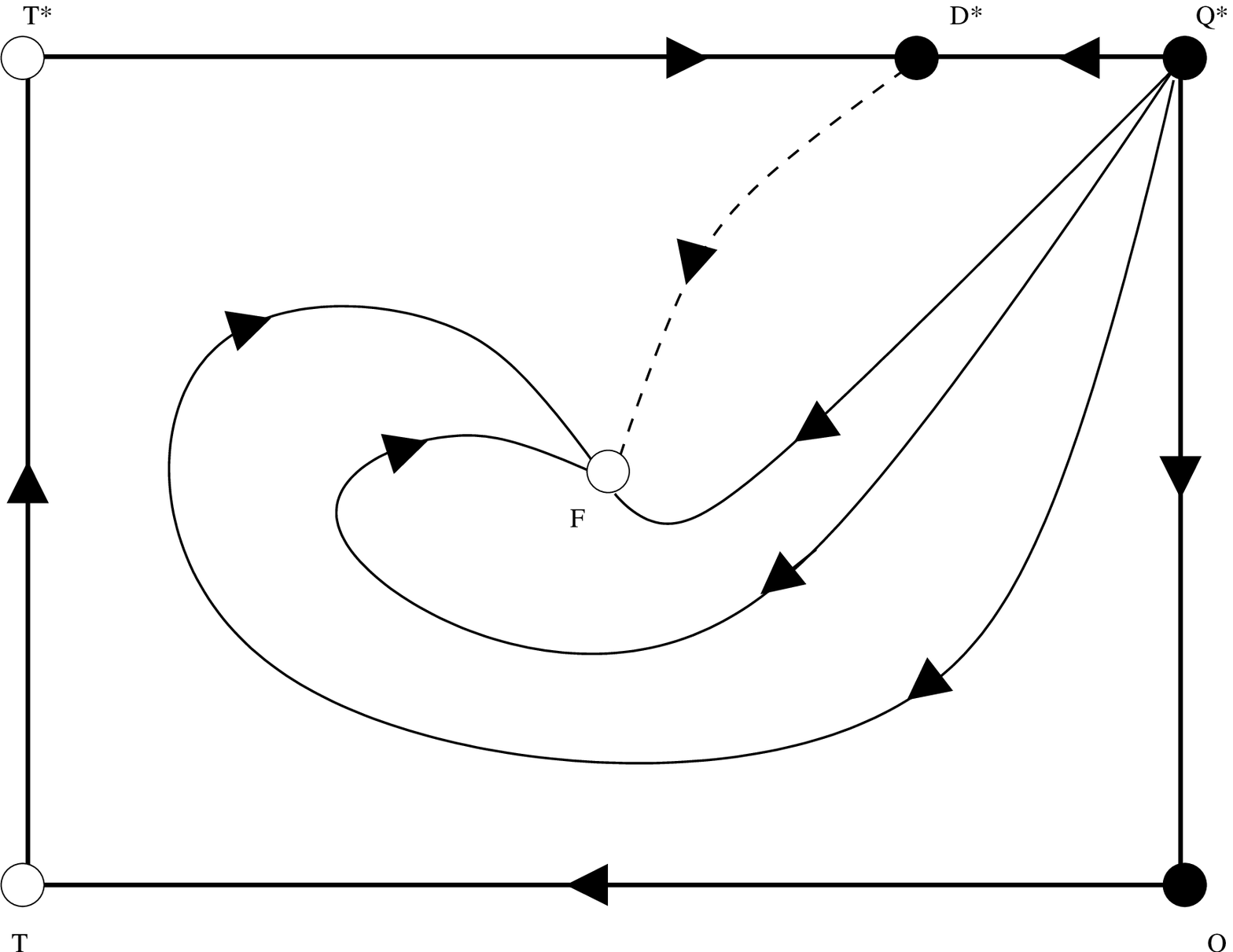}\label{VlasovI}}\quad
\psfrag{T}[cc][cc][0.7][0]{$\mathrm{T}_\flat$}
\psfrag{T*}[cc][cc][0.7][0]{$\mathrm{T}_\sharp$}
\psfrag{Q}[cc][cc][0.7][0]{$\mathrm{Q}_\flat$}
\psfrag{Q*}[cc][cc][0.7][0]{$\mathrm{Q}_\sharp$}
\psfrag{D*}[cc][cr][0.7][0]{$\mathrm{R}_\sharp$}
\psfrag{F}[cc][cc][0.7][0]{$\mathrm{F}$}
\subfigure[\Dplus]{\includegraphics[width=0.3\textwidth]{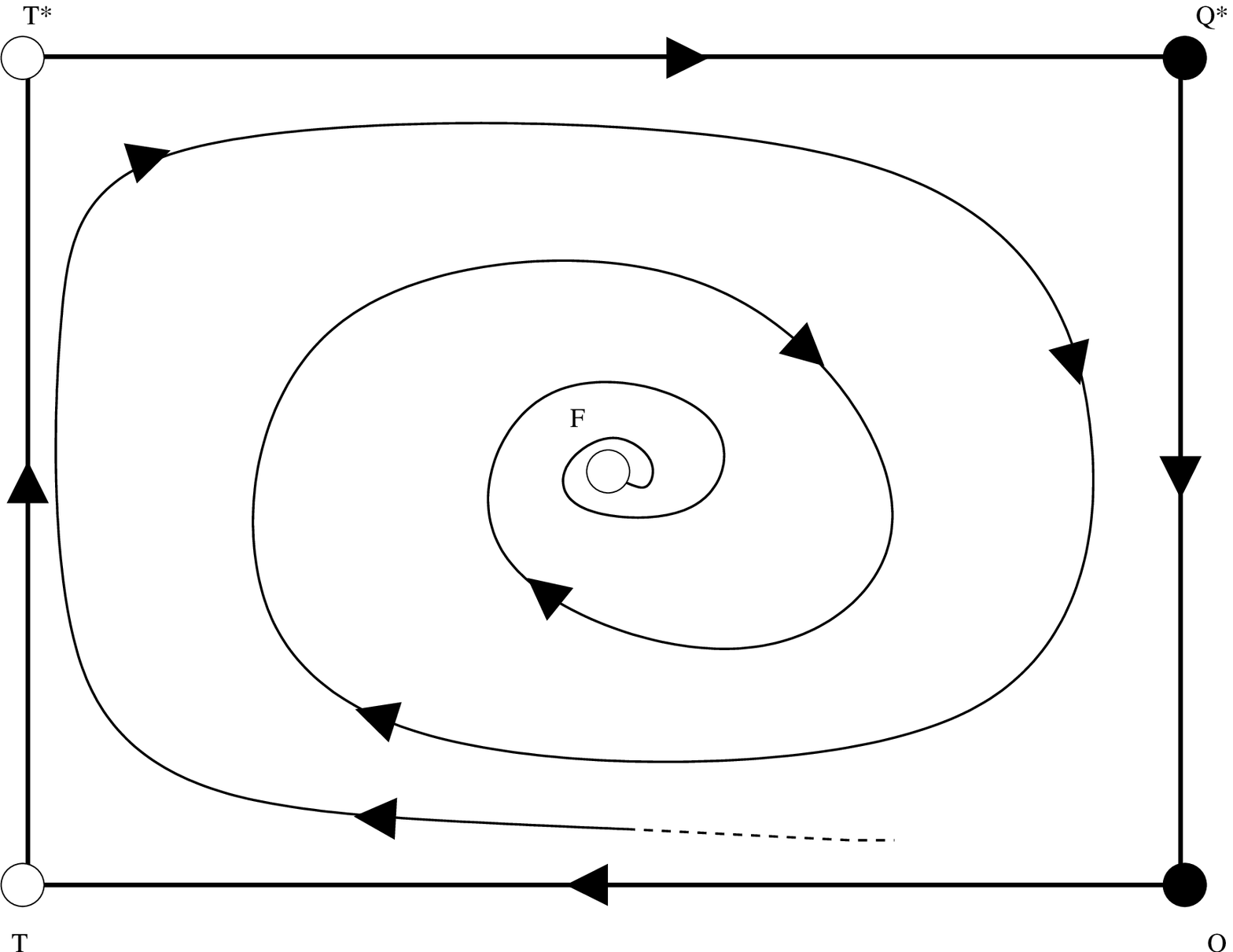}}\quad
\psfrag{s}[cc][cc][0.8][0]{$s$}
\psfrag{sigma}[cc][cc][0.8][0]{$\Sigma_+$}
\psfrag{t}[cc][cc][0.7][0]{$\vartheta$}
\psfrag{0}[cc][cc][0.7][0]{$0$}
\psfrag{1/2}[cc][cc][0.7][0]{$\textfrac{1}{2}$}
\psfrag{-1}[cc][cc][0.7][0]{$-1$}
\psfrag{+1}[cc][cc][0.7][0]{$+1$}
\psfrag{pi2}[cc][cc][0.7][0]{$\textfrac{\pi}{2}$}
\psfrag{X}[cc][cc][1.2][0]{$\mathcal{X}_{\,\mathrm{I}}$}
%\psfrag{bl}[cc][cc][1][0]{($\mathcal{B}, \mathcal{L}^*)$}
%\psfrag{V1}[cc][cc][1][0]{$\mathcal{V}_{\,\mathrm{I}}$}
\subfigure[Coordinates]{\label{Icoo}\includegraphics[width=0.3\textwidth]{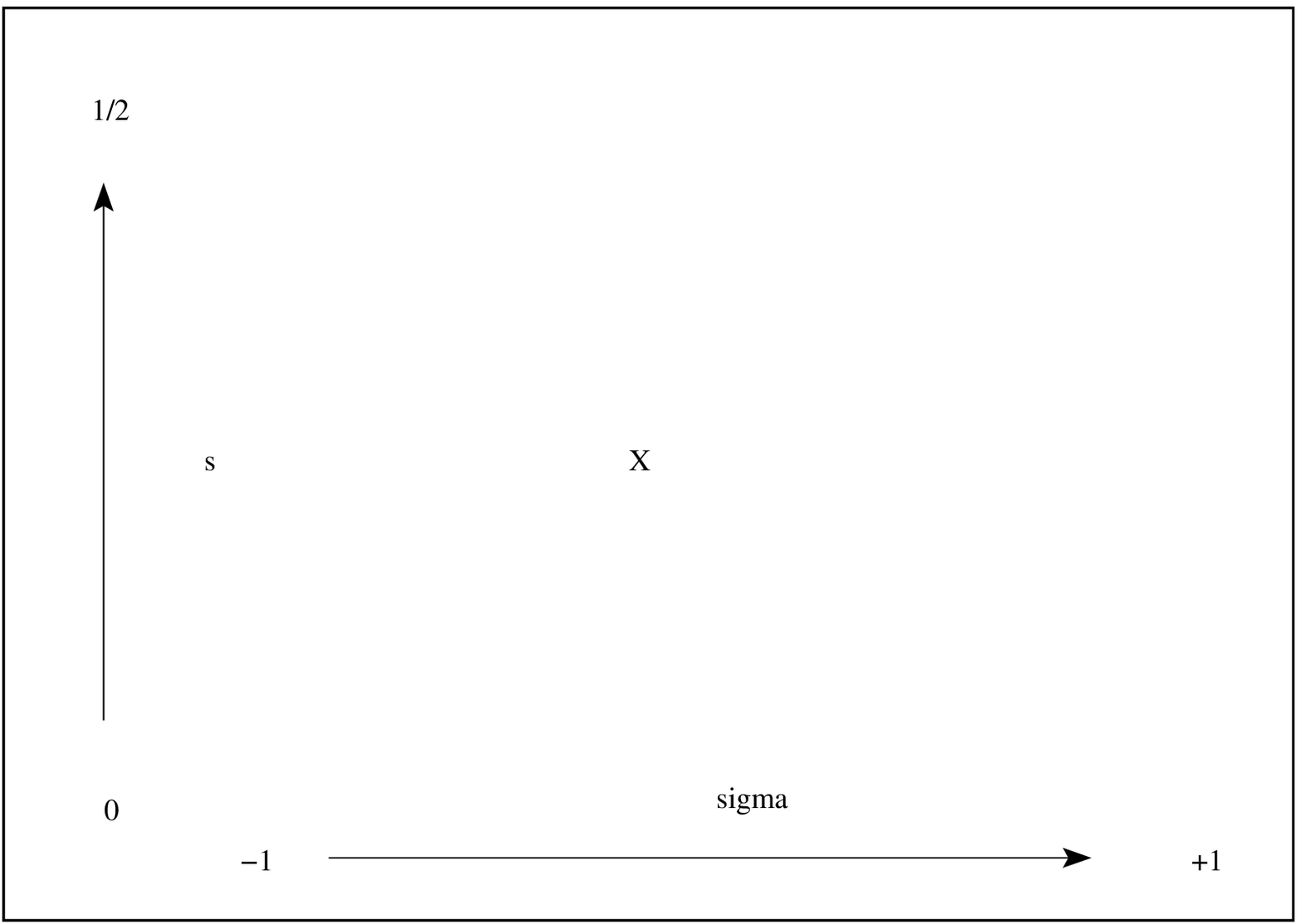}}
\end{center}
\caption{Phase portraits of orbits representing LRS Bianchi type~I models;
the state space is $\mathcal{X}_{\,\mathrm{I}}$ and the dynamical system 
is~\eqref{dynsysbianchiI}. 
Continuous lines represent typical orbits---each of these orbits is a member of a one-parameter set of orbits 
with the same asymptotic behavior. 
Dashed lines represent orbits whose future and/or past asymptotics are non-generic.
Certain features in these diagrams are included for future purposes:
First, the fixed points are color-coded. When the Bianchi type~I state space $\mathcal{X}_{\,\mathrm{I}}$ appears
as the boundary subset of the type~II, there exists an
orthogonal direction. A black fixed point is an attractor, a white fixed point is a repellor in the orthogonal direction. 
Second, the cases \Aminus\ and \Aplus\ are split into subcases, which become relevant for Bianchi type~II;
note that the difference is merely in the color-coding (which is irrelevant for type~I itself).}
\label{BianchiIfig}
\end{figure}

\begin{table}
\begin{center}
\begin{tabular}{|c|c|c|}
\hline  & &   \\[-2ex]
Type  & Past Attractor & Future Attractor \\ \hline  & &  \\[-2ex] 
\Dminus & $\mathrm{T}_\flat$, $\mathrm{Q}_\sharp$ & $\mathrm{T}_\sharp$, $\mathrm{Q}_\flat$\\
\Cminus, \Bminus & $\mathrm{T}_\flat$, $\mathrm{Q}_\sharp$ & $\mathrm{R}_\sharp$, $\mathrm{Q}_\flat$\\
\Aminus, \Azerominus &  $\mathrm{T}_\flat$, $\mathrm{Q}_\sharp$ & $\mathrm{R}_\sharp$, $\mathrm{R}_\flat$\\
\Aplus, \Azeroplus & $\mathrm{T}_\flat$, $\mathrm{Q}_\sharp$ & $\mathrm{F}$\\
\Bplus, \Cplus &  $\mathrm{Q}_\sharp$ & $\mathrm{F}$\\
\Dplus & Heteroclinic cycle~\eqref{heteroI} & $\mathrm{F}$\\
\hline 
\end{tabular}
\caption{The past and the future attractor for Bianchi type~I orbits encode
the asymptotic behavior of typical solutions.
The possible $\alpha$/$\omega$-limits of non-generic orbits can be read off from Fig.~\ref{BianchiIfig}.}
\label{attractorsI}
\end{center}
\end{table}

\begin{Remark}
The LRS subcase of Bianchi type~$\mathrm{VII}_0$ coincides with LRS Bianchi type~I.
Setting $\hat{n}_1 = 0$ in~\eqref{domsys} yields the type~I system~\eqref{dynsysbianchiI} irrespective
of the value of $\hat{n}_2 = \hat{n}_3$.
\end{Remark}

%%%%%%%%%%%%%%%%%%%%%%%%%%%%%%%%%%%%%%%%%%%%%%%%%%%%%%%%%%%%%%%%%%%%%%%%%%%%%%%%%%%%%%%%%%%%%%%%%%%%%%%%%%%%
\section{Bianchi type II}
\label{sec:II}
%%%%%%%%%%%%%%%%%%%%%%%%%%%%%%%%%%%%%%%%%%%%%%%%%%%%%%%%%%%%%%%%%%%%%%%%%%%%%%%%%%%%%%%%%%%%%%%%%%%%%%%%%%%%

According to Table~\ref{tab1}, Bianchi type~II corresponds to setting 
$\hat{n}_1=1$ and $\hat{n}_2 = \hat{n}_3 = 0$ in~\eqref{Ddef} and in the 
reduced dynamical system~\eqref{domsys}.
We infer that $H_D\equiv 1$, which reduces the 
normalization~\eqref{domvars} to the standard Hubble-normalization,
and the equation for $M_2$ 
decouples from~\eqref{domsys}. Therefore, 
the essential dynamics are described by the equations for the variables 
$(\Sigma_+, M_1, s)$, which are given by 
\begin{subequations}\label{dynsysbianchiII}
\begin{align}
\label{IIsig+}
&\Sigma_+'=\textfrac{1}{6}M_1^2\left(2-\Sigma_+\right)-3\Omega\left[\textfrac{1}{2}(1-w)\Sigma_+-\left(u(s)-w\right)\right]\:,\\[0.5ex]
\label{IIs} 
&s'=-6s(1-2s)\Sigma_+\:, \\[0.5ex]
\label{IIM1}
&M_1'=M_1\left[2\left(1-2\Sigma_+\right)-\textfrac{1}{6}M_1^2-\textfrac{3}{2}(1-w)\Omega\right]\:.
\end{align}
\end{subequations}
The Hamiltonian constraint~\eqref{gaussconsimple} becomes
\begin{equation}\label{IIcon}
\Sigma_+^2 + \textfrac{1}{12} \,M_1^2 + \Omega =1 \:.
\end{equation}
Therefore, the dynamical system~\eqref{dynsysbianchiI} is defined on the state space 
\begin{equation}\label{statespaceII}
\mathcal{X}_{\,\mathrm{II}}=\big\{(\Sigma_+,s,M_1)\,\big|\, {-1}< \Sigma_+ < 1,\: 
0<s<\textfrac{1}{2},\: 0 < M_1 < \sqrt{12(1-\Sigma_+^2)}\: \big\}\:,
\end{equation}
which resembles a tent; see Fig.~\ref{b2state}.

\begin{figure}[Ht]
\begin{center}
\psfrag{-1}[cc][cc][0.8][0]{$-1$}
\psfrag{1}[cc][cc][0.8][0]{$1$}
\psfrag{0}[cc][cc][0.8][0]{$0$}
\psfrag{12}[cc][cc][0.8][0]{$\textfrac{1}{2}$}
\psfrag{sig}[cc][cc][0.8][0]{$\Sigma_+$}
\psfrag{m1}[cc][cc][0.8][0]{$M_1$}
\psfrag{sb}[cc][cc][1][0]{$\mathcal{S}_\flat$}
\psfrag{ss}[cc][cc][1][0]{$\mathcal{S}_\sharp$}
\psfrag{v}[cc][cc][1][0]{$\mathcal{V}_\mathrm{II}$}
\psfrag{xi}[cc][rr][1][25]{$\mathcal{X}_{\,\mathrm{I}}$}
\psfrag{s}[cc][cc][0.8][0]{$s$}
\includegraphics[width=0.5\textwidth]{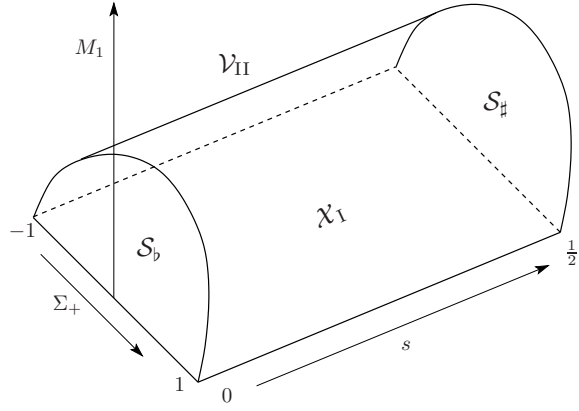}
\end{center}
\caption{The state space $\mathcal{X}_{\,\mathrm{II}}$. The base of the `tent' is the type~I
state space $\mathcal{X}_{\,\mathrm{I}}$; the `roof' $\mathcal{V}_\mathrm{II}$ 
is the vacuum boundary; the sides are denoted by $\mathcal{S}_\flat$ and $\mathcal{S}_\sharp$,
respectively. Recall that the subscript~$\flat$ corresponds to $s$ taking the value $s = 0$,
while $\sharp$ corresponds to $s = 1/2$.}
\label{b2state}
\end{figure}

The dynamical system~\eqref{dynsysbianchiII} admits a regular extension to the boundary of $\mathcal{X}_{\,\mathrm{II}}$,
which consists of four components,
\begin{equation}\label{IIboundarycomps}
\partial\mathcal{X}_{\,\mathrm{II}}=\overline{\mathcal{X}}_{\,\mathrm{I}}
\cup\overline{\mathcal{S}}_\flat\cup\overline{\mathcal{S}}_\sharp\cup\overline{\mathcal{V}}_\mathrm{II}\:,
\end{equation}
which are defined as  
\begin{align*}
&\mathcal{X}_{\,\mathrm{I}} = \big\{ M_1 = 0\,;\: {-1}< \Sigma_+ < 1, \:0<s<\textfrac{1}{2}\:\big\} \\[0.5ex]
&\mathcal{S}_\flat= \Big\{s=0\,;\: {-1}< \Sigma_+ < 1, \: 0< M_1 < \sqrt{12(1-\Sigma_+^2)}\:\Big\}\:,\\[0.5ex]
&\mathcal{S}_\sharp=\Big\{s=\textfrac{1}{2}\,;\: {-1}< \Sigma_+ < 1, \: 0< M_1 < \sqrt{12(1-\Sigma_+^2)}\:  \Big\}\:,\\[0.5ex]
&\mathcal{V}_\mathrm{II}=\Big\{\Omega = 0\,; \:{-1}< \Sigma_+ < 1, \:0<s<\textfrac{1}{2},\: M_1 = \sqrt{12(1-\Sigma_+^2)} \:\Big\}\:.
\end{align*}
The symbol $\mathcal{X}_{\,\mathrm{I}}$ is used with slight abuse of notation; 
the type~I state space $\mathcal{X}_{\,\mathrm{I}}$ is embedded into $\overline{\mathcal{X}}_{\mathrm{II}}$ 
as the boundary subset $\{M_1 = 0\}$, see the discussion below.
The component $\mathcal{V}_\mathrm{II}$ is the vacuum boundary;
the components $\mathcal{S}_\flat$ and $\mathcal{S}_\sharp$ we call the `sides'; see Fig.~\ref{b2state}.
Let us study the dynamics of the flow induced on each of the above boundary subsets.

\textbf{The Bianchi type~I boundary} $\bm{\mathcal{X}_{\,\mathrm{I}}}$. 
The dynamical system induced by~\eqref{dynsysbianchiII} on the invariant subset $M_1=0$ 
is identical to the type~I system~\eqref{dynsysbianchiI};
its dynamics are described by Theorem~\ref{BianchiItheo}, i.e., Fig.~\ref{BianchiIfig}. 
We thus refer to $\mathcal{X}_{\,\mathrm{I}}$ as the Bianchi type~I boundary of $\mathcal{X}_{\,\mathrm{II}}$. 
Let us replicate the list of fixed points:
\begin{list}{point}{\leftmargin1.15cm\labelwidth4cm\labelsep0cm}
\item[$\mathrm{T}_\flat$/$\mathrm{T}_\sharp$] \quad The `Taub points': $\mathrm{T}_\flat: (\Sigma_+,s,M_1)=(-1,0,0)$ and 
$\mathrm{T}_\sharp: (\Sigma_+,s,M_1)=(-1,\textfrac{1}{2},0)$.
\item[$\mathrm{Q}_\flat$/$\mathrm{Q}_\sharp$] \quad The `non-flat LRS points': $\mathrm{Q}_\flat: (\Sigma_+,s,M_1)=(1,0,0)$ and $\mathrm{Q}_\sharp: (\Sigma_+,s,M_1)=(1,\textfrac{1}{2},0)$.
\item[$\mathrm{R}_\flat$] \quad For $\beta \in (-1,1)$ we have $\mathrm{R}_\flat: (\Sigma_+,s,M_1)=(-\beta,0,0)$.
\item[$\mathrm{R}_\sharp$] \quad For $\beta \in (-2,2)$ we have $\mathrm{R}_\sharp: (\Sigma_+,s,M_1)=(\textfrac{\beta}{2},\textfrac{1}{2},0)$.
\item[$\mathrm{F}$] \quad  The Friedmann point F is given by $(\Sigma_+,s,M_1)=(0,\bar{s},0)$, where $\bar{s}$ solves $u(\bar{s})=w$.
\end{list}
It is straightforward to study the role of the fixed points w.r.t.\ the direction orthogonal
to $\mathcal{X}_{\,\mathrm{I}}$; from~\eqref{dynsysbianchiII} we have
\[
\left[\textfrac{d}{d\tau}(\log M_1)\right]_{\,|\mathcal{X}_{\mathrm{I}}}=2(1 -2 \Sigma_+) -\textfrac{3}{2} (1-w) (1 -\Sigma_+^2)
= \textfrac{1}{2} (1 +3 w) + \Sigma_+ \big( \textfrac{3}{2} (1-w) \Sigma_+ -4 \big)\:.
\]
As an example, consider the point $\mathrm{F}$: Evaluated at F, this derivative is positive, 
so that there exists (at least) one orbit
that emanates from F into the interior of $\mathcal{X}_{\,\mathrm{II}}$.
In other words, F is a repellor in the orthogonal direction (i.e., the $M_1$ direction);
this behavior is color-coded into Fig.~\ref{BianchiIfig} --- the fixed point F is
represented by a white circle. (Points that act as attractors in the
orthogonal direction are depicted in black in Fig.~\ref{BianchiIfig}.)

\textbf{The vacuum boundary} $\bm{\mathcal{V}_{\,\mathrm{II}}}$. 
The dynamical system induced on $\mathcal{V}_{\,\mathrm{II}}$ 
is obtained by setting $\Omega=0$ in~\eqref{dynsysbianchiII}; since $\Omega=0$ entails $\rho=0$, we refer to $\mathcal{V}_{\,\mathrm{II}}$ 
as the vacuum boundary. The essential dynamics are described by the system
\begin{equation*}
\Sigma_+'=2(1-\Sigma_+^2)(2-\Sigma_+)\:,\qquad
s'=-6s(1-2s)\Sigma_+\:,
\end{equation*}
since $\Omega=0$ implies $M_1^2=12 (1-\Sigma_+^2)$. 
The closure of $\mathcal{V}_{\,\mathrm{II}}$  is the rectangle $(\Sigma_+,s)\in [-1,1]\times [0,\textfrac{1}{2}]$; 
the vertices are the only fixed points: $\mathrm{T}_\flat$, $\mathrm{T}_\sharp$, $\mathrm{Q}_\flat$, $\mathrm{Q}_\sharp$. 
The boundary of the rectangle is characterized by a simple flow, see Fig.~\ref{vacuumIIflow}.

The orbits $\mathrm{T}_\flat\rightarrow \mathrm{T}_\sharp$ and $\mathrm{Q}_\flat \leftarrow \mathrm{Q}_\sharp$ coincide
with the type~I vacuum orbits and thus represent the Taub solution~\eqref{taub} 
and the non-flat LRS solution~\eqref{solQ}, respectively.
Analogously, every orbit in the interior of $\mathcal{V}_{\,\mathrm{II}}$
represents a type~II vacuum solution.
Since $\Sigma_+'>0$ in $\mathcal{V}_{\,\mathrm{II}}$, 
the analysis of the dynamics is straightforward: 
The fixed points $\mathrm{T}_\flat$ and $\mathrm{Q}_\flat$ are the $\alpha$-limit set and the $\omega$-limit set, 
respectively, of every orbit, see Fig.~\ref{vacuumIIflow}.

\begin{figure}[Ht]
\begin{center}
\psfrag{sig}[cc][cc][1][0]{$\Sigma_+$}
\psfrag{tb}[cc][cc][0.8][0]{$\mathrm{T}_\flat$}
\psfrag{ts}[cc][cc][0.8][0]{$\mathrm{T}_\sharp$}
\psfrag{qb}[cc][cc][0.8][0]{$\mathrm{Q}_\flat$}
\psfrag{qs}[cc][cc][0.8][0]{$\mathrm{Q}_\sharp$}
\psfrag{s}[cc][cc][1][0]{$s$}
\includegraphics[width=0.5\textwidth]{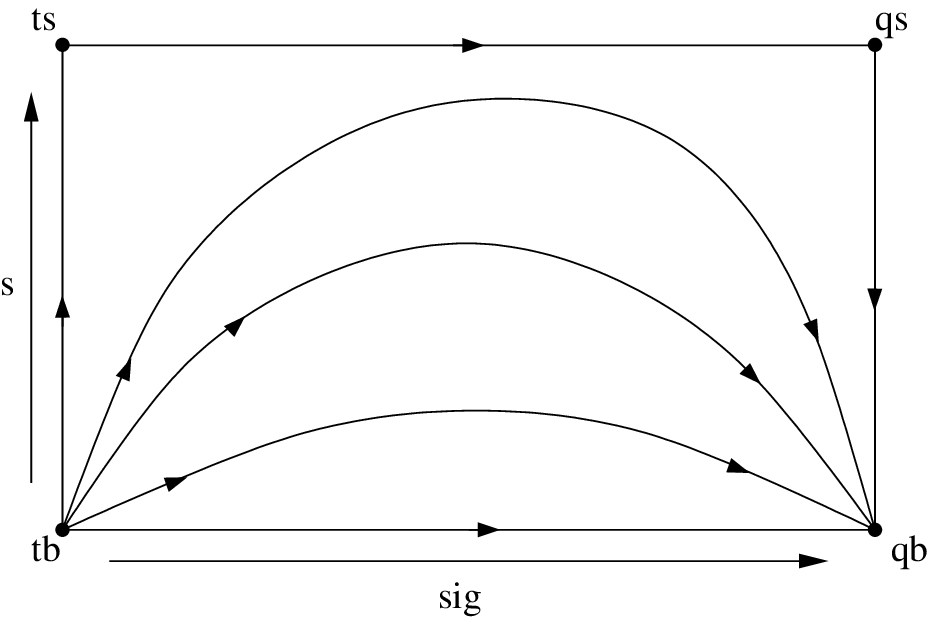}
\end{center}
\caption{The flow induced on $\mathcal{V}_{\,\mathrm{II}}$.}
\label{vacuumIIflow}
\end{figure}

\textbf{The side} $\bm{\mathcal{S}_\sharp}$. The dynamical system 
that is induced by~\eqref{dynsysbianchiII} on $\mathcal{S}_\sharp$ is given by
\begin{subequations}\label{dynsysAsharp}
\begin{align}
\label{dynsysAsharp1}
&\Sigma_+'=\textfrac{1}{6}M_1^2(2-\Sigma_+)-\textfrac{3}{2}\Omega(1-w)(\Sigma_+-\textfrac{\beta}{2})\:,\\[0.5ex]
\label{dynsysAsharp2}
& M_1'=M_1 \big[ \textfrac{3}{2} (1-w) \Sigma_+^2 - 4 \Sigma_+ + 
\textfrac{1}{2}(1 +3 w) \big(1 -\textfrac{1}{12}M_1^2\big)\big]\:,
\end{align}
\end{subequations}
where $\Omega=1-\Sigma_+^2-\textfrac{1}{12}M_1^2$. 
The state space is (the interior of) a closed semi-circle.
The two vertices are the fixed points $\mathrm{T}_\sharp$ and $\mathrm{Q}_\sharp$.

On the vacuum part $\Omega = 0$ of the boundary of $\mathcal{S}_\sharp$ (which is the semi-circle) 
we observe 
\begin{equation}\label{boundIIa}
{\Sigma_+'}_{\,|\Omega=0}=2(1-\Sigma_+^2)(2-\Sigma_+)\,,
\end{equation}
which is positive since $|\Sigma_+| < 1$.
Therefore, the semi-circle corresponds to an orbit
$\mathrm{T}_\sharp \rightarrow \mathrm{Q}_\sharp$; see Fig.~\ref{Asharpfig}. (The same orbit is depicted in Fig.~\ref{vacuumIIflow}
and the color-coding in Fig.~\ref{BianchiIfig} is consistent as well: In Fig.~\ref{BianchiIfig}, $\mathrm{T}_\sharp$ 
is white and $\mathrm{Q}_\sharp$ is black.)

The equation induced on the subset $M_1 = 0$ of the boundary of $\mathcal{S}_\sharp$ (which is the base of the semi-circle) 
is
\[ 
{\Sigma_+'}_{\,|M_1 = 0}=-\textfrac{3}{2}(1-\Sigma_+^2)(1-w)(\Sigma_+-\textfrac{\beta}{2})\:,
\]
which coincides with~\eqref{eqcalIsharp}; 
this is because the base of the semi-circle coincides with 
the $\mathcal{I}_\sharp$ subset of $\partial\mathcal{X}_{\,\mathrm{I}}$.
Accordingly, there is a fixed point, $\mathrm{R}_\sharp$, if $|\beta|<2$ (i.e., in the cases \A, \B, \C),
which acts as an attractor on this boundary component;
in contrast, we obtain an orbit $\mathrm{T}_\sharp \leftarrow \mathrm{Q}_\sharp$ for $\beta \leq {-2}$~(\Dminus)
and an orbit $\mathrm{T}_\sharp \rightarrow \mathrm{Q}_\sharp$ for $\beta \geq 2$~(\Dplus);
see Figs.~\ref{BianchiIfig} and~\ref{Asharpfig}.

To investigate which role the point $\mathrm{R}_\sharp$ plays in the context of the
flow on $\mathcal{S}_\sharp$ we consider 
\begin{equation}\label{dm1}
\left[\textfrac{d}{d\tau}\log M_1\right]_{\,|\mathrm{R}_\sharp}=\textfrac{3}{8}\beta^2(1-w)-2\beta+\textfrac{1}{2}(1+3w)\:.
\end{equation}
Let 
\begin{equation}\label{betaDdef}
\beta_\sharp(w) :=\frac{8-2\sqrt{9w^2-6w+13}}{3(1-w)} = \frac{8-2\sqrt{[3(1-w)]^2 +4(1+3 w)}}{3(1-w)}\,\in\, (0,1)\:.
\end{equation}
The r.h.s. of~\eqref{dm1} is positive for $-2<\beta<\beta_\sharp(w)$; hence, in this case, 
there exists one orbit that emanates from the fixed point $\mathrm{R}_\sharp$ into the interior of $\mathcal{S}_\sharp$, i.e.,
$\mathrm{R}_\sharp$ is a saddle in $\overline{\mathcal{S}}_\sharp$.
If, on the other hand, $\beta_\sharp(w)<\beta<2$, then 
the r.h.s.\ of~\eqref{dm1} is negative, which makes $\mathrm{R}_\sharp$ a sink in $\overline{\mathcal{S}}_\sharp$.
(In the exceptional case $\beta = \beta_\sharp(w)$, the fixed point $\mathrm{R}_\sharp$ is a 
non-hyperbolic sink, 
as will be proved below by using global methods.) 
Note that the distinction $\beta \gtrless \beta_\sharp(w)$ leads to 
two subcases of the case \Aplus, since $\beta_\sharp(w) \in (0,1)$,
cf.~Fig.~\ref{Asharpfig} and the color-coding in Fig.~\ref{BianchiIfig}.

Let us turn to the interior of $\mathcal{S}_\sharp$.
For $\beta<\beta_\sharp(w)$ (i.e., in the $\boldsymbol{-}$ cases and the
respective subcase of \Aplus) 
there exists a fixed point in the interior of $\mathcal{S}_\sharp$, which we call $\mathrm{C}_\sharp$:
\begin{itemize}
\item[$\mathrm{C}_\sharp$] \quad The coordinates of $\mathrm{C}_\sharp$ are 
\,${\Sigma_+}=\textfrac{2(1+3w)}{16-3\beta(1-w)}$, 
${M_1^2}=\textfrac{36(1-w)[3\beta^2(1-w)-16\beta+4(1+3w)]}{[16-3\beta(1-w)]^2}$\,.
\end{itemize}
This fixed point represents a solution that generalizes the Collins-Stewart solution~\eqref{CSsol};
it reduces to~\eqref{CSsol} in the case $\beta=0$.
The exact solution that corresponds to the fixed point $\mathrm{C}_\sharp$ is calculated in Appendix~\ref{exact}.
The condition $\beta<\beta_\sharp(w)$ assures that $(M_1)_{\,|\mathrm{C}_\sharp}>0$; 
the point $\mathrm{C}_\sharp$ converges to $\mathrm{R}_\sharp$ when $\beta \rightarrow \beta_\sharp(w)$. 
The point $\mathrm{C}_\sharp$ is a sink in $\mathcal{S}_\sharp$; this will be proved
below by using global methods.
Finally, to investigate which role $\mathrm{C}_\sharp$ plays in the context
of the system~\eqref{dynsysbianchiII} on the entire 
type~II state space $\overline{\mathcal{X}}_{\mathrm{II}}$, we compute
\[
\left[\textfrac{d}{d\tau}\log (1-2s)\right]_{\,|\mathrm{C}_\sharp}={6\Sigma_+}_{\,|\mathrm{C}_\sharp}>0\;,
\]
which implies that $\mathrm{C}_\sharp$ repels orbits from the interior of $\mathcal{X}_\mathrm{II}$; in other
words, $\mathrm{C}_\sharp$ is a saddle in $\overline{\mathcal{X}}_\mathrm{II}$.
This fact is depicted in Fig.~\ref{Asharpfig} by representing $\mathrm{C}_\sharp$ by a white circle.  
 
To perform a global dynamical systems analysis of the flow on the invariant set $\mathcal{S}_\sharp$, 
we apply the monotonicity principle. 

Consider first the case \Dplus. Since $\beta\geq 2$ in this case, we have $\Sigma_+'>0$ 
in the interior of $\mathcal{S}_\sharp$. Therefore, the $\alpha$- and $\omega$-limit set of every
interior orbit must lie on the boundary $\partial\mathcal{S}_\sharp$. 
From the simple structure of the flow on $\partial\mathcal{S}_\sharp$,
we deduce that the $\alpha$- and $\omega$-limit set of every orbit in the interior of $\mathcal{S}_\sharp$ 
is the fixed point $\mathrm{T}_\sharp$ and $\mathrm{Q}_\sharp$, respectively; see Fig.~\ref{Asharpfig}. 

In the remaining cases, where $\beta<2$, we define
\[
c=\textfrac{2(1+3w)}{16-3\beta(1-w)}\:,\quad\delta=\textfrac{3(1-w)(c-\beta/2)c}{(1+3w)(1-c^2)}\:.
\]
It is elementary to verify that $c\in (0,1)$ and $\delta<1$ for all $\beta<2$; moreover,
\[
\delta<0\ \text{ for } \beta\in (\beta_\sharp(w),2)\:,\quad 
\delta=0\ \text{ for } \beta=\beta_\sharp(w)\:,\quad\text{and}\quad 
\delta>0\ \text{ for }\beta<\beta_\sharp(w)\:.
\]
Consider the function
\[
Z_2=\frac{M_1^{2\delta}\,\Omega^{1-\delta}}{(1-c\Sigma_+)^2}\:.
\]
(For $\beta=0$ the function $Z_2$ reduces to the monotone function for perfect fluid models given in~\cite[eq.~(6.47)]{WE}.)
A straightforward calculation shows that 
\[
Z_2'=\chi Z_2\:,\qquad \text{where }\,
\chi=\textfrac{3(1-w)(1-c\beta/2)}{(1-c\Sigma_+)(1-c^2)}(\Sigma_+-c)^2\geq 0\quad \text{($\forall\,\beta<2$)}
\]
and $\chi=0$ if and only if $\Sigma_+=c$. If $\beta<\beta_\sharp(w)$, the fixed point $\mathrm{C}_\sharp$ is 
located in the interior of $\mathcal{S}_\sharp$. In this case the function $Z_2$ is 
monotonically increasing on $\mathcal{S}_\sharp\backslash\{\mathrm{C}_\sharp\}$. 
The monotonicity principle then implies that the fixed point $\mathrm{C}_\sharp$ is 
the $\omega$-limit set of every orbit in the interior of $\mathcal{S}_\sharp$, 
whereas the $\alpha$-limit set is contained on the boundary 
$\partial\mathcal{S}_\sharp$. The generic $\alpha$-limit set is the point $\mathrm{T}_\sharp$
in the cases \Aplus\ (where $\beta<\beta_\sharp(w)$), \Aminus, \Bminus, \Cminus;
see Fig.~\ref{Asharpfig}.
In the case \Dminus, the entire boundary $\partial\mathcal{S}_\sharp$, which forms a heteroclinic cycle,
is the $\alpha$-limit set.
Finally, for $\beta\in [\beta_\sharp(w),2)$, which comprises the respective subcase
of \Aplus\ and the cases \Bplus, \Cplus,
the function $Z_2$ is monotonically increasing 
in the whole interior and thus the $\alpha$- and the $\omega$-limit sets of interior orbits 
are both contained on the boundary of $\mathcal{S}_\sharp$. The $\alpha$-limit set
is $\mathrm{T}_\sharp$, the $\omega$-limit set is $\mathrm{R}_\sharp$. 
These results are summarized in Fig.~\ref{Asharpfig}.

\begin{figure}[Ht!]
\begin{center}
\psfrag{T}[cc][cr][0.7][0]{$\mathrm{T}_\sharp$}
\psfrag{Q}[cc][cc][0.7][0]{$\mathrm{Q}_\sharp$}
\psfrag{C}[cc][cc][0.7][0]{$\mathrm{C}_\sharp$}
\subfigure[\Dminus]{\includegraphics[width=0.45\textwidth]{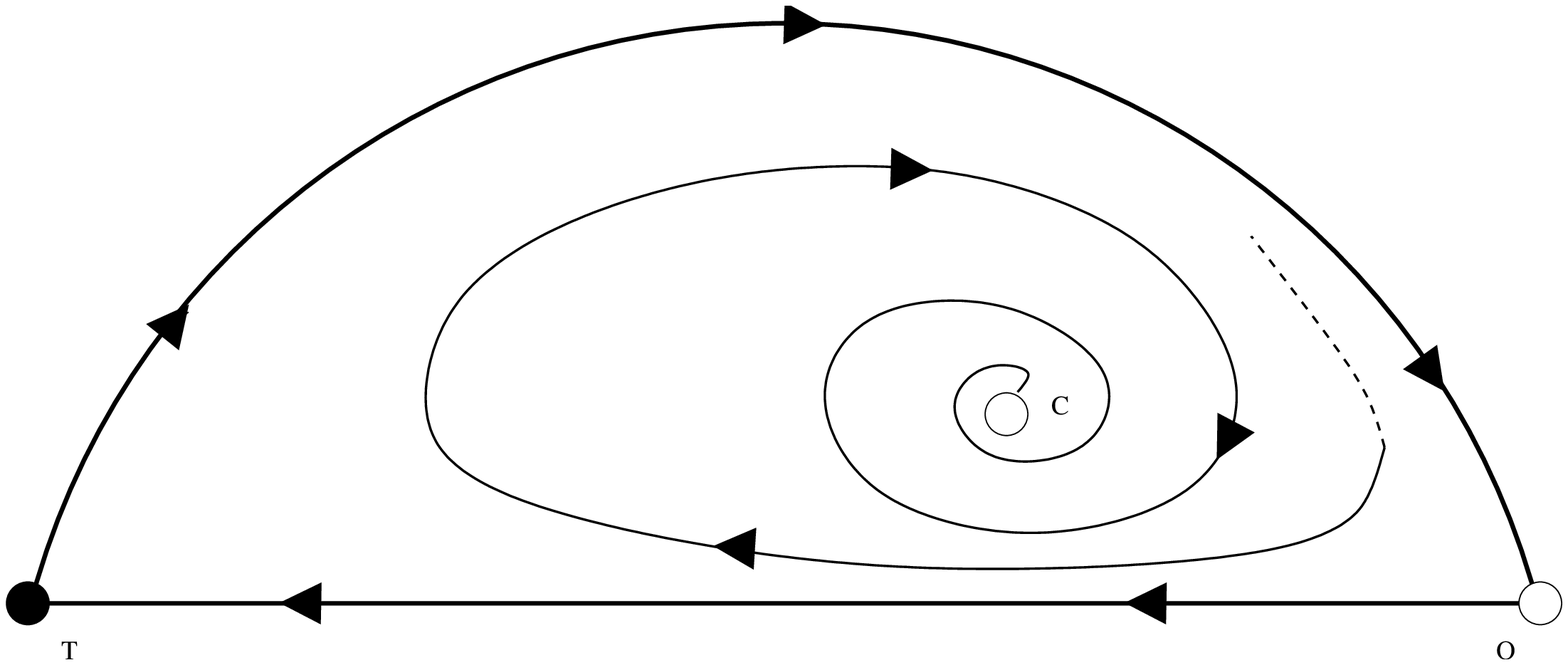}}\qquad
\psfrag{T}[cc][cr][0.7][0]{$\mathrm{T}_\sharp$}
\psfrag{Q}[cc][cc][0.7][0]{$\mathrm{Q}_\sharp$}
\psfrag{D}[cc][cr][0.7][0]{$\mathrm{R}_\sharp$}
\psfrag{C}[cc][cc][0.7][0]{$\mathrm{C}_\sharp$}
\subfigure[\Azerominus, \Aminus, \Bminus, \Cminus]{\includegraphics[width=0.45\textwidth]{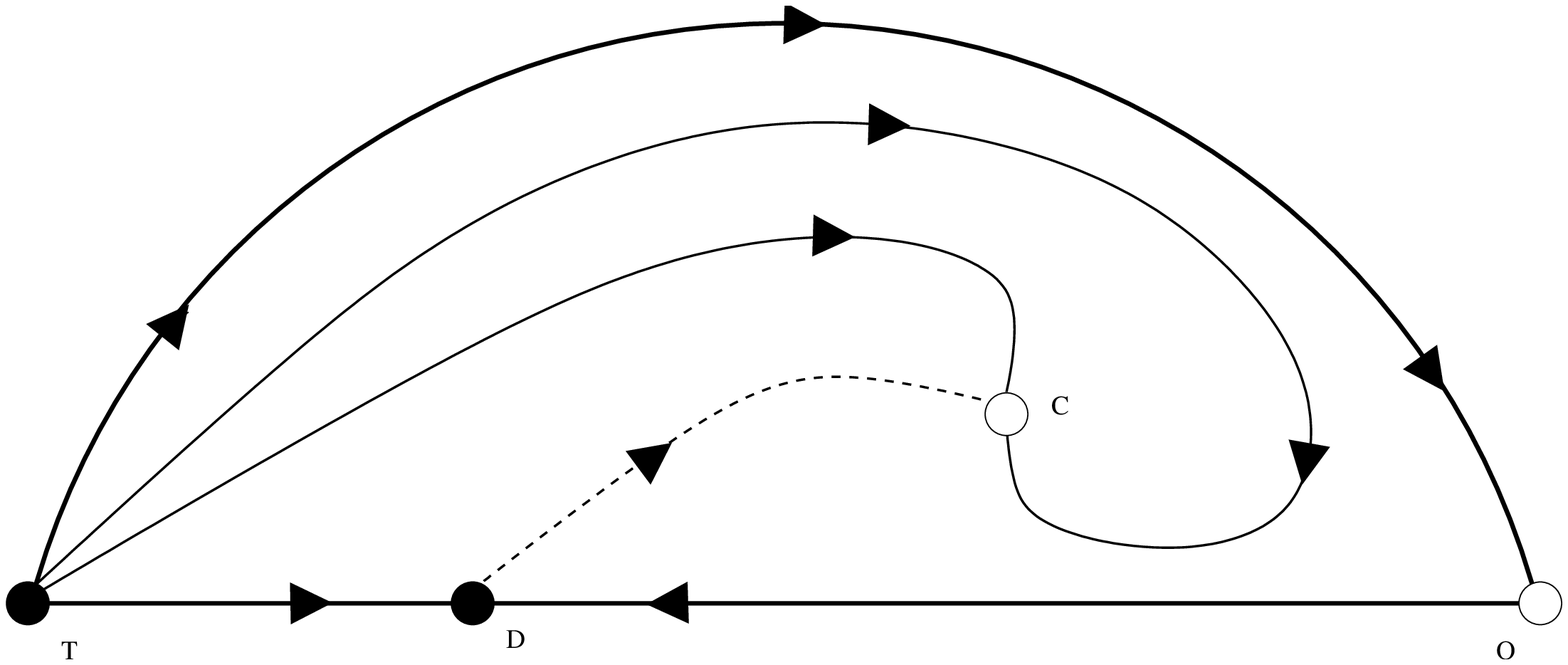}}\\
\psfrag{T}[cc][cr][0.7][0]{$\mathrm{T}_\sharp$}
\psfrag{Q}[cc][cc][0.7][0]{$\mathrm{Q}_\sharp$}
\psfrag{D}[cc][cr][0.7][0]{$\mathrm{R}_\sharp$}
\psfrag{C}[cc][cc][0.7][0]{$\mathrm{C}_\sharp$}
\subfigure[\Azeroplus, \Aplus\ $(0<\beta<\beta_\sharp)$]{\includegraphics[width=0.45\textwidth]{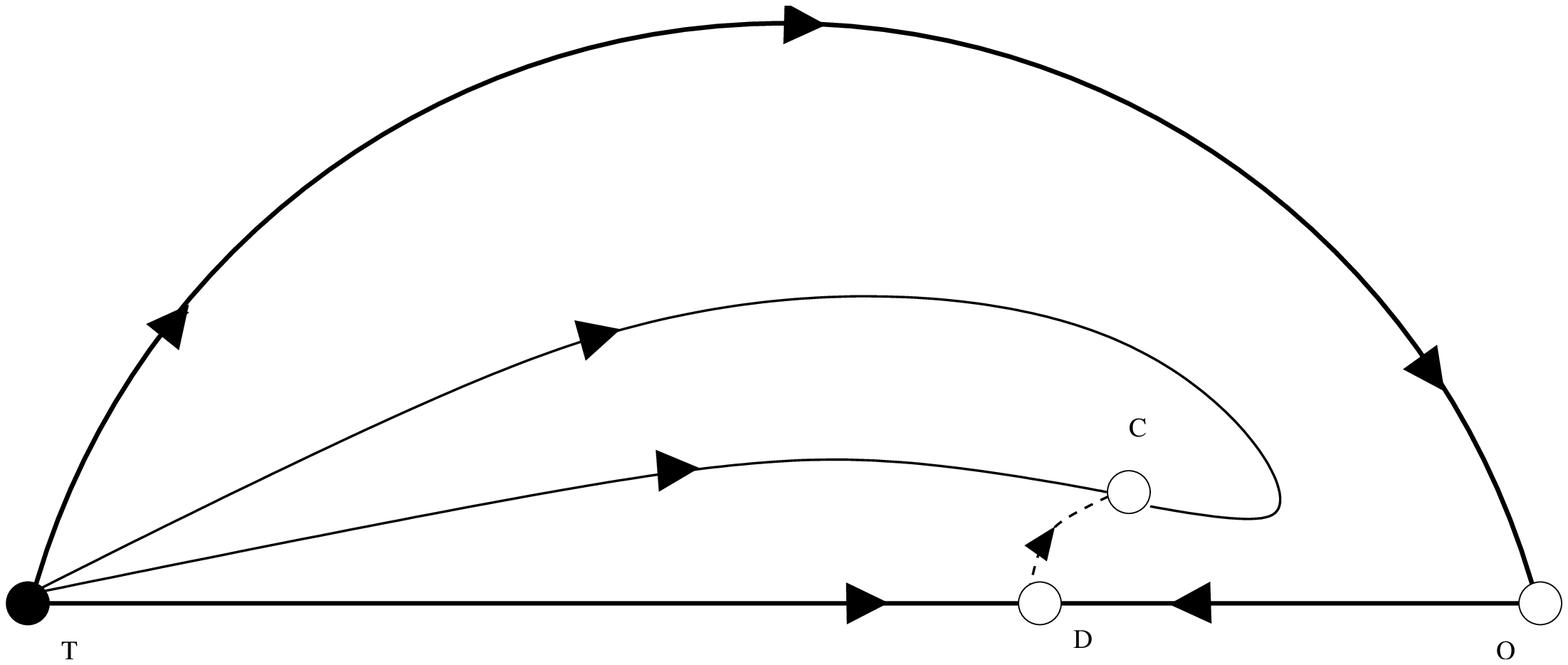}}\qquad
\psfrag{T}[cc][cc][0.7][0]{$\mathrm{T}_\sharp$}
\psfrag{Q}[cc][cc][0.7][0]{$\mathrm{Q}_\sharp$}
\psfrag{D}[cc][cc][0.7][0]{$\mathrm{R}_\sharp$}
\subfigure[\Aplus\ $(\beta_\sharp\leq\beta<1)$, \Bplus, \Cplus]{\includegraphics[width=0.45\textwidth]{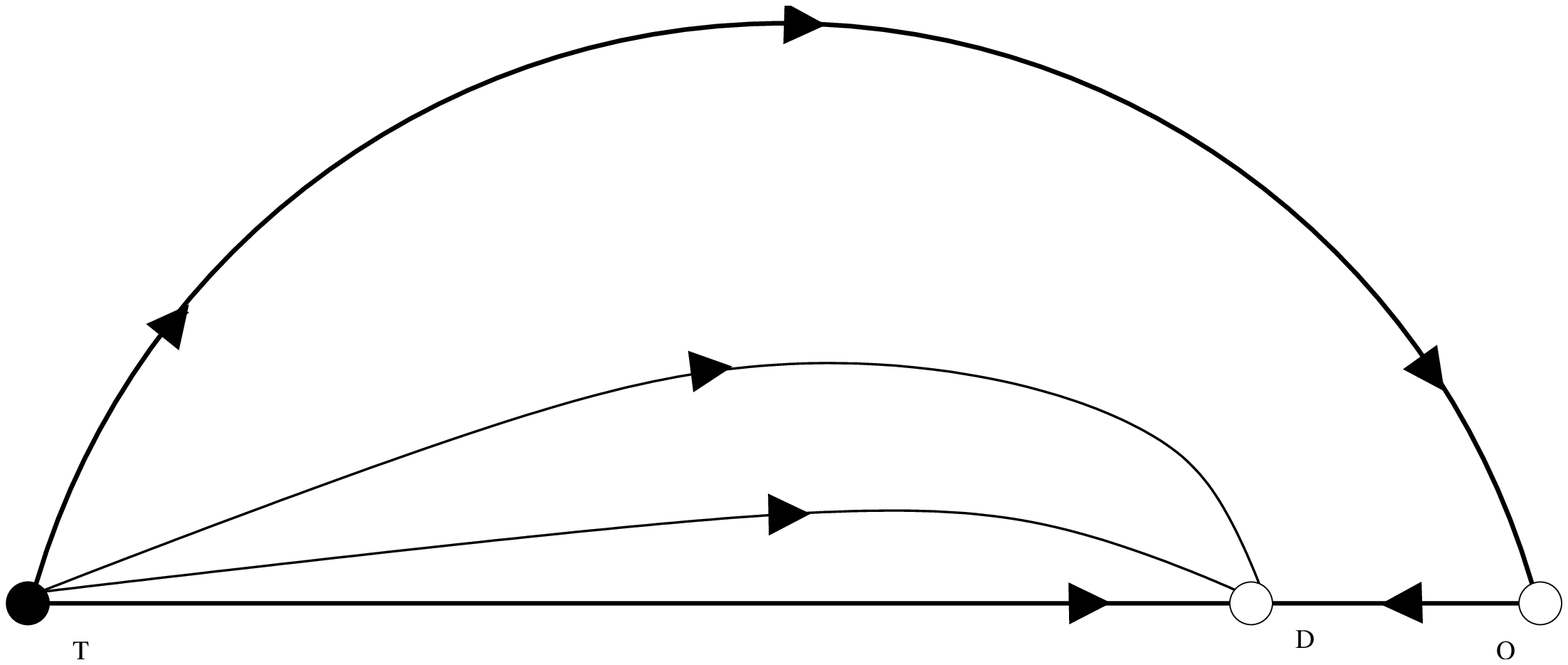}}\\
\psfrag{T}[cc][cc][0.7][0]{$\mathrm{T}_\sharp$}
\psfrag{Q}[cc][cc][0.7][0]{$\mathrm{Q}_\sharp$}
\psfrag{C}[cc][cc][0.7][0]{$\mathrm{C}_\sharp$}
\subfigure[\Dplus]{\label{SsharpD+}\includegraphics[width=0.45\textwidth]{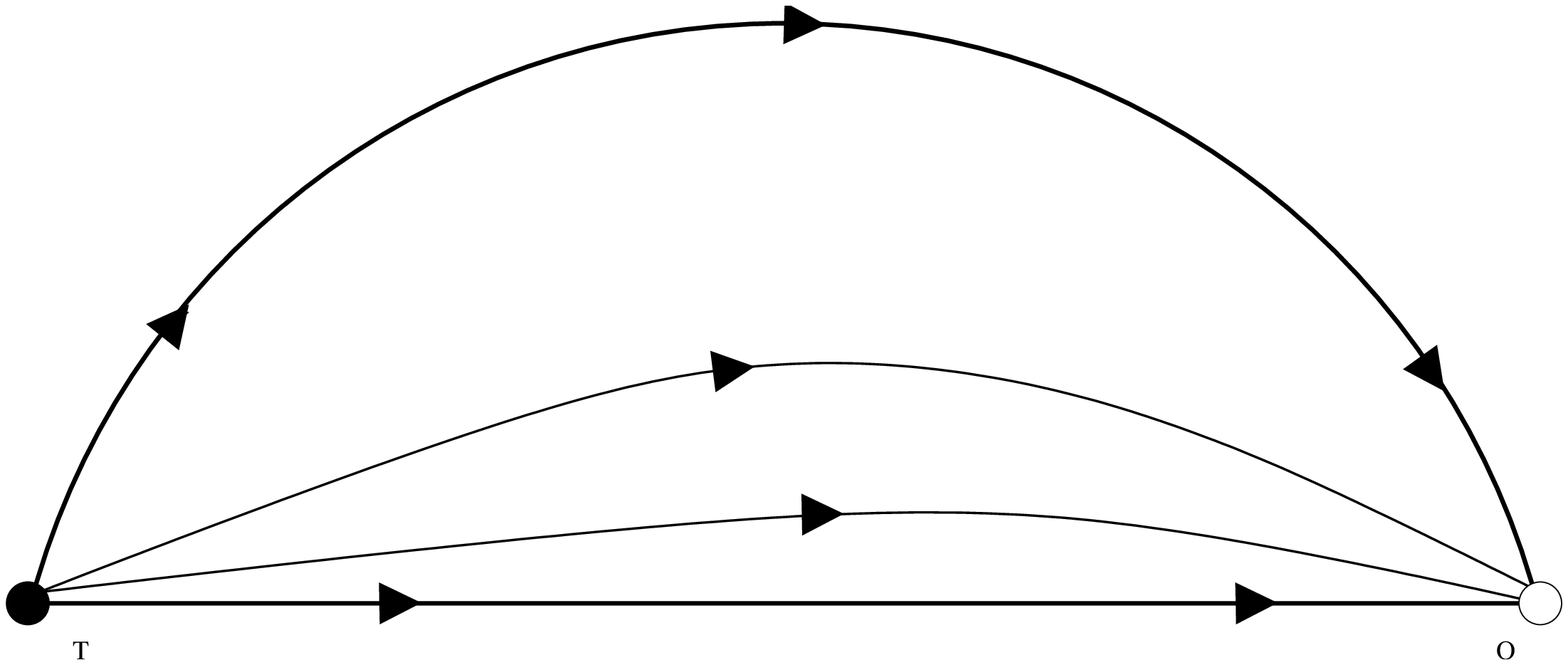}\label{magneticAsharp}}\qquad
\psfrag{-1}[cc][cc][0.7][0]{$-1$}
\psfrag{+1}[cc][cc][0.7][0]{$+1$}
\psfrag{M1}[cc][cc][0.7][0]{$M_1$}
\psfrag{sigma}[cc][cc][0.7][0]{$\Sigma_+$}
\psfrag{Asharp}[cc][cc][1.2][0]{$\mathcal{S}_\sharp$}
\subfigure[Coordinates]{\includegraphics[width=0.45\textwidth]{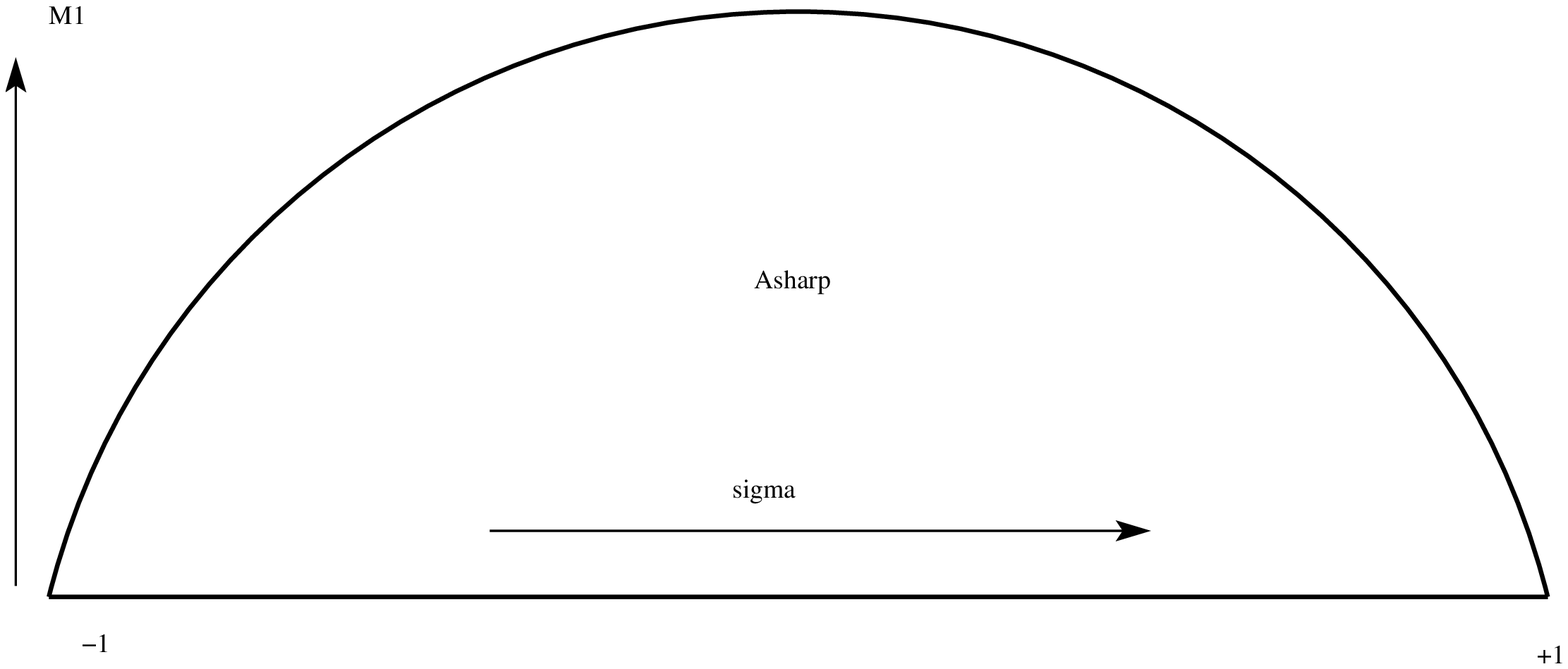}}
\end{center}
\caption{Phase portraits of the flows on $\mathcal{S}_\sharp$ in dependence on the anisotropy parameter $\beta$. 
Continuous lines represent typical orbits---each of these orbits is a member of a one-parameter set of orbits 
(i.e., a two-parameter set of solutions). 
Dashed lines represent orbits whose future and/or past asymptotic behavior is non-generic.
The fixed points are color-coded. When the set $\mathcal{S}_\sharp$ is viewed as the
boundary subset of the type~II state space $\mathcal{X}_{\,\mathrm{II}}$, there exists an
orthogonal direction. A black fixed point is an attractor, a white fixed point is a repellor in the orthogonal direction,
cf.~Fig.~\ref{BianchiIfig}. 
Analogously, the set $\mathcal{S}_\sharp$ appears as a boundary subset of the type~VIII and the type~IX
state spaces; the color-coding then refers to the orthogonal direction of $\mathcal{S}_\sharp$
when viewed as an embedded set in these state spaces.}
\label{Asharpfig}
\end{figure}

\textbf{The side} $\bm{\mathcal{S}_\flat}$. 
The dynamical system induced on $\mathcal{S}_\flat$ is given by
\begin{subequations}\label{dynsysAflat}
\begin{align}
\tag{\ref{dynsysAsharp1}${}^\prime$}
&\Sigma_+'=\textfrac{1}{6}M_1^2(2-\Sigma_+)-\textfrac{3}{2}\Omega(1-w)(\Sigma_++\beta)\:,\\[0.5ex]
\tag{\ref{dynsysAsharp2}${}^\prime$}
& M_1'=M_1 \big[ \textfrac{3}{2} (1-w) \Sigma_+^2 - 4 \Sigma_+ + 
\textfrac{1}{2}(1 +3 w) \big(1 -\textfrac{1}{12}M_1^2\big)\big]\:,
\end{align}
\end{subequations}
where $\Omega=1-\Sigma_+^2-\textfrac{1}{12}M_1^2$. 
The state space $\mathcal{S}_\flat$ is an identical copy of the state space $\mathcal{S}_\sharp$,
where the dynamical system~\eqref{dynsysAflat} on $\mathcal{S}_\flat$ is obtained from the 
system~\eqref{dynsysAsharp} on $\mathcal{S}_\sharp$ upon the formal substitution $\beta\mapsto -2\beta$. 
Accordingly, we simply adapt the results of the analysis of $\mathcal{S}_\sharp$
by replacing $\beta\mapsto -2\beta$.

The flow on the vacuum part $\Omega = 0$ of $\partial\mathcal{S}_\flat$ (which is the semi-circle) 
is determined by an equation identical to~\eqref{boundIIa};
hence, 
the semi-circle is an orbit $\mathrm{T}_\flat \rightarrow \mathrm{Q}_\flat$; see Fig.~\ref{Aflatfig}.
The part $M_1 = 0$ of $\partial\mathcal{S}_\flat$ (which is the base of the semi-circle) 
coincides with 
the $\mathcal{I}_\flat$ subset of $\partial\mathcal{X}_{\,\mathrm{I}}$; the equation is~\eqref{eqcalIflat}.
Accordingly, there is the fixed point $\mathrm{R}_\flat$ if $|\beta|< 1$ (i.e., in the \A\ cases),
which acts as an attractor on this boundary component; if $\beta \leq -1$ (\Bminus, \Cminus, \Dminus),
we obtain an orbit $\mathrm{T}_\flat \rightarrow \mathrm{Q}_\flat$;
if $\beta \geq 1$ (\Bplus, \Cplus, \Dplus),
we obtain an orbit $\mathrm{T}_\flat \leftarrow \mathrm{Q}_\flat$;
see Figs.~\ref{BianchiIfig} and~\ref{Aflatfig}.

To investigate which role the point $\mathrm{R}_\flat$ plays in the context of the
flow on $\mathcal{S}_\flat$ we consider 
\begin{equation}\label{dm2}
\left[\textfrac{d}{d\tau}\log M_1\right]_{\,|\mathrm{R}}=\textfrac{3}{2}\beta^2(1-w)+4\beta+\textfrac{1}{2}(1+3w)\:.
\tag{\ref{dm1}${}^\prime$}
\end{equation}
Let 
\begin{equation}\label{betaRdef}
\beta_\flat(w) :=\frac{-4-\sqrt{9w^2-6w+13}}{3(1-w)} = \frac{-4-\sqrt{[3(1-w)]^2 +4(1+3 w)}}{3(1-w)}\,\in\, 
(-\textfrac{1}{2},0)\:.
\tag{\ref{betaDdef}${}^\prime$}
\end{equation}
The r.h.s. of~\eqref{dm2} is positive for $\beta_\flat(w)<\beta<1$; hence, in this case,
there exists one orbit that emanates from the fixed point $\mathrm{R}_\flat$ into the interior of $\mathcal{S}_\flat$, i.e.,
$\mathrm{R}_\flat$ is a saddle in $\overline{\mathcal{S}}_\flat$.
If, on the other hand, $-1<\beta<\beta_\flat(w)$, then the r.h.s. of~\eqref{dm2} is negative, 
which makes $\mathrm{R}_\flat$ a sink in $\overline{\mathcal{S}}_\flat$. 
(The fixed point $\mathrm{R}_\flat$ attracts interior orbits in the case $\beta=\beta_\flat(w)$ as well, 
as will be proved below by using global methods.) 
Since $\beta_\flat(w)\in (-\textfrac{1}{2},0)$ for $w\in (-\textfrac{1}{3},1)$, 
$\mathrm{R}_\flat$ is always a sink in the case \Aplus, 
whereas we need to distinguish two subcases of the case \Aminus, see Fig.~\ref{BianchiIfig}
and the color-coding in Fig.~\ref{BianchiIfig}.

For $\beta>\beta_\flat(w)$ (i.e., in the $\boldsymbol{+}$ cases and the
respective subcase of \Aminus) 
there exists a fixed point in the interior of $\mathcal{S}_\flat$, which we call $\mathrm{C}_\flat$:
\begin{itemize}
\item[$\mathrm{C}_\flat$] \quad The coordinates of $\mathrm{C}_\flat$ are 
\,${\Sigma_+}=\textfrac{1+3w}{8+3\beta(1-w)}$,
${M_1^2}=\textfrac{36(1-w)[3\beta^2(1-w)+8\beta+(1+3w)]}{[8+3\beta(1-w)]^2}$\,.
\end{itemize}
This fixed point represents a solution that generalizes the Collins-Stewart solution~\eqref{CSsol};
it reduces to~\eqref{CSsol} in the 
case $\beta=0$.
The point $\mathrm{C}_\flat$ is a sink in $\mathcal{S}_\flat$; this will be proved
below by using global methods.
When $\beta$ converges to $\beta_\flat(w)$, the point 
$\mathrm{C}_\flat$ converges to $\mathrm{R}_\flat$.
Finally, the role of $\mathrm{C}_\flat$ in the 
type~II state space $\overline{\mathcal{X}}_{\mathrm{II}}$ is determined by
\[
\left[\textfrac{d}{d\tau}\log s\right]_{\,|\mathrm{C}_\flat}=-{6\Sigma_+}_{|\mathrm{C}_\flat}<0\:,
\]
which implies that the fixed point $\mathrm{C}_\flat$ attracts orbits from the interior 
of $\mathcal{X}_\mathrm{II}$; it is therefore represented by a black dot.

To complete the dynamical systems analysis of the set $\mathcal{S}_\flat$, 
we use the monotonicity principle with the function
\begin{subequations}\label{z3}
\begin{equation}
Z_3=\frac{M_1^{2\delta}\,\Omega^{1-\delta}}{(1-c\Sigma_+)^2}\:,
\end{equation}
where
\begin{equation}
c=\textfrac{1+3w}{8+3\beta(1-w)}\:,\quad\delta=\textfrac{3(1-w)(c+\beta)c}{(1+3w)(1-c^2)}\:.
\end{equation}
\end{subequations}
This function is well-defined when $\beta > -1$ (which is the difficult case);
in fact, $Z_3$ satisfies
\[
Z_3'=\chi Z_3\:,\qquad \text{where }\,
\chi=\textfrac{3(1-w)(1+c\beta)}{(1-c\Sigma_+)(1-c^2)}(\Sigma_+-c)^2\geq 0\quad
 \text{($\forall\,\beta>-1$)}
\]
and $\chi=0$ if and only if $\Sigma_+=c$. 
The implications of the monotonicity principle are analogous
to those for the set $\mathcal{S}_\sharp$.
The results are summarized in 
the phase portraits depicted in Fig.~\ref{Aflatfig}.

\begin{figure}[Ht!]
\begin{center}
\psfrag{T}[cc][cc][0.7][0]{$\mathrm{T}_\flat$}
\psfrag{Q}[cc][cc][0.7][0]{$\mathrm{Q}_\flat$}
\psfrag{C}[cc][cc][0.7][0]{$\mathrm{C}_\flat$}
\subfigure[\Bminus, \Cminus, \Dminus]{\includegraphics[width=0.45\textwidth]{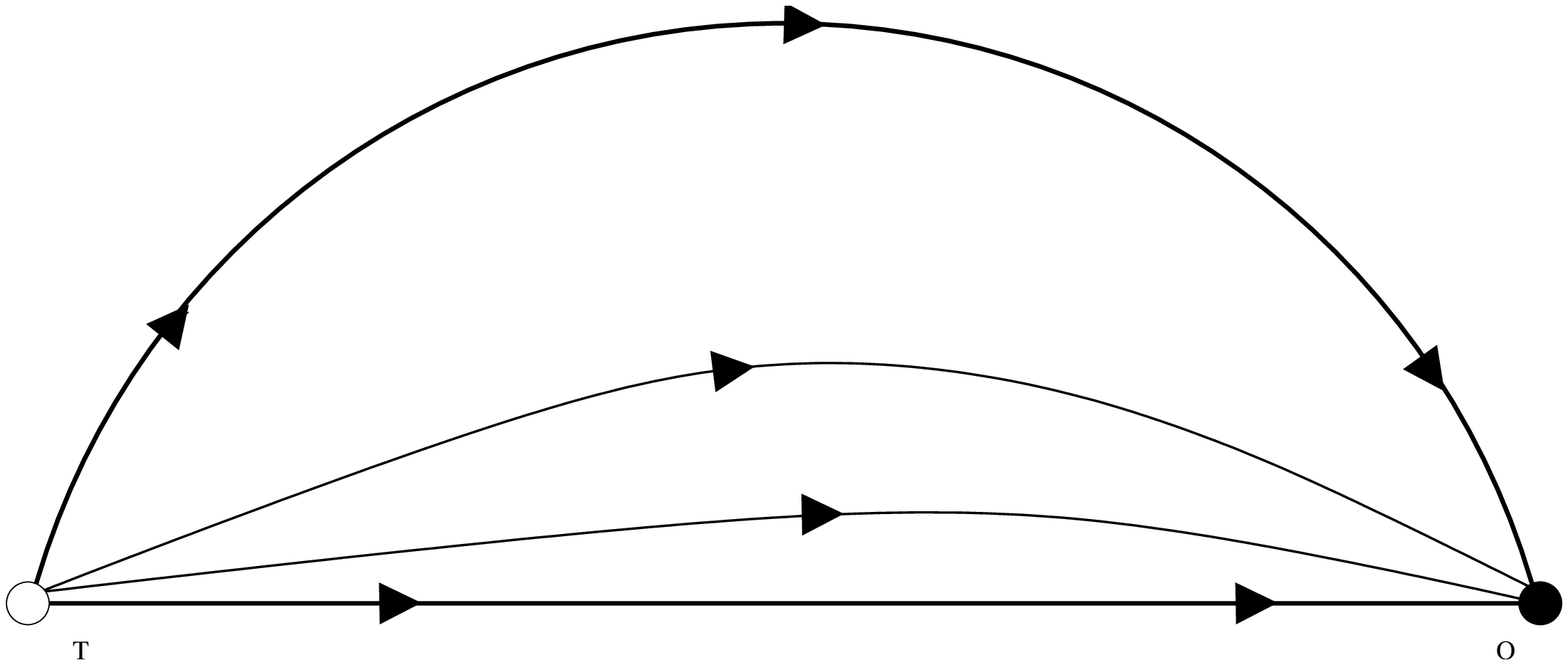}\label{magneticAflat}}\qquad
\psfrag{T}[cc][cc][0.7][0]{$\mathrm{T}_\flat$}
\psfrag{Q}[cc][cc][0.7][0]{$\mathrm{Q}_\flat$}
\psfrag{D}[cc][cc][0.7][0]{$\mathrm{R}_\flat$}
\psfrag{C}[cc][cc][0.7][0]{$\mathrm{C}_\flat$}
\subfigure[\Aminus\ $(-1<\beta\leq\beta_\flat)$]{\includegraphics[width=0.45\textwidth]{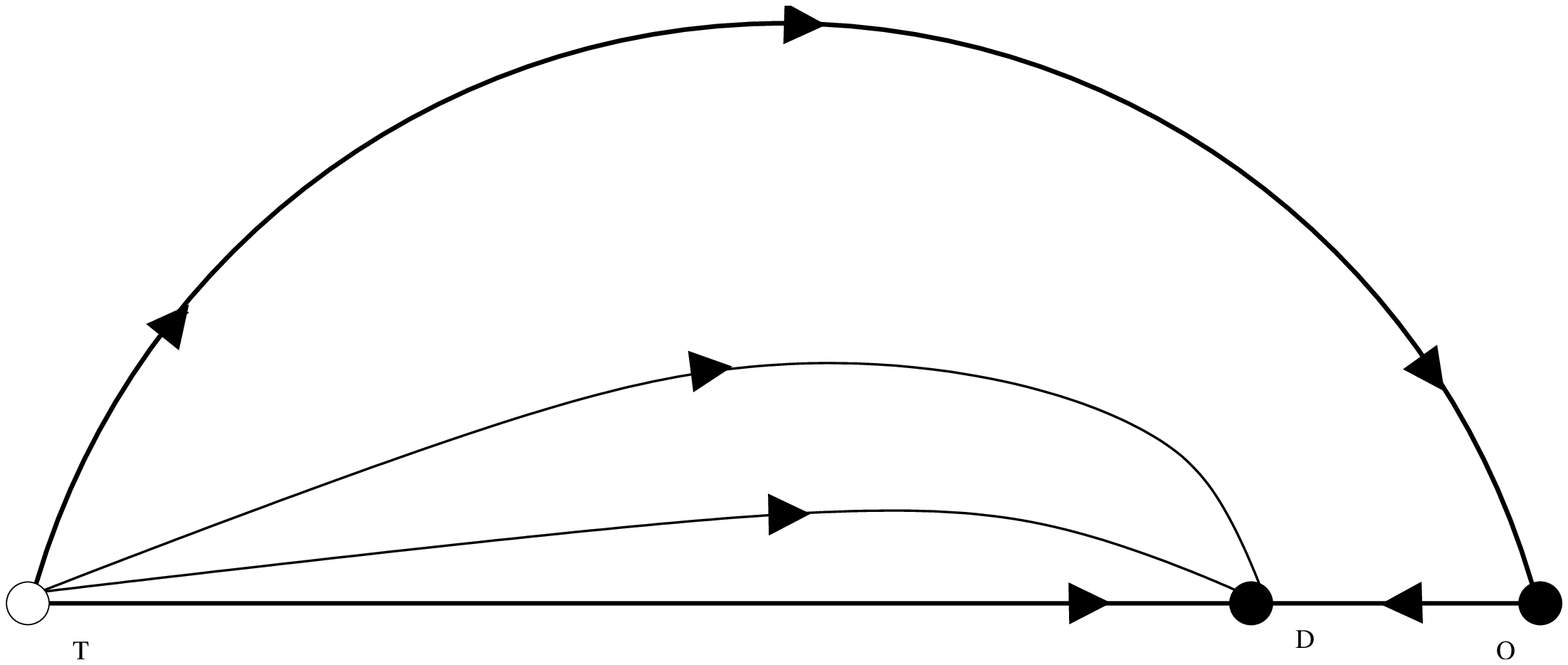}}\\
\psfrag{T}[cc][cc][0.7][0]{$\mathrm{T}_\flat$}
\psfrag{Q}[cc][cc][0.7][0]{$\mathrm{Q}_\flat$}
\psfrag{D}[cc][cr][0.7][0]{$\mathrm{R}_\flat$}
\psfrag{C}[cc][cc][0.7][0]{$\mathrm{C}_\flat$}
\subfigure[\Azerominus, \Aminus\ $(\beta_\flat<\beta<0)$]{\includegraphics[width=0.45\textwidth]{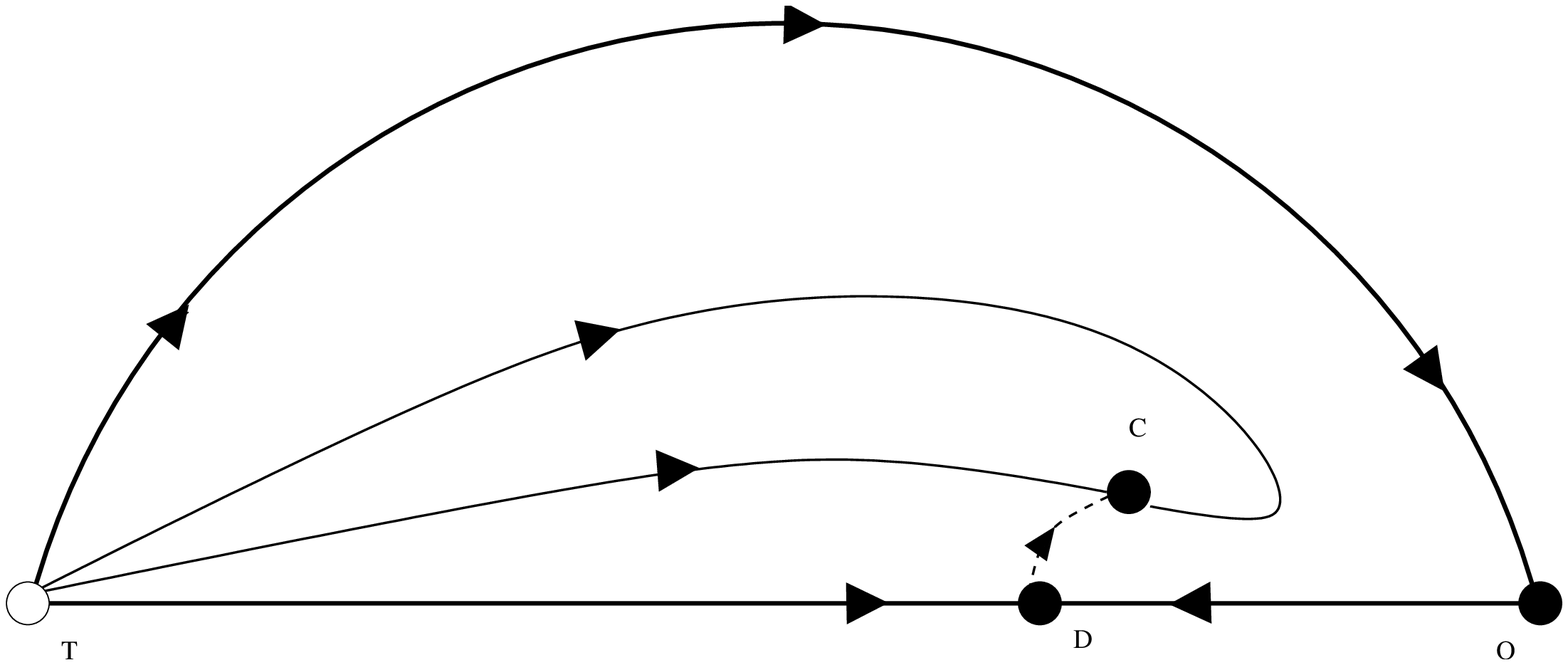}}\qquad
\psfrag{T}[cc][cc][0.7][0]{$\mathrm{T}_\flat$}
\psfrag{Q}[cc][cc][0.7][0]{$\mathrm{Q}_\flat$}
\psfrag{D}[cc][cr][0.7][0]{$\mathrm{R}_\flat$}
\subfigure[\Azeroplus, \Aplus] {\includegraphics[width=0.45\textwidth]{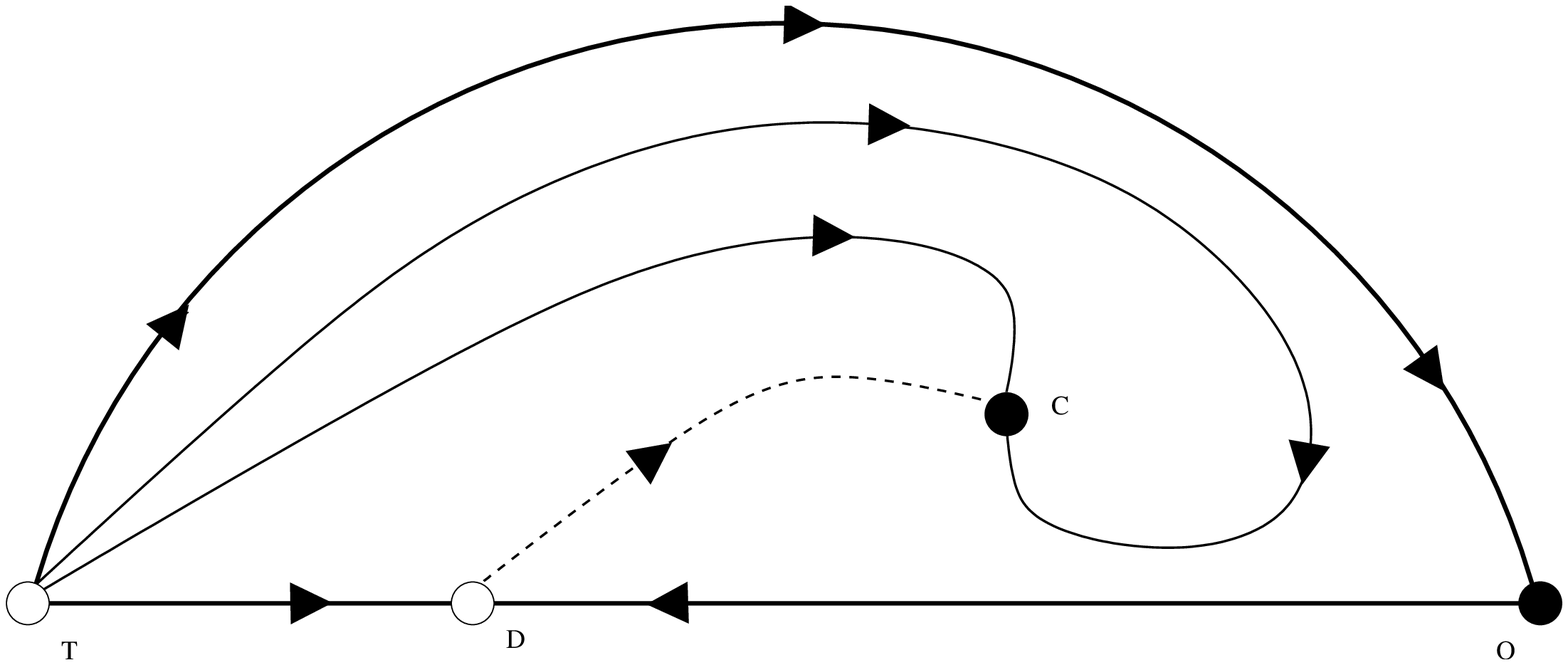}}\\
\psfrag{T}[cc][cc][0.7][0]{$\mathrm{T}_\flat$}
\psfrag{Q}[cc][cc][0.7][0]{$\mathrm{Q}_\flat$}
\psfrag{C}[cc][cc][0.7][0]{$\mathrm{C}_\flat$}
\subfigure[\Bplus, \Cplus, \Dplus]{\includegraphics[width=0.45\textwidth]{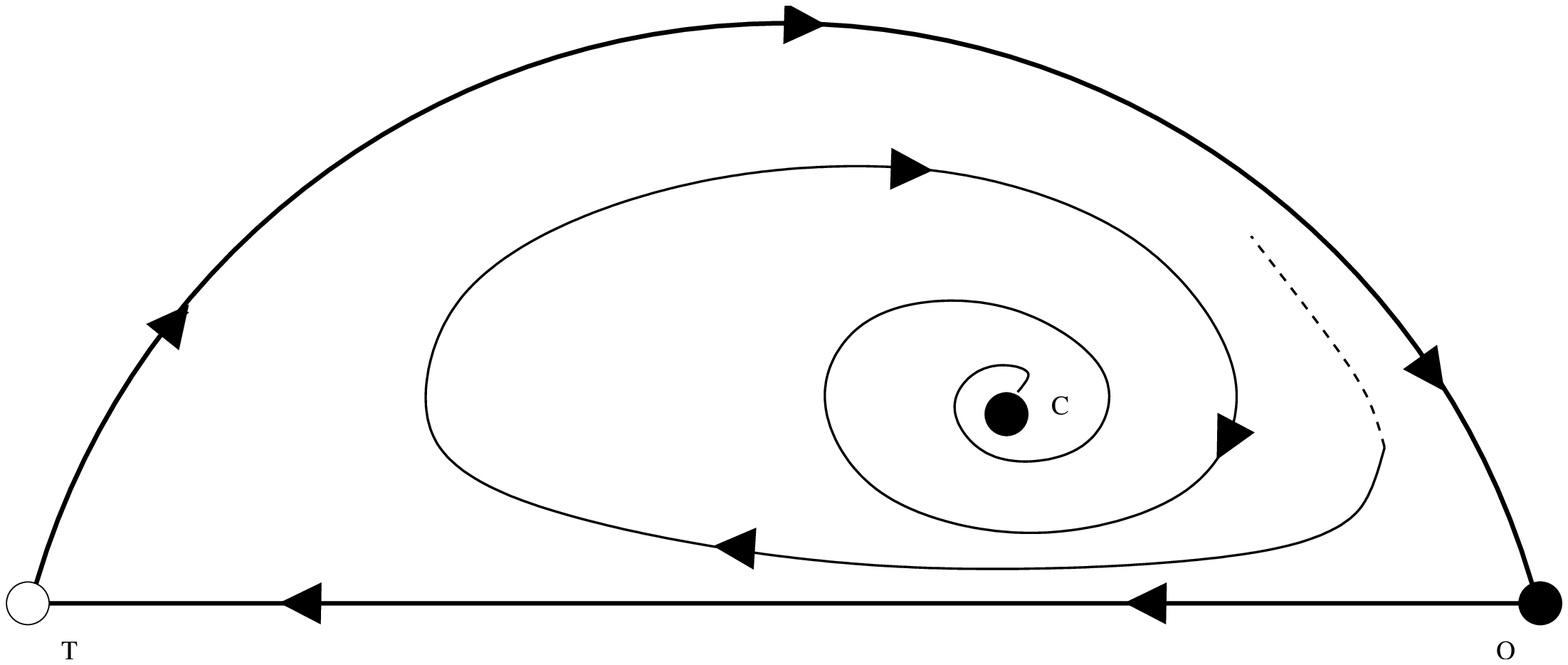}}\qquad
\psfrag{-1}[cc][cc][0.7][0]{$-1$}
\psfrag{+1}[cc][cc][0.7][0]{$+1$}
\psfrag{M1}[cc][cc][0.7][0]{$M_1$}
\psfrag{sigma}[cc][cc][0.7][0]{$\Sigma_+$}
\psfrag{Asharp}[cc][cc][1.2][0]{$\mathcal{S}_\flat$}
\subfigure[Coordinates]{\includegraphics[width=0.45\textwidth]{Asharpcoord.eps}}
\end{center}
\caption{Phase portraits of the flow on $\mathcal{S}_\flat$. The comments of Fig.~\ref{Asharpfig}
apply analogously.}
\label{Aflatfig}
\end{figure}

\textbf{Global dynamics on the type~II state space} $\bm{\mathcal{X}_{\mathrm{II}}}$.
We are now in a position to complete the analysis of Bianchi type~II models:
By collecting the previous results concerning the four boundary
components~\eqref{IIboundarycomps} of the type~II state space $\mathcal{X}_{\mathrm{II}}$ 
and applying the monotonicity principle, 
we get a complete description of the dynamics in the interior of  $\mathcal{X}_\mathrm{II}$
and thus of Bianchi type~II models.

\begin{Theorem}\label{BianchiIItheo}
The qualitative behavior of typical (in fact, generic) orbits of the dynamical 
system~\eqref{dynsysbianchiII}---which represents LRS Bianchi type~II models---in the various anisotropy cases 
is the one sketched in Fig.~\ref{BianchiIIfig}. 
The past and the future attractors for the various cases 
are listed in Table~\ref{attractorsII}. 
There exist non-generic
solutions that are characterized by different asymptotic behavior;
the complete list of possible $\alpha$- and $\omega$-limit sets
of orbits 
is given in Table~\ref{alfaII}.
\end{Theorem}

\begin{proof}
Consider the function $Z_4:\mathcal{X}_{\,\mathrm{II}}\cup\mathcal{V}_\mathrm{II}\to (0,\infty)$ given by
\[
Z_4=M_1^{-3/2}\frac{s}{1-2s}\:,
\]
which satisfies
\[
Z_4'=-3\left(\Sigma_+^2+\textfrac{1}{4}(1+3w)\Omega\right)Z_4 \quad\text{and}\quad
Z_4'''=-\textfrac{2}{3}Z_4M_1^4\: \text{ when }\,(\Sigma_+,\Omega)=(0,0)\:.
\]
It follows that the function $Z_4$ is strictly monotonically 
decreasing on $\mathcal{X}_\mathrm{II}\cup\mathcal{V}_\mathrm{II}$. 
By the monotonicity principle, the $\alpha$-limit set $\alpha(\gamma)$ and the $\omega$-limit 
set $\omega(\gamma)$ of every orbit $\gamma$ in the interior of $\mathcal{X}_\mathrm{II}$ must satisfy
\begin{equation}\label{IIwherealom}
\alpha(\gamma)\subset\overline{\mathcal{S}}_\sharp\cup\overline{\mathcal{X}}_{\mathrm{I}}\:,
\qquad\omega(\gamma)\subset\overline{\mathcal{S}}_\flat\:.
\end{equation}
It merely remains to study the flow induced on these boundary subsets to determine 
the possible $\alpha/\omega$-limit sets. 
This is done by inspection of the phase 
portraits of Figs.~\ref{BianchiIfig},~\ref{Asharpfig} and~\ref{Aflatfig}.

The invariant structures on the boundary subsets~\eqref{IIwherealom} 
that are potential 
$\alpha/\omega$-limit sets 
are fixed points, periodic orbits, and heteroclinic cycles/networks.
From Figs.~\ref{BianchiIfig},~\ref{Asharpfig} and~\ref{Aflatfig} we see that
periodic orbits do not occur, so that the
$\alpha/\omega$-limit sets of orbits in $\mathcal{X}_{\,\mathrm{II}}$
must be fixed points or heteroclinic cycles/networks (if present).
We begin by investigating the fixed points.

Consider, for instance, the interior of $\mathcal{X}_{\,\mathrm{I}}$,
see Fig.~\ref{BianchiIfig}. The only possible limit set in the 
interior of $\mathcal{X}_{\,\mathrm{I}}$ is the fixed point $\mathrm{F}$.
In the $\boldsymbol{-}$ cases, $\mathrm{F}$ is a saddle in $\mathcal{X}_{\,\mathrm{I}}$;
moreover, it is represented by a white circle in the subfigures of Fig.~\ref{BianchiIfig}, hence 
it is a repellor in the orthogonal direction (i.e., into the interior of
$\mathcal{X}_{\,\mathrm{II}}$). We infer that $\mathrm{F}$ is
the $\alpha$-limit set for a one-parameter set of orbits in $\mathcal{X}_{\,\mathrm{II}}$,
which makes $\mathrm{F}$  a possible $\alpha$-limit set, but not a generic repellor.
In the $\boldsymbol{+}$ cases, $\mathrm{F}$ is a sink in $\mathcal{X}_{\,\mathrm{I}}$;
since it is a repellor in the orthogonal direction, there exists
exactly one orbit in $\mathcal{X}_{\,\mathrm{II}}$ that possesses $\mathrm{F}$ as 
an $\alpha$-limit set.

Inspection of the fixed points on $\partial\mathcal{X}_{\,\mathrm{I}}$ yields
only one (generic) source of orbits in $\mathcal{X}_{\,\mathrm{II}}$:
the Taub point $\mathrm{T}_\flat$. This is in the $\boldsymbol{-}$ cases
and in the case \Aplus; see Figs.~\ref{biDmin}--\ref{biAplu}, where $\mathrm{T}_\flat$ is
a repellor within $\overline{\mathcal{X}}_{\mathrm{I}}$ and color-coded as
a repellor in the orthogonal direction.
The sinks we find are $\mathrm{Q}_\flat$ in the cases \Bminus, \Cminus, \Dminus\
and 
$\mathrm{R}_\flat$ in the subcase ${-1}<\beta\leq \beta_\flat(w)$ of \Aminus;
in the remaining cases there does not exist any sink on $\overline{\mathcal{X}}_{\mathrm{I}}$.

Turning to $\overline{\mathcal{S}}_\sharp$ we do not find any sources or
sinks, see Fig.~\ref{Asharpfig}. 
On $\overline{\mathcal{S}}_\flat$ there is the Taub point $\mathrm{T}_\flat$, which 
is a source in the $\boldsymbol{-}$ cases and in the case \Aplus, which is consistent
with what we already know.
The sinks are $\mathrm{Q}_\flat$ in the cases \Bminus, \Cminus, \Dminus,
$\mathrm{R}_\flat$ in the subcase ${-1}<\beta\leq \beta_\flat(w)$ of \Aminus,
and $\mathrm{C}_\flat$ in the subcase $\beta_\flat(w)< \beta < 0$ of \Aminus\
and in the $\boldsymbol{+}$ cases.

While in the cases \Apm, \Bminus, and \Cminus, fixed points are 
the only possible $\alpha$- and $\omega$-limit sets of 
interior orbits in 
$\mathcal{X}_{\,\mathrm{II}}$, in the cases \Dminus, \Bplus, \Cplus, and \Dplus,
there exist heteroclinic cycles on the boundary which may 
belong to the past/future attractor. Let us study these cases in more detail.

In the cases \Bplus, \Cplus, \Dplus, the
boundaries of the boundary subsets $\mathcal{S}_\sharp$, $\mathcal{X}_{\,\mathrm{I}}$,
and $\mathcal{S}_\flat$ form a heteroclinic network, which is a set of
entangled heteroclinic cycles, see Figs.~\ref{IIBCplus} and~\ref{IIDplus}.
In particular, in the cases \Bplus\ and \Cplus, 
the boundary of the set $\overline{\mathcal{S}}_\sharp\cup\overline{\mathcal{X}}_{\mathrm{I}}$,
where the $\alpha$-limit set of generic orbits must be located by~\eqref{IIwherealom},
is a heteroclinic cycle, which is of the type~\eqref{heteroI};
see the bold dashed lines in Fig.~\ref{IIBCplus}.
This heteroclinic cycle remains as the only possible $\alpha$-limit
set of (generic) orbits in $\mathcal{X}_{\,\mathrm{II}}$.
In the case \Dplus, 
the set $\overline{\mathcal{S}}_\sharp\cup\partial\mathcal{X}_{\mathrm{I}}$
consists of a heteroclinic network, cf.~Fig.~\ref{IIDplus}; the structure of this network 
is that of~\eqref{heteroI}, where there exists a one-parameter family of possible
paths from $\mathrm{T}_\sharp$ to $\mathrm{Q}_\sharp$
(since every orbit in $\overline{\mathcal{S}}_\sharp$ connects $\mathrm{T}_\sharp$ with
$\mathrm{Q}_\sharp$, see Fig.~\ref{SsharpD+}).
It is not clear which part of the heteroclinic 
network acts as the actual $\alpha$-limit set of orbits in $\mathcal{X}_{\,\mathrm{II}}$.
In any case, it is immediate that the 
approach toward the singularity is oscillatory for generic solutions in $\mathcal{X}_{\,\mathrm{II}}$.

To investigate the
future asymptotics in the cases \Bplus, \Cplus, \Dplus, we note 
that the boundary of $\mathcal{S}_\flat$ forms a heteroclinic cycle 
connecting the fixed points $\mathrm{T}_\flat$ and $\mathrm{Q}_\flat$. 
We claim that there are no interior orbits that have this heteroclinic 
cycle as the $\omega$-limit set, which leaves 
the fixed point $\mathrm{C}_\flat$ as the only possible $\omega$-limit 
for all interior orbits. To prove this claim, 
consider again the function $Z_3$ given by~\eqref{z3}; since $c<1$ for 
$\beta\geq 1$, we may regard $Z_3$ as a function that is defined on 
the entire state space $\mathcal{X}_{\,\mathrm{II}}$ (except at $\Sigma_+ = c$). 
Moreover,  since $\delta\in (0,1)$ 
for $\beta\geq 1$, $Z_3$ vanishes on the boundary of $\mathcal{S}_\flat$. 
Therefore, to prove our claim, it is enough to show that along any interior orbit, 
$Z_3$ would be increasing if the orbit approached some point on the 
boundary of $\mathcal{S}_\flat$. In fact, by a direct computation one can verify that
\begin{equation}\label{dz3}
\left[\textfrac{d}{d\tau}\log Z_3\right]_{\,|s=0}=
\textfrac{4(2+\beta)(1+3w)(\Sigma_+-c)^2}{3c(1-w)(1-\Sigma_+c)(\beta+\beta_1)(\beta+\beta_2)}\:,
\end{equation}
where
\[
\beta_1=\textfrac{7-3w}{3(1-w)}<\beta_2=\textfrac{3+w}{1-w}\:.
\]
Since $\beta_1,\beta_2>2$, for $w\in(-\textfrac{1}{3}, 1)$, the r.h.s. of~\eqref{dz3} is 
positive (unless $\Sigma_+=c$), which establishes the claim.

Finally, 
in the case \Dminus, there exists a family of heteroclinic cycles
containing the fixed points $\mathrm{T}_\sharp$ and $\mathrm{Q}_\sharp$, 
see Fig.~\ref{IIDminus}. 
The orbit that connects $\mathrm{T}_\sharp$ with $\mathrm{Q}_\sharp$ is the semi-circle of $\partial\mathcal{S}_\sharp$, 
whereas the orbits connecting $\mathrm{Q}_\sharp$ with $\mathrm{T}_\sharp$ 
lie on the Bianchi type~I boundary $\mathcal{X}_{\,\mathrm{I}}$, cf.~Fig.~\ref{biDmin}.
To show that there do not exist orbits in $\mathcal{X}_{\,\mathrm{II}}$ 
that have any of these heteroclinic cycles as the $\alpha$-limit set a rather technical analysis
is required; we merely outline the argument here.
Assume that there exists an orbit $\gamma \subset \mathcal{X}_{\,\mathrm{II}}$ 
such that $\alpha(\gamma) = \partial \mathcal{S}_\sharp$ (or any other of the cycles).
By assumption, there exists a sequence of intervals $(\tau_{n+1},\tau_n)$, $n\in \mathbb{N}$, 
such that the solution $\gamma(\tau)$ is in a neighborhood $\Sigma_+ \in (-1,-1+\epsilon)$ 
(and $s > \textfrac{1}{2} - \epsilon^\prime$) of $\mathrm{T}_\sharp$, i.e., $-1<\Sigma_+(\tau)<-1+\epsilon$ for
$\tau_{n+1} < \tau < \tau_n$, $n$ even,
and in a neighborhood $\Sigma_+ \in (1-\epsilon,1)$ of $\mathrm{Q}_\sharp$, 
i.e., $1-\epsilon < \Sigma_+ <1$ for $\tau_{n+1} < \tau < \tau_n$, $n$ odd.
In the limit $\tau\rightarrow -\infty$, eq.~\eqref{IIs} decouples from~\eqref{IIsig+} and~\eqref{IIM1},
and we are able to compute the times $\Delta\tau_n = \tau_n - \tau_{n+1}$
the solution spends in the neighborhood of 
$\mathrm{T}_\sharp$ and $\mathrm{Q}_\sharp$, respectively.
We find that $\Delta\tau_n$ is increasing in such a way that $\Delta\tau_n$ is strictly larger
than the $\Delta\tau_{n-1}$; more specifically, there exists $c > 1$ 
(which is independent of $\epsilon$) such that
$\Delta\tau_n > c \,\Delta\tau_{n-1}$.
Considering~\eqref{IIs} we see that $\textfrac{1}{2} - s$ increases (in the past direction) at least like
$e^{6 (1-\epsilon) \tau}$ for $\tau\in (\tau_{n+1},\tau_n)$, $n$ even,
i.e., while the solution is in a neighborhood of $\mathrm{T}_\sharp$.
Conversely, $\textfrac{1}{2} - s$ decreases (toward the past) at most like
$e^{-6 \tau}$ for $\tau\in (\tau_{n+1},\tau_n)$, $n$ odd,
i.e., while the solution is in a neighborhood of $\mathrm{Q}_\sharp$.
Since, in the limit $\tau\rightarrow -\infty$, 
the time the solution spends away from the fixed points becomes
negligible as compared to the times $\Delta\tau_n$, 
the behavior of $s-\textfrac{1}{2}$ 
is represented by the alternating sequence
\[
\textfrac{1}{2} - s \sim 
\exp\Big[6 \Big( (1-\epsilon) \Delta\tau_0 - \Delta\tau_1 +(1-\epsilon) \Delta\tau_2 - \Delta\tau_3 + 
(1-\epsilon)\Delta\tau_4 +\cdots\Big)\Big] \:.
\]
Since $\Delta\tau_n > (1-\epsilon)^{-1} \Delta\tau_{n-1}$, $n$ even,
this behavior contradicts the assumption that $\alpha(\gamma) = \partial \mathcal{S}_\sharp$,
since $\textfrac{1}{2} -s \not\rightarrow 0$.
This establishes the claims of the theorem.
\end{proof}

\textit{Interpretation of Theorem~\ref{BianchiIItheo}}. 
In the $\boldsymbol{+}$ cases and in the subcase $\beta_\flat(w)<\beta< 0$
of \Aminus, the future asymptotic behavior of LRS Bianchi type~II solutions is
governed by the approach to the solution represented by the fixed point $\mathrm{C}_\flat$.
In Appendix~\ref{exact} we compute this solution explicitly, see~\eqref{cfsol},
and demonstrate that it constitutes a natural generalization of 
the Collins-Stewart perfect fluid solution~\eqref{CSsol}.
(Recall that for type~II perfect fluid cosmologies, the Collins-Stewart solution is
the future attractor.)
In the subcase $-1 < \beta \leq \beta_\flat(w)$ of \Aminus, generic type~II solutions
are asymptotic to the solution represented by $\mathrm{R}_\flat$, which is
given as~\eqref{rfsol} in Appendix~\ref{exact}; note that this is a non-vacuum
solution because $\Omega \neq 0$.
Finally, in the cases \Bminus, \Cminus, \Dminus, every generic solution 
is future asymptotic to the non-flat LRS Kasner solution~\eqref{solQ}, which is a vacuum solution.
The past asymptotic behavior of anisotropic type~II solutions resembles the 
behavior of perfect fluid solutions in the $\boldsymbol{-}$ cases and
in the case \Aplus: Towards the singularity, solutions behave
asymptotically like the Taub solution. ``Matter does not matter'' in these cases, since
the behavior coincides with that of vacuum solutions. 
However, in the cases \Bplus, \Cplus, \Dplus, 
\textit{anisotropic matter matters}. (Note that \Bplus\ is consistent
with the energy conditions.) 
The approach to the singularity is oscillatory in these cases; the solution oscillates between
the Taub family~\eqref{taub} and the non-flat LRS family~\eqref{solQ}.
Since $\Omega\not\rightarrow 0$ for these solutions, this is an example of asymptotic
behavior that differs from the behavior of vacuum solutions.

\begin{Remark}
An example of a type \Bplus~matter model is collisionless (Vlasov)
matter for zero mass particles. In this special case, the existence of 
oscillations toward the past singularity for LRS Bianchi type~II solutions 
was already observed in~\cite{RT}. We refer to~\cite{letter} for an extension of the analysis in~\cite{RT}.   
\end{Remark}

\begin{figure}[Ht!]
\begin{center}
\psfrag{tb}[cc][cc][0.7][0]{$\mathrm{T}_\flat$}
\psfrag{ts}[cc][cl][0.7][0]{$\mathrm{T}_\sharp$}
\psfrag{qb}[cc][cc][0.7][0]{$\mathrm{Q}_\flat$}
\psfrag{qs}[cc][cr][0.7][0]{$\mathrm{Q}_\sharp$}
\psfrag{f}[cc][cr][0.7][0]{$\mathrm{F}$}
\psfrag{cs}[cc][cc][0.7][0]{$\mathrm{C}_\sharp$}
\psfrag{d}[cc][cr][0.7][0]{$\mathrm{R}_\sharp$}
\psfrag{r}[cc][cc][0.7][0]{$\mathrm{R}_\flat$}
\psfrag{cb}[cc][cc][0.7][0]{$\mathrm{C}_\flat$}
\subfigure[\Dminus]{\label{IIDminus}\includegraphics[width=0.35\textwidth]{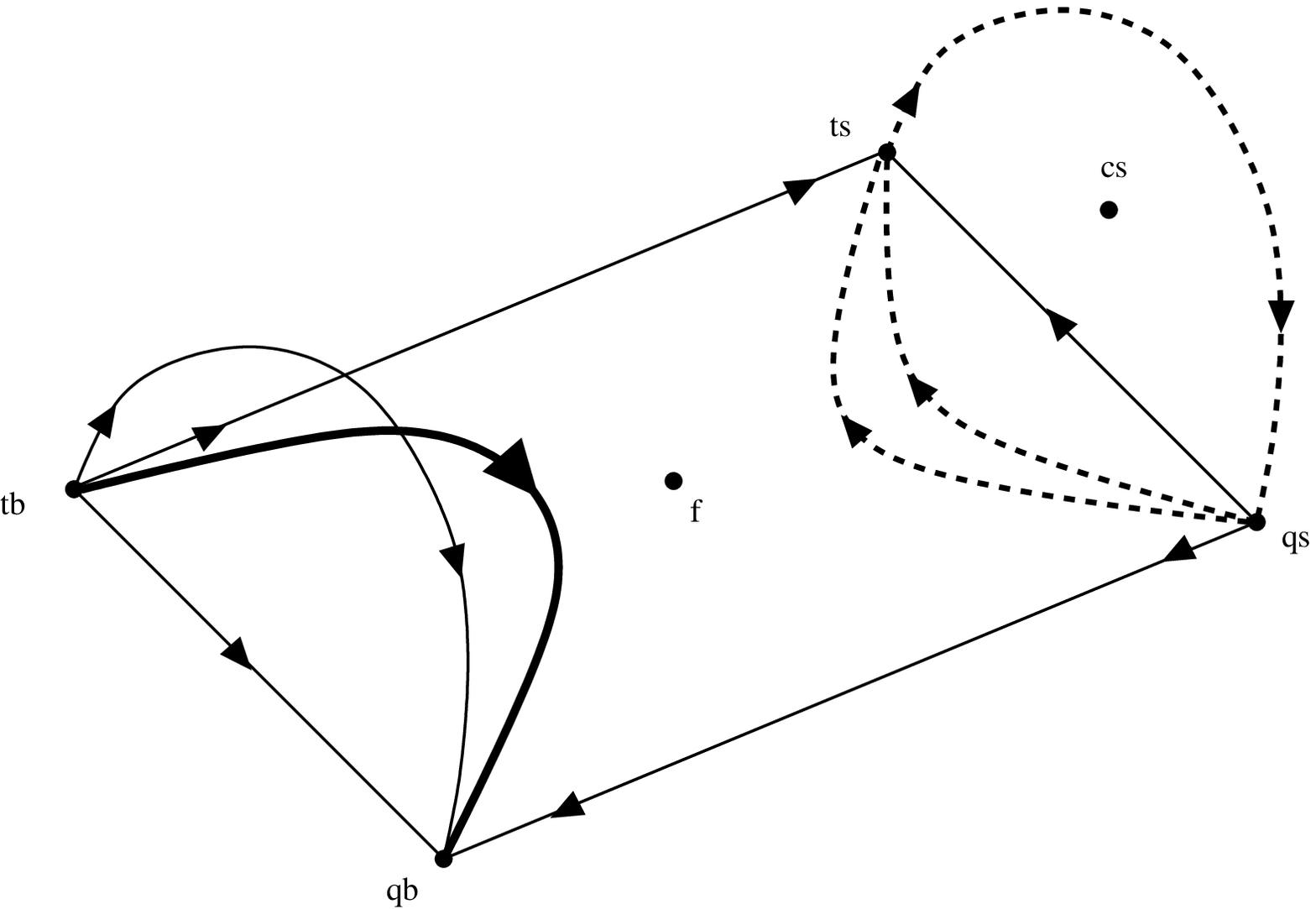}}\qquad\quad
\psfrag{T}[cc][cc][0.7][0]{$\mathrm{T}_\flat$}
\psfrag{T*}[cc][cc][0.7][0]{$\mathrm{T}_\sharp$}
\psfrag{Q}[cc][cc][0.7][0]{$\mathrm{Q}_\flat$}
\psfrag{Q*}[cc][cc][0.7][0]{$\mathrm{Q}_\sharp$}
\psfrag{D*}[cc][cc][0.7][0]{$\mathrm{R}_\sharp$}
\psfrag{F}[cc][cc][0.7][0]{$\mathrm{F}$}
\subfigure[\Bminus, \Cminus]{\includegraphics[width=0.35\textwidth]{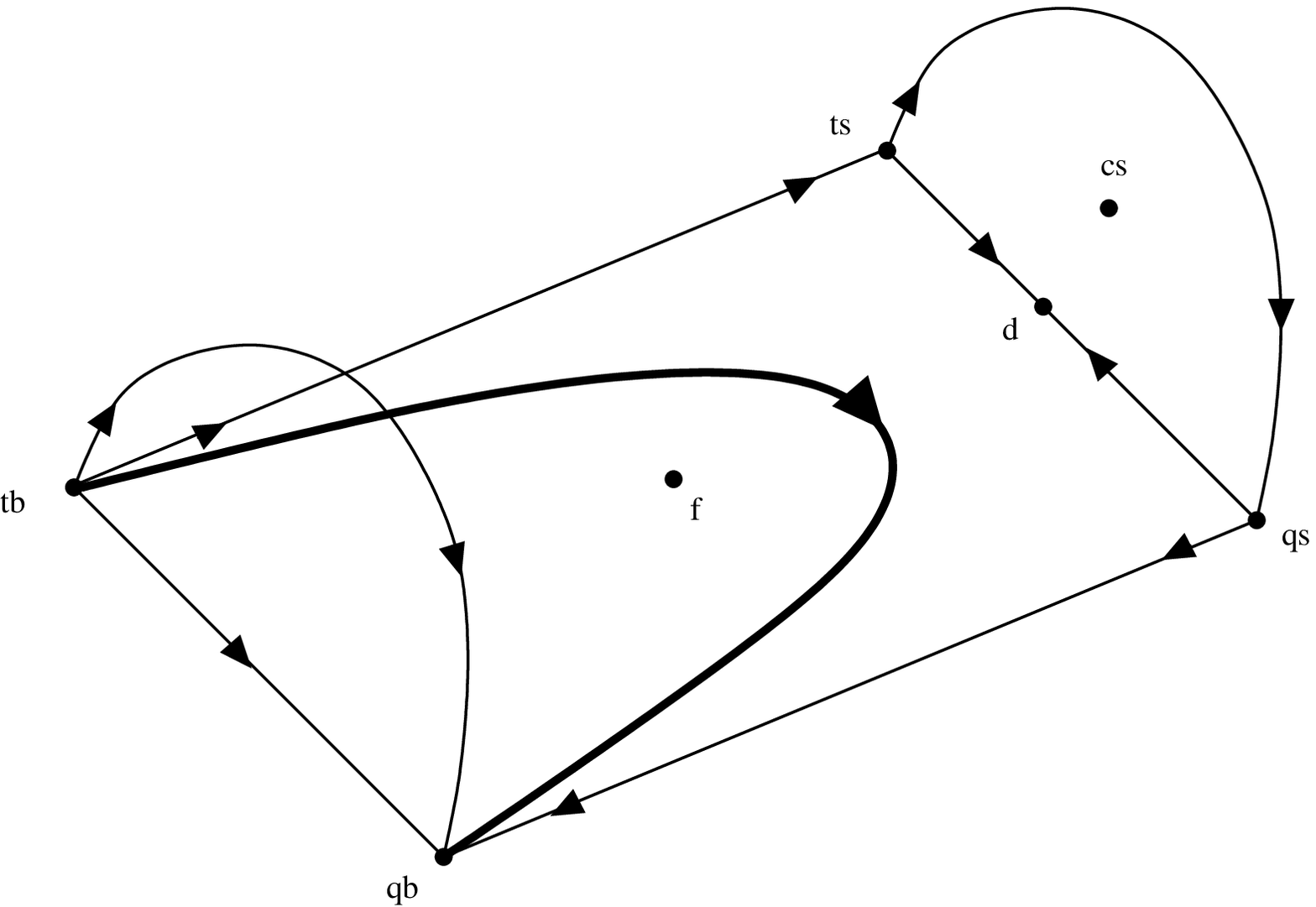}}\\
\psfrag{T}[cc][cc][0.7][0]{$\mathrm{T}_\flat$}
\psfrag{T*}[cc][cc][0.7][0]{$\mathrm{T}_\sharp$}
\psfrag{Q}[cc][cc][0.7][0]{$\mathrm{Q}_\flat$}
\psfrag{Q*}[cc][cc][0.7][0]{$\mathrm{Q}_\sharp$}
\psfrag{R}[cc][cc][0.7][0]{$\mathrm{R}_\flat$}
\psfrag{D*}[cc][cc][0.7][0]{$\mathrm{R}_\sharp$}
\psfrag{F}[cc][cc][0.7][0]{$\mathrm{F}$}
\subfigure[\Aminus~$(-1<\beta\leq\beta_\flat)$]{\includegraphics[width=0.35\textwidth]{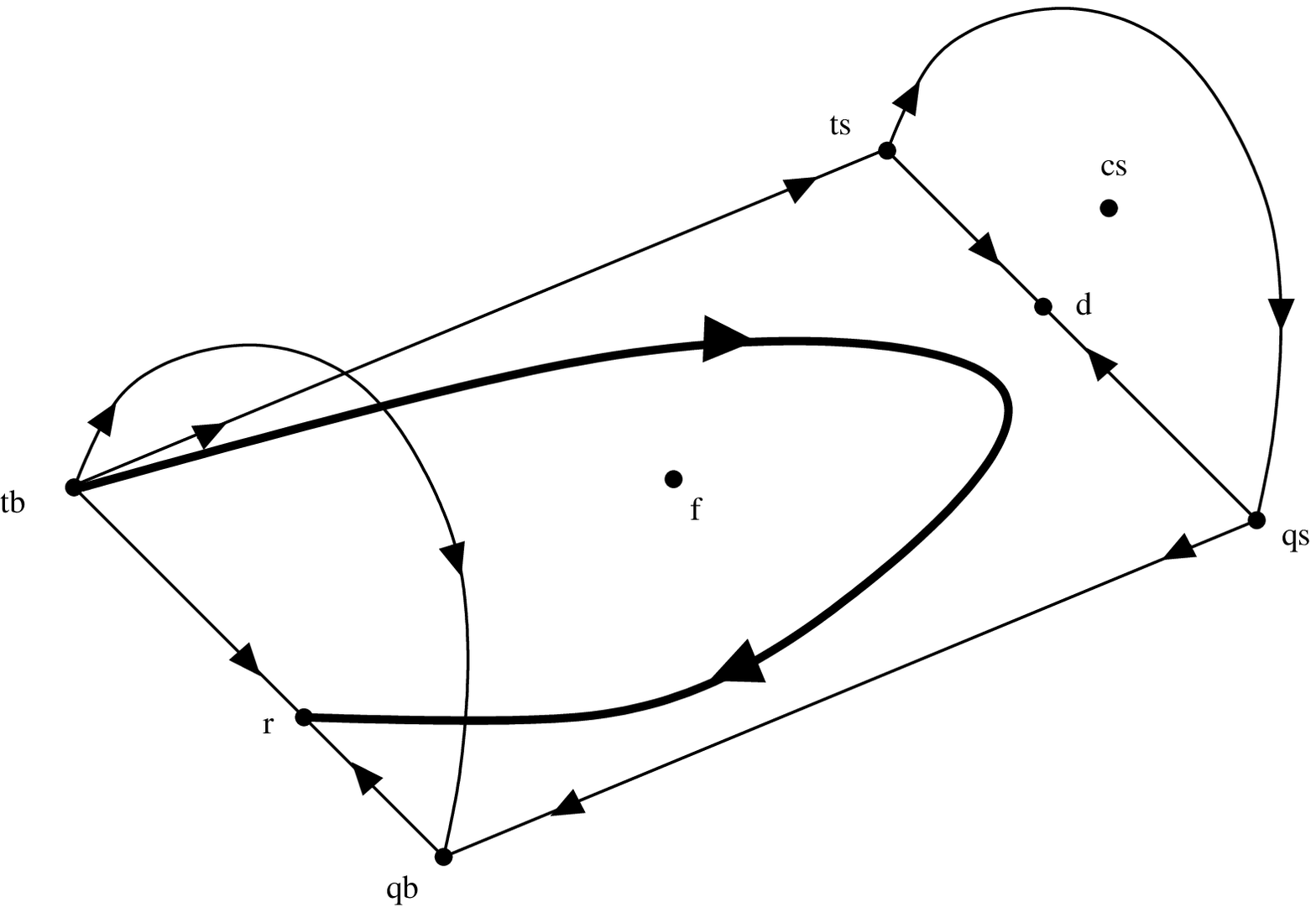}}\qquad\quad
\psfrag{T}[cc][cc][0.7][0]{$\mathrm{T}_\flat$}
\psfrag{T*}[cc][cc][0.7][0]{$\mathrm{T}_\sharp$}
\psfrag{Q}[cc][cc][0.7][0]{$\mathrm{Q}_\flat$}
\psfrag{Q*}[cc][cc][0.7][0]{$\mathrm{Q}_\sharp$}
\psfrag{R}[cc][cc][0.7][0]{$\mathrm{R}_\flat$}
\psfrag{D*}[cc][cc][0.7][0]{$\mathrm{R}_\sharp$}
\psfrag{F}[cc][cc][0.7][0]{$\mathrm{F}$}
\subfigure[\Azerominus, \Aminus $(\beta_\flat<\beta<0)$]{\includegraphics[width=0.35\textwidth]{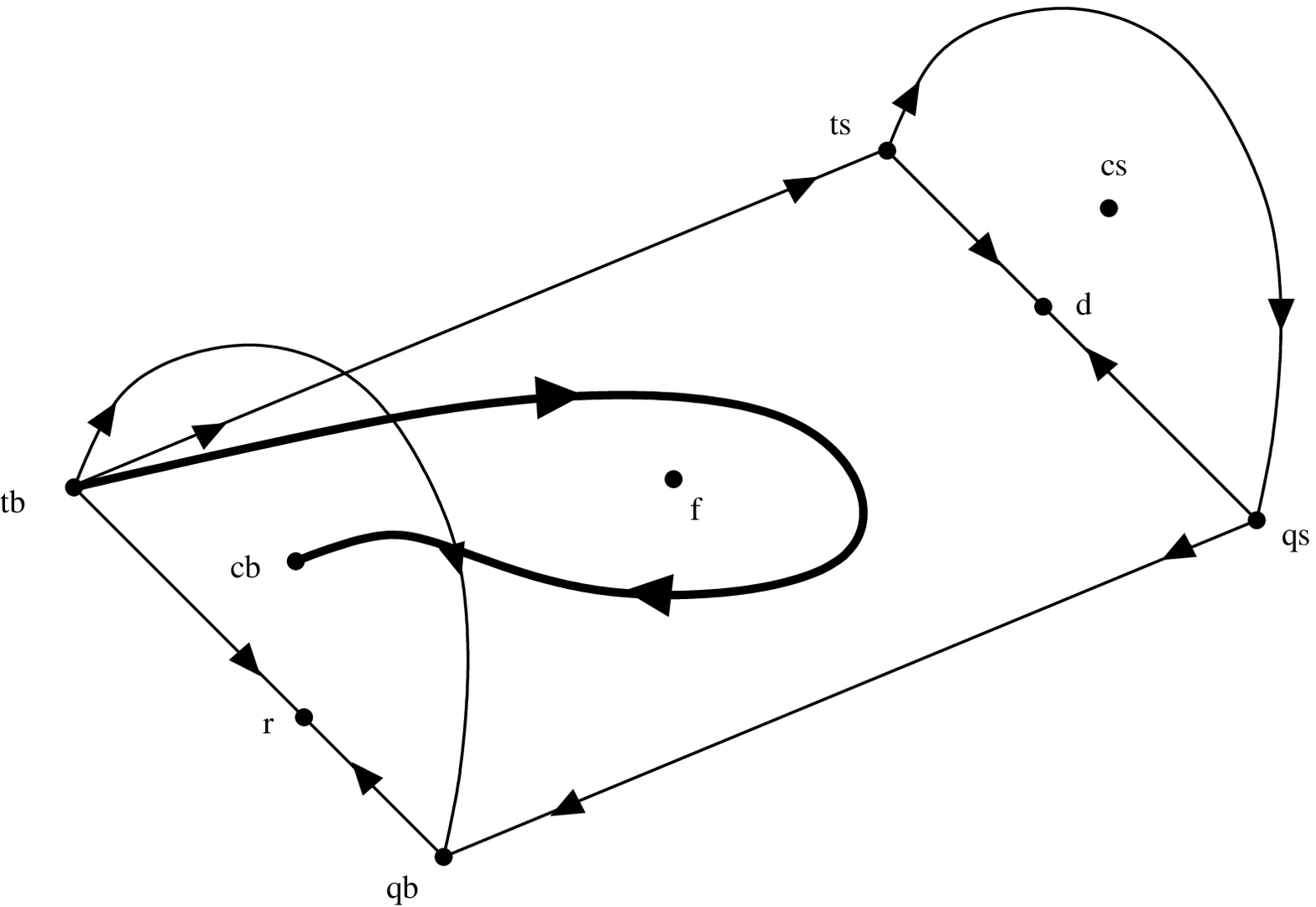}}\\
\psfrag{T}[cc][cc][0.7][0]{$\mathrm{T}_\flat$}
\psfrag{T*}[cc][cc][0.7][0]{$\mathrm{T}_\sharp$}
\psfrag{Q}[cc][cc][0.7][0]{$\mathrm{Q}_\flat$}
\psfrag{Q*}[cc][cc][0.7][0]{$\mathrm{Q}_\sharp$}
\psfrag{R}[cc][cc][0.7][0]{$\mathrm{R}_\flat$}
\psfrag{D*}[cc][cc][0.7][0]{$\mathrm{R}_\sharp$}
\psfrag{F}[cc][cc][0.7][0]{$\mathrm{F}$}
\subfigure[\Azeroplus, \Aplus~$(0<\beta<\beta_\sharp)$]{\includegraphics[width=0.35\textwidth]{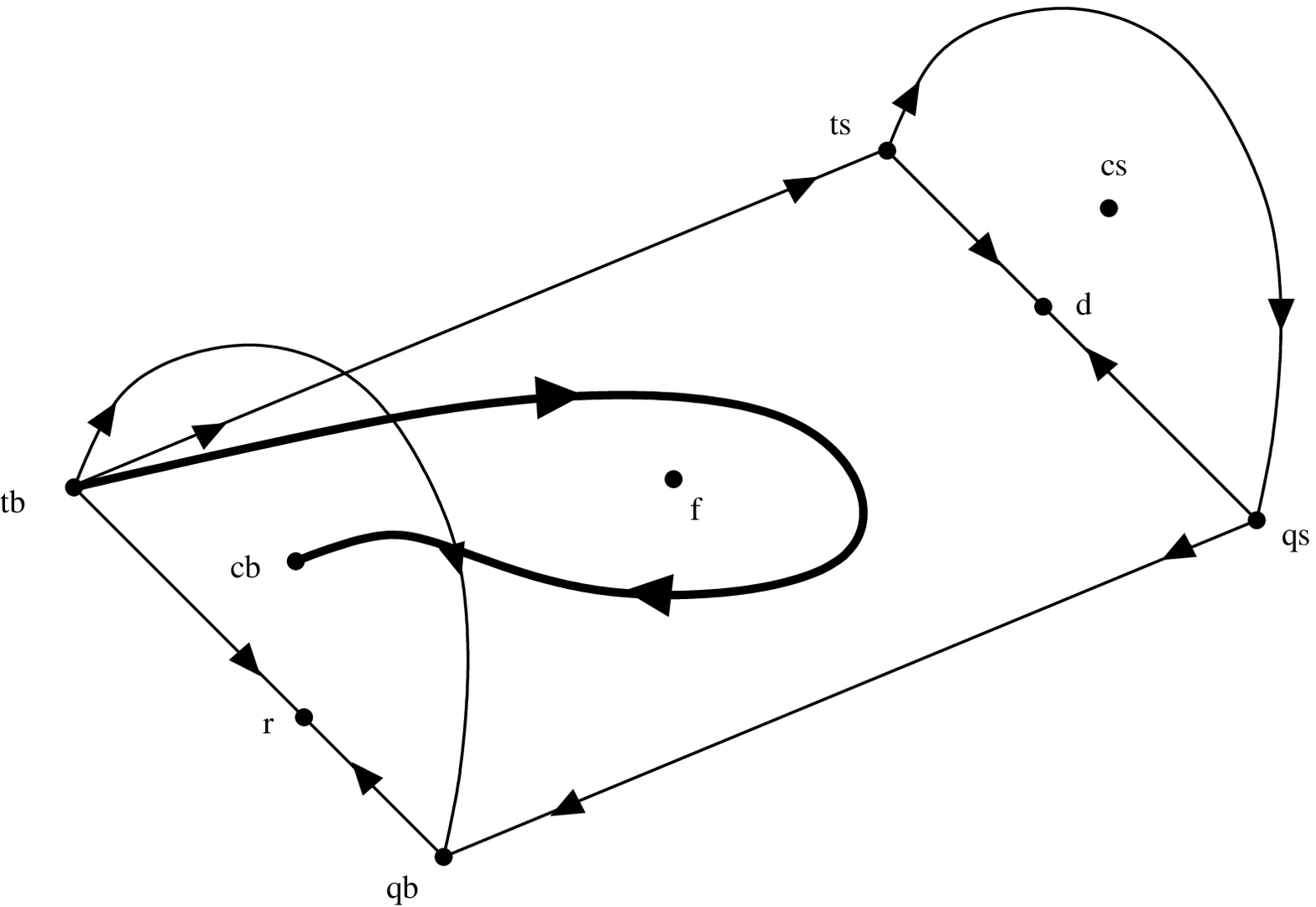}}\qquad\quad
\psfrag{T}[cc][cc][0.7][0]{$\mathrm{T}_\flat$}
\psfrag{T*}[cc][cc][0.7][0]{$\mathrm{T}_\sharp$}
\psfrag{Q}[cc][cc][0.7][0]{$\mathrm{Q}_\flat$}
\psfrag{Q*}[cc][cc][0.7][0]{$\mathrm{Q}_\sharp$}
\psfrag{R}[cc][cc][0.7][0]{$\mathrm{R}_\flat$}
\psfrag{D*}[cc][cc][0.7][0]{$\mathrm{R}_\sharp$}
\psfrag{F}[cc][cc][0.7][0]{$\mathrm{F}$}
\subfigure[\Aplus$~(\beta_\sharp\leq\beta<1)$]{\includegraphics[width=0.35\textwidth]{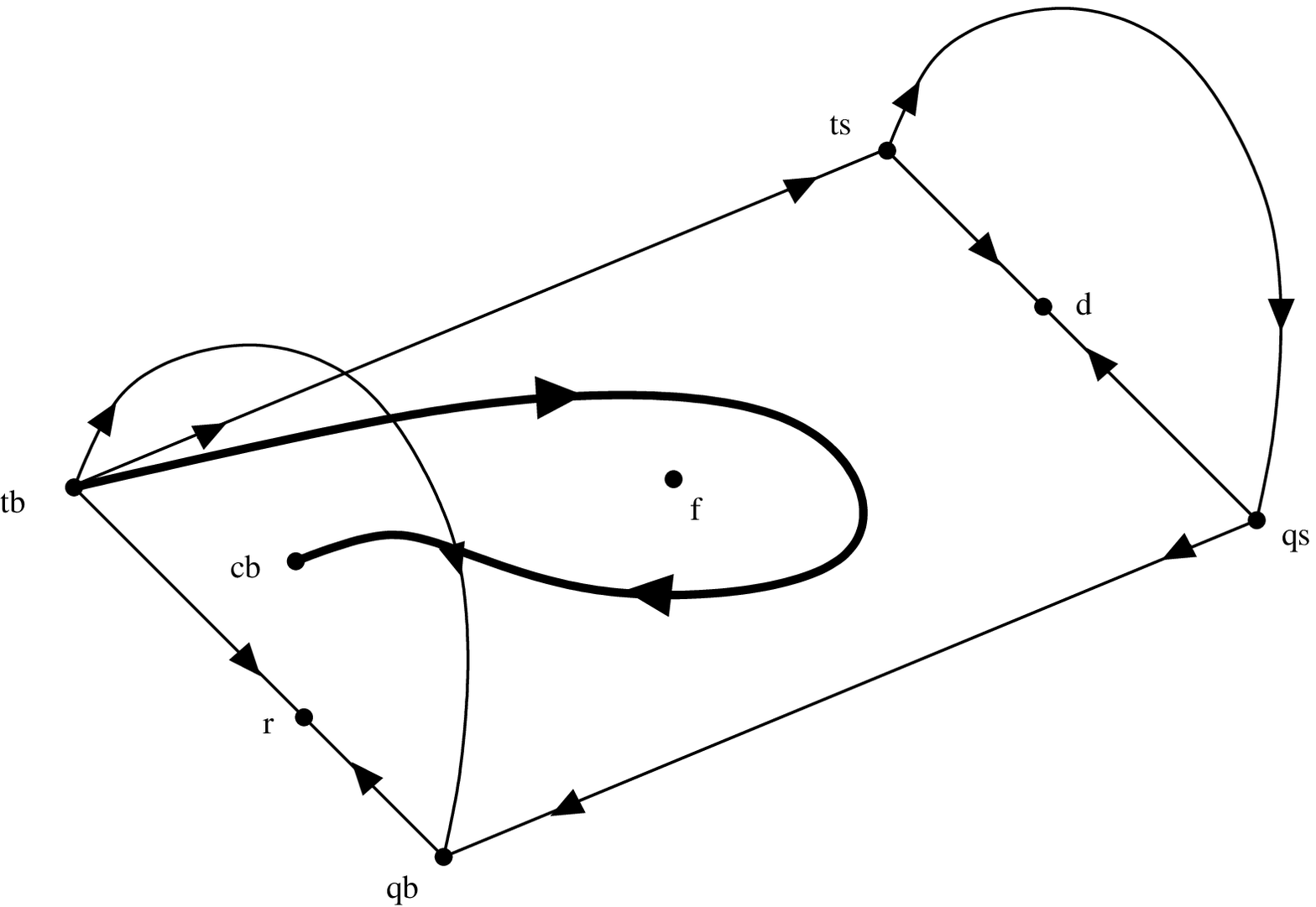}}\\
\psfrag{T}[cc][cc][0.7][0]{$\mathrm{T}_\flat$}
\psfrag{T*}[cc][cc][0.7][0]{$\mathrm{T}_\sharp$}
\psfrag{Q}[cc][cc][0.7][0]{$\mathrm{Q}_\flat$}
\psfrag{Q*}[cc][cc][0.7][0]{$\mathrm{Q}_\sharp$}
\psfrag{D*}[cc][cc][0.7][0]{$\mathrm{R}_\sharp$}
\psfrag{F}[cc][cc][0.7][0]{$\mathrm{F}$}
\subfigure[\Bplus, \Cplus]{\label{IIBCplus}\includegraphics[width=0.35\textwidth]{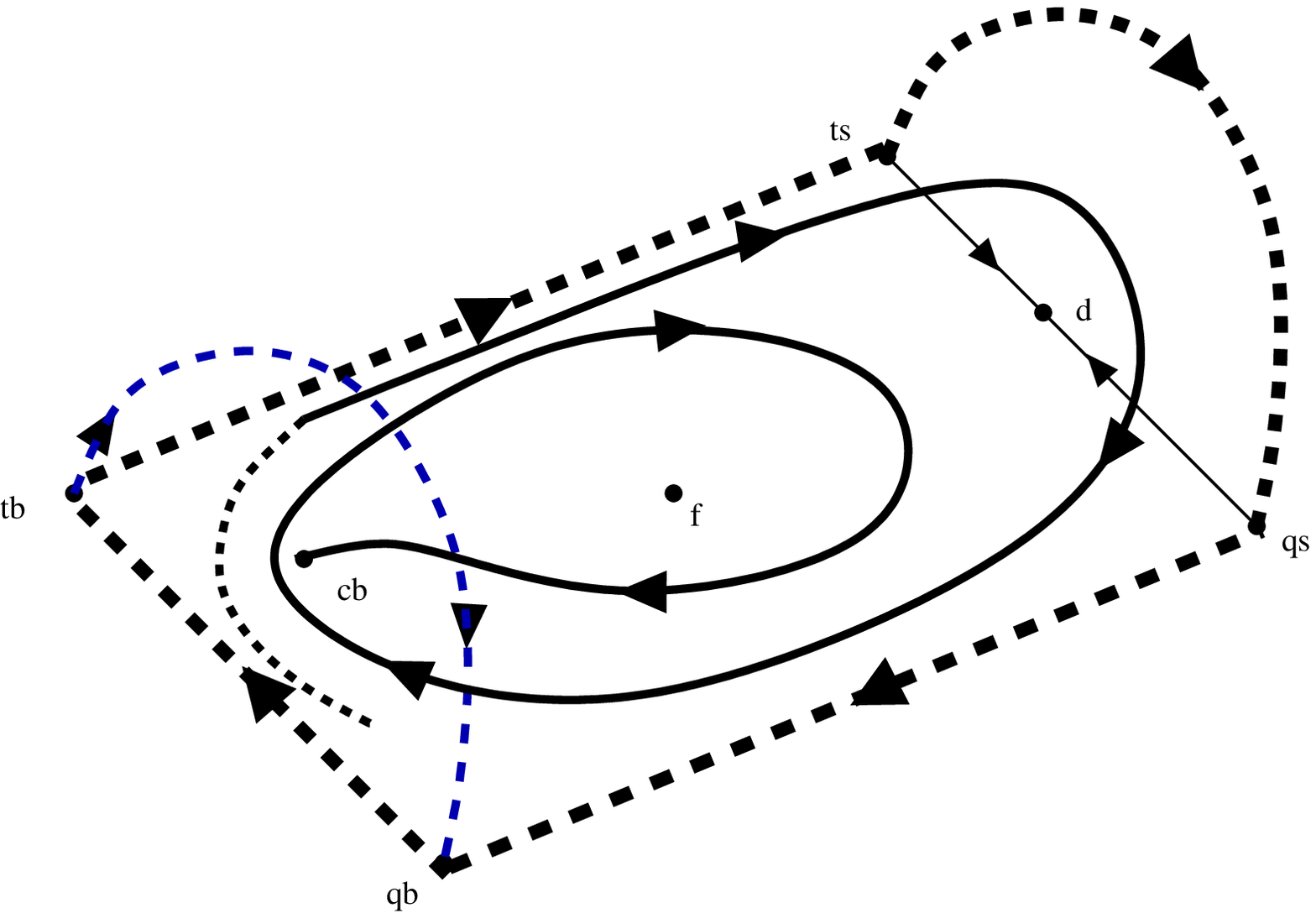}\label{VlasovII}}\qquad\quad
\psfrag{T}[cc][cc][0.7][0]{$\mathrm{T}_\flat$}
\psfrag{T*}[cc][cc][0.7][0]{$\mathrm{T}_\sharp$}
\psfrag{Q}[cc][cc][0.7][0]{$\mathrm{Q}_\flat$}
\psfrag{Q*}[cc][cc][0.7][0]{$\mathrm{Q}_\sharp$}
\psfrag{D*}[cc][cc][0.7][0]{$\mathrm{R}_\sharp$}
\psfrag{F}[cc][cc][0.7][0]{$\mathrm{F}$}
\subfigure[\Dplus]{\label{IIDplus}\includegraphics[width=0.35\textwidth]{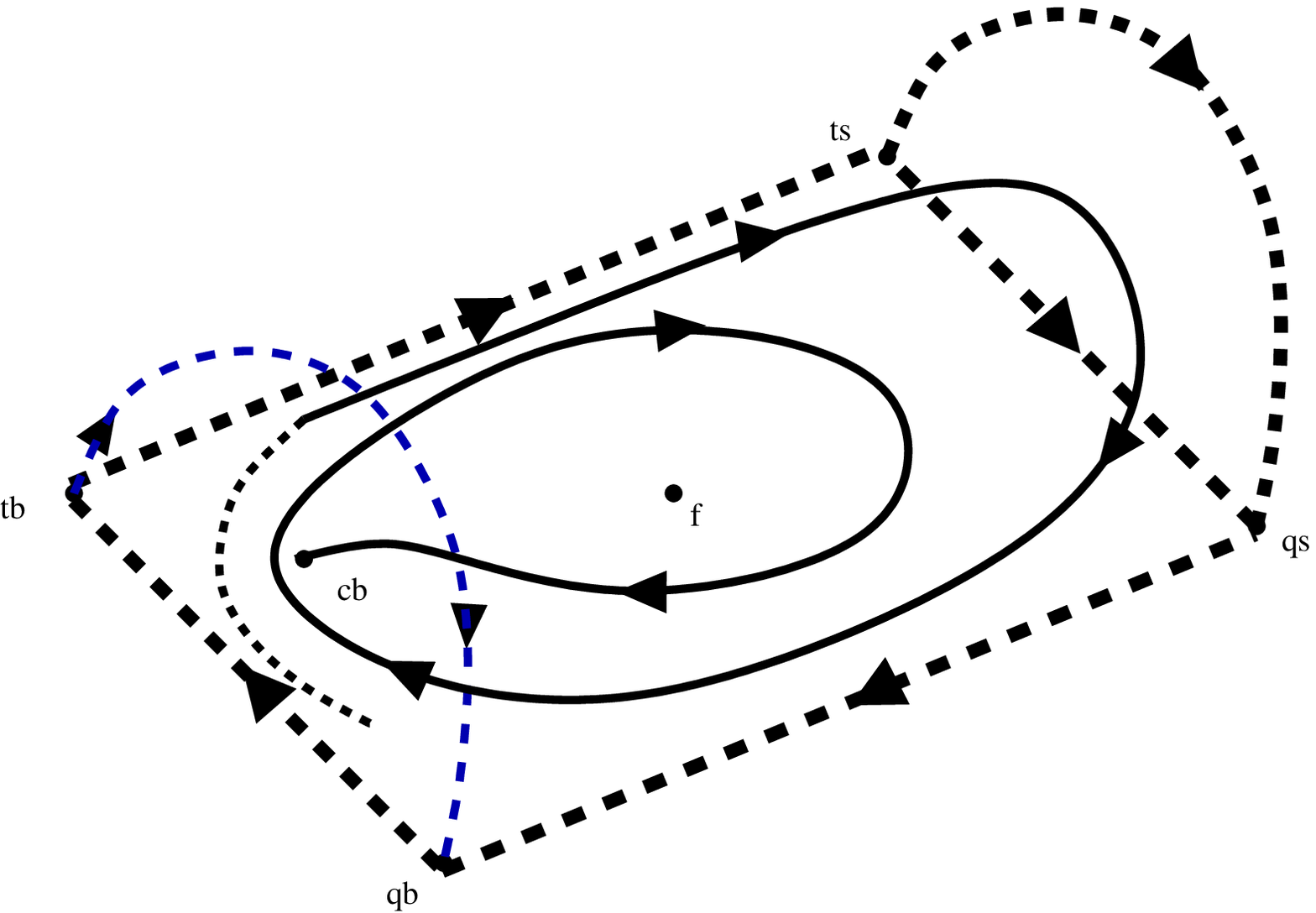}}
\end{center}
\caption{Phase portraits of Bianchi type~II orbits. 
Bold lines depict the behavior of typical orbits. 
Dashed lines are orbits which lie on the boundary of the state space 
and form heteroclinic cycles/networks. 
In the cases \Bplus, \Cplus, \Dplus, the heteroclinic cycle/network 
depicted with bold dashed lines is the past attractor of interior orbits. 
In all other cases, the past attractor is the fixed point $\mathrm{T}_\flat$. 
The future attractor is a fixed point in all cases. There exist non-generic orbits whose
behavior toward the past or toward the future is different; 
see Table~\ref{alfaII}.}
\label{BianchiIIfig}
\end{figure}

\begin{table}
\begin{center}
\begin{tabular}{|c|c|c|}
\hline  & &   \\[-2ex]
Type  & Past Attractor & Future Attractor \\ \hline  & &  \\[-2ex] 
\Bminus, \Cminus, \Dminus & $\mathrm{T}_\flat$ & $\mathrm{Q}_\flat$\\
\Aminus {\scriptsize ($\beta \leq \beta_\flat$)}  & $\mathrm{T}_\flat$ & $\mathrm{R}_\flat$\\
\Aminus {\scriptsize($\beta_\flat<\beta$)}, \Azero, \Aplus &  $\mathrm{T}_\flat$ & $\mathrm{C}_\flat$\\
\Bplus, \Cplus & Heteroclinic cycle & $\mathrm{C}_\flat$\\
\Dplus & Heteroclinic cycle/network & $\mathrm{C}_\flat$\\
\hline 
\end{tabular}
\caption{Past and future attractor of orbits in $\mathcal{X}_{\,\mathrm{II}}$, which
represent LRS Bianchi type~II cosmological models.
The exact solutions represented by the fixed points are given in Sec.~\ref{perfectfluid} and 
in Appendix~\ref{exact}.}
\label{attractorsII}
\end{center}
\end{table}

\begin{table}
\begin{center}
\begin{tabular}{|c|c|c|c|c|c|c|}
\hline 
\multirow{2}{*}{Fixed point} & \multicolumn{6}{|c|}{Dimension of unstable [{\small stable}] 
manifold intersected with $\mathcal{X}_{\,\mathrm{II}}$} 
\\[0.2ex]
& \Dminus, \Cminus, \Bminus & \Aminus {\scriptsize ($\beta \leq \beta_\flat$)} & 
\Aminus  {\scriptsize($\beta_\flat<\beta$)} & \Aplus  {\scriptsize($\beta < \beta_\sharp$)} & 
\Aplus  {\scriptsize($\beta_\sharp\leq \beta$)} & \Bplus, \Cplus, \Dplus \\ \hline 
& & & & & &  \\[-1.3ex]
$\mathrm{T}_\flat$ & \textbf{3} & \textbf{3} & \textbf{3} & \textbf{3} & \textbf{3} &  \\
$\mathrm{Q}_\flat$ & {\small [3]} & & & & & \\
$\mathrm{R}_\flat$ & & {\small [3]} & & 2 & 2 &   \\
$\mathrm{R}_\sharp$ & & & & 2 & & \\
$\mathrm{F}$ & 2 & 2 & 2 & 1 & 1 & 1  \\
$\mathrm{C}_\flat$ &  & & {\small [3]} & {\small [3]} & {\small [3]} & {\small [3]}  \\
$\mathrm{C}_\sharp$ & 1 & 1 & 1 & 1 & & \\
het.\ cycle & & & & & &  \textbf{3} \\\hline
\end{tabular}
\caption{Possible $\alpha$- and $\omega$-limit sets of LRS Bianchi type~II cosmological models.
Each Bianchi type~II solution corresponds to an orbit in $\mathcal{X}_{\,\mathrm{II}}$; hence 
the intersection of the unstable [stable] manifold of a fixed point with $\mathcal{X}_{\,\mathrm{II}}$ yields the set of type~II
solutions converging to that fixed point as $t\rightarrow 0$ [$t\rightarrow \infty$]. If 
this set is three-dimensional, the fixed point is a source [sink]; if the dimension is two,
there exists a one-parameter set of orbits converging to the fixed point; if the dimension is one,
there is only one orbit with that property. Center manifold analysis shows that 
the cases \Azeroplus\ and \Azerominus\ behave like
\Aplus\ ($0<\beta< \beta_\sharp$) and \Aminus\ ($\beta_\flat <\beta < 0$), respectively.}
\label{alfaII}
\end{center}
\end{table}

%%%%%%%%%%%%%%%%%%%%%%%%%%%%%%%%%%%%%%%%%%%%%%%%%%%%%%%%%%%%%%%%%%%%%%%%%%%%%%
\section{Bianchi type IX: Setup}\label{B89sec}
%%%%%%%%%%%%%%%%%%%%%%%%%%%%%%%%%%%%%%%%%%%%%%%%%%%%%%%%%%%%%%%%%%%%%%%%%%%%%%

%---------------------------------------------------------------------------------
\subsection{The reduced dynamical system for Bianchi types~IX (and~VIII)}
\label{redVIIIIX}
%---------------------------------------------------------------------------------

For the Bianchi types VIII and IX, the structure constants $\hat{n}_1$
and $\hat{n}_2 = \hat{n}_3$ are non-zero; 
according to Table~\ref{tab1} 
we have $\hat{n}_1^2 = 1$ and
$\hat{n}_1 \hat{n}_2 = \mp 1$ in type~VIII and type~IX, respectively.
Setting $\hat{n}_1^2 = 1$ and $\hat{n}_1 \hat{n}_2 = \mp 1$
in the dynamical system~\eqref{domsys} 
yields a system of five coupled equations; however, 
the constraints~\eqref{bothcons} imply that of the 
five variables in~\eqref{domsys} only three are independent.

\begin{Remark}
We note that $H_D \neq 1$ (i.e., $H \neq D$),
hence the normalization in~\eqref{normvars} is not the standard Hubble-normalization.
For Bianchi type~VIII (where $\hat{n}_1 \hat{n}_2 = {-1}$)
the constraint~\eqref{1stcon} implies that
$H_D > 1$ (when we restrict our attention to expanding cosmological models).
Bianchi type~IX (where $\hat{n}_1 \hat{n}_2 ={+1}$),
is characterized by $H_D \in (-1,1)$. 
\end{Remark}

It is not immediately evident, whether 
a three-dimensional dynamical system with suitable regularity properties
can be constructed from the constrained system~\eqref{domsys}.
To achieve this aim, 
in a first step, we closely follow~\cite{HRU}:
We use~\eqref{1stcon} to express $M_2$ in terms of $H_D$ and $M_1$, 
so that~\eqref{2ndcon} becomes
\begin{equation}\label{sofHDM1}
s = \Big( 2 + \frac{3 \hat{n}_1 \hat{n}_2 ( 1- H_D^2)}{M_1^2}\Big)^{-1} =
\Big( 2 + \frac{3 |1- H_D^2|}{M_1^2}\Big)^{-1}\:,
\end{equation}
which expresses $s$ in terms of $H_D$ and $M_1$.
The independent variables we choose
are thus $H_D$, $\Sigma_+$, and $M_1$, and the dynamical system~\eqref{domsys}
reduces to
\begin{subequations}\label{89syst}
\begin{align}
\label{HDEq2}
H_D^\prime & = -(1-H_D^2) (q - H_D \Sigma_+) \:,\\[0.5ex]
\label{Sig+Eq}
\Sigma_+^\prime & = -(2- q) H_D\Sigma_+ - (1-H_D^2) (1-\Sigma_+^2) + \textfrac{1}{3}\, \hat{n}_1^2 M_{1}^2 
+ 3\Omega \,\big(u(s) -w\big) \:,\\[0.5ex]
\label{M1Eq}
M_{1}^\prime & = M_{1} \big( q H_D - 4 \Sigma_+ + (1-H_D^2) \Sigma_+ \big)\:, 
\end{align}
\end{subequations}
where $s$ is regarded as a function of $H_D$ and $M_1$, see~\eqref{sofHDM1}.
The Hamiltonian constraint~\eqref{gaussconsimple},
\begin{equation}
\Omega = 1 - \Sigma_+^2 - \textfrac{1}{12} \, \hat{n}_1^2 M_{1}^2 \:,
\end{equation}
is used to solve for $\Omega$,
and
$q$ is given by $q = 2 \Sigma_+^2 + \textfrac{1}{2} ( 1 + 3 w) \Omega$.
Note that the system~\eqref{89syst} is identical for Bianchi type~VIII and
type~IX; however, the state spaces are different: $H_D > 1$ for
type~VIII and $H_D \in (-1,1)$ for type~IX.

Anticipating the analysis to follow,
let us note that the system~\eqref{89syst} does \textit{not}
have the regularity properties that are necessary 
to perform a satisfactory dynamical systems analysis. This is due to the fact that the variable $s$ 
and thus the function $u(s)$ in~\eqref{Sig+Eq} does 
not have a limit when $H^2_D\to 1$ and $M_1\to 0$. 
To remedy this problem, it will become inevitable to
introduce yet another system of coordinates in a neighborhood of the 
`bad' part of the boundary.

\begin{Remark}
Let us comment on the role of the boundaries.
In the case of Bianchi spacetimes where the matter is assumed to be 
vacuum or a perfect fluid, 
the Bianchi types~VIII and~IX are characterized by the maximal
number of `true degrees of freedom': The dynamical system
describing the dynamics of (diagonal) vacuum [perfect fluid] models of
Bianchi type~VIII or~IX is of dimension four [five];
in the LRS case the dimension is two [three].
In contrast, for the lower
Bianchi types, one or several degrees of freedom are `gauge',
and the dynamical systems describing the lower Bianchi types
are of lower dimension.
In particular, one finds that the state spaces and the dynamics 
of the lower Bianchi types appear on the boundary 
of the state spaces of the higher Bianchi types. This
hierarchy of invariant subsets (which is called the Lie contraction hierarchy,
see Appendix~\ref{Liecontractions})
plays a fundamental 
role in the dynamical systems analysis of cosmological models, see~\cite{HU}. 
In the case of Bianchi spacetimes where the matter is assumed to be
anisotropic, such a hierarchy exists only partially.
While the Bianchi type~I state space appears as the boundary subset of the type~II state space, see Sec.~\ref{sec:II},
the type~II state space is not directly related to a boundary subset of the type~VIII/IX
state space. This is a simple consequence of the fact that 
the two state spaces have the same dimension.
\end{Remark}

%---------------------------------------------------------------------------------
\subsection{Kantowski-Sachs and LRS Bianchi type~III models}
\label{KSIII}
%---------------------------------------------------------------------------------

In this paper we refrain from giving an analysis of 
the dynamics Kantowski-Sachs and LRS Bianchi type~III models. However,
Kantowski-Sachs and LRS Bianchi type~III play a certain role in the analysis 
of the flow on the boundary of the state spaces of Bianchi type~VIII and~IX.

As observed in Appendix~\ref{lrsexplained}, 
the spatial Ricci curvature of Kantowski-Sachs
and LRS Bianchi type~III models are obtained from the
Ricci tensor of Bianchi class~A models by formally setting
$\hat{n}_1^2$ to zero 
and $\hat{n}_1 \hat{n}_2$ to ${+1}$ or ${-1}$, respectively.
As a consequence, the same is true for every system of equations
describing these models. In particular,
the (reduced) dynamical system representing Kantowski-Sachs models and
LRS Bianchi type~III models is obtained from~\eqref{domsys} by 
formally setting $\hat{n}_1^2$ to zero 
and $\hat{n}_1 \hat{n}_2$ to ${+1}$ or ${-1}$, respectively.
In complete analogy to the considerations in the type~VIII and the type~IX case,
we are led to the dynamical system~\eqref{89syst} and the constraint~\eqref{gaussconsimple},
where 
\begin{equation*} 
\hat{n}_1^2  \equiv 0 \:.
\end{equation*}
Again, the variable $s$ is given in terms of $M_1$ and $H_D$ by~\eqref{sofHDM1}.
Note that the system~\eqref{89syst} is identical for Kantowski-Sachs and 
LRS type~III; however, the state spaces are different: 
\begin{subequations}
\begin{align}
&\mathcal{X}_{\mathrm{KS}} = \Big\{ (H_D, \Sigma_+, M_1) \:\big|\: H_D \in (-1,1)\,,\: M_1 > 0\,,\: 
\Sigma_+^2  < 1 \Big\}\:,\\
\label{typeIIIstatespace}
&\mathcal{X}_{\mathrm{III}} = \Big\{ (H_D, \Sigma_+, M_1) \:\big|\: H_D > 1\,,\: M_1 > 0\,,\: 
\Sigma_+^2 < 1 \Big\}\:.
\end{align}
\end{subequations}

\begin{Remark}
In the case where the matter is a perfect fluid, i.e., $u(s) \equiv w$,
the equation for $M_1$ decouples from~\eqref{89syst} and the
reduced dynamical system for Kantowski-Sachs models and LRS type~III
models is merely two-dimensional, the equations being~\eqref{HDEq2} and~\eqref{Sig+Eq}.
However, in the case of an anisotropic matter source, the system~\eqref{89syst} 
is a coupled system, since the variable $s$, which appears in~\eqref{Sig+Eq},
depends on $M_1$ by~\eqref{sofHDM1}.
\end{Remark}

%%%%%%%%%%%%%%%%%%%%%%%%%%%%%%%%%%%%%%%%%%%%%%%%%%%%%%%%%%%%%%%%%%%%%%%%%%%%%%
\section{Bianchi type IX: Analysis}\label{B9sec}
%%%%%%%%%%%%%%%%%%%%%%%%%%%%%%%%%%%%%%%%%%%%%%%%%%%%%%%%%%%%%%%%%%%%%%%%%%%%%%

The dynamics 
of LRS models of Bianchi type~IX are represented by
the reduced dynamical system~\eqref{89syst}, where $H_D \in (-1,1)$ (and $\hat{n}_1^2 = 1$).
Accordingly, 
the state space of LRS Bianchi type~IX models is
\[
\mathcal{X}_{\mathrm{IX}} = \Big\{ (H_D, \Sigma_+, M_1) \:\big|\: H_D \in (-1,1)\,,\: M_1 > 0\,,\: 
\Sigma_+^2 + \textfrac{1}{12} M_1^2 < 1 \Big\}\:,
\]
which is relatively compact.
The state space $\mathcal{X}_{\mathrm{IX}}$ is depicted in Fig.~\ref{b9state};
like $\mathcal{X}_{\,\mathrm{II}}$ it has the form of a tent; note, however,
the difference in the set of variables.

\begin{figure}[Ht]
\begin{center}
\psfrag{-1}[cc][cc][0.7][0]{$-1$}
\psfrag{1}[cc][cc][0.7][0]{$1$}
\psfrag{0}[cc][cc][0.7][0]{$-1$}
\psfrag{12}[cc][cc][0.7][0]{$1$}
\psfrag{sig}[cc][cc][0.8][0]{$\Sigma_+$}
\psfrag{m1}[cc][cc][0.8][0]{$M_1$}
\psfrag{ss}[cc][cc][1][0]{$\mathcal{S}_\sharp$}
\psfrag{v}[cc][cc][1][0]{$\mathcal{V}_\mathrm{IX}$}
\psfrag{xi-}[cc][cc][1][0]{$\mathcal{B}_{\mathrm{IX}}^-$}
\psfrag{xi+}[cc][cc][1][0]{$\mathcal{B}_{\mathrm{IX}}^+$}
\psfrag{s}[cc][cc][0.8][0]{$H_D$}
\psfrag{l}[cc][cc][1][-40]{$\mathcal{L}_\mathrm{I}$}
\psfrag{h0}[cc][cc][0.8][-45]{$\mathbf{H_D=0}$}
\includegraphics[width=0.5\textwidth]{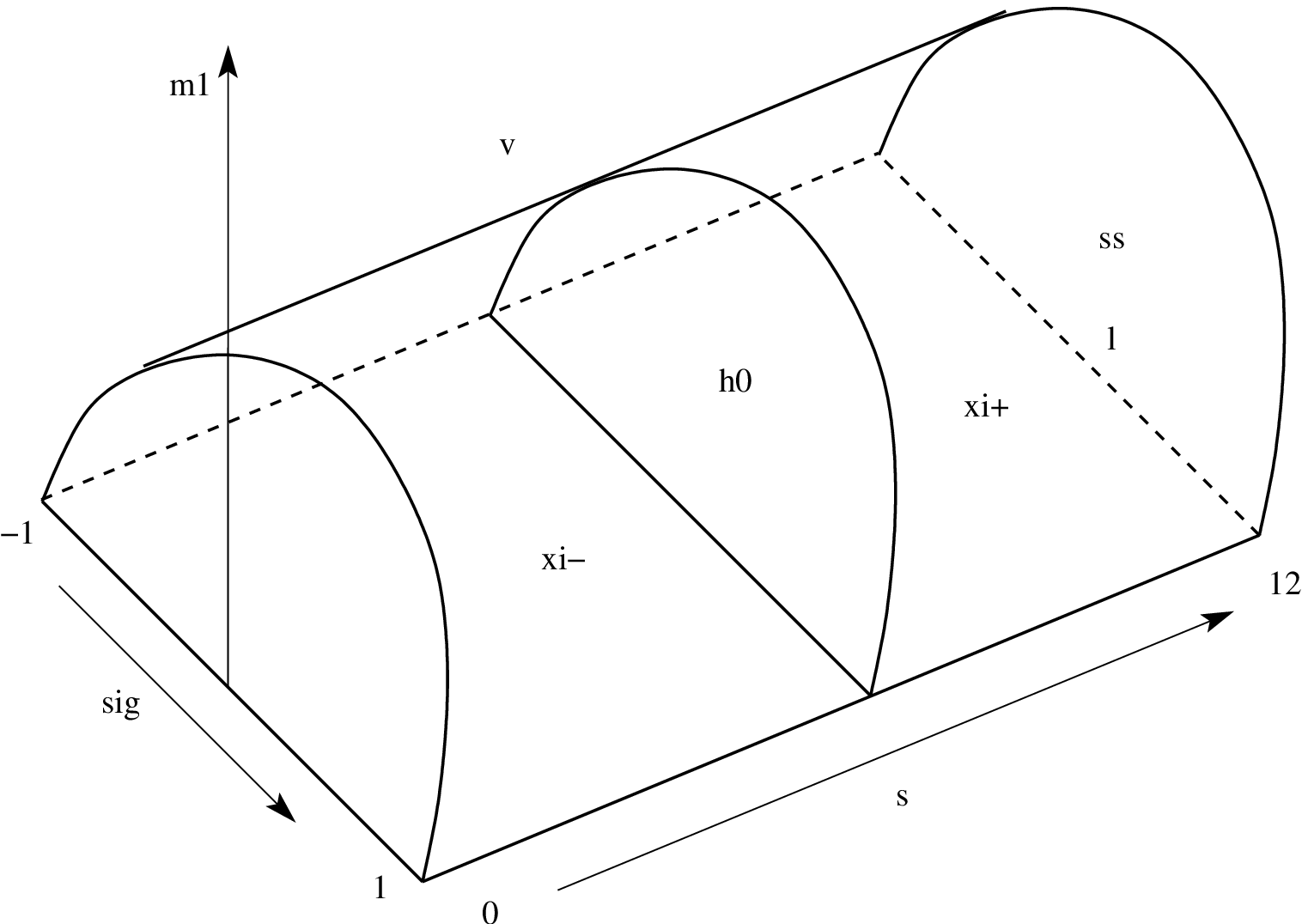}
\end{center}
\caption{The state space $\mathcal{X}_\mathrm{IX}$.}
\label{b9state}
\end{figure}

The state space $\mathcal{X}_\mathrm{IX}$ admits a natural
decomposition into two `halves' that are separated by a plane, i.e.,
\[
\mathcal{X}_{\mathrm{IX}}=\mathcal{X}_{\mathrm{IX}}^{+} \cup \mathcal{X}_{\mathrm{IX}}^{0} \cup \mathcal{X}_{\mathrm{IX}}^{-}\:,
\]
where 
\[
\mathcal{X}_{\mathrm{IX}}^{+} = \{H_D > 0 \} \subset \mathcal{X}_\mathrm{IX}\:,
\quad
\mathcal{X}_{\mathrm{IX}}^{0} = \{H_D = 0 \} \subset \mathcal{X}_\mathrm{IX}\:,
\quad
\mathcal{X}_{\mathrm{IX}}^{-} = \{H_D < 0 \} \subset \mathcal{X}_\mathrm{IX}\:.
\]
The plane $\mathcal{X}_\mathrm{IX}^{0}$ acts as a `semipermeable membrane'
for the flow of~\eqref{89syst}, since
\[
{H_D'}_{\,|H_D=0}=-\big(2\Sigma_+^2+\textfrac{1}{2}(1+3w)\Omega\big)<0\:.
\]
This 
$\mathcal{X}_\mathrm{IX}^{+}$ a subset that is 
past invariant under the flow of~\eqref{89syst}, while
$\mathcal{X}_\mathrm{IX}^{+}$ is future-invariant.
The physical interpretation is straightforward: A type~IX model
that is expanding at $t=t_0$, i.e., $H(t_0) > 0$ ($\Leftrightarrow H_D(t_0) > 0$), must have been
expanding up to that time, i.e., for $0 < t < t_0$;
conversely, a model that is contracting at $t=t_1$(i.e., $H(t_1) < 0$) 
must continue to contract $\forall t > t_1$ (which eventually leads to a big crunch);
finally, a model that satisfies $H = 0$ at some time, is a model
that starts from an initial singularity, expands to a state
of maximum expansion ($H = 0$), from which it then recontracts 
to a final singularity. (This ``closed universe recollapse'' behavior 
is the well-known behavior of vacuum and perfect fluid solutions of Bianchi type~IX,
see~\cite{W} and~\cite{HRU}.
However, there exist anisotropic matter models that deviate
from ``closed universe recollapse'', see~\cite{CHletter}
and the results presented in Sec.~\ref{B9res}.)

\begin{Remark}
The dynamical system~\eqref{89syst} is invariant under the discrete symmetry
\begin{equation}\label{symIX}
H_D\rightarrow -H_D\:, \ \Sigma_+\rightarrow -\Sigma_+\:,\ \tau\rightarrow -\tau\:;
\end{equation}
the flow in $\mathcal{X}_\mathrm{IX}^{-}$ is the image of 
the flow in $\mathcal{X}_\mathrm{IX}^{+}$ under this discrete symmetry.
In particular, it suffices to analyze the asymptotic behavior
of solutions in $\mathcal{X}_\mathrm{IX}^{+}$;
the asymptotic behavior of solutions in $\mathcal{X}_\mathrm{IX}^{-}$, 
follows by applying the discrete symmetry (where the roles of
$\alpha$- and $\omega$-limit sets are reversed).
\end{Remark}

%----------------------------------------------------------------------
\subsection{The boundaries of $\bm{\mathcal{X}_{\mathrm{IX}}}$}
\label{subsect:BIXbound}
%----------------------------------------------------------------------

The boundary of the type~IX state space 
$\mathcal{X}_{\mathrm{IX}}$ consists of the invariant sets
\begin{subequations}\label{IXboun}
\begin{align}
\label{Ssharpagain} &\mathcal{S}_\sharp = \big\{H_D=1\,;\:{-1}< \Sigma_+ < 1, \: 0< M_1 < \sqrt{12(1-\Sigma_+^2)}\:\big\}\:,   \\[0.5ex]
\label{basedef}
&\mathcal{B}_{\mathrm{IX}}=\big\{M_1=0\,;\:-1<\Sigma_+<1,\:-1<H_D<1\:\big\}\:,\\[0.5ex]
&\mathcal{V}_\mathrm{IX}=\big\{\Omega = 0\,; \:{-1}< \Sigma_+ < 1, \:-1<H_D<1,\: 
M_1 = \sqrt{12(1-\Sigma_+^2)} \:\big\}\:.
\end{align}
\end{subequations}
The symbol $\mathcal{S}_\sharp$ is used with slight abuse of notation; we refer to
the discussion of~\eqref{Sidesys} below 
(and to the analogous comment following~\eqref{IIboundarycomps}).
There exists a second copy of the set~\eqref{Ssharpagain} with $H_D=-1$; since this set 
is merely the image of $\mathcal{S}_\sharp$ under the discrete symmetry~\eqref{symIX},
we may restrict our attention to $\mathcal{S}_\sharp$. Analogously,
concerning the `base' $\mathcal{B}_{\mathrm{IX}}$, 
it is enough to restrict our attention to the 
past invariant half 
$\mathcal{B}_{\mathrm{IX}}^+=\mathcal{B}_{\mathrm{IX}}\cap\{H_D\geq 0\}$, see Fig.~\ref{b9state}. 

The dynamical system~\eqref{89syst} admits a regular extension 
to each of the invariant boundary subsets~\eqref{IXboun}. 
However, the system can \textit{not} be extended to the
entire boundary $\partial\mathcal{X}_{\mathrm{IX}}$; 
this is because the variable $s$ and thus the r.h.s.\ of~\eqref{Sig+Eq}
does not possess a well-defined limit when $(H_D,\Sigma_+,M_1)$ converges 
to a point on the closure of the line
\begin{equation}\label{badboundary}
\mathcal{L}_{\mathrm{I}}=\partial\mathcal{S}_\sharp\cap\partial\mathcal{B}_{\mathrm{IX}}\:;
\end{equation}
the use of the subscript $\mathrm{I}$ will become clear in Subsec.~\ref{blowup},
when we remedy this defect of the system~\eqref{89syst}.

The first step of our analysis is to study 
the flow that is induced on the boundaries~\eqref{IXboun} 
of $\mathcal{X}_\mathrm{IX}$, where the dynamical system~\eqref{89syst} admits a regular extension.

\textbf{The vacuum boundary} $\bm{\mathcal{V}_\mathrm{IX}}$. 
The dynamical system~\eqref{89syst} admits a smooth extension onto $\overline{\mathcal{V}}_\mathrm{IX}$ which is given by
setting $\Omega = 0$ in~\eqref{89syst}. Since $\Omega = 0$ and thus $\rho = 0$ 
on $\mathcal{V}_\mathrm{IX}$, the orbits on this set represent type~IX vacuum models.
The induced system is
\begin{subequations}\label{Roofsys}
\begin{align}
\label{RoofHD}
H_D^\prime & = (1-H_D^2) ( H_D - 2 \Sigma_+ ) \Sigma_+   \:,\\[0.5ex]
\label{RoofSig+}
\Sigma_+^\prime & = (1-\Sigma_+^2) \big[ 2 + (1 -\Sigma_+^2) + (H_D - \Sigma_+)^2 \big] \:;
\end{align}
\end{subequations}
the variable $M_1$ is given by $M_1 = \sqrt{12(1-\Sigma_+^2)}$.
The closure of the vacuum boundary is a square, 
$\overline{\mathcal{V}}_\mathrm{IX} = [-1,1] \times [-1,1]$.
Eq.~\eqref{RoofSig+} entails that $\Sigma_+$ is strictly monotonically
increasing whenever $\Sigma_+ \neq \pm 1$.
It follows straightforwardly that the vertices of the square are the only fixed points;
two of these fixed point satisfy $H_D = 1$:
\begin{itemize}
\item[$\mathrm{T}$] \quad The Taub point $\mathrm{T}$ is given by $(H_D, \Sigma_+) = (1,-1)$.
\item[$\mathrm{Q}$] \quad The non-flat LRS point $\mathrm{Q}$ is given by $(H_D, \Sigma_+) = (1,1)$.
\end{itemize}
The Taub point $\mathrm{T}$ corresponds to the Taub solution~\eqref{taub} and
$\mathrm{Q}$ to the non-flat LRS Kasner solution~\eqref{solQ}.
The reason for the lack of a subscript ${}_\sharp$ or ${}_\flat$ 
will become clear in Subsec.~\ref{blowup}.
(The fixed points that correspond to $\mathrm{T}$ and $\mathrm{Q}$ on the subset 
$H_D = -1$ are $(H_D,\Sigma_+) = (-1,1)$
and $(H_D,\Sigma_+)=({-1},{-1})$; the former is a Taub point, the latter
a non-flat LRS point. The properties of these fixed points follows from
the discrete symmetry~\eqref{symIX}.)

Since $\Sigma_+$ is a strictly monotonically increasing function 
on $[-1,1] \times (-1,1)$, the monotonicity principle applies.
In combination with the fact that $H_D^\prime < 0$ on the boundaries $H_D = \pm 1$, 
we arrive at the conclusion that the Taub point $\mathrm{T}$ 
is the $\alpha$-limit set of every orbit in $\mathcal{V}_\mathrm{IX}$
(while the Taub point $(-1,1)$ is the $\omega$-limit set); see Fig.~\ref{vacuumIXflow}.

\begin{figure}[Ht]
\begin{center}
\psfrag{Sp}[cc][cc][1][0]{$\Sigma_+$}
\psfrag{T}[cc][cc][1][0]{$\mathrm{T}$}
\psfrag{Q}[cc][cc][1][0]{$\mathrm{Q}$}
\psfrag{HD}[cc][cc][1][0]{$H_D$}
\includegraphics[width=0.4\textwidth]{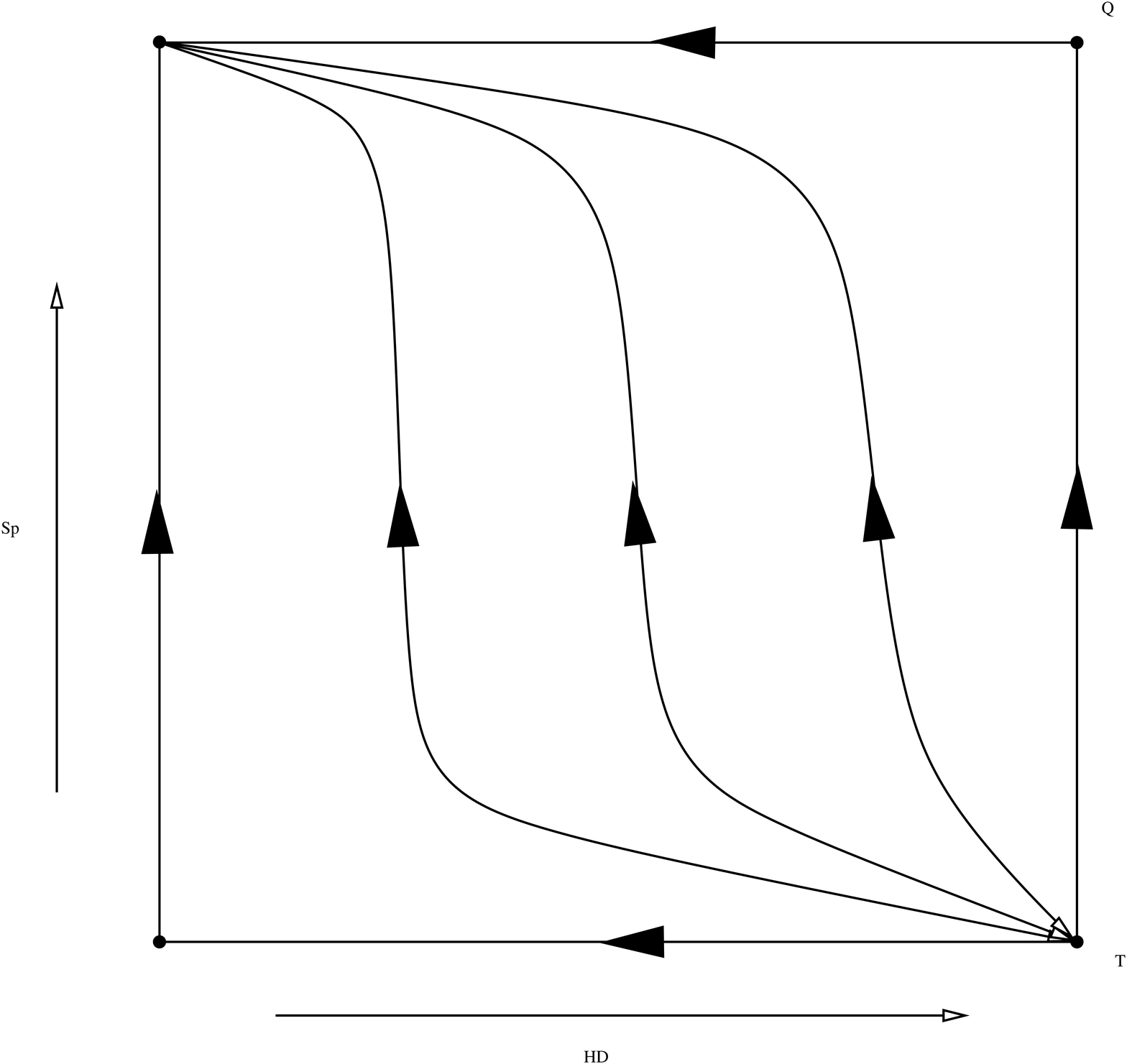}
\end{center}
\caption{The flow induced on $\mathcal{V}_{\mathrm{IX}}$.}
\label{vacuumIXflow}
\end{figure}

\textbf{The side} $\bm{\mathcal{S}_\sharp}$.
The dynamical system that is induced by~\eqref{89syst} on the side 
$\mathcal{S}_\sharp$ of the type~IX state space is identical to the system~\eqref{dynsysAsharp}: 
\begin{subequations}\label{Sidesys}
\begin{align}
\label{SideSig+}
&\Sigma_+'=\textfrac{1}{6}M_1^2(2-\Sigma_+)-\textfrac{3}{2}\Omega(1-w)(\Sigma_+-\textfrac{\beta}{2})\:,\\[0.5ex]
\label{SideM1}
& M_1'=M_1 \big[ \textfrac{3}{2} (1-w) \Sigma_+^2 - 4 \Sigma_+ + 
\textfrac{1}{2}(1 +3 w) \big(1 -\textfrac{1}{12}M_1^2\big)\big]\:,
\end{align}
\end{subequations}
where $\Omega=1-\Sigma_+^2- \textfrac{1}{12} M_1^2$. 
The system~\eqref{Sidesys} on $\mathcal{S}_\sharp$ possesses a smooth extension
to $\overline{\mathcal{S}}_\sharp$ (which does not contradict the fact
that~\eqref{89syst} cannot be extended to $\overline{\mathcal{S}}_\sharp$).

The phase portraits for the various cases are depicted in Fig.~\ref{Asharpfig}. 
Note that the fixed point $\mathrm{C}_\sharp$ is represented by a white dot 
in the figures, which means that it acts as a repellor in a direction transversal to $\mathcal{S}_\sharp$, 
i.e., into the interior of $\mathcal{X}_\mathrm{IX}$. To see this we use~\eqref{HDEq2}
and find
\begin{equation}\label{Cb9}
\left[\textfrac{d}{d\tau}\log|1-H_D|\right]_{\,|\mathrm{C}_\sharp} =
2 \big(2 \Sigma_+^2- \Sigma_+ + \textfrac{1}{2}\,(1+ 3 w) \Omega \big)\,|_{\mathrm{C}_\sharp} =  
\textfrac{6(1+3w)}{16-3\beta(1-w)}>0\:.
\end{equation}

The vertices of the boundary $\partial\mathcal{S}_\sharp$ 
coincide with the fixed points $\mathrm{T}_\sharp$ and $\mathrm{Q}_\sharp$,
see Fig.~\ref{Asharpfig}; in addition, in the cases \A, \B, \C, 
the fixed point $\mathrm{R}_\sharp$ is present.
The significance of these fixed points in the
type~IX context will become
apparent in Subsec.~\ref{blowup}.

\begin{Remark}
The solutions of~\eqref{Sidesys} on the side $\mathcal{S}_\sharp$ 
are type~II solutions associated with an anisotropic matter source 
with `rigid' rescaled principal pressures, i.e.,
$u(s) \equiv \mathrm{const}$.
For such an anisotropic matter source, eq.~\eqref{IIs} 
decouples from the system~\eqref{dynsysbianchiII};
the remaining equations~\eqref{IIsig+} and~\eqref{IIM1}
coincide with~\eqref{Sidesys}.
Analogously, the orbits of~\eqref{Basesys} on the base $\mathcal{B}_{\mathrm{IX}}$,
to be discussed next, represent solution of Kantowski-Sachs type, see~\eqref{KSmetric},
where again 
the rescaled principal pressures are `rigid'.
For such an anisotropic matter source, eq.~\eqref{M1Eq} 
decouples from the system~\eqref{89syst}. (Recall that
$\hat{n}_1^2$ is set to zero for the Kantowski-Sachs system.)
Therefore, the associated Kantowski-Sachs models are
described by the reduced system consisting of~\eqref{HDEq2} and~\eqref{Sig+Eq},
which is identical to the system~\eqref{Basesys} induced on $\mathcal{B}_{\mathrm{IX}}$.
\end{Remark}

\textbf{The base} $\bm{\mathcal{B}_{\mathrm{IX}}}$.
The dynamical system~\eqref{89syst} 
extends smoothly to the base $\mathcal{B}_{\mathrm{IX}}$ of the state space 
(but not to $\overline{\mathcal{B}}_{\mathrm{IX}}$).
The induced system is obtained by setting $M_1 = 0$ in~\eqref{89syst}:
\begin{subequations}\label{Basesys}
\begin{align}
\label{BaseHD}
H_D^\prime & = -(1-H_D^2) \Big[ 2 - \textfrac{3}{2} ( 1 -w) ( 1-\Sigma_+^2) - H_D \Sigma_+ \Big]    \:,\\[0.5ex]
\label{BaseSig+}
\Sigma_+^\prime & = -(1-\Sigma_+^2) \Big[ (1-H_D^2) + \textfrac{3}{2}(1-w) \big(H_D  \Sigma_+ +\beta \big) \Big] \:.
\end{align}
\end{subequations}
The induced system~\eqref{Basesys} on $\mathcal{B}_{\mathrm{IX}}$ possesses a smooth extension
to $\overline{\mathcal{B}}_{\mathrm{IX}}=[-1,1]\times[-1,1]$
(which does not contradict the fact
that~\eqref{89syst} cannot be extended to $\overline{\mathcal{B}}_{\mathrm{IX}}$).
 
The boundary of $\mathcal{B}_{\mathrm{IX}}$ contains several fixed points.
\begin{itemize}
\item[$\mathrm{T}_\flat$] \quad The Taub point $\mathrm{T}_\flat$ is given by $(H_D, \Sigma_+) = (1,-1)$.
\item[$\mathrm{Q}_\flat$] \quad The non-flat LRS point $\mathrm{Q}_\flat$ is given by $(H_D, \Sigma_+) = (1,1)$.
\end{itemize}

\begin{Remark}
The subscripts (i.e., ${}_\flat$ or ${}_\sharp$) are chosen with care; their
significance in the present context will become clear in Subsec.~\ref{blowup}. 
Although the points $\mathrm{T}$ on $\partial\mathcal{V}_{\mathrm{IX}}$,
$\mathrm{T}_\sharp$ on $\partial\mathcal{S}_\sharp$, and 
$\mathrm{T}_\flat$ on $\partial\mathcal{B}_{\mathrm{IX}}$
have formally the same coordinates 
(i.e., $H_D =1$, $\Sigma_+ = -1$, $M_1 = 0$),
these points cannot be identified. 
We will see in Subsec.~\ref{blowup} 
that the Taub points $\mathrm{T}_\flat$, $\mathrm{T}_\sharp$, $\mathrm{T}$  
provide different (but equivalent) representations of the Taub 
solution~\eqref{taub}.
The reason for this is that the system~\eqref{89syst} cannot be 
extended to the line $\overline{\mathcal{L}}_{\mathrm{I}}$.
While
the extension of the system~\eqref{Basesys} on $\mathcal{B}_{\mathrm{IX}}$
to the boundary subset $H_D =1$ yields 
eq.~\eqref{eqcalIflat}, 
the extension of the system~\eqref{Sidesys} on
$\mathcal{S}_\sharp$ to (the base of) $\partial\mathcal{S}_\sharp$
leads to eq.~\eqref{eqcalIsharp}. 
Unless $\beta = 0$, these two equations are different.
\end{Remark}

If and only if $|\beta| < 1$ (i.e., in the cases \Aplus\ and \Aminus), 
the boundary of $\mathcal{B}_{\mathrm{IX}}$ contains the point $\mathrm{R}_\flat$.
\begin{itemize}
\item[$\mathrm{R}_\flat$] \quad This fixed point has the coordinates $(H_D,\Sigma_+) = (1, -\beta)$.
\end{itemize}
Applying the discrete symmetry~\eqref{symIX} we obtain an analogous fixed point with $H_D = -1$. 
In Appendix~\ref{exact} we calculate the explicit solution associated with $\mathrm{R}_\flat$
and give its interpretation.

Searching for interior fixed points in $\mathcal{B}_{\mathrm{IX}}$ we find a fixed point $\mathrm{P}$ in $\mathcal{B}_{\mathrm{IX}}^+$
(and an analogous point in $\mathcal{B}_{\mathrm{IX}}^-$ generated by the discrete symmetry).
\begin{itemize}
\item[$\mathrm{P}$] \quad The coordinates of $\mathrm{P}$ are \,$H_D = \textfrac{2 + 3\beta (1-w)}{\sqrt{(1-3 w)^2 + 6 \beta (1-w)}}$,\,
$\Sigma_+ = \textfrac{1 + 3 w}{\sqrt{(1-3 w)^2 + 6 \beta (1-w)}}$\:.
\end{itemize}
There are, however, two requirements for the fixed point $\mathrm{P}$ to exist:
First, $w$ must satisfy
\begin{subequations}\label{domain}
\begin{equation}\label{wdomain}
-\textfrac{1}{3} \,<\, w\,<\, \textfrac{1-\sqrt{3}}{3} \approx -0.244 
\end{equation}
and, second, $\beta$ must be in the range
\begin{equation}\label{betadomain}
\beta_- < \beta < \beta_+  \quad\text{with}\quad 
\beta_{\pm} = \textfrac{-1\pm \sqrt{-3+(1-3 w)^2}}{3(1-w)} \:.
\end{equation}
\end{subequations}
The conditions~\eqref{domain} ensure that ${H_D}_{\,|\mathrm{P}}$ and ${\Sigma_+}_{\,|\mathrm{P}}$
are real (and positive) numbers strictly less then $1$.
The range of admissible values of $w$ and $\beta$ is depicted in Fig.~\ref{betaminmax};
since ${-}\textfrac{1}{2} < \beta_- < {-2}+\sqrt{3}$
and ${-2}+\sqrt{3} < \beta_+ < 0$, cf.~Fig.~\ref{betaminmax},
we find that the conditions~\eqref{domain} distinguish a subcase
of \Aminus, which we call the \textit{Bianchi type~IX special case} and denote by \AminusP.

An alternative representation of~\eqref{betadomain} in terms of $v_-$, see~\eqref{betadef1}, is
\begin{equation}\label{vmdomain}
\textfrac{1}{6} \left( 1 + 6 w - \sqrt{-3+(1-3 w)^2} \right) \, < \, v_- \, <\, 
\textfrac{1}{6} \left( 1 + 6 w + \sqrt{-3+(1-3 w)^2} \right)\:,
\tag{\ref{betadomain}${}^\prime$}
\end{equation}
cf.~Fig.~\ref{vpmminmax}. 
Finally, note that $\mathrm{P}$ coincides with $\mathrm{R}_\flat$ 
in the limiting cases $\beta=\beta_\pm$.  The exact solution corresponding to the fixed point 
$\mathrm{P}$ is of Kantowski-Sachs type and is derived in Appendix~\ref{exact}.

\begin{figure}[Ht]
\begin{center}
\psfrag{a}[bc][bc][0.7][0]{$-\textfrac{1}{3}$}
\psfrag{c}[bc][bc][0.7][0]{$-0.30$}
\psfrag{d}[bc][bc][0.7][0]{$-0.27$}
\psfrag{e}[bc][bc][0.7][0]{$\textfrac{1-\sqrt{3}}{3}$}
\psfrag{b}[bc][bc][1][0]{$\beta$}
\psfrag{vm}[bc][bc][1][0]{$v_-$}
\psfrag{m}[lc][lc][0.7][0]{$-0.1$}
\psfrag{n}[lc][lc][0.7][0]{$-0.268$}
\psfrag{o}[lc][lc][0.7][0]{$-0.5$}
\psfrag{w}[lc][lc][1][0]{$w$}
\psfrag{v}[lc][lc][0.7][0]{$-0.077$}
\psfrag{x}[lc][lc][0.7][0]{$-0.2$}
\psfrag{z}[lc][lc][0.7][0]{$-\textfrac{1}{3}$}
\includegraphics[width=0.65\textwidth]{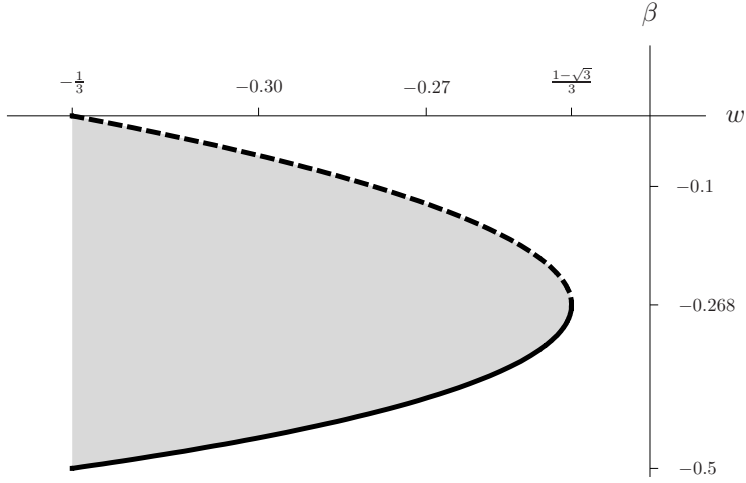}
\caption{The Bianchi type~IX special case \AminusP\ is defined by conditions on $w$ and the anisotropy parameter $\beta$.
The admissible values are depicted as
the gray region. The axes of the diagram are $\beta = 0$ and $w = -0.23$.}
\label{betaminmax}
\end{center}
\end{figure}

\begin{figure}[Ht]
\begin{center}
\psfrag{a}[bc][bc][0.7][0]{$-\textfrac{1}{3}$}
\psfrag{c}[bc][bc][0.7][0]{$-0.30$}
\psfrag{d}[bc][bc][0.7][0]{$-0.27$}
\psfrag{e}[bc][bc][0.7][0]{$\textfrac{1-\sqrt{3}}{3}$}
\psfrag{b}[bc][bc][1][0]{$\beta$}
\psfrag{vm}[bc][bc][1][0]{$v_-$}
\psfrag{m}[lc][lc][0.7][0]{$-0.1$}
\psfrag{n}[lc][lc][0.7][0]{$-0.268$}
\psfrag{o}[lc][lc][0.7][0]{$-0.5$}
\psfrag{w}[lc][lc][1][0]{$w$}
\psfrag{v}[lc][lc][0.7][0]{$-0.077$}
\psfrag{x}[lc][lc][0.7][0]{$-0.2$}
\psfrag{z}[lc][lc][0.7][0]{$-\textfrac{1}{3}$}
\psfrag{s}[lc][lc][0.7][0]{$-\textfrac{1}{3}$}
\psfrag{u}[lc][lc][0.7][0]{$-0.577$}
\psfrag{r}[lc][lc][0.7][0]{$-1$}
\psfrag{vp}[bc][bc][1][0]{$v_+$}
\subfigure{\includegraphics[width=0.45\textwidth]{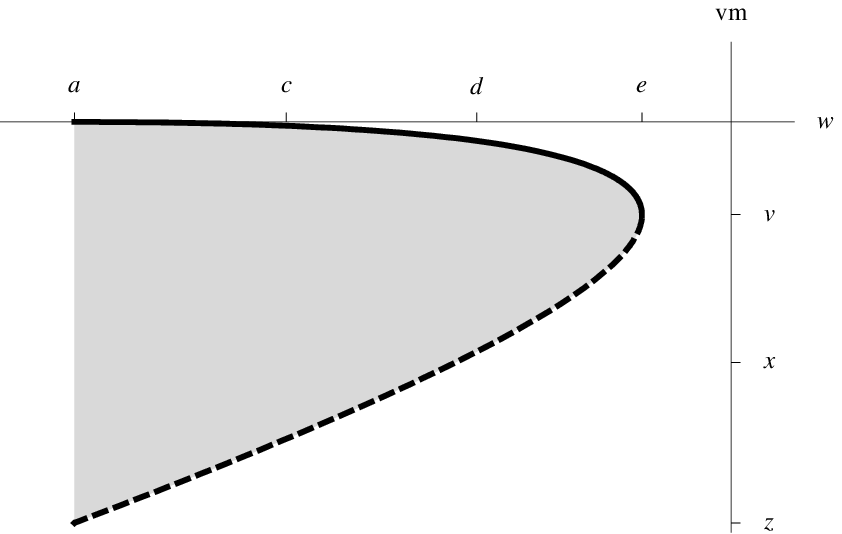}}\qquad\qquad
\subfigure{\includegraphics[width=0.45\textwidth]{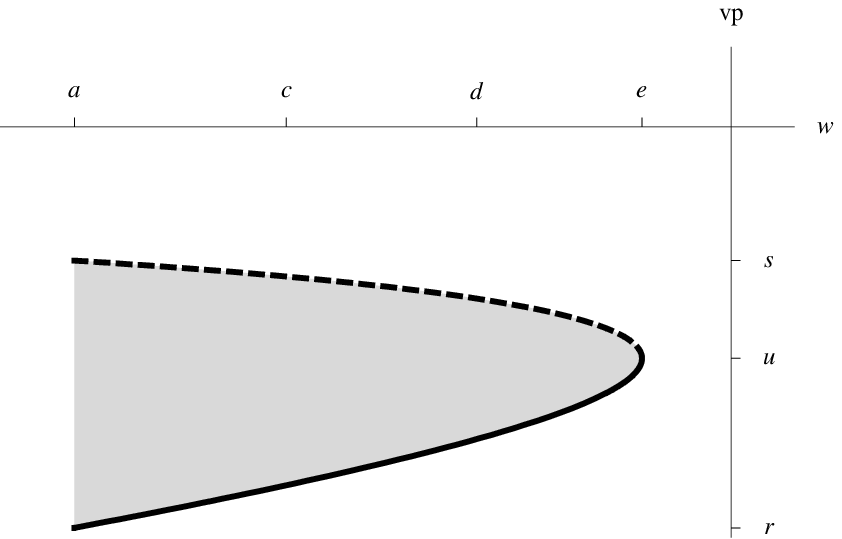}}
\caption{The admissible values of $w$ and $v_-$/$v_+$ for the Bianchi type~IX special case \AminusP.
The axes are $v_\pm = 0$ and $w = -0.23$. From this figure, in combination with
the considerations of Sec.~\ref{LRSsols},
it is evident that the dominant energy condition and the
strong energy condition are satisfied,
since $|v_\pm| <1$ and $w > -1/3$.}
\label{vpmminmax}
\end{center}
\end{figure}

The local dynamical systems analysis of the fixed points is straightforward.
The role of the fixed points $\mathrm{T}_\flat$ and $\mathrm{Q}_\flat$
is determined by the
flow on the boundaries of $\mathcal{B}_{\mathrm{IX}}$, which is obtained in a straightforward 
way; we simply refer to Fig.~\ref{Bfig}.
The fixed point $\mathrm{R}_\flat$ is an attractor on the line $H_D = 1$;
in the transversal direction we have
\begin{equation}\label{Rtransvers}
\left[\textfrac{d}{d\tau}\log\left(1-H_D\right)\right]_{\,|\mathrm{R}_\flat} = (1+ 3 w) + 2 \beta + 3(1-w) \beta^2\:,
\end{equation}
which is negative if and only if $w$ satisfies~\eqref{wdomain} and $\beta$ satisfies~\eqref{betadomain}.
Hence, in the Bianchi type~IX special case \AminusP, the fixed point $\mathrm{R}_\flat$ is
a sink in $\overline{\mathcal{B}}_{\mathrm{IX}}$; in all other cases, $\mathrm{R}_\flat$ is a saddle.
(In the marginal case $\beta = \beta_\pm$, this does not follow from the local
analysis since $\mathrm{R}_\flat$ is not hyperbolic in these cases.
However, the statement is still correct, as will be proved by using global methods.)

The fixed point $\mathrm{P}$ which exists in the Bianchi type~IX 
special case \AminusP\ is a saddle in $\mathcal{B}_{\mathrm{IX}}$.
This is a direct consequence of the 
linearization of the dynamical system~\eqref{Basesys} at $\mathrm{P}$;
in fact the determinant of the linearized system is
\[
{-}\textfrac{27(1-w)^3}{(1-3w)^2+6\beta(1-w)}
\left(\beta-\textfrac{2w}{1-w}\right)\left(\beta-\beta_-\right)\left(\beta_+ - \beta\right) < 0\:,
\]
which is negative because the three terms in brackets are positive under
the conditions~\eqref{wdomain} and~\eqref{betadomain} of case \AminusP.
Accordingly, the linearization is indefinite and the point $\mathrm{P}$
is a saddle. 

Since the fixed point $\mathrm{P}$ lies within $\mathcal{B}_{\mathrm{IX}}$ 
(and not on the boundary $H_D = 1$ like $\mathrm{T}_\flat$, $\mathrm{Q}_\flat$, $\mathrm{R}_\flat$),
the dynamical system~\eqref{89syst} on $\mathcal{X}_{\mathrm{IX}}$
extends smoothly to a neighborhood of $\mathrm{P}$ in $\overline{\mathcal{X}}_{\mathrm{IX}}$.
Using~\eqref{M1Eq} we find
\begin{equation}\label{Portho}
\left[\textfrac{d}{d\tau}\log M_1\right]_\mathrm{P}=-\textfrac{3(1+3w)}{\sqrt{(1-3w)^2+6\beta(1-w)}}<0\:.
\end{equation}
Therefore, $\mathrm{P}$ is a saddle in $\overline{\mathcal{X}}_{\mathrm{IX}}$
with a two-dimensional stable manifold and a one-dimensional unstable manifold.
We reemphasize that analogous statements for the fixed points $\mathrm{T}_\flat$, $\mathrm{Q}_\flat$, $\mathrm{R}_\flat$
are unavailable in the present dynamical systems formulation, since the system~\eqref{89syst} does not extend to the line
$\mathcal{L}_{\mathrm{I}}$, cf.~\eqref{badboundary}. 
In Subsec.~\ref{blowup} we will remedy this defect.

Let us turn to the \textbf{global dynamics on} $\bm{\mathcal{B}_{\mathrm{IX}}}$.
\begin{Lemma}
The phase portraits of the dynamical system~\eqref{Basesys} on the base $\mathcal{B}_{\mathrm{IX}}$ 
for the various anisotropy cases are as depicted
in Fig.~\ref{Bfig}. 
\end{Lemma}

\begin{proof}
The eq.~\eqref{BaseHD} for $H_D$ can be written as 
\begin{equation}\label{HDprimeB}
H_D^\prime = -(1-H_D^2) \,\Big[ \textfrac{1}{2} (1 + 3 w) - H_D \Sigma_+ + \textfrac{3}{2} (1 -w) \Sigma_+^2\Big] \:.
\end{equation}
For $H_D\geq 0$ and $\Sigma_+ \leq 0$, the expression in brackets is manifestly positive.
For positive $\Sigma_+$ the expression attains its (unique) minimum for the value $H_D = 1$;
in that case we obtain a second order
polynomial in $(-\Sigma_+)$, which coincides with~\eqref{Rtransvers} (modulo a factor of $2$)
and is therefore non-negative if $w > \textfrac{1-\sqrt{3}}{3}$.
We conclude that the expression in brackets is positive for all $H_D \in[0,1)$, $\Sigma_+\in (-1,1)$,
if $w \geq \textfrac{1-\sqrt{3}}{3}$.

This implies that, in the case $\bm{w \geq \textfrac{1-\sqrt{3}}{3}}\,$,
$H_D$ is strictly monotonically decreasing on the subset $H_D \geq 0$
of $\mathcal{B}_{\mathrm{IX}}$, i.e., on $\mathcal{B}_{\mathrm{IX}}^+$. Applying the symmetry~\eqref{symIX} we see that
$H_D$ is then strictly monotonically decreasing on the entire state space $\mathcal{B}_{\mathrm{IX}}$.
Taking into account the simple structure of the flow on $\partial\mathcal{B}_{\mathrm{IX}}$,
we conclude that the $\alpha$-limit set of an orbit in $\mathcal{B}_{\mathrm{IX}}$ must
be one of the fixed points with $H_D = 1$: $\mathrm{T}_\flat$, $\mathrm{Q}_\flat$, or
$\mathrm{R}_\flat$ (which requires $|\beta| <1$, i.e., one of the \A\ cases); 
the former two possibilities represent the past asymptotic behavior 
of typical orbits, which is because $\mathrm{R}_\flat$ is a saddle
when $w \geq \textfrac{1-\sqrt{3}}{3}$
(except in the case $w = \textfrac{1-\sqrt{3}}{3}$, $\beta = \beta_\pm$,
where it is a center saddle, as follows from a center manifold analysis).
By applying the symmetry~\eqref{symIX} we obtain an analogous statement for
the possible $\omega$-limit sets.
In the \A\ cases, when the fixed point $\mathrm{R}_\flat$ exists,
there is one special orbit emanating from $\mathrm{R}_\flat$
that acts as a separatrix in the state space $\mathcal{B}_{\mathrm{IX}}$.
To obtain the $\omega$-limit point of this special orbit we consider
the function $r_\flat = \beta H_D + \Sigma_+$ on $\mathcal{B}_{\mathrm{IX}}$;
the set $r_\flat = 0$ represents a straight line in $\overline{\mathcal{B}}_{\mathrm{IX}}$ that
connects the fixed point $\mathrm{R}_\flat$ with its counterpart on $H_D = -1$.
A straightforward computation shows that
\begin{equation}\label{rflat}
(r_\flat^\prime)_{|r_\flat = 0} = -(1-H_D^2) (1 + 2 \beta) \:,
\end{equation}
which is negative for $\beta > -1/2$, zero for $\beta = -1/2$, and positive for $\beta < -1/2$.
Therefore, the subset $r_\flat \leq 0$ of $\mathcal{B}_{\mathrm{IX}}$ 
is future-invariant if $\beta \geq -1/2$,
while $r_\flat \geq 0$ is future-invariant if $\beta \leq -1/2$.
(The set $r_\flat = 0$ is an actual orbit in the case $\beta = -1/2$).
The local analysis of the system~\eqref{Basesys} at the fixed point $\mathrm{R}_\flat$
shows that the unstable manifold of $\mathrm{R}_\flat$ and thus
the separatrix orbit emanating from $\mathrm{R}_\flat$ is contained 
in the future-invariant subset. 
From this observation it is straightforward to determine the $\omega$-limit
point of this orbit.
We conclude that the structure of the flow 
is indeed as depicted in Fig.~\ref{Bfig}.

In the case $\bm{-\textfrac{1}{3} < w < \textfrac{1-\sqrt{3}}{3}}\,$.
the monotonicity of $H_D$ on $\mathcal{B}_{\mathrm{IX}}$ is violated; however,
the property $H_D^\prime > 0$ is restricted to a rather small region of $\mathcal{B}_{\mathrm{IX}}$,
see Fig.~\ref{BHDposneg}.
We distinguish two subcases: $\beta \not\in (\beta_-, \beta_+)$ and $\beta \in (\beta_-, \beta_+)$, 
where the latter coincides with the special case \AminusP.
First, consider $\beta \not\in (\beta_-, \beta_+)$.
In this case, there does not exist a fixed point
in the interior of $\mathcal{B}_{\mathrm{IX}}$. 
This fact excludes the presence of periodic 
orbits and heteroclinic cycles in $\mathcal{B}_{\mathrm{IX}}$.
(For a two-dimensional system, a periodic orbit necessarily encircles
a fixed point which is a sink or a source.)
The absence of these structures in $\mathcal{B}_{\mathrm{IX}}$ and the 
Poincar\'e-Bendixson theorem (see Appendix~\ref{dynsysapp}) imply
that an orbit in $\mathcal{B}_{\mathrm{IX}}$ must leave every 
compact subset of $\mathcal{B}_{\mathrm{IX}}$; in other words, the $\alpha$-
and $\omega$-limit set must intersect the boundary.
Taking into account the flow on the boundary it follows
that the possible $\alpha$-limit sets are the fixed points on $H_D = 1$;
the possible $\omega$-limits sets are the fixed points on $H_D = {-1}$.
Using again~\eqref{rflat} we conclude that Fig.~\ref{Bfig} depicts the correct qualitative behavior
of the flow on $\mathcal{B}_{\mathrm{IX}}$ (where, however, $H_D$ is not necessarily monotonically
decreasing).

Second, consider $\beta \in (\beta_-, \beta_+)$, i.e., the Bianchi type~IX special case \AminusP,
in which the presence of the fixed point $\mathrm{P}$ slightly complicates matters.
However, the fixed point $\mathrm{P}$ (as well as its counterpart in
the subset $H_D < 0$) is a saddle; therefore,
there do not exist any periodic orbits in $\mathcal{B}_{\mathrm{IX}}$
and the following argument excludes the existence of heteroclinic cycles as well.
Since $\mathrm{P}$ lies on 
the boundary of the region $H_D^\prime > 0$,
it is not difficult to prove that one of the two
orbits emanating from $\mathrm{P}$ must converge to $\mathrm{R}_\flat$.
(The first step is to construct a future-invariant subset of the set $H_D^\prime > 0$,
whose closure contains the sink $\mathrm{R}_\flat$. 
Taking into account the flow on the boundary $H_D = 1$, we find
that $\mathrm{R}_\flat$ is the attractor for every orbit in this
future-invariant set; see Fig.~\ref{BHDposneg}.
The second step is to show that one half of the unstable subspace of $\mathrm{P}$
is contained in this future-invariant set. 
We conclude that one orbit emanating from $\mathrm{P}$
is contained in the future-invariant set and thus converges to $\mathrm{R}_\flat$, which
proves the claims.)
The $\omega$-limit set of the second orbit emanating from $\mathrm{P}$ must
be contained in the future-invariant half $H_D < 0$ of the state space $\mathcal{B}_{\mathrm{IX}}$.
More specifically, this orbit converges to the fixed point $({-1},{-1})$ (which
represents a non-flat LRS point) on $H_D = {-1}$.
(To prove this claim we note that ${r_\flat}_{|\mathrm{P}} < 0$, i.e.,
$\mathrm{P}$ is contained in the subset $r_\flat < 0$
of the state space $\overline{\mathcal{B}}_{\mathrm{IX}}$; this set
is future-invariant since $\beta > \beta_- > -1/2$.
The only possible attractor in the future-invariant set $\{r_\flat < 0\}$
set is the fixed point $({-1},{-1})$.)
Collecting the results we have proved that Fig.~\ref{yesspec} 
depicts the correct qualitative behavior of the flow on $\mathcal{B}_{\mathrm{IX}}$
in the special case \AminusP.
This complete the proof of the lemma.
\end{proof}

\begin{figure}[Ht]
\begin{center}
\psfrag{R}[cc][cc][1][0]{$\mathrm{R}_\flat$}
\psfrag{T}[cc][cc][1][0]{$\mathrm{T}_\flat$}
\psfrag{HD}[lc][lc][1.1][0]{$H_D^\prime > 0$}
\psfrag{P}[cc][cc][1][0]{$\mathrm{P}$}
\psfrag{Q}[cc][cc][1][0]{$\mathrm{Q}_\flat$}
\subfigure{\includegraphics[width=0.3\textwidth]{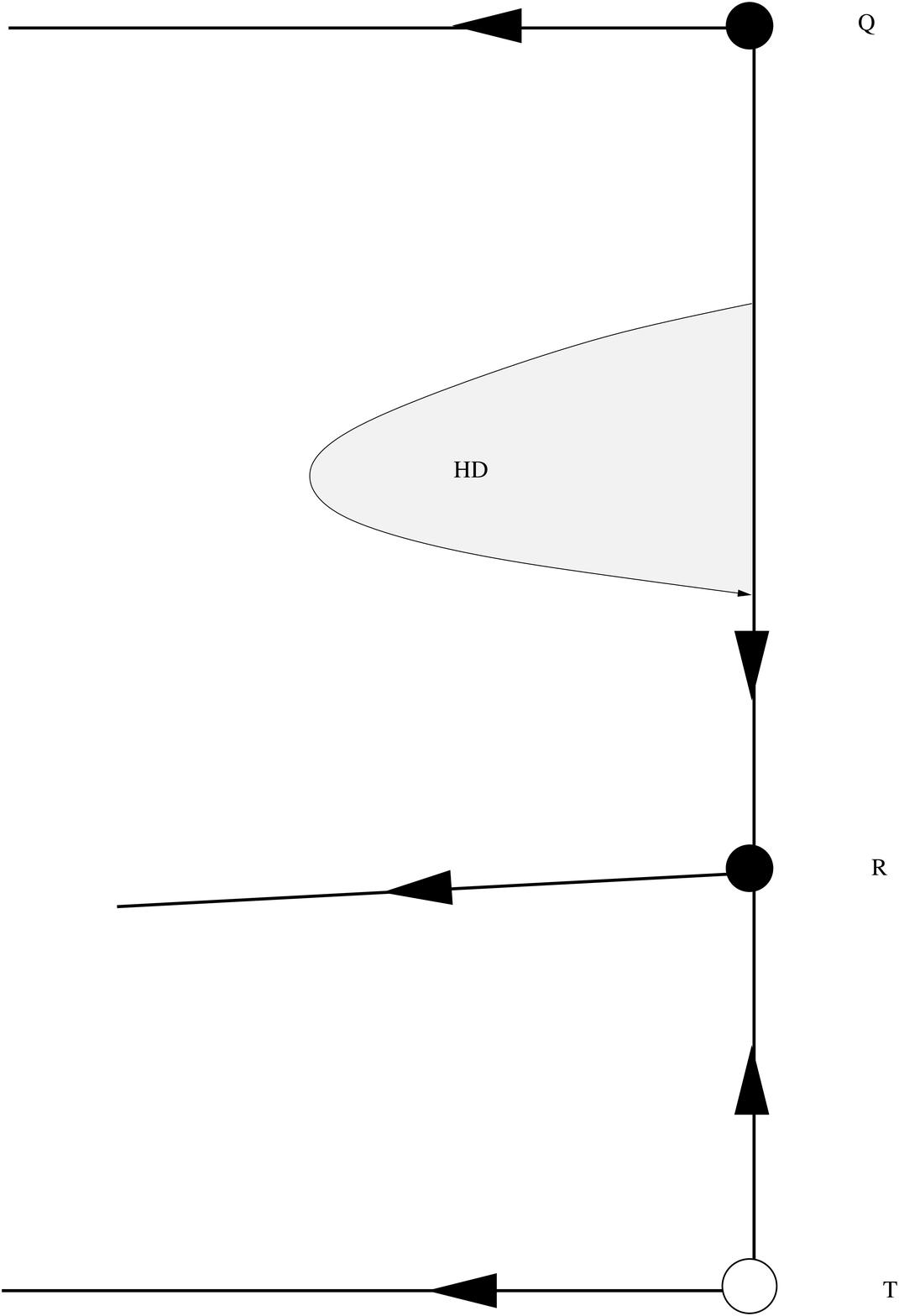}}\qquad\qquad\qquad\qquad
\subfigure{\includegraphics[width=0.3\textwidth]{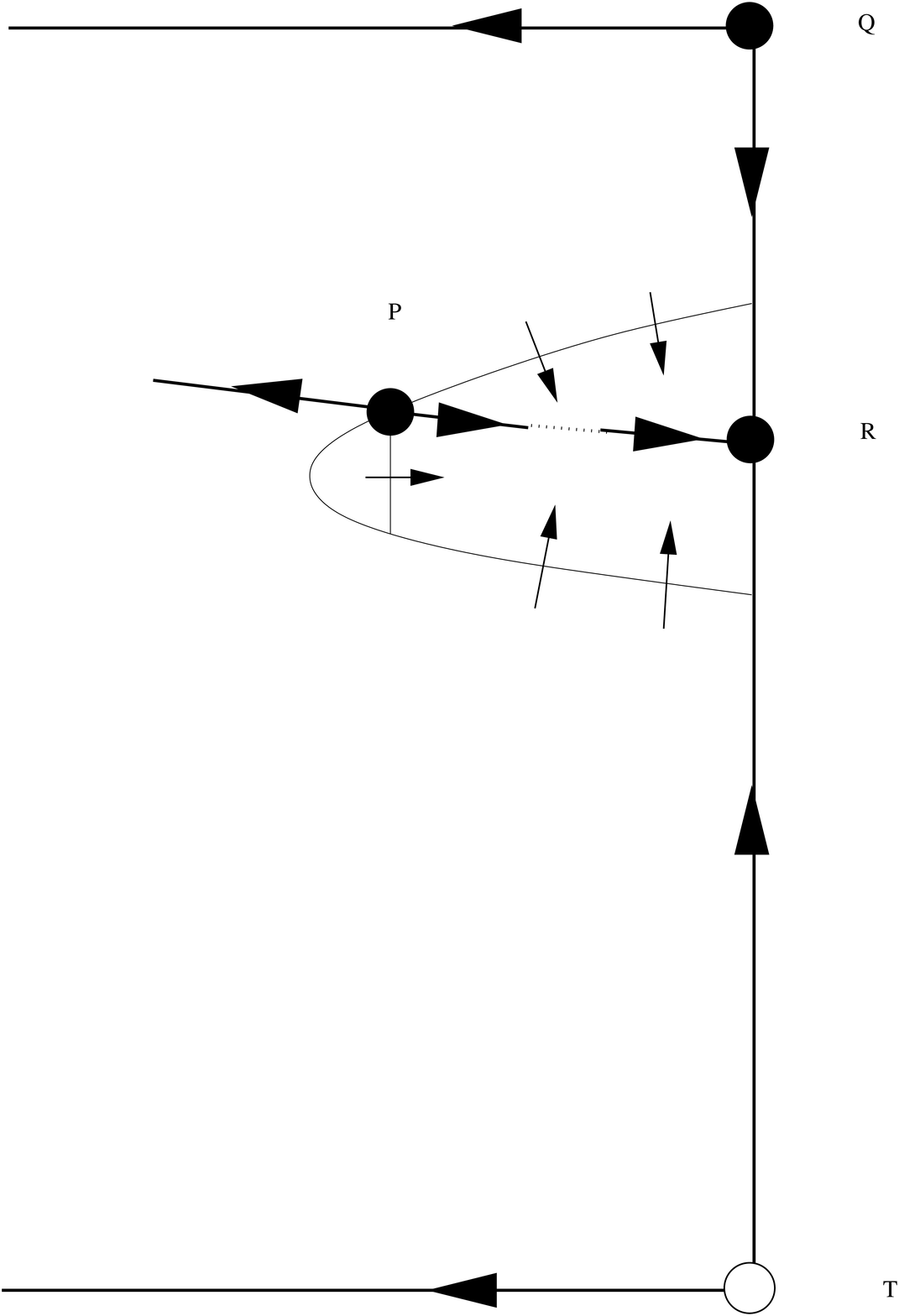}}
\caption{These figures depict a neighborhood of the boundary $H_D =1$ on $\overline{\mathcal{B}}_{\mathrm{IX}}$.
On $H_D =1$ (which is the vertical line) there are the fixed points $\mathrm{T}_\flat$, $\mathrm{Q}_\flat$,
and, in the \A\ cases, the fixed point $\mathrm{R}_\flat$. If $-\textfrac{1}{3} < w < \textfrac{1-\sqrt{3}}{3}$,
there exists a small region in $\mathcal{B}_{\mathrm{IX}}$, where $H_D^\prime > 0$.
If $\beta \not\in (\beta_-,\beta_+)$, which is the case for the figure on the l.h.\ side,
this does not have any influence on the qualitative global dynamics on $\mathcal{B}_{\mathrm{IX}}$.
However, if $\beta \in (\beta_-,\beta_+)$, which corresponds to the special case \AminusP, 
then there exists a fixed point, $\mathrm{P}$, in $\mathcal{B}_{\mathrm{IX}}$, and $\mathrm{R}_\flat$
is a sink. It is not difficult to construct a future-invariant set that is a neighborhood
of $\mathrm{R}_\flat$ and contains an orbit $\mathrm{P}\rightarrow \mathrm{R}_\flat$;
see the figure on the r.h.\ side.}
\label{BHDposneg}
\end{center}
\end{figure}

\begin{figure}[Ht!]
\begin{center}
\psfrag{+1}[cc][cc][0.7][0]{$+1$}
\psfrag{12}[cc][cc][0.7][0]{$\textfrac{1}{2}$}
\psfrag{-1}[cc][cc][0.7][0]{$-1$}
\psfrag{0}[cc][cc][0.7][0]{$0$}
\psfrag{rc}[cc][cc][1][0]{$r$}
\psfrag{h}[cc][cc][1][0]{$H_D$}
\psfrag{h0}[cc][cc][0.7][90]{$H_D=0$}
\psfrag{sig}[cc][cc][1][0]{$\Sigma_+$}
\psfrag{r}[cc][cc][0.7][0]{$\mathrm{R}_\flat$}
\psfrag{t}[cc][cc][0.7][0]{$\mathrm{T}_\flat$}
\psfrag{q}[c][cc][0.7][0]{$\mathrm{Q}_\flat$}
\psfrag{R}[cc][cc][0.7][0]{$\mathrm{R}_\flat$}
\psfrag{T}[cc][cc][0.7][0]{$\mathrm{T}_\flat$}
\psfrag{Q}[cc][cc][0.7][0]{$\mathrm{Q}_\flat$}
\psfrag{B}[cc][cc][1][0]{$\mathcal{B}_{\mathrm{IX}}$}
\subfigure[\Bminus, \Cminus, \Dminus]{\includegraphics[width=0.37\textwidth]{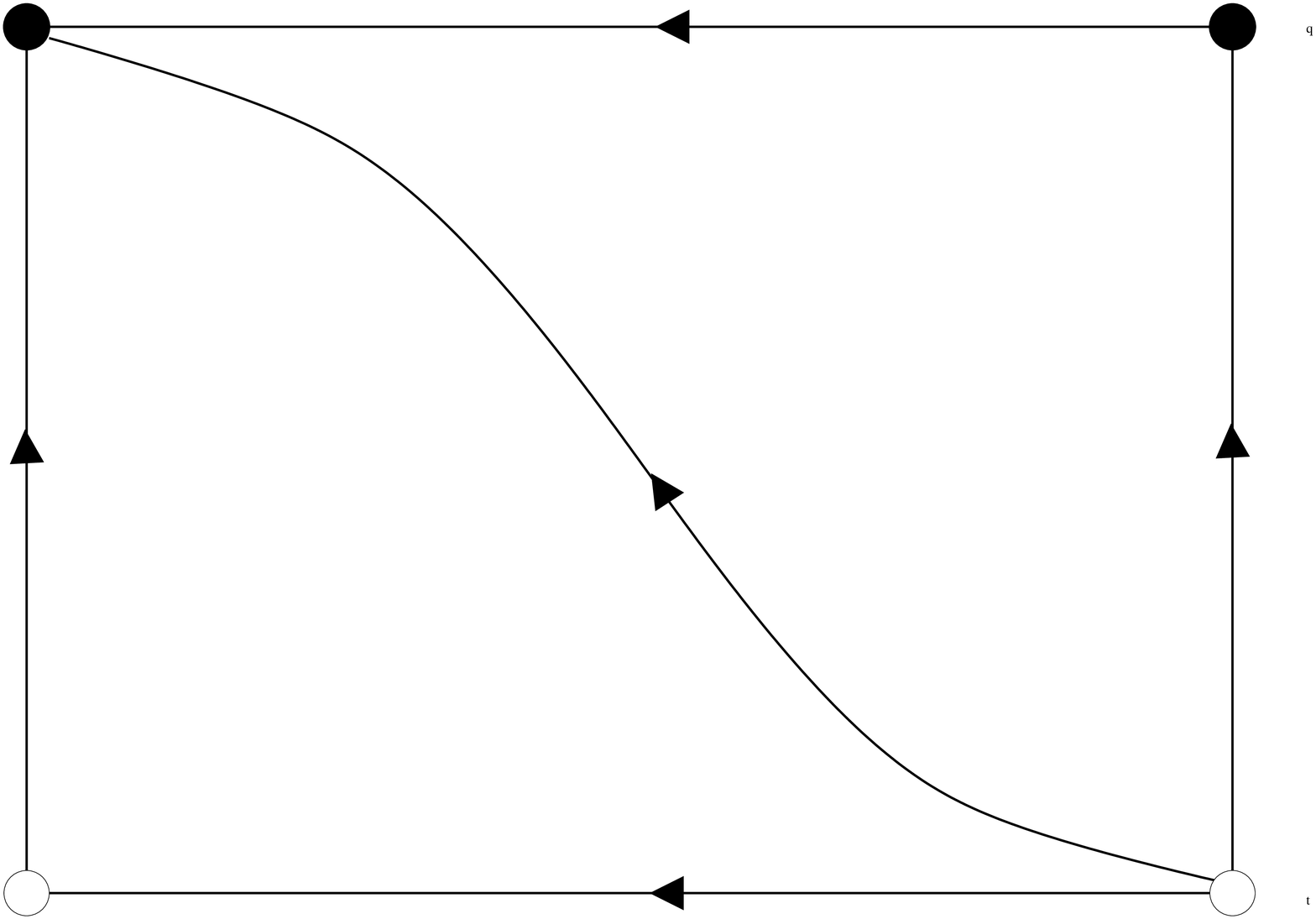}\label{Bfigmag}}\qquad\qquad
\subfigure[\Aminus\ {(\scriptsize $\beta<-1/2$)}]{\includegraphics[width=0.37\textwidth]{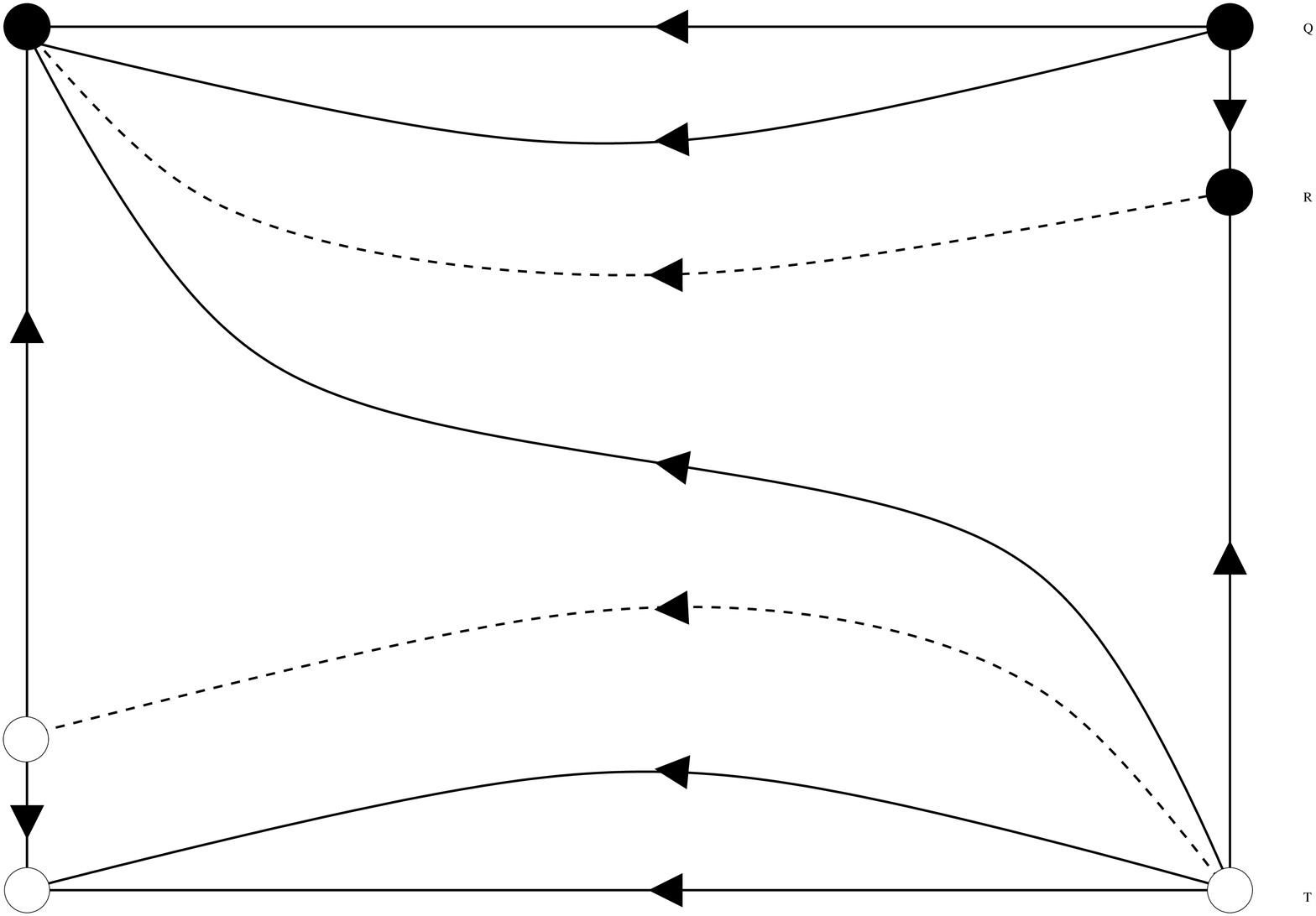}}\\
\subfigure[\Aminus\ {(\scriptsize $\beta=-1/2$)}]{\includegraphics[width=0.37\textwidth]{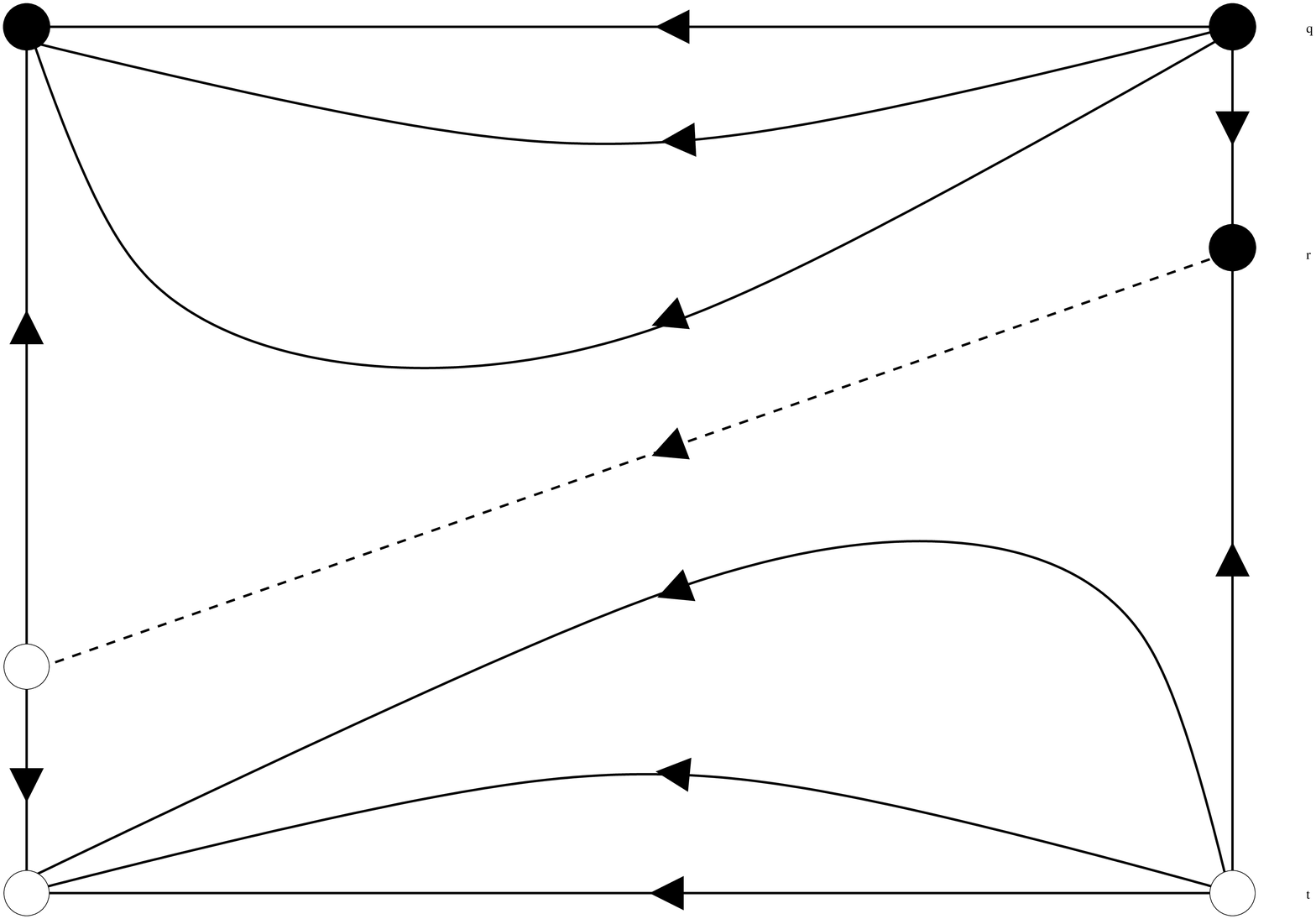}}\qquad\qquad
\subfigure[\Azerominus, \Aminus\ {(\scriptsize $\beta>-1/2$)}, not \AminusP]{\label{nospec}\includegraphics[width=0.37\textwidth]{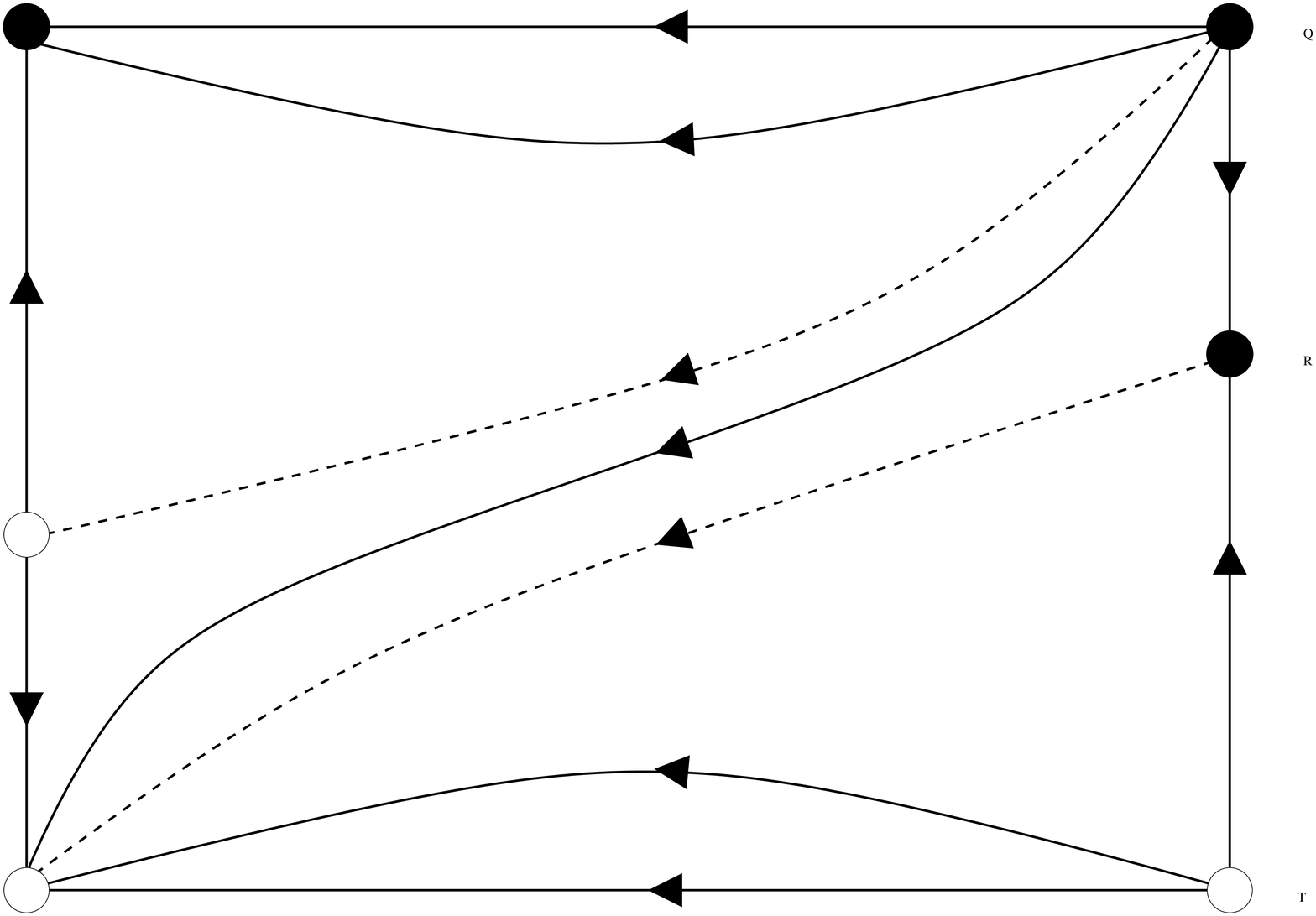}}\\
\subfigure[\AminusP]{\label{yesspec}\includegraphics[width=0.37\textwidth]{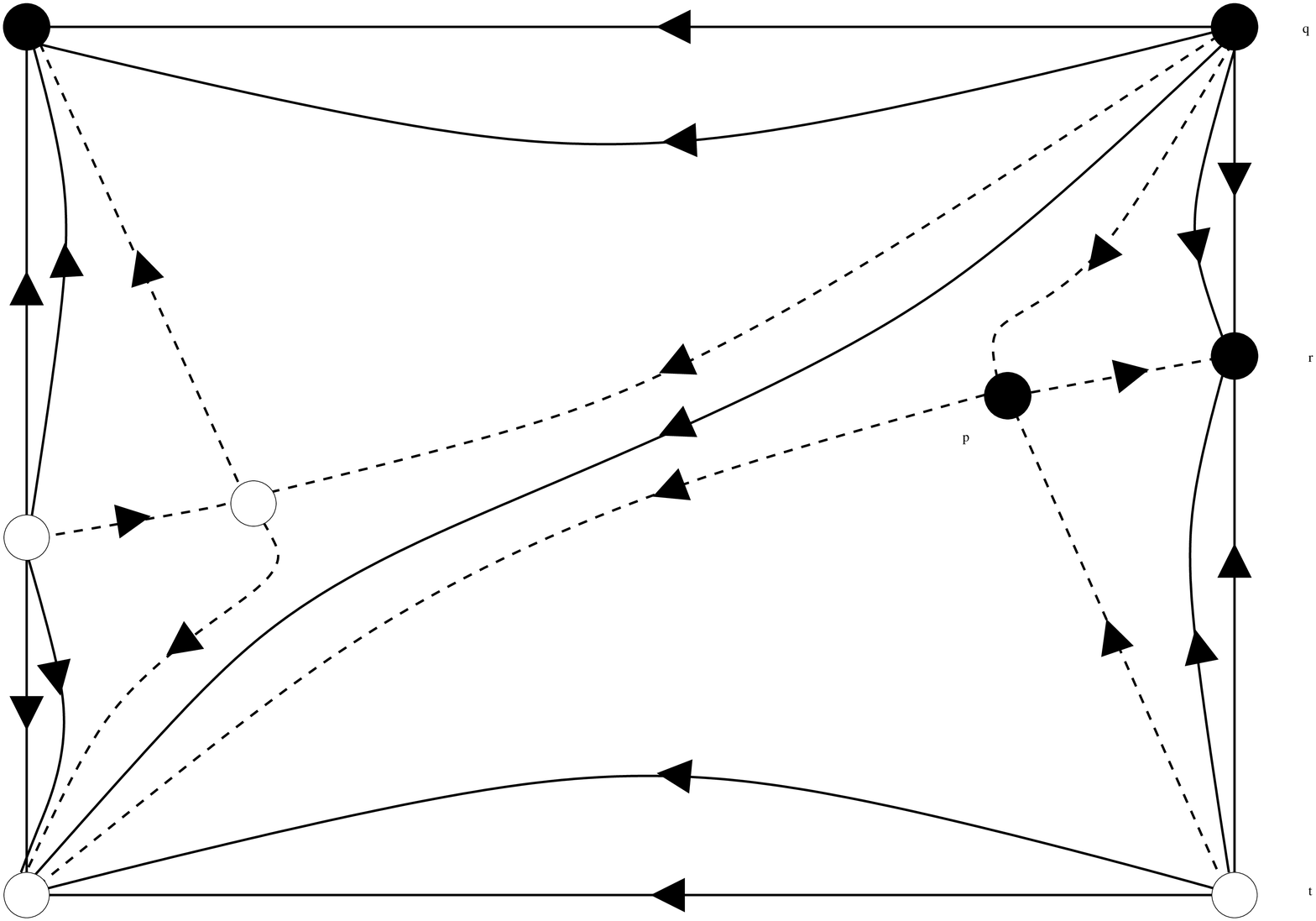}}\qquad\qquad 
\subfigure[\Azeroplus, \Aplus]{\includegraphics[width=0.37\textwidth]{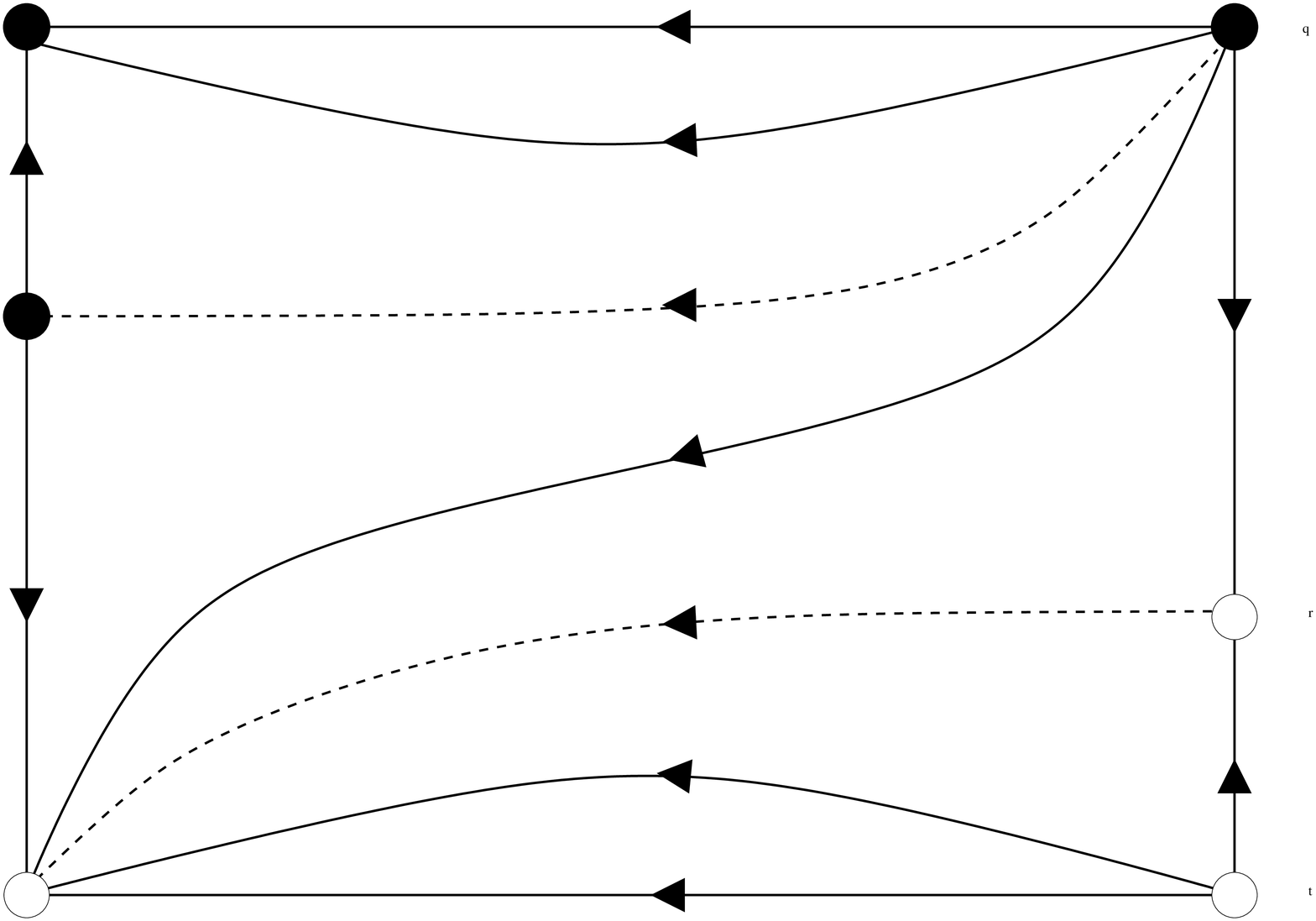}}\\
\subfigure[\Bplus, \Cplus, \Dplus]{\includegraphics[width=0.37\textwidth]{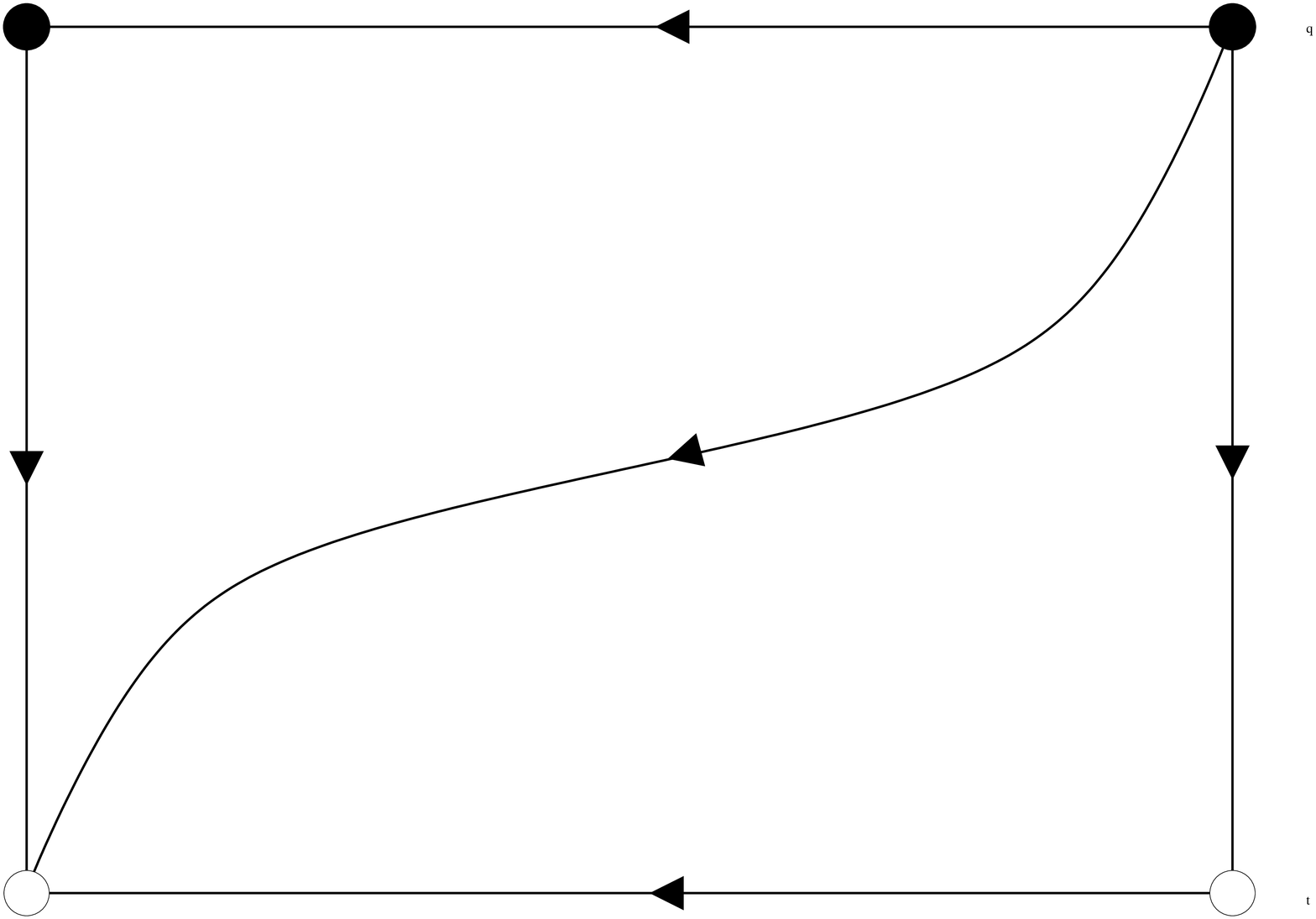}}\qquad\qquad
\subfigure[Coordinates]{\label{Bcoords}\includegraphics[width=0.37\textwidth]{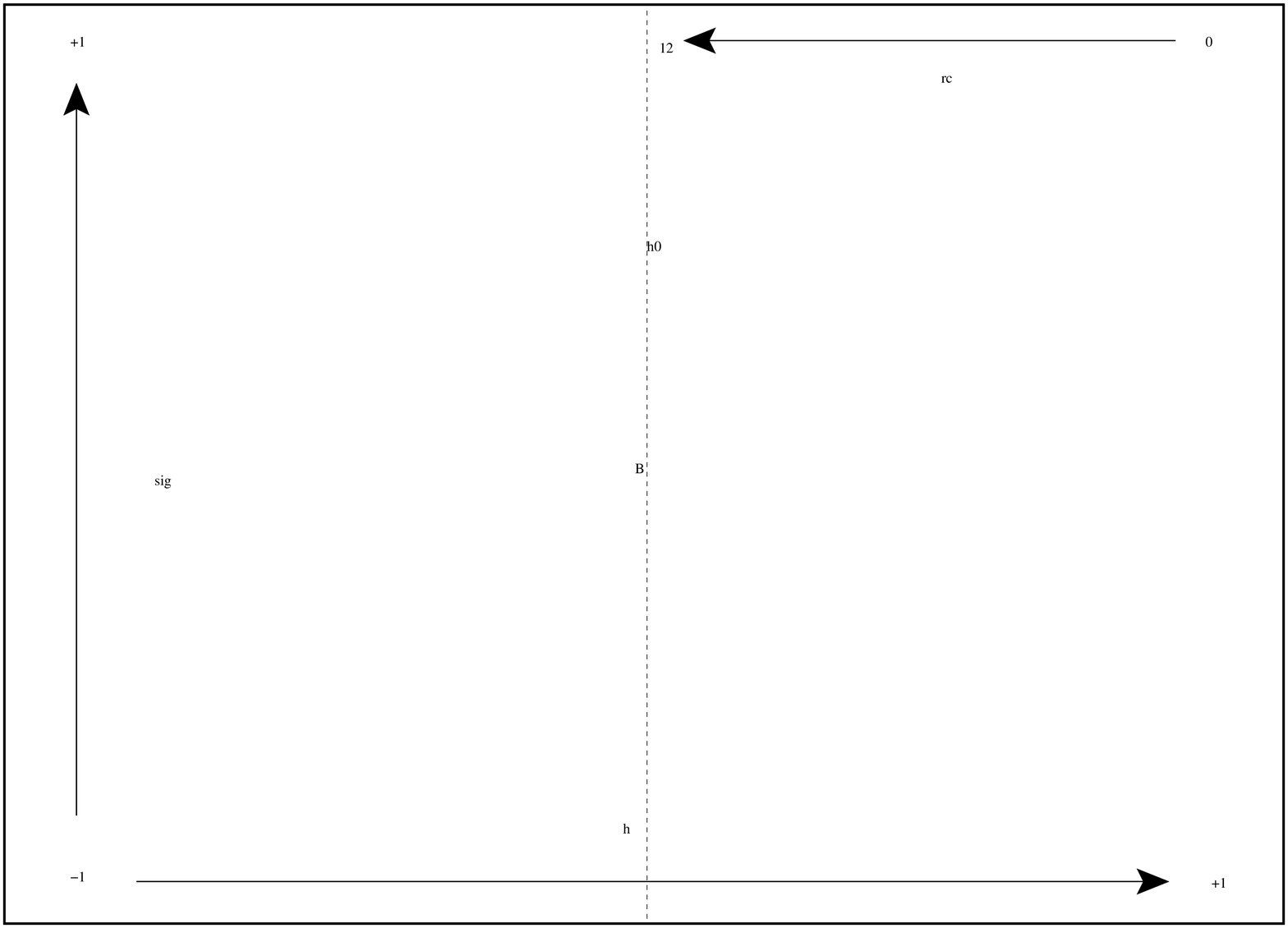}}
\caption{Phase portraits of the dynamical system~\eqref{Basesys} on $\mathcal{B}_{\mathrm{IX}}$.
Typical orbits are represented by continuous lines, non-generic orbits by dashed lines.
The color-coding of the fixed points $\mathrm{T}_\flat$, $\mathrm{Q}_\flat$, $\mathrm{R}_\flat$ makes
sense only in the context of the dynamical systems formulation of Sec.~\ref{blowup};
the coordinate $r$ of Subfig.~(h) is used in that section.}
\label{Bfig}
\end{center}
\end{figure}

%----------------------------------------------------------------------
\subsection{The Bianchi type~I boundary}\label{blowup}
%----------------------------------------------------------------------

In this subsection we introduce the tools that are needed to study the behavior 
of solutions of the dynamical system~\eqref{89syst} 
in a neighborhood of the line $\mathcal{L}_\mathrm{I} = \partial\mathcal{S}_\sharp\cap\partial\mathcal{B}_{\mathrm{IX}}$, 
cf.~\eqref{badboundary}.
A careful analysis is necessary since the system~\eqref{89syst} does not
admit a smooth extension to that part of the boundary of the state space
$\mathcal{X}_{\mathrm{IX}}$.

To remedy this defect of the system~\eqref{89syst} 
we introduce a set of `polar coordinates' centered on the line $\mathcal{L}_\mathrm{I}$. 
Let 
\begin{subequations}\label{polartransf}
\begin{align}
M_1^2 & =  3 \,r \sin \vartheta\;, \\[0.2ex]
\label{1minusHD2}
1 - H_D^2 & = 2 \,r \cos\vartheta\;, \\[0.2ex]
\Sigma_+ & =\text{unchanged}\:.
\end{align}
\end{subequations}
The transformation of variables from $(H_D, M_1,\Sigma_+)$ to $(r, \vartheta,\Sigma_+)$, where $r>0$ and
$0<\vartheta< \textfrac{\pi}{2}$,
is a diffeomorphism on the domain $(H_D, M_1,\Sigma_+) \in\mathcal{X}_\mathrm{IX}^{+}$.
We define the domain $\mathcal{Y}_{\mathrm{IX}}^+$ of the variables $(r,\vartheta,\Sigma_+)$ to be 
the preimage of the state space $\mathcal{X}_\mathrm{IX}^{+}$ under the transformation~\eqref{polartransf}.
We obtain $r \cos\vartheta < \textfrac{1}{2}$ from~\eqref{1minusHD2}; the constraint
$1 - \Sigma_+^2 - \textfrac{1}{12}\,M_1^2 = \Omega > 0$ implies
$r\sin\vartheta < 4 (1 -\Sigma_+^2)$. Therefore, $\mathcal{Y}_{\mathrm{IX}}^+$ can be written
as 
\begin{equation}\label{statespaceIXpolar}
\mathcal{Y}_{\mathrm{IX}}^+ = \Big\{ r>0, \vartheta \in(0,\textfrac{\pi}{2}), \Sigma_+ \in (-1,1)\:\big|\:
r < \min\left[\textfrac{1}{2 \cos\vartheta},\textfrac{4(1-\Sigma_+^2)}{\sin\vartheta}\right] \,\Big\}\:,
\end{equation}
see Fig.~\ref{Yspace}.
By construction, the state space $\mathcal{Y}_{\mathrm{IX}}^+$ is past invariant under the flow
of~\eqref{dynsyspolar} (because $\mathcal{X}_\mathrm{IX}^{+}$ is past invariant).

\begin{figure}[Ht]\label{Y}
\begin{center}
\psfrag{sigma}[cc][cc][0.7][0]{$\Sigma_+$}
\psfrag{r}[cc][cc][0.7][0]{$r$}
\psfrag{theta}[cc][cc][0.7][0]{$\vartheta$}
\psfrag{theta}[cc][cc][0.7][0]{$\vartheta$}
\psfrag{L}[cc][cc][1][20]{$\mathcal{X}_\mathrm{I}$}
\psfrag{B}[cc][cc][1][0]{$\mathcal{B}_{\mathrm{IX}}^+$}
\psfrag{S}[cc][cc][1][0]{$\mathcal{S}_\sharp$}
\psfrag{h0}[cl][cl][1][-45]{$H_D=0$}
\subfigure[The space $\mathcal{Y}_{\mathrm{IX}}^+$]{\includegraphics[width=0.4\textwidth]{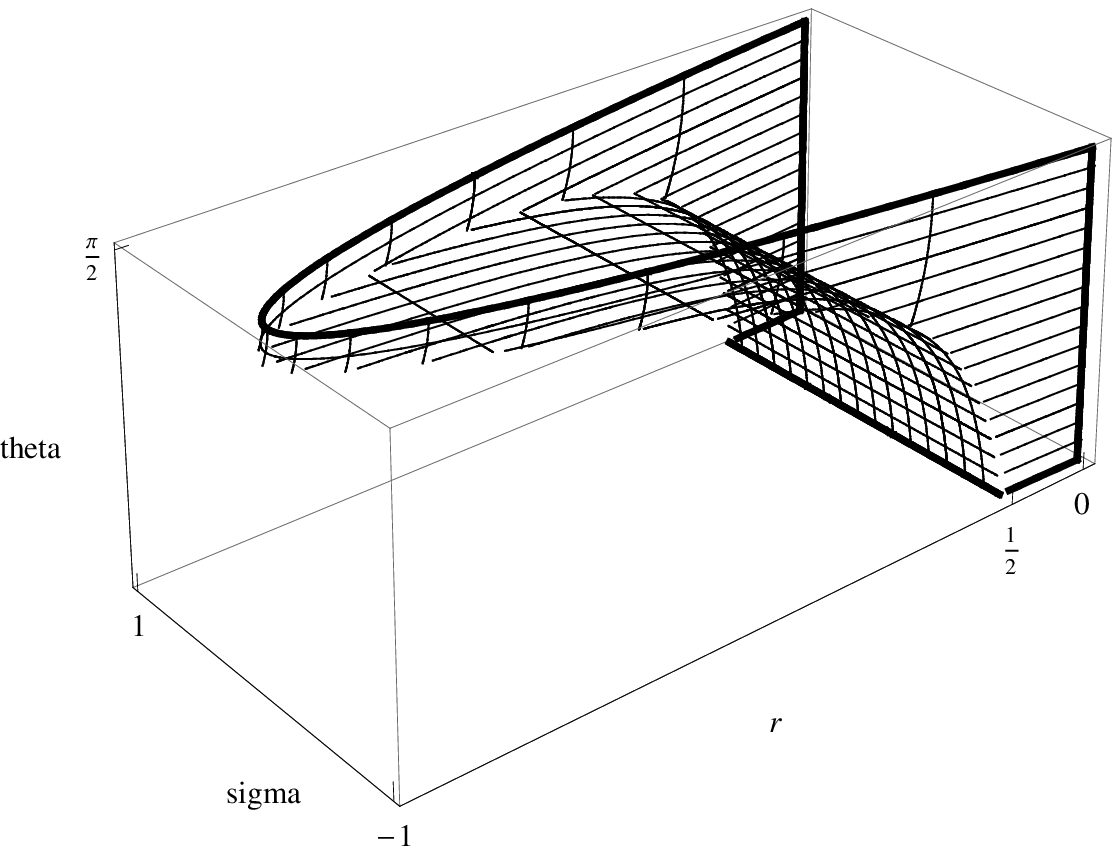}\label{Yspace}}\qquad\qquad
\subfigure[A schematic depiction of $\partial\mathcal{Y}_{\mathrm{IX}}^+$]{\includegraphics[width=0.3\textwidth]{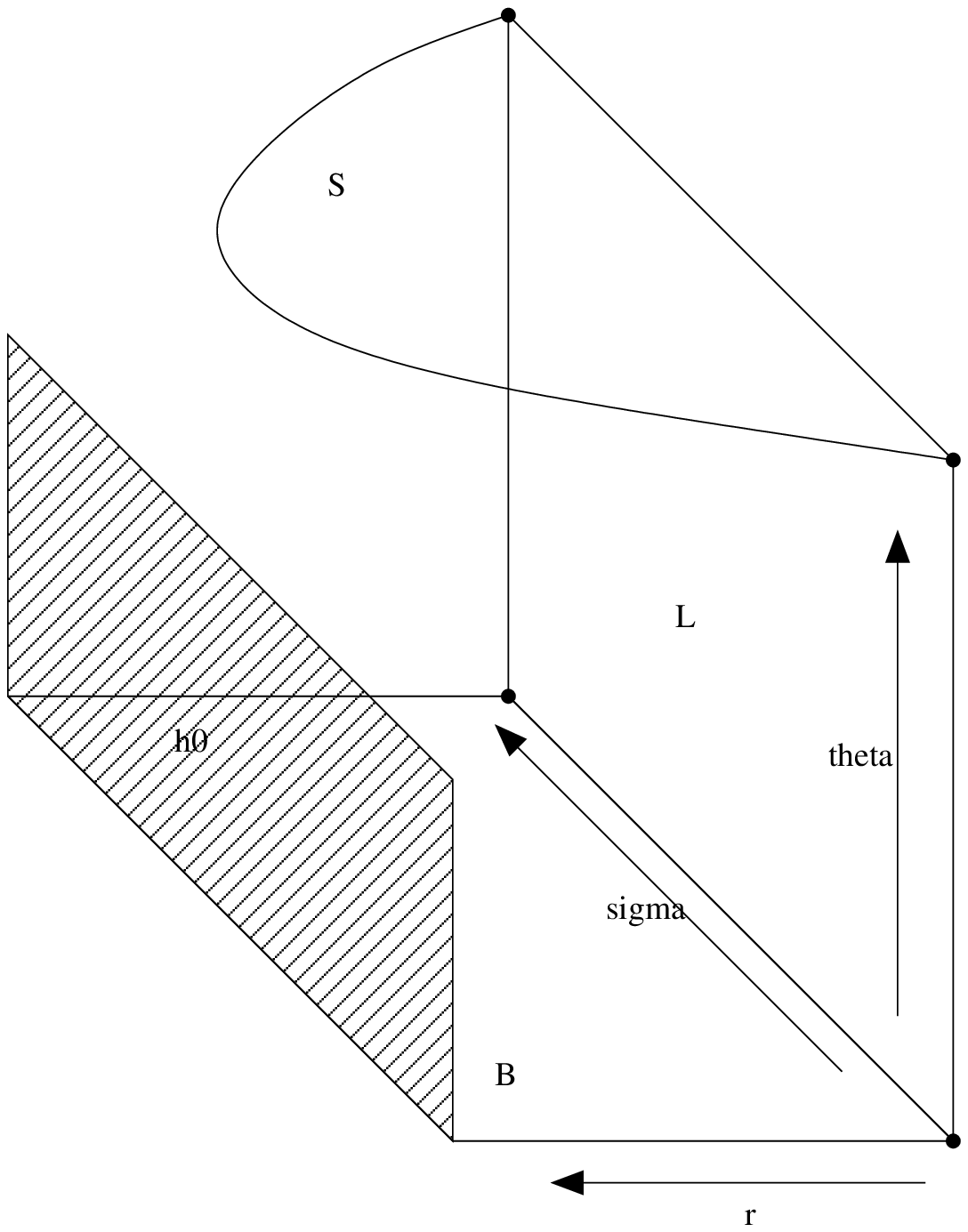}}
\caption{The space $\mathcal{Y}_{\mathrm{IX}}^+$ and its boundary.}
\end{center}
\end{figure}

In the new variables, the dynamical system~\eqref{89syst} takes the form
\begin{subequations}\label{dynsyspolar}
\begin{align}
r^\prime & = 2\, r \left( H_D (q - H_D \Sigma_+) - 3 \Sigma_+ \sin^2 \vartheta \right), \\[0.5ex]
\vartheta^\prime & = - 3 \Sigma_+ \sin (2\vartheta)\:, \\[0.5ex]
\label{polarSigma+}
\Sigma_+^\prime & = r \sin\vartheta -1 + (H_D -\Sigma_+)^2 + H_D \Sigma_+ (q- H_D\Sigma_+) + 3 \Omega\, (u(s)-w)\:,
\end{align}
\end{subequations}
where $q = 2 \Sigma_+^2 + \textfrac{1}{2} (1+ 3 w) \Omega$ and 
$\Omega =  1 - \Sigma_+^2 - \textfrac{1}{4}\, r \sin\vartheta$.
In addition, $H_D$ and $s$ are regarded as functions of $r$ and $\vartheta$ in~\eqref{dynsyspolar},
\begin{equation}\label{ho}
H_D  = \sqrt{1 -2 r \cos\vartheta} \;,\qquad
s=\frac{1}{2}\:\frac{\tan \vartheta}{1 + \tan \vartheta} \:.
\end{equation}
Note that these are smooth functions of $r$ and $\vartheta$ 
on the state space $\mathcal{Y}_{\mathrm{IX}}^+$. In particular,
\begin{equation}\label{sjprime}
\frac{\partial s}{\partial \vartheta} = \frac{1}{\sin(2 \vartheta)} \,2 s (1- 2s)  = 
\frac{1}{2} \frac{1}{1+ \sin (2\vartheta)} \:,
\qquad\text{hence}\quad
\frac{1}{4} \leq \frac{\partial s}{\partial \vartheta} \leq \frac{1}{2} 
\quad\forall \vartheta \in [0,\textfrac{\pi}{2}]
\end{equation}
and 
$\partial s/\partial \vartheta = \textfrac{1}{2}$ at $\vartheta = 0$ and $\vartheta = \textfrac{\pi}{2}$.

In contrast to the system~\eqref{89syst}, 
the dynamical system~\eqref{dynsyspolar} on the state space $\mathcal{Y}_{\mathrm{IX}}^+$
\textit{admits a regular extension to the boundaries} (where we assume $H_D \geq \mathrm{const} > 0$).
The reason for this is that $s$ and thus the function $u(s)$ in~\eqref{polarSigma+} 
is well-defined for each point on the boundary of the state space.

\textbf{The side} $\bm{\mathcal{S}_\sharp}$.
In terms of the variables~\eqref{polartransf}, the set $\mathcal{S}_\sharp$, cf.~\eqref{Ssharpagain},
corresponds to 
\[
\mathcal{S}_\sharp = \big\{ \vartheta = \frac{\pi}{2}\,; \:{-1}< \Sigma_+ < 1, \: 0< r <  4(1-\Sigma_+^2)\:\big\}\:.
\]
(We use the symbol $\mathcal{S}_\sharp$ with slight abuse of notation.) 
The induced dynamical system reduces to the system~\eqref{Sidesys}, where $M_1$ is replaced by $\sqrt{3 r}$.
The phase portraits for the various cases are depicted in Fig.~\ref{Asharpfig}. 

\textbf{The base} $\bm{\mathcal{B}_{\mathrm{IX}}}$.
In terms of the variables~\eqref{polartransf}, the (past invariant half of the) set $\mathcal{B}_{\mathrm{IX}}$, cf.~\eqref{basedef},
corresponds to 
\[
\mathcal{B}^+_\mathrm{IX} = \big\{ \vartheta = 0\,; \:{-1}< \Sigma_+ < 1, \: 0< r <  \textfrac{1}{2}\:\big\}\:.
\]
It is an invariant subset of the system~\eqref{dynsyspolar} on $\overline{\mathcal{Y}}_{\mathrm{IX}}$ as well.
The induced system on $\mathcal{B}^+_\mathrm{IX}$ is equivalent to the system~\eqref{Basesys} on 
the past invariant half of $\mathcal{B}_{\mathrm{IX}}$,
when we set $H_D = \sqrt{1 - 2 r}$.
The phase portraits for the various cases are depicted in Fig.~\ref{Bfig}, where
we restrict our attention to the past invariant (i.e., the right) half.

\textbf{The vacuum boundary}.
Setting $\Omega = 0$ in~\eqref{dynsyspolar} yields 
the (past invariant half of the) vacuum subset.
The vacuum subset of $\overline{\mathcal{Y}}_{\mathrm{IX}}^+$ is different
from $\overline{\mathcal{V}}_{\mathrm{IX}}^+$, which is depicted in Fig.~\ref{vacuumIXflow}. 
The main difference concerns the fixed points: While there are
two fixed points, $\mathrm{T}$ and $\mathrm{Q}$, on the past invariant half of the 
vacuum subset $\overline{\mathcal{V}}_{\mathrm{IX}}^+$ of $\overline{\mathcal{X}}_{\mathrm{IX}}^+$,
there are four fixed points on the vacuum subset of $\overline{\mathcal{Y}}_{\mathrm{IX}}^+$:
\begin{itemize}
\item[$\mathrm{T}_\flat$] \quad The Taub point $\mathrm{T}_\flat$ is given by $(r,\vartheta, \Sigma_+) = (0,0,-1)$.
\item[$\mathrm{T}_\sharp$] \quad The Taub point $\mathrm{T}_\sharp$ is given by $(r,\vartheta, \Sigma_+) = (0,\textfrac{\pi}{2},-1)$.
\item[$\mathrm{Q}_\flat$] \quad The non-flat LRS point $\mathrm{Q}_\flat$ is given by $(r,\vartheta, \Sigma_+) = (0,0,1)$.
\item[$\mathrm{Q}_\sharp$] \quad The non-flat LRS point $\mathrm{Q}_\sharp$ is given by $(r,\vartheta, \Sigma_+) = (0,\textfrac{\pi}{2},1)$.
\end{itemize}
The transformation~\eqref{polartransf} duplicates the fixed points $\mathrm{T}$ and $\mathrm{Q}$;
more specifically, there exist two orbits on the vacuum subset of $\overline{\mathcal{Y}}_{\mathrm{IX}}^+$,
an orbit $\mathrm{T}_\flat \rightarrow \mathrm{T}_\sharp$ and an orbit $\mathrm{Q}_\flat \leftarrow \mathrm{Q}_\sharp$,
that are collapsed to the fixed points $\mathrm{T}$ and $\mathrm{Q}$ on $\overline{\mathcal{V}}_{\mathrm{IX}}^+$.
This becomes apparent when we consider the fourth boundary subset of $\mathcal{Y}_{\mathrm{IX}}^+$,
which does not have a direct correspondence on $\overline{\mathcal{X}}_{\mathrm{IX}}^+$.

\textbf{The Bianchi type~I boundary} $\bm{\mathcal{X}_{\,\mathrm{I}}}$.
There exists an additional two-dimensional boundary subset of $\mathcal{Y}_{\mathrm{IX}}^+$:
The set $r=0$.
\[
\mathcal{X}_{\,\mathrm{I}} = \big\{ r = 0\,; \:{-1}< \Sigma_+ < 1, \: 0< \vartheta <  \textfrac{\pi}{2}\:\big\}\:.
\]
This boundary subset can be regarded as the preimage of the line $\mathcal{L}_\mathrm{I}$ under the
map~\eqref{polartransf}.
Since the map~\eqref{polartransf} is not injective for $r=0$ 
(because $(r,\vartheta)=(0,\vartheta) \mapsto (M_1,H_D) = (0,1)$, $\forall \vartheta$),
the set $\mathcal{X}_{\,\mathrm{I}}$ is two-dimensional;
the introduction of the coordinates $(r,\vartheta)$ generates a blow-up of 
the line $\mathcal{L}_\mathrm{I}$.
The dynamical system~\eqref{dynsyspolar}
\textit{extends regularly} to the subset $\overline{\mathcal{X}}_{\,\mathrm{I}}$; the extension reads
\begin{subequations}\label{Ibou}
\begin{align}
\vartheta^\prime & =- 3 \Sigma_+ \sin(2\vartheta) \\
\Sigma_+^\prime & =-3(1-\Sigma_+^2)\left[\textfrac{1}{2}(1-w)\Sigma_+-\left(u(s)-w\right)\right]\:.
\end{align}
\end{subequations}
The variable $s$ is given in terms of $\vartheta$ by~\eqref{ho};
from~\eqref{sjprime} it is immediate that
the system~\eqref{Ibou} coincides with the Bianchi type~I system~\eqref{dynsysbianchiI};
therefore, we call $\mathcal{X}_{\,\mathrm{I}}$ the Bianchi type~I subset.
The phase portraits for the various cases are depicted in Fig.~\ref{BianchiIfig2};
the difference between Fig.~\ref{BianchiIfig} and Fig.~\ref{BianchiIfig2} 
is in the color-coding of 
the fixed points $\mathrm{Q}_\flat$ and $\mathrm{R}_\flat$. 
Note in particular that $\mathrm{R}_\flat$ attracts orbits 
from the orthogonal direction in the Bianchi type~IX special case \AminusP, 
see Fig.~\ref{AminusspecI}; this follows by computing
\[
\textfrac{d}{d\tau}\left[\log r\right]_{|\mathrm{R}_\flat}=3(1-w)\beta^2+2\beta+(1+3w)\:,
\]
which is negative in the case \AminusP\ and thus consistent with Fig.~\ref{yesspec}. %Fig.~\ref{Bfig}.  

\begin{figure}[Ht!]
\begin{center}
\psfrag{T}[cc][cc][0.7][0]{$\mathrm{T}_\flat$}
\psfrag{T*}[cc][cc][0.7][0]{$\mathrm{T}_\sharp$}
\psfrag{Q}[cc][cc][0.7][0]{$\mathrm{Q}_\flat$}
\psfrag{Q*}[cc][cc][0.7][0]{$\mathrm{Q}_\sharp$}
\psfrag{F}[cc][cc][0.7][0]{$\mathrm{F}$}
\subfigure[\Dminus]{\label{biDmin2}\includegraphics[width=0.3\textwidth]{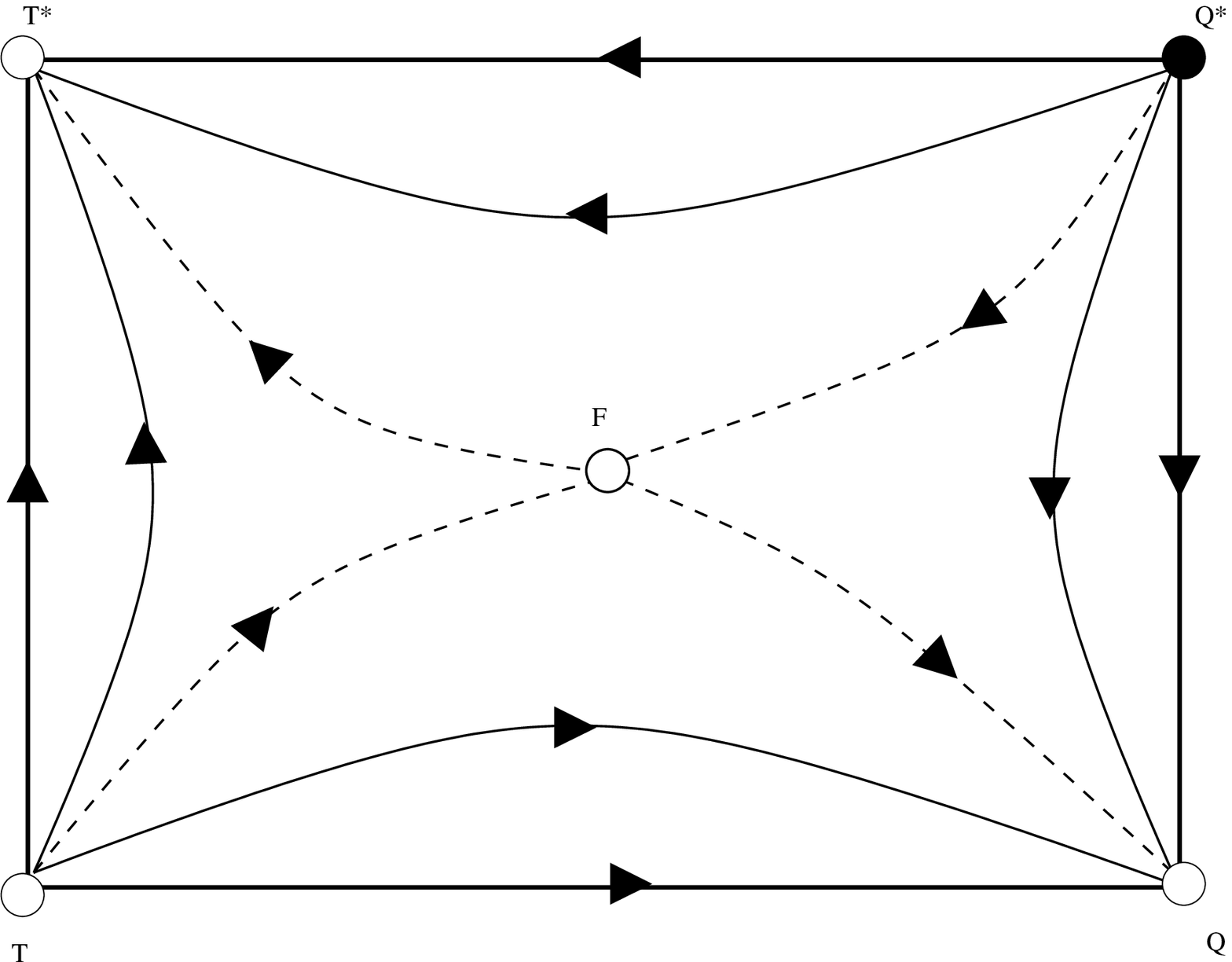}}\quad
\psfrag{T}[cc][cc][0.7][0]{$\mathrm{T}_\flat$}
\psfrag{T*}[cc][cc][0.7][0]{$\mathrm{T}_\sharp$}
\psfrag{Q}[cc][cc][0.7][0]{$\mathrm{Q}_\flat$}
\psfrag{Q*}[cc][cc][0.7][0]{$\mathrm{Q}_\sharp$}
\psfrag{D*}[cc][cr][0.7][0]{$\mathrm{R}_\sharp$}
\psfrag{F}[cc][cc][0.7][0]{$\mathrm{F}$}
\subfigure[\Bminus, \Cminus]{\includegraphics[width=0.3\textwidth]{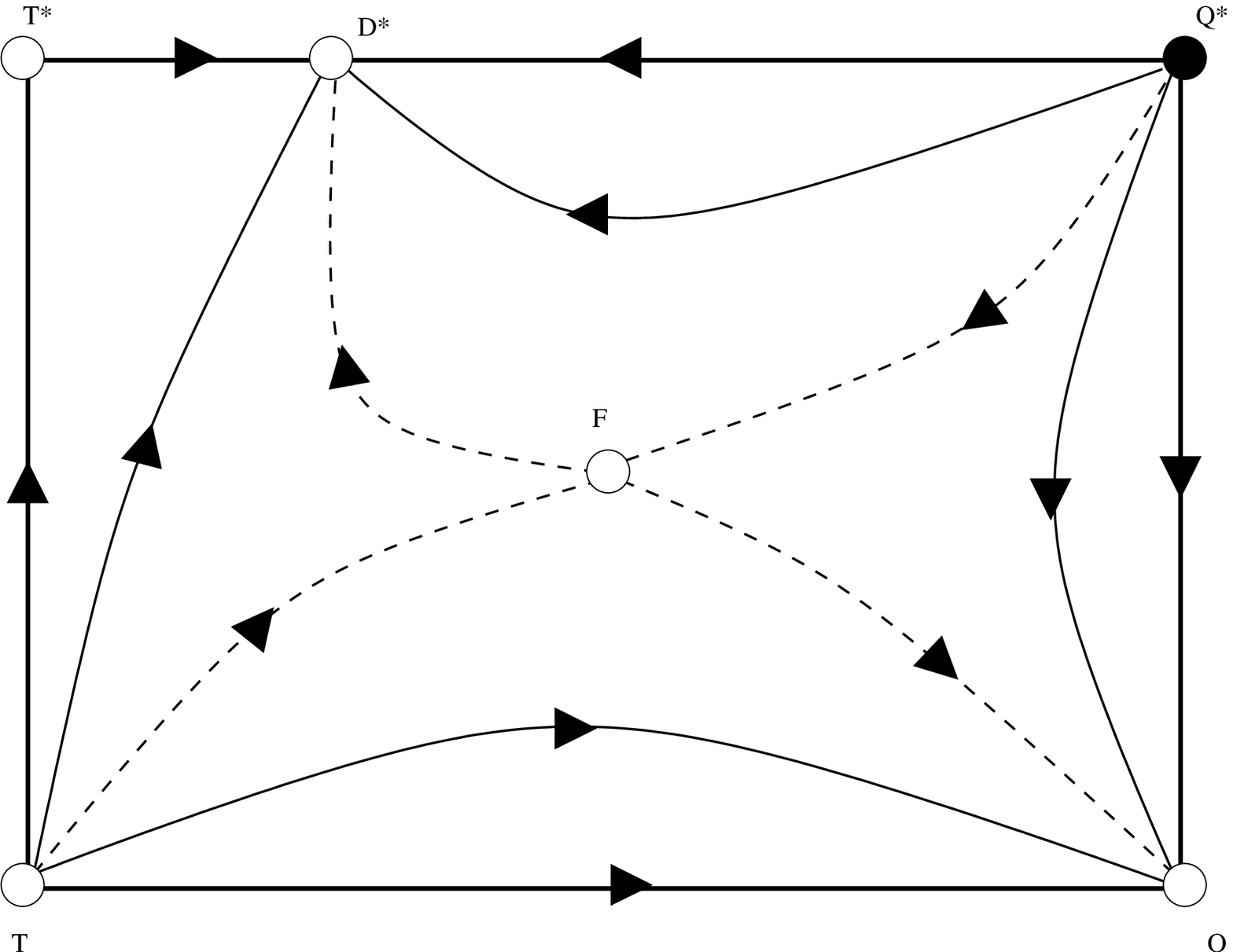}}\quad
\psfrag{T}[cc][cc][0.7][0]{$\mathrm{T}_\flat$}
\psfrag{T*}[cc][cc][0.7][0]{$\mathrm{T}_\sharp$}
\psfrag{Q}[cc][cc][0.7][0]{$\mathrm{Q}_\flat$}
\psfrag{Q*}[cc][cc][0.7][0]{$\mathrm{Q}_\sharp$}
\psfrag{R}[rl][cc][0.7][0]{$\mathrm{R}_\flat$}
\psfrag{D*}[cc][cr][0.7][0]{$\mathrm{R}_\sharp$}
\psfrag{F}[cc][cc][0.7][0]{$\mathrm{F}$}
\subfigure[\AminusP]{\label{AminusspecI}\includegraphics[width=0.3\textwidth]{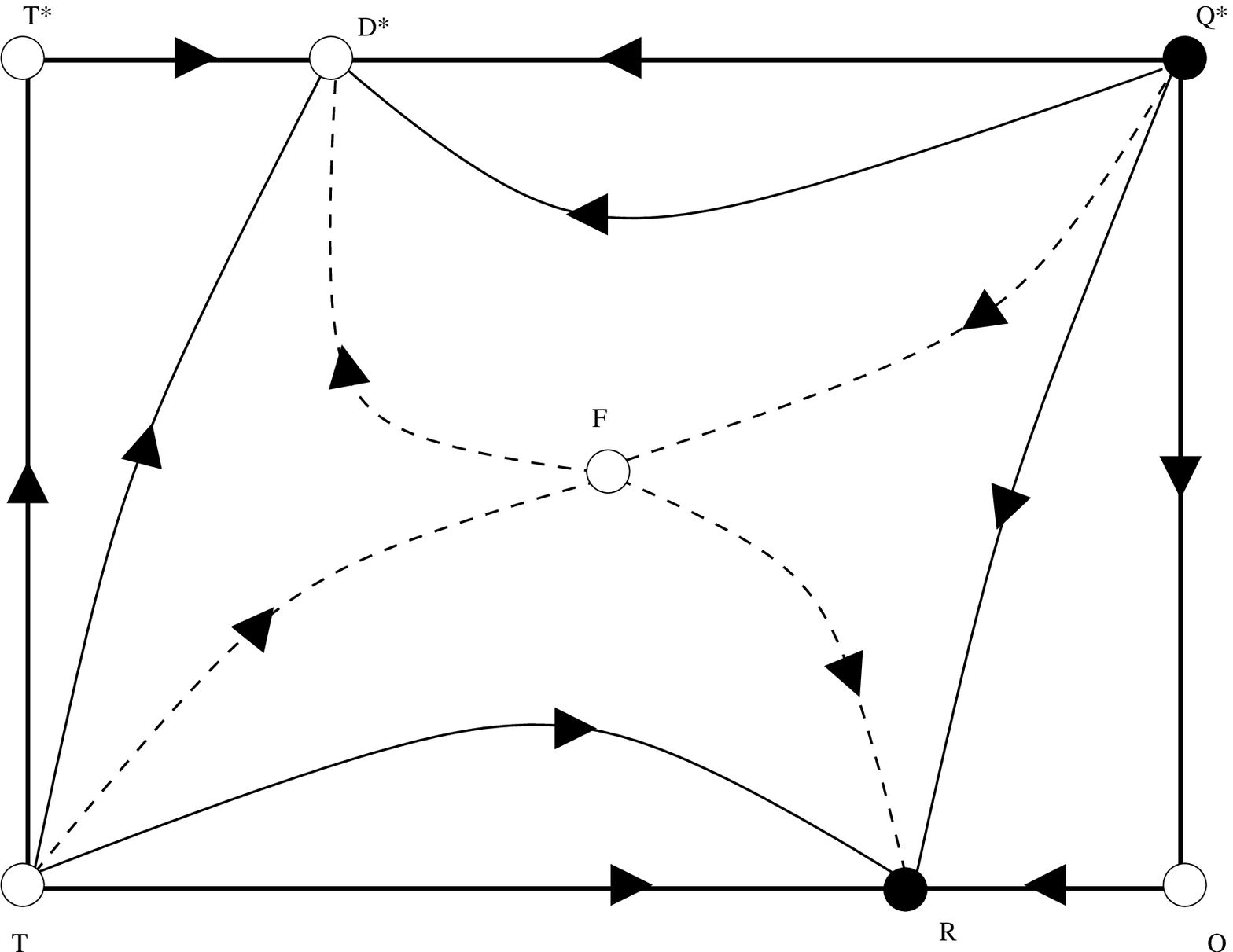}}\\
\psfrag{T}[cc][cc][0.7][0]{$\mathrm{T}_\flat$}
\psfrag{T*}[cc][cc][0.7][0]{$\mathrm{T}_\sharp$}
\psfrag{Q}[cc][cc][0.7][0]{$\mathrm{Q}_\flat$}
\psfrag{Q*}[cc][cc][0.7][0]{$\mathrm{Q}_\sharp$}
\psfrag{R}[rl][cc][0.7][0]{$\mathrm{R}_\flat$}
\psfrag{D*}[cc][cr][0.7][0]{$\mathrm{R}_\sharp$}
\psfrag{F}[cc][cc][0.7][0]{$\mathrm{F}$}
\subfigure[\Aminus$\backslash$\AminusP, \Azerominus]{\includegraphics[width=0.3\textwidth]{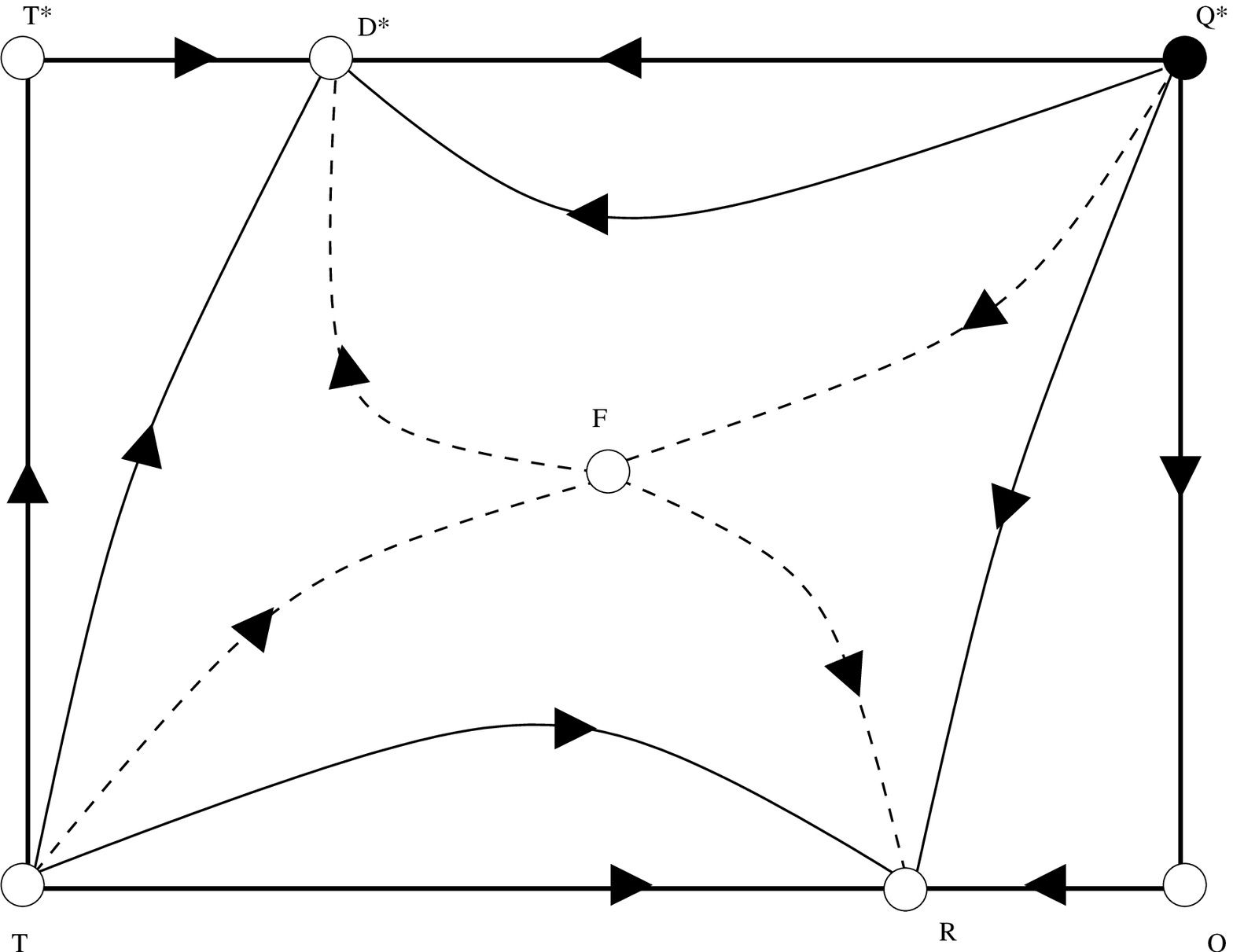}}\quad
\psfrag{T}[cc][cc][0.7][0]{$\mathrm{T}_\flat$}
\psfrag{T*}[cc][cc][0.7][0]{$\mathrm{T}_\sharp$}
\psfrag{Q}[cc][cc][0.7][0]{$\mathrm{Q}_\flat$}
\psfrag{Q*}[cc][cc][0.7][0]{$\mathrm{Q}_\sharp$}
\psfrag{R}[rl][cc][0.7][0]{$\mathrm{R}_\flat$}
\psfrag{D*}[cc][cc][0.7][0]{$\mathrm{R}_\sharp$}
\psfrag{F}[cc][cc][0.7][0]{$\mathrm{F}$}
\subfigure[\Azeroplus, \Aplus~$(0<\beta<\beta_\sharp)$]{\includegraphics[width=0.3\textwidth]{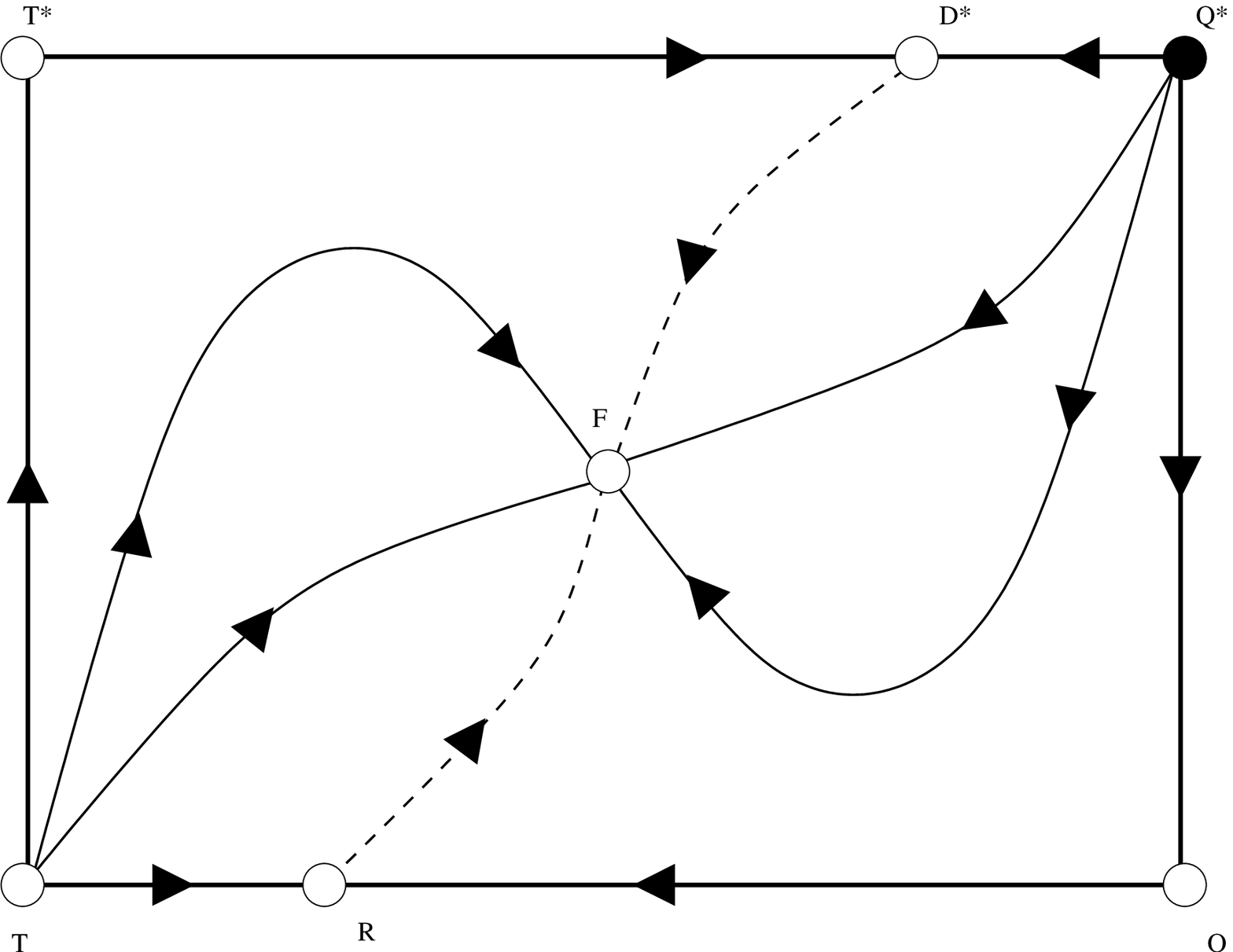}}\quad
\psfrag{T}[cc][cc][0.7][0]{$\mathrm{T}_\flat$}
\psfrag{T*}[cc][cc][0.7][0]{$\mathrm{T}_\sharp$}
\psfrag{Q}[cc][cc][0.7][0]{$\mathrm{Q}_\flat$}
\psfrag{Q*}[cc][cc][0.7][0]{$\mathrm{Q}_\sharp$}
\psfrag{R}[rl][cc][0.7][0]{$\mathrm{R}_\flat$}
\psfrag{D*}[cc][cr][0.7][0]{$\mathrm{R}_\sharp$}
\psfrag{F}[cc][cc][0.7][0]{$\mathrm{F}$}
\subfigure[\Aplus~$(\beta_\sharp\leq\beta<1)$]{\label{biAplu2}\includegraphics[width=0.3\textwidth]{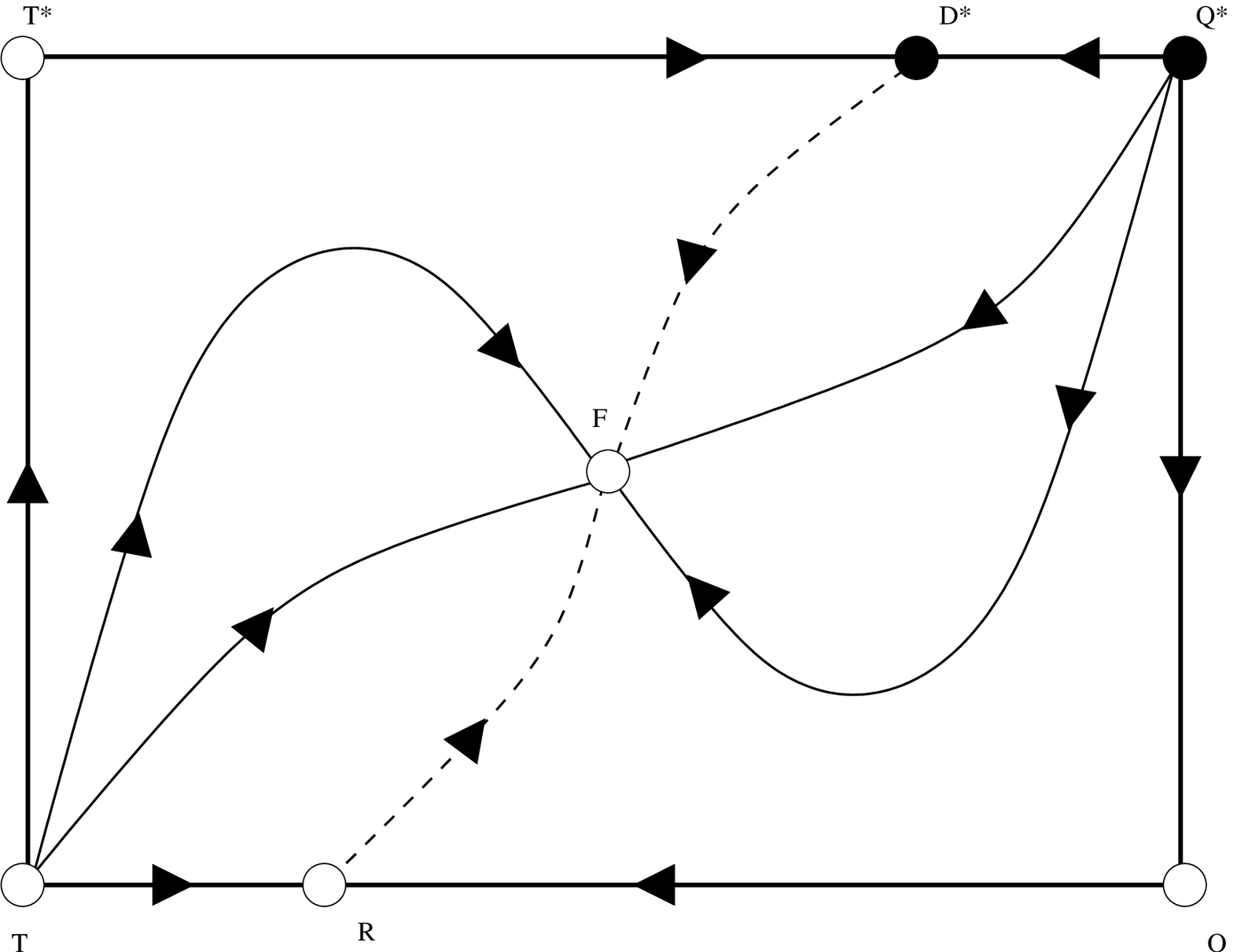}}\\
\psfrag{T}[cc][cc][0.7][0]{$\mathrm{T}_\flat$}
\psfrag{T*}[cc][cc][0.7][0]{$\mathrm{T}_\sharp$}
\psfrag{Q}[cc][cc][0.7][0]{$\mathrm{Q}_\flat$}
\psfrag{Q*}[cc][cc][0.7][0]{$\mathrm{Q}_\sharp$}
\psfrag{D*}[cc][cr][0.7][0]{$\mathrm{R}_\sharp$}
\psfrag{F}[cc][cc][0.7][0]{$\mathrm{F}$}
\subfigure[\Bplus, \Cplus]{\includegraphics[width=0.3\textwidth]{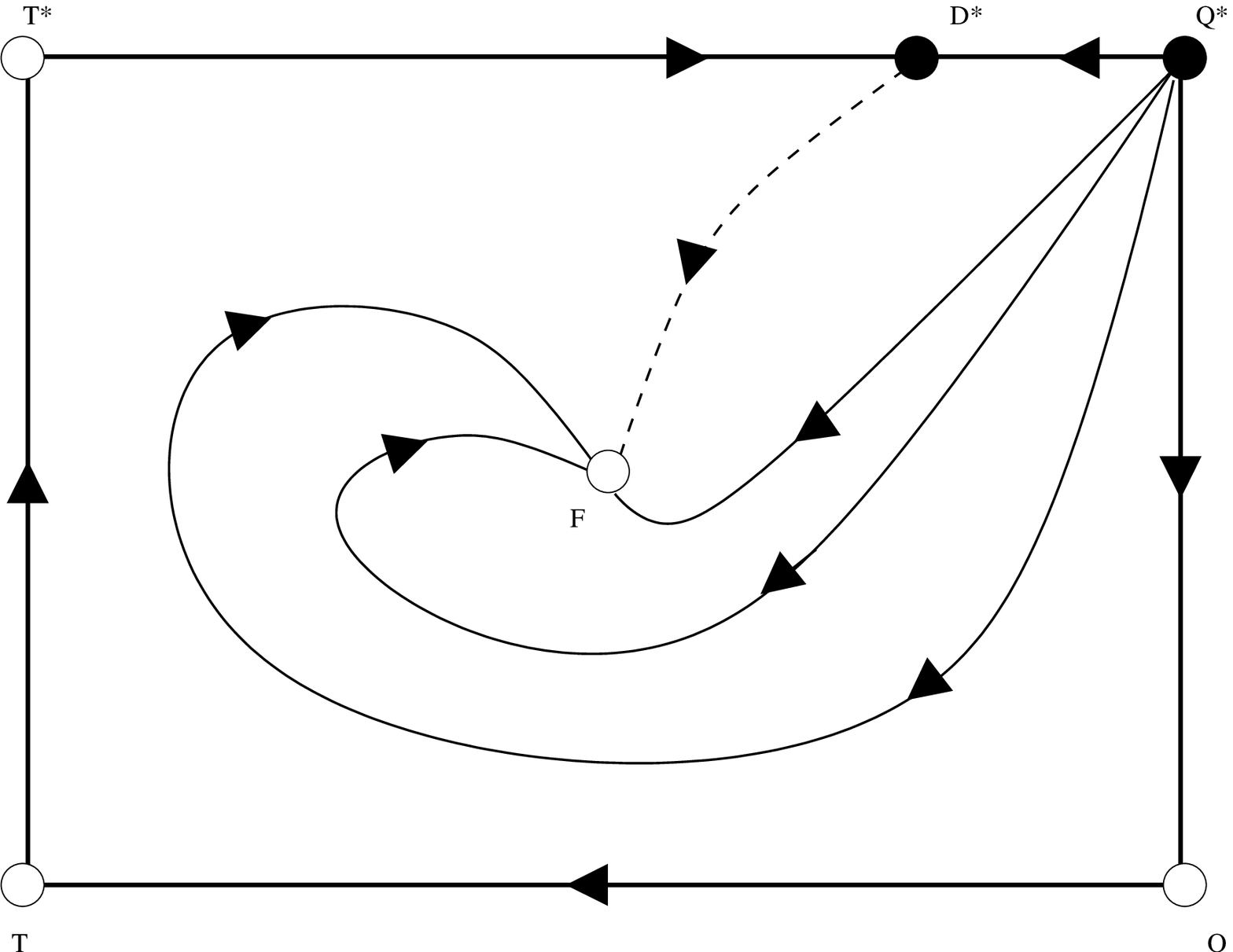}}\quad
\psfrag{T}[cc][cc][0.7][0]{$\mathrm{T}_\flat$}
\psfrag{T*}[cc][cc][0.7][0]{$\mathrm{T}_\sharp$}
\psfrag{Q}[cc][cc][0.7][0]{$\mathrm{Q}_\flat$}
\psfrag{Q*}[cc][cc][0.7][0]{$\mathrm{Q}_\sharp$}
\psfrag{D*}[cc][cr][0.7][0]{$\mathrm{R}_\sharp$}
\psfrag{F}[cc][cc][0.7][0]{$\mathrm{F}$}
\subfigure[\Dplus]{\includegraphics[width=0.3\textwidth]{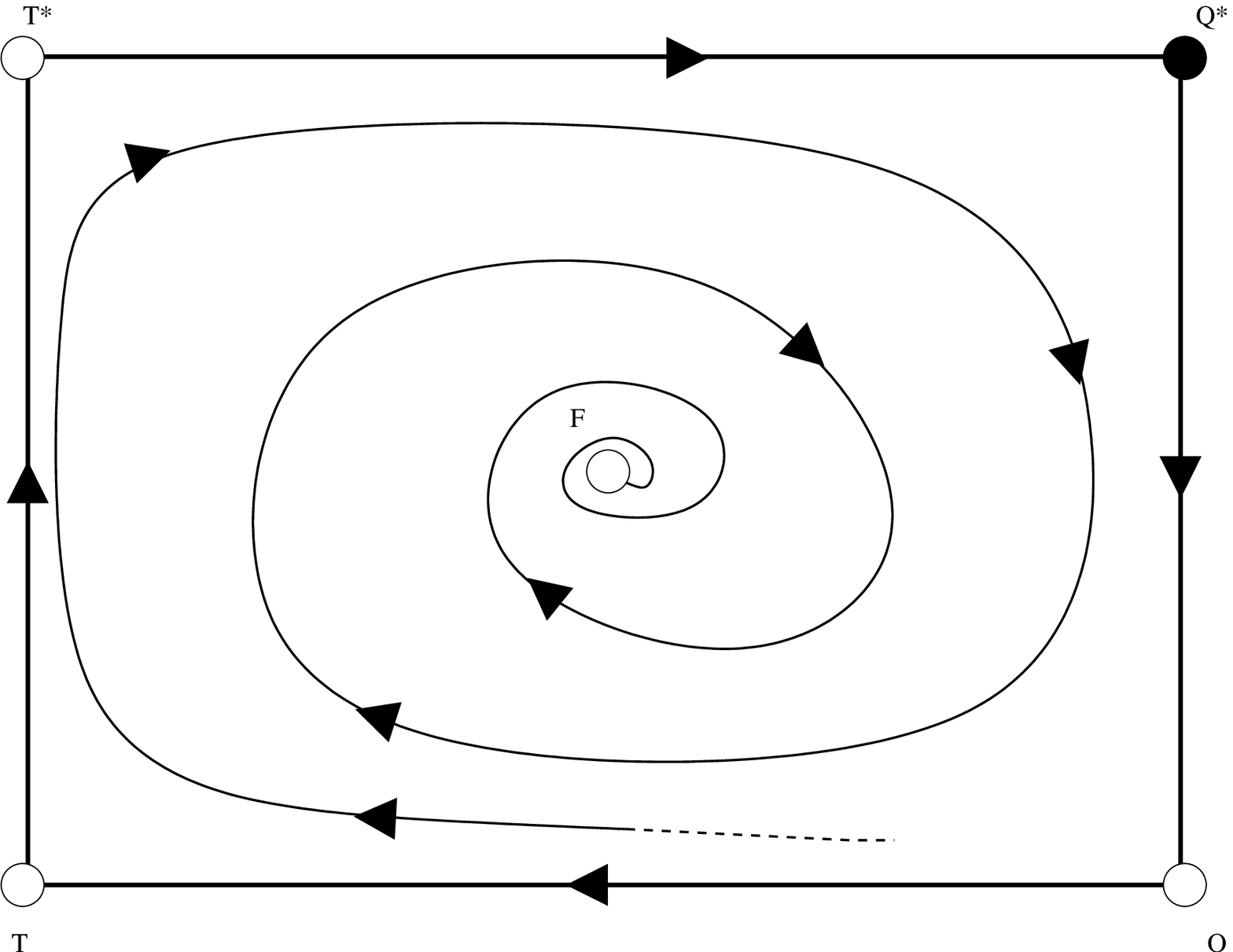}}\quad
\psfrag{s}[cc][cc][0.7][0]{$\vartheta$}
\psfrag{sigma}[cc][cc][0.7][0]{$\Sigma_+$}
\psfrag{t}[cc][cc][0.7][0]{$\vartheta$}
\psfrag{0}[cc][cc][0.7][0]{$0$}
\psfrag{1/2}[cc][cc][0.7][0]{$\textfrac{\pi}{2}$}
\psfrag{-1}[cc][cc][0.7][0]{$-1$}
\psfrag{+1}[cc][cc][0.7][0]{$+1$}
\psfrag{pi2}[cc][cc][0.7][0]{$\textfrac{\pi}{2}$}
\psfrag{X}[cc][cc][1.2][0]{$\mathcal{X}_{\,\mathrm{I}}$}
\subfigure[Coordinates]{\label{Icoo2}\includegraphics[width=0.3\textwidth]{coordinatesI.eps}}
\end{center}
\caption{Phase portraits of orbits on the Bianchi type~I subset $\mathcal{X}_{\,\mathrm{I}}$
of $\overline{\mathcal{Y}}_{\mathrm{IX}}^+$.
The fixed points are color-coded: 
A black fixed point is an attractor, a white fixed point is a repellor in the orthogonal direction.
This figure differs from Fig.~\ref{BianchiIfig} in the color-coding of the fixed points $\mathrm{Q}_\flat$ and $\mathrm{R}_\flat$.}
\label{BianchiIfig2}
\end{figure}

%%%%%%%%%%%%%%%%%%%%%%%%%%%%%%%%%%%%%%%%%%%%%%%%%%%%%%%%%%%%%%%%%%%%%%%%%%%%%%
\section{Bianchi type IX: Results}\label{B9res}
%%%%%%%%%%%%%%%%%%%%%%%%%%%%%%%%%%%%%%%%%%%%%%%%%%%%%%%%%%%%%%%%%%%%%%%%%%%%%%

In this section we use the building blocks of Sec.~\ref{B9sec} to derive
the main results on the \textit{global dynamics of LRS Bianchi type~IX models}.
The physical interpretation of these results is stated explicitly.

Let $\gamma$ denote an orbit of the dynamical system~\eqref{89syst}
on $\mathcal{X}_{\mathrm{IX}}$, which represents an LRS Bianchi type~IX model. 
There exist two fundamentally different cases in the context of
Bianchi type~IX dynamics.
\begin{itemize}
\item The recollapsing case: $\gamma \cap \mathcal{X}_{\mathrm{IX}}^{0} \neq \emptyset$ 
($\Rightarrow$ $\gamma \not\subset \mathcal{X}_{\mathrm{IX}}^{+}$). 
\item The (forever) expanding case:  $\gamma \subset \mathcal{X}_{\mathrm{IX}}^{+}$ ($\gamma \cap \mathcal{X}_{\mathrm{IX}}^{0} = \emptyset$).
\end{itemize}

In the \textbf{recollapsing case} there exists $\tau_0 > 0$ such that
$\gamma(\tau) \in \mathcal{X}_{\mathrm{IX}}^{+}$ (i.e., $H_D > 0$) $\forall\tau \in({-\infty}, \tau_0)$, 
$\gamma(\tau_0) \in \mathcal{X}_{\mathrm{IX}}^{0}$ (i.e., $H_D(\tau_0) = 0$), and
$\gamma(\tau) \in \mathcal{X}_{\mathrm{IX}}^{-}$ (i.e., $H_D < 0$) $\forall\, \tau_0 < \tau < \infty$. 
The cosmological model expands from an initial singularity, reaches a point
of maximum size, and then recollapses to finish in a big crunch.

In the \textbf{expanding case} (non-recollapsing case) we have $\gamma \subset \mathcal{X}_{\mathrm{IX}}^+$,
i.e., $\gamma(\tau) \in \mathcal{X}_{\mathrm{IX}}^+$ (and thus $H_D > 0$) $\forall \tau\in(-\infty,\infty)$.
In this case the cosmological model expands forever.

A third case, the collapsing case, arises from the expanding case by the discrete symmetry~\eqref{symIX}
and corresponds to forever contracting models. 

\begin{Lemma}\label{nointerlimitpoints}
Let $\gamma\subset\mathcal{X}_\mathrm{IX}$ be a recollapsing or expanding 
orbit of the dynamical system~\eqref{89syst}.
Then
\[
\alpha(\gamma) \subseteq \overline{\mathcal{S}}_\sharp \:.
\]
In the recollapsing case, $\omega(\gamma)$ is contained on the image of $\overline{\mathcal{S}}_\sharp$
under the discrete symmetry~\eqref{symIX}, in particular, $H_D \rightarrow -1$ as $\tau\rightarrow \infty$.
In the expanding case, $\omega(\gamma) \subseteq \overline{\mathcal{B}}_{\mathrm{IX}}$.
\end{Lemma}

\begin{Remark}
Note that both $\overline{\mathcal{S}}_\sharp$ and $\overline{\mathcal{B}}_{\mathrm{IX}}$ include
the type~I set $\overline{\mathcal{L}}_{\mathrm{I}}$.
\end{Remark}

\begin{proof}
Consider the $\mathcal{C}^1$ function $Z_5:\mathcal{X}_\mathrm{IX}\cup\mathcal{V}_\mathrm{IX}\to\R$ given by
\begin{equation}\label{delta}
Z_5=H_D\textfrac{M_1^{1/3}}{|1-H_D^2|^{2/3}}\:.
\end{equation}
Since $Z_5$ satisfies 
\[
Z_5'=-\textfrac{M_1^{1/3}}{|1-H_D^2|^{2/3}}\,\big(2\Sigma_+^2+\textfrac{1}{2}(1+3w)\Omega\big)\:
\]
and
\[
Z'''_5=-\textfrac{4M_1^{1/3}}{|1-H^2_D|^{2/3}} (\Sigma_+')^2\quad\text{if }\, (\Sigma_+,\Omega)=(0,0)\:,
\]
this function is strictly monotonically decreasing on $\mathcal{X}_\mathrm{IX}\cup\mathcal{V}_\mathrm{IX}$.
The monotonicity principle implies that
the $\alpha$- and the $\omega$-limit set of every orbit
in $\mathcal{X}_{\mathrm{IX}}$ must be contained 
on $\partial\mathcal{X}_\mathrm{IX}\backslash\mathcal{V}_\mathrm{IX}$
(which is the union of $\overline{\mathcal{B}}_{\mathrm{IX}}$, $\overline{\mathcal{S}}_\sharp$,
and the counterpart of $\overline{\mathcal{S}}_\sharp$ with $H_D = -1$). 

For a recollapsing or an expanding orbit we have $Z_5 > 0$ for some time $\tau_-$.
Because $Z_5$ is zero on $\mathcal{B}_{\mathrm{IX}}$, we infer that $\alpha(\gamma) \cap \mathcal{B}_{\mathrm{IX}} = \emptyset$, 
The first claim of the lemma ensues.

For a recollapsing orbit we have $Z_5 < 0$ for some time $\tau_+$.
Therefore, $\omega(\gamma) \cap \mathcal{B}_{\mathrm{IX}} = \emptyset$, because $Z_5$ is zero on $\mathcal{B}_{\mathrm{IX}}$.
For an expanding orbit, $Z_5$ is positive for all times; 
this leaves $\omega(\gamma) \subseteq \overline{\mathcal{B}}_{\mathrm{IX}}$ as the only possible asymptotic behavior.
This concludes the proof of the lemma.
\end{proof}

To proceed we make use of the dynamical systems formulation of Subsec.~\ref{blowup}.
An orbit $\gamma$ of the dynamical system~\eqref{dynsyspolar} on $\mathcal{Y}_{\mathrm{IX}}^+$ 
represents the expanding phase of an LRS Bianchi type~IX model; this is because
the past invariant subset $\mathcal{X}_\mathrm{IX}^{+}$ of $\mathcal{X}_\mathrm{IX}$ (which is defined by $H_D > 0$) 
and $\mathcal{Y}_{\mathrm{IX}}^+$ are diffeomorphic.
(In the expanding case, there is but the expanding phase; hence, in that case, $\gamma$
represents the entire type~IX solution.)
Since the dynamical system~\eqref{dynsyspolar} on $\mathcal{Y}_{\mathrm{IX}}^+$ possesses a regular extension to
the boundaries of the state space, we are able to apply standard methods
from the theory of dynamical systems.

\textbf{Lemma \ref{nointerlimitpoints}${}^\prime$.} 
\itshape In the context of the formulation of Subsec.~\ref{blowup},
the results of Lemma~\ref{nointerlimitpoints} imply that 
$\alpha(\gamma) \subseteq \overline{\mathcal{S}}_\sharp \cup \overline{\mathcal{X}}_{\mathrm{\,I}}$.
Furthermore, $\omega(\gamma)$ is either contained in $\overline{\mathcal{B}}_{\mathrm{IX}}\cup \overline{\mathcal{X}}_{\mathrm{\,I}}$
(expanding models) or there is a big crunch (recollapsing models).
\upshape

\begin{Theorem}\label{equalfootingthm}
In the Bianchi type~IX special case {\textnormal \AminusP}\ defined by~\eqref{domain}, 
recollapsing solutions and forever expanding solutions
occur on an equal footing.
The qualitative behavior of typical solutions
is the one sketched in Fig.~\ref{special}.
In all other anisotropy cases, every model is a recollapsing model.
\end{Theorem}

\begin{proof}
Assume an anisotropy case different from \AminusP.
If there were an expanding model $\gamma$, then 
$\omega(\gamma) \subseteq \overline{\mathcal{B}}_{\mathrm{IX}}\cup \overline{\mathcal{X}}_{\mathrm{\,I}}$ by 
Lemma~\ref{nointerlimitpoints}${}^\prime$.
Inspection of the flow on $\overline{\mathcal{B}}_{\mathrm{IX}}$ reveals that there do not exist any potential $\omega$-limit sets
on $\overline{\mathcal{B}}_{\mathrm{IX}}$; in particular, the stable manifold of every fixed point
has an empty intersection with the state space $\mathcal{Y}_{\mathrm{IX}}^+$; see Fig.~\ref{Bfig}.
The same is true for $\overline{\mathcal{X}}_{\mathrm{\,I}}$, see Fig.~\ref{BianchiIfig2}.
The only slightly non-trivial case is the case \Dplus, where $\partial\mathcal{X}_{\mathrm{\,I}}$ is a heteroclinc
cycle. However, we may use the function $Z_1$, see~\eqref{Z1def}: If $\omega(\gamma) = \partial\mathcal{X}_{\mathrm{\,I}}$,
then $Z_1 \rightarrow \infty$ along $\gamma$ as $\tau\rightarrow \infty$. 
However, a straightforward computation shows that
$Z_1$ is decreasing in the neighborhood of the fixed points of $\partial\mathcal{X}_{\mathrm{\,I}}$ 
which contradicts the divergence of $Z_1$.

Second, assume the Bianchi type~IX special case \AminusP. 
Figs.~\ref{yesspec} and~\ref{AminusspecI} imply that there are two potential $\omega$-limit
sets for orbits $\gamma \subset \mathcal{Y}_{\mathrm{IX}}^+$: First, the fixed point $\mathrm{R}_\flat$,
and second, the fixed point $\mathrm{P}$.
The latter is a saddle whose stable manifold (intersected with $\mathcal{Y}_{\mathrm{IX}}^+$)
is two-dimensional; hence, there exists a one-parameter family of orbits that
converge to $\mathrm{P}$ as $\tau\rightarrow \infty$.
The fixed point $\mathrm{R}_\flat$, however, is a sink, 
which means that there exists a two-parameter family
of orbits that converge to $\mathrm{R}_\flat$ as $\tau\rightarrow \infty$.
Accordingly, in the case \AminusP, there exists an open set of forever expanding cosmological models;
the asymptotics of these models is governed by the approach to  $\mathrm{R}_\flat$; see Fig.~\ref{special}.
(The existence of recollapsing solutions is evident from the properties of the flow
on $\mathcal{X}_{\mathrm{IX}}^0$.)
To find the past attractor of orbits in $\mathcal{Y}_{\mathrm{IX}}^+$
we follow the well-established principles that have 
already been successfully applied in the proof of
Theorem~\ref{BianchiIItheo}: Analyzing 
Figs.~\ref{Asharpfig},~\ref{Bfig}, and~\ref{BianchiIfig2},
it is immediate that the only possible $\alpha$-limit set of
typical orbits is the source $\mathrm{T}_\flat$, cf.~Table~\ref{alfaIX}.
This completes the proof of the lemma.
\end{proof}

{\it Interpretation of Theorem~\ref{equalfootingthm}}. 
Cosmological models of Bianchi type~IX with anisotropic matter
satisfy the closed-universe-recollapse 
conjecture~\cite{Barrow/Galloway/Tipler:1986,CHletter,Lin/Wald:1989},
i.e., these models recollapse.
However, there is one restriction: The anisotropic matter model
is required to be different from \AminusP, see~\eqref{domain} and
Figs.~\ref{mattermodelsfig},~\ref{betaminmax} and~\ref{vpmminmax}. (Recall that
the energy conditions are satisfied for the \AminusP\ matter model.)
In the case of anisotropic matter of that kind, 
type~IX solutions need not necessarily recollapse;
there exists a set of typical solutions (corresponding
to an open set of initial data) that are forever expanding.
We conclude that ``anisotropic matter of the type \AminusP\ matters'' (a lot).
Let us thus discuss the anisotropy case \AminusP\ in detail: 
Toward the initial singularity, every typical solution with anisotropic matter of
the type \AminusP\ is asymptotic to the Taub solution~\eqref{taub};
see Fig.~\ref{special}. However, there exist non-generic solutions
that exhibit different kinds of behavior: We observe
solutions with isotropic singularities, cf.~\eqref{FRWsol}, and solutions that approach~\eqref{cssol} 
as $t\rightarrow 0$. 
Toward the future, recollapsing solutions and forever expanding
solutions exist on an equal footing. Recollapsing solution 
exhibit a final singularity of the Taub type.
Expanding solutions, on the other hand, are geodesically complete toward the future;
the future asymptotics are characterized by an approach to the
asymptotic metric~\eqref{rfsol}.

\begin{Theorem}\label{BianchiIXtheo}
The qualitative behavior of typical orbits of the dynamical system~\eqref{dynsyspolar}
that represent the expanding phase of LRS Bianchi type~IX models (in the various
anisotropy cases) is the one sketched in Figs.~\ref{BianchiIXfig} and Fig.~\ref{special}.
The past attractors for the various cases and the possible limit $\alpha$-limit sets for
non-generic solutions are given in Table~\ref{alfaIX}.
\end{Theorem}

\begin{Remark}
Theorem~\ref{equalfootingthm} states that, for every anisotropy case except \AminusP, 
all models are recollapsing models. Theorem~\ref{BianchiIXtheo} discloses
the possible past asymptotic behavior;
the possible future asymptotic behavior of these recollapsing models
arises from the past asymptotic behavior by applying the discrete
symmetry~\eqref{symIX}.
\end{Remark}

\begin{proof}
The proof of Theorem~\ref{BianchiIXtheo} follows the 
well-established principles that have already been successfully applied in the proof of
Theorem~\ref{BianchiIItheo}.
Lemma~\ref{nointerlimitpoints}${}^\prime$ restricts the possible $\alpha$-limit sets 
to a subset of the boundary of the state space $\mathcal{Y}_{\mathrm{IX}}^+$.
Since by the analysis of Sec.~\ref{B9sec}, the dynamics on these boundary subsets
is fully understood, the proof of Theorem~\ref{BianchiIXtheo} essentially amounts 
to gluing together the phase portraits of
Figs.~\ref{Asharpfig},~\ref{Bfig}, and~\ref{BianchiIfig2}. 

The invariant structures on the boundary subsets~\eqref{IIwherealom} 
that are potential 
$\alpha$-limit sets
are fixed points, periodic orbits, and heteroclinic cycles/networks.
Figs.~\ref{Asharpfig},~\ref{Bfig}, and~\ref{BianchiIfig2} show that
periodic orbits do not occur, so that the
$\alpha$/$\omega$-limit sets of orbits in $\mathcal{Y}_{\mathrm{IX}}^+$
must be fixed points or heteroclinic cycles/networks (if present).
The investigation of the fixed points is straightforward. The procedure
is completely analogous to the one in the proof of Theorem~\ref{BianchiIItheo}.
The results of the local analysis of the fixed points is summarized
in Table~\ref{alfaIX}.
Let us thus restrict ourselves to discussing the role of the
heteroclinic cycles/networks.

Heteroclinic cycles are present
only in the cases \Dminus, \Bplus, \Cplus, and \Dplus, see Fig.~\ref{BianchiIXfig}. 
In the cases \Bplus\ and \Cplus, inspection of Figs.~\ref{Asharpfig} and~\ref{BianchiIfig2}
reveals that there does not exist any fixed point on $\overline{\mathcal{S}}_\sharp \cup \overline{\mathcal{X}}_{\mathrm{\,I}}$
(cf.~Lemma~\ref{nointerlimitpoints}${}^\prime$) that acts as a source; we simply use the color-coding
of these figures. This leaves only one 
structure on $\overline{\mathcal{S}}_\sharp \cup \overline{\mathcal{X}}_{\mathrm{\,I}}$
that can be the $\alpha$-limit set of typical orbits: the heteroclinic cycle depicted
in Fig.~\ref{IXBCplus}.
Case~\Dplus\ is analogous; however, instead of a heteroclinic cycle there is a heteroclinic network,
see Fig.~\ref{IXDplus}. It is this network (or a subset thereof) that it
the past attractor. 
Finally, consider the case \Dminus. In analogy with the arguments presented
in the proof of Theorem~\ref{BianchiIItheo} we find that the heteroclinic cycle
$\partial\mathcal{S}_\sharp$, see Fig.~\ref{IXDminus}, is not a possible
$\alpha$-limit set of orbits of $\mathcal{Y}_{\mathrm{IX}}^+$.
(The same is true for the cycles involving an orbit connecting
$\mathrm{T}_\sharp$ with $\mathrm{Q}_\sharp$ on $\mathcal{X}_{\mathrm{\,I}}$.)
This concludes the proof of the theorem.
\end{proof}

{\it Interpretation of Theorem~\ref{BianchiIXtheo}}. 
Anisotropic matter of the types \Dminus, \Cminus, \Bminus, \Aminus$\backslash$\AminusP, and \Aplus\
``does not matter''. The asymptotic behavior toward the initial singularity
of typical Bianchi type~IX solutions (associated with anisotropic matter of one of these types)
is the same as for typical perfect fluid solutions, see Sec.~\ref{perfectfluid}.
We observe that the typical solutions behave like the Taub solution~\eqref{taub} as
$t\rightarrow 0$. However, there exist non-generic solutions that exhibit
different kinds of asymptotic behavior: There are
solutions with isotropic singularities, cf.~\eqref{FRWsol}, 
solutions that approach~\eqref{cssol}, and solutions that approach~\eqref{rssol} or~\eqref{rfsol}
as $t\rightarrow 0$; see Table~\ref{alfaIX} for details.
Toward the future we observe recollapse; the asymptotics toward the final singularity
is analogous to the asymptotics toward the initial singularity.

In contrast, anisotropic matter of the types \Bplus, \Cplus, and \Dplus\ ``matters''.
Generically, the approach of the associated LRS type~IX models 
toward the (initial and final) singularity is \textit{oscillatory}:
The solutions oscillate between
the Taub family~\eqref{taub} and the non-flat LRS family~\eqref{solQ}.
This kind of asymptotic behavior differs considerably 
from that of (LRS) vacuum and perfect fluid solutions
(as discussed in Sec.~\ref{perfectfluid}).

\begin{figure}[Ht!]
\begin{center}
\psfrag{tf}[cc][cc][0.7][0]{$\mathrm{T}_\flat$}
\psfrag{ts}[cc][cc][0.7][0]{$\mathrm{T}_\sharp$}
\psfrag{qf}[cc][cc][0.7][0]{$\mathrm{Q}_\flat$}
\psfrag{qs}[cc][cc][0.7][0]{$\mathrm{Q}_\sharp$}
\psfrag{f}[cc][cc][0.7][0]{$\mathrm{F}$}
\psfrag{c}[cc][cc][0.7][0]{$\mathrm{C}_\sharp$}
\psfrag{d}[cc][cc][0.7][0]{$\mathrm{R}_\sharp$}
\psfrag{r}[cc][cc][0.7][0]{$\mathrm{R}_\flat$}
\psfrag{h0}[cc][cc][0.7][-45]{$H_D=0$}
\subfigure[\Dminus]{\label{IXDminus}\includegraphics[width=0.25\textwidth]{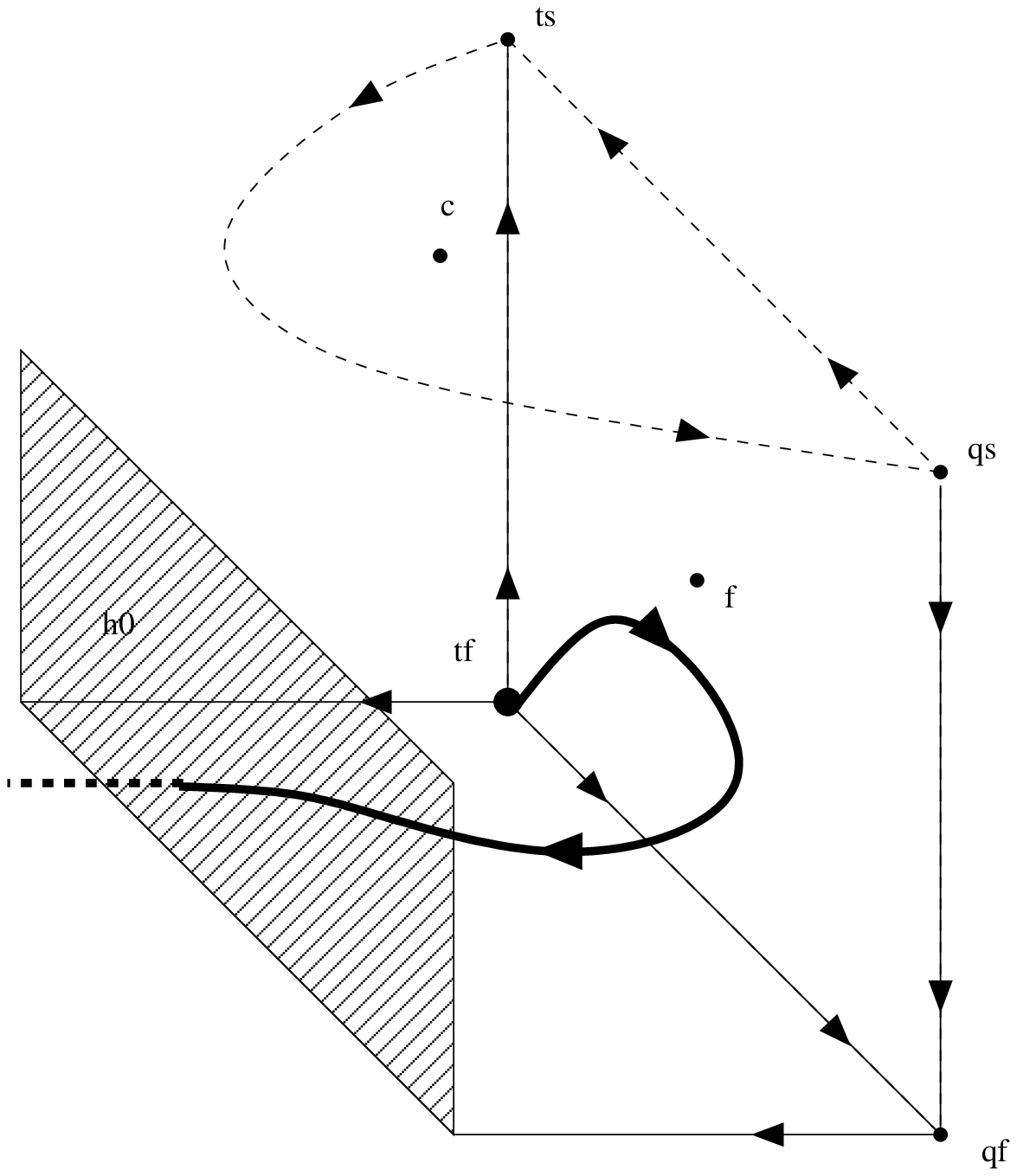}}\qquad
\subfigure[\Bminus,\Cminus]{\includegraphics[width=0.25\textwidth]{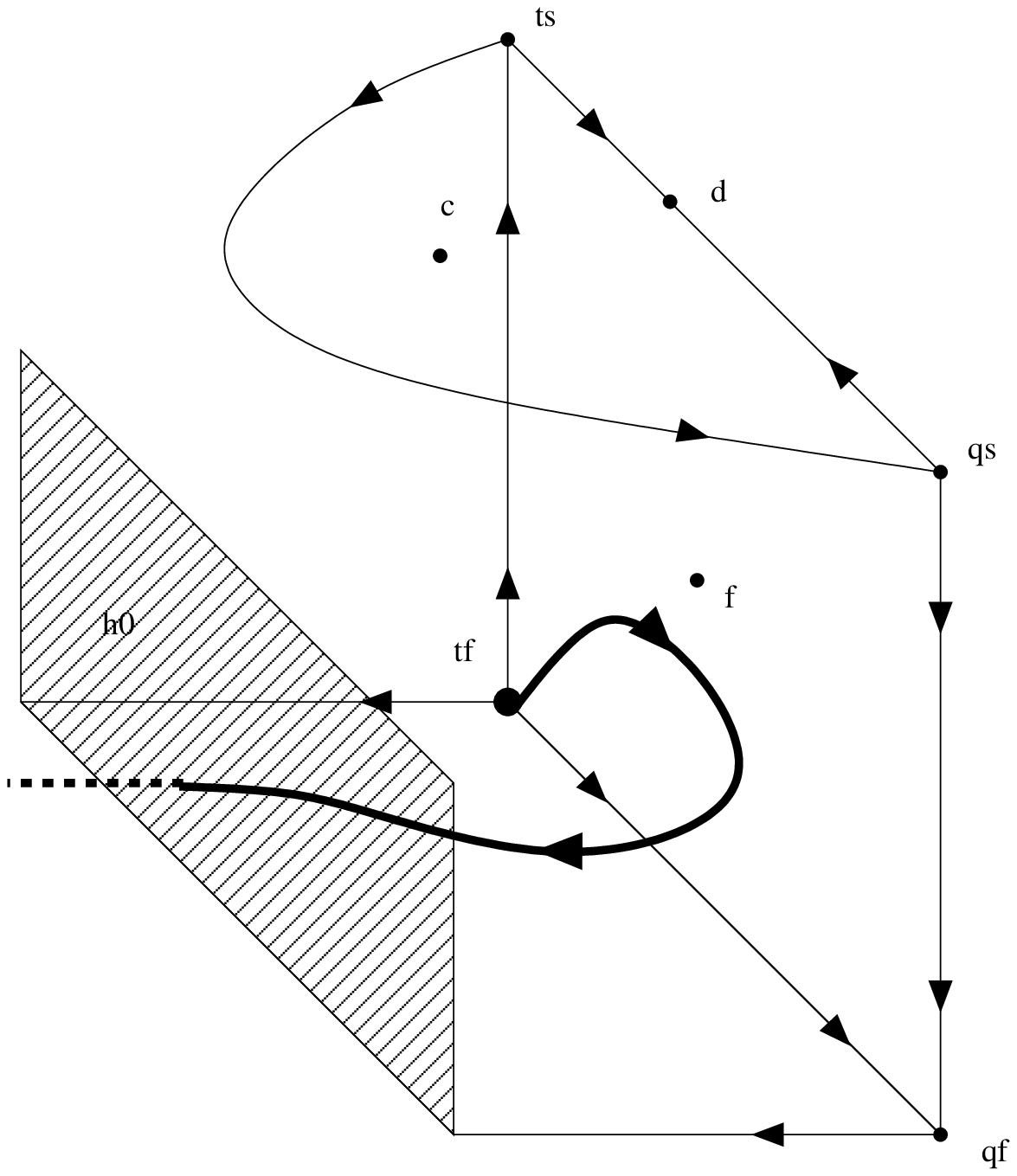}}\qquad
\subfigure[\Aminus\!\!$\setminus$\AminusP,~\Aplus$(0<\beta<\beta_\sharp)$]%
{\includegraphics[width=0.25\textwidth]{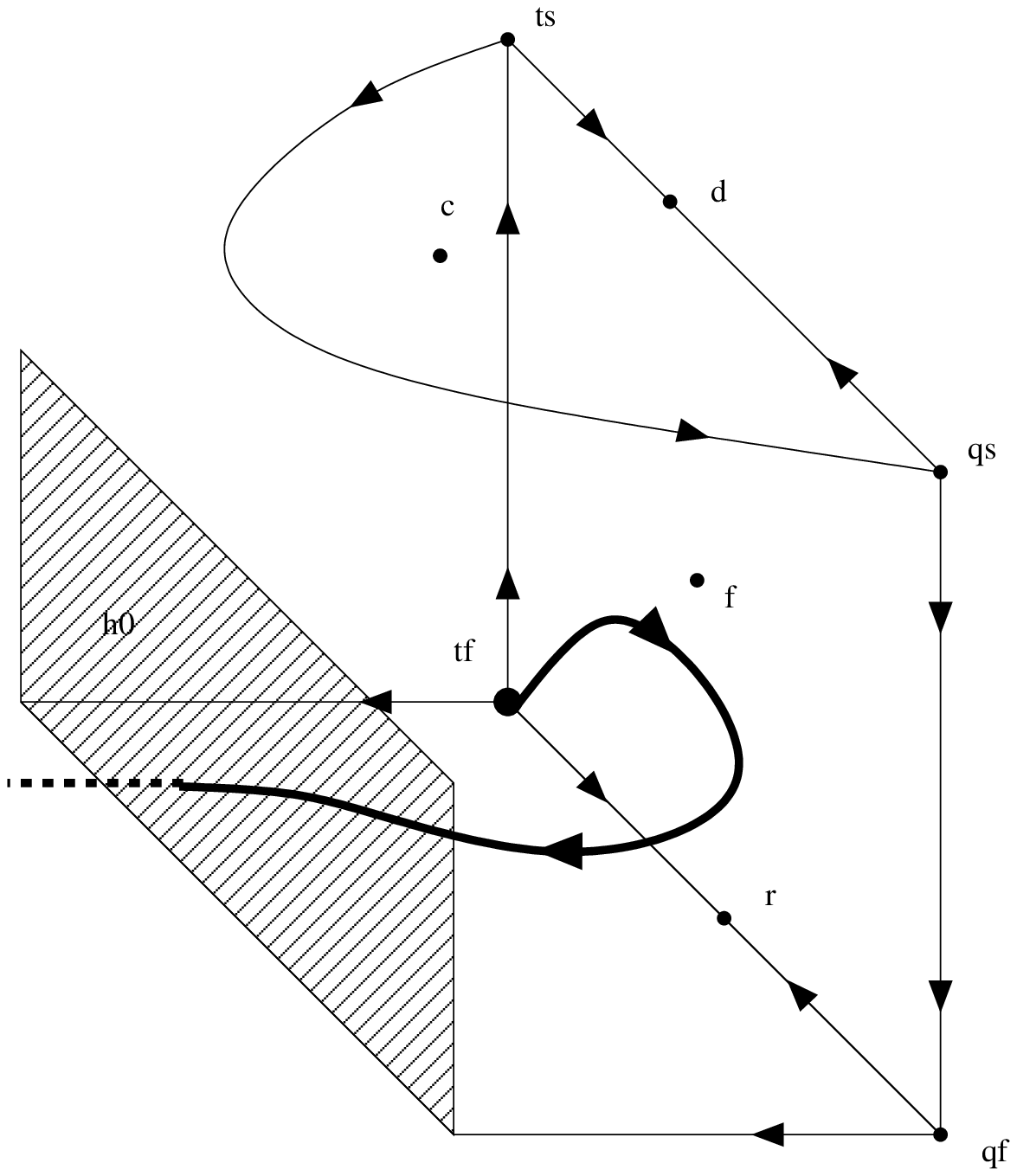}}\\
\subfigure[\Aplus~$(\beta_\sharp\leq\beta<1)$]{\includegraphics[width=0.25\textwidth]{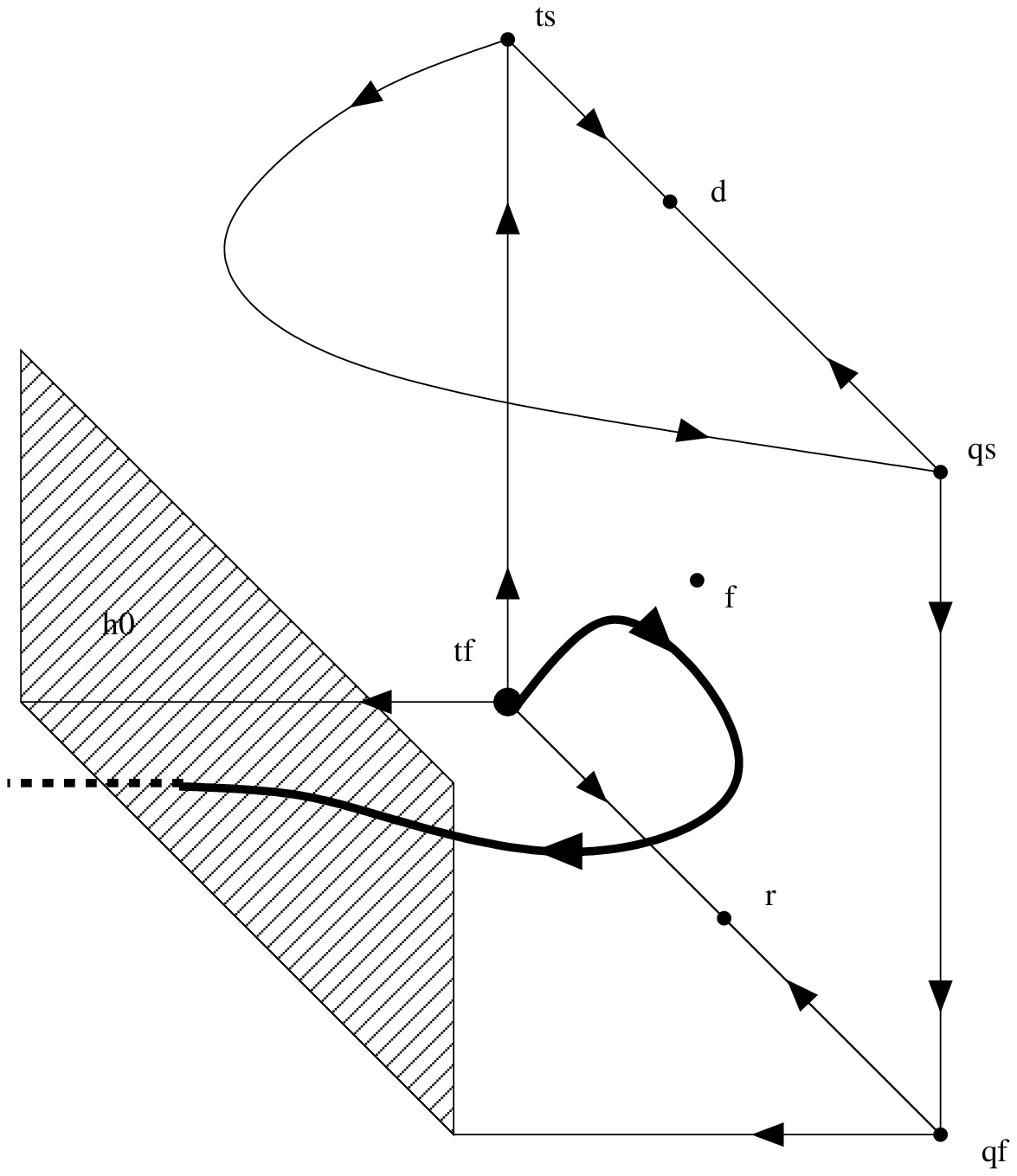}}\qquad
\subfigure[\Bplus, \Cplus]{\label{IXBCplus}\includegraphics[width=0.25\textwidth]{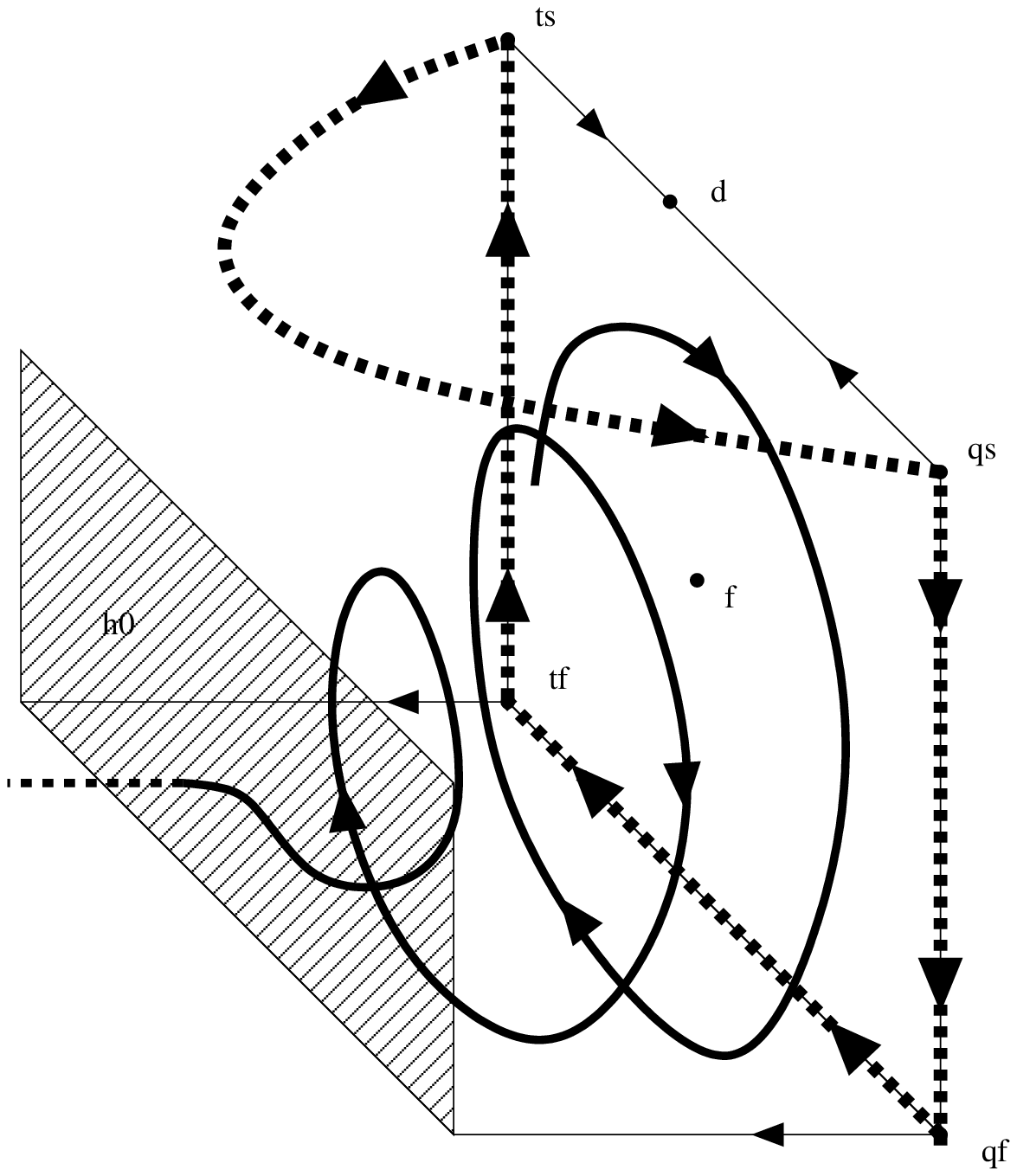}\label{VlasovIX}}\qquad
\subfigure[\Dplus]{\label{IXDplus}\includegraphics[width=0.25\textwidth]{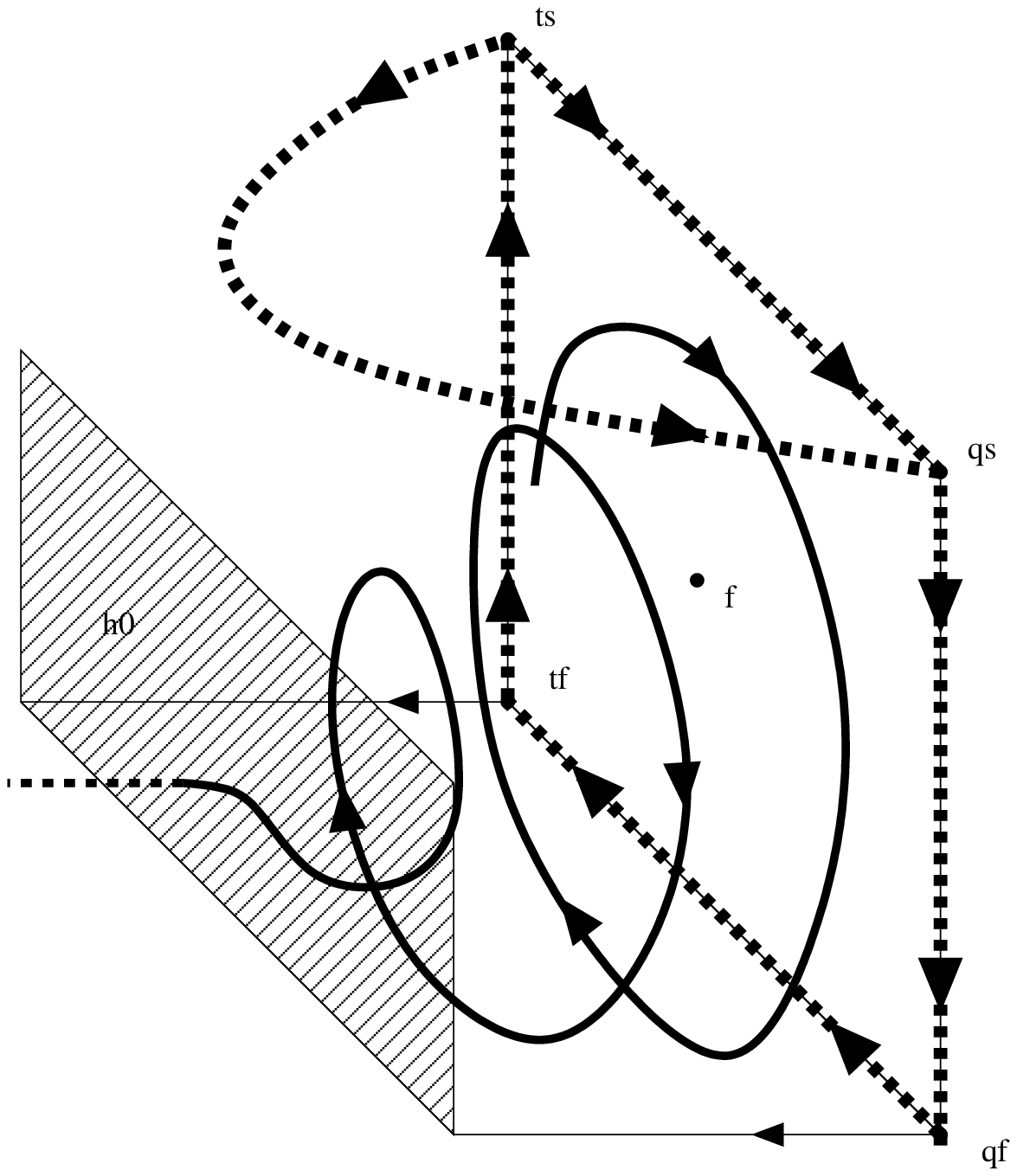}}\end{center}
\caption{Phase portraits of Bianchi type~IX orbits in the various anisotropy cases---the 
special case \AminusP\ is depicted in Fig.~\ref{special}. 
Bold lines are typical (in fact generic) orbits; dashed lines 
are orbits that lie on the boundary of the state space 
and connect to form heteroclinic cycles/networks.
In the cases \Bplus\ and \Cplus,
the past attractor is the heteroclinic cycle;
in case \Dplus, the past attractor is the heteroclinic network (or a subset thereof).
In the remaining cases, the past attractor is the fixed point $\mathrm{T}_\flat$. 
There exist non-generic orbits that display different asymptotic behavior toward the past, see Table~\ref{alfaIX}.
The phase portraits in the cases \Azeropm\ are the same as for \Aminus.}
\label{BianchiIXfig}
\end{figure}

\begin{figure}[Ht!]
\begin{center}
\psfrag{tf}[cc][cc][0.7][0]{$\mathrm{T}_\flat$}
\psfrag{ts}[cc][cc][0.7][0]{$\mathrm{T}_\sharp$}
\psfrag{qf}[cc][cc][0.7][0]{$\mathrm{Q}_\flat$}
\psfrag{qs}[cc][cc][0.7][0]{$\mathrm{Q}_\sharp$}
\psfrag{f}[cc][cc][0.7][0]{$\mathrm{F}$}
\psfrag{c}[cc][cc][0.7][0]{$\mathrm{C}_\sharp$}
\psfrag{d}[cc][cc][0.7][0]{$\mathrm{R}_\sharp$}
\psfrag{r}[cc][cc][0.7][0]{$\mathrm{R}_\flat$}
\psfrag{p}[cc][cc][0.7][0]{$\mathrm{P}$}
\psfrag{h0}[cc][cc][0.7][-45]{$H_D=0$}
\includegraphics[width=0.5\textwidth]{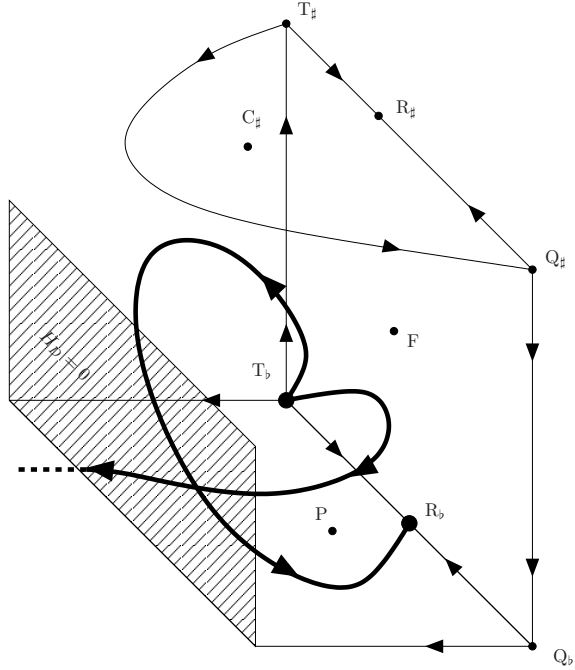}
\caption{Phase portrait of typical LRS Bianchi type~IX orbits in the special case \AminusP\ 
defined by~\eqref{domain}. 
This is the only anisotropy case where expanding (i.e., non-recollapsing) solutions exist;
these solutions satisfy $H_D > 0$ for all times.
The future attractor of these solutions is 
the fixed point $\mathrm{R}_\flat$. Note that this type of behavior is typical.
A second class of typical orbits corresponds to recollapsing models.
Finally, there 
exist non-generic solutions whose $\omega$-limit set is the point $\mathrm{P}$; these
orbits are not depicted; instead we refer to Table~\ref{alfaIX}.}\label{special}
\end{center}
\end{figure}

\begin{table}
\begin{center}
\begin{tabular}{|c|c|c|c|c|c|c|c|c|c|c|}
\hline 
\multirow{2}{*}{Fixed point} & \multicolumn{10}{|c|}{Dimension of unstable [{\small stable}] 
manifold intersected with $\mathcal{Y}_{\mathrm{IX}}^+$} 
\\[0.2ex]
& \Dminus & \Cminus & \Bminus & \Aminus\ {\footnotesize (not \AminusP)} & 
\AminusP & \Aplus  {\scriptsize($\beta < \beta_\sharp$)} & 
\Aplus  {\scriptsize($\beta_\sharp\leq \beta$)} & \Bplus & \Cplus & \Dplus \\ \hline 
& & & & & & & & & & \\[-1.3ex]
$\mathrm{T}_\flat$ & \textbf{3} & \textbf{3} & \textbf{3} & \textbf{3} & \textbf{3} & \textbf{3} & \textbf{3} & & & \\
$\mathrm{R}_\flat$ & & & &  & [{\small 3}] & 2 & 2 & & &  \\
$\mathrm{R}_\sharp$ & & & & & & 2 & & & & \\
$\mathrm{F}$ & 2 & 2 & 2 & 2 & 2 & 1 & 1 & 1 & 1 & 1 \\
$\mathrm{C}_\sharp$ & 1 & 1 & 1 & 1 & 1 & 1 & & & & \\
het.\ cycle & & & & & & & & \textbf{3}  & \textbf{3} & \textbf{3} \\
$\mathrm{P}$ &  &  &  &  & [{\small 2}] &  & & & & \\\hline
\end{tabular}
\caption{Possible $\alpha$-limit sets of LRS Bianchi type~IX models
and $\omega$-limit sets of forever expanding LRS Bianchi type~IX models.
Each orbit in $\mathcal{Y}_{\mathrm{IX}}^+$ corresponds to 
(the expanding phase of) a Bianchi type~IX solution (and conversely); hence
the intersection of the unstable [stable] manifold of a fixed point with $\mathcal{Y}_{\mathrm{IX}}^+$ yields the set of type~IX
solutions converging to the solution represented by that 
fixed point as $t\rightarrow 0$ [$t\rightarrow \infty$]. If 
this set is three-dimensional, the fixed point is a source [sink]; if the dimension is two,
there exists a one-parameter set of orbits converging to the fixed point; if the dimension is one,
there is only one orbit with that property. Center manifold analysis shows that 
the cases \Azeroplus\ and \Azerominus\ behave like
\Aplus\ ($0<\beta< \beta_\sharp$) and \Aminus\ (not \AminusP), respectively.
The exact solutions represented by the fixed points are given in Sec.~\ref{perfectfluid} and 
in Appendix~\ref{exact}.}
\label{alfaIX}
\end{center}
\end{table}

%%%%%%%%%%%%%%%%%%%%%%%%%%%%%%%%%%%%%%%%%%%%%%%%%%%%%%
%%%%%%%%%%%%%%%%%%%%%%%%%%%%%%%%%%%%%%%%%%%%%%%%%%%%%%
\section{Examples of matter models}
\label{mattermodelssec}
%%%%%%%%%%%%%%%%%%%%%%%%%%%%%%%%%%%%%%%%%%%%%%%%%%%%

In this section we present in detail 
two examples of matter models to which one can apply the main results of this paper: 
Ensembles of collisionless massless particles, described by the Vlasov equation, 
and magnetic fields. (For a discussion of a third matter model, elastic matter,
we refer to~\cite{CH}.)
Collisionless matter with massless particles satisfies our assumptions and the results of this paper  
apply directly. Magnetic fields violate Assumption~\ref{asswi}, but we show below that our analysis 
extends straightforwardly to this matter model as well.  In fact we believe that, up to minor changes, our 
method should work for an even larger class of matter models than the one considered in this paper. 
For a general introduction to collisionless matter and the Vlasov-Einstein system 
we refer to~\cite{A}; the Bianchi type~I case is discussed in detail in~\cite{HU}. For a 
generalization of the results presented here to Vlasov matter with massive particles we refer to~\cite{letter}.  

%%%%%%%%%%%%%%%%%%%%%%%%%%%%%%%%%%%%%%%%%
\subsection{Collisionless matter}\label{vlasov}
%%%%%%%%%%%%%%%%%%%%%%%%%%%%%%%%%%%%

Consider an ensemble of
particles with mass $m$ in the spacetime $(M,{}^4\mathbf{g})$. 
This ensemble of particles is represented by a distribution function (`phase space density') 
$f\geq 0$, which is defined on the mass shell, i.e., 
on the subset of the tangent bundle given by ${}^4\mathbf{g}(v,v) = -m^2$,
where $v$ denotes the (future-directed) four momentum.
If $(t,x^i)$ is a system of coordinates on $M$ such that $\partial_t$ is timelike and $\partial_{x^i}$ is spacelike, 
then the spatial coordinates
$v^i$ of the four-momentum are coordinates on the mass shell, and
we can regard $f$ as a function $f = f(t,x^i, v^j)$, $i,j =1,2,3$. The energy-momentum tensor is defined as
\begin{equation}\label{Tvlasov}
T^{\mu\nu}=\int f v^\mu v^\nu \sqrt{|\det{}^4\mathbf{g}|}\:|v_0|^{-1} \: dv^1 dv^2 dv^3\:.
\end{equation}
The equation satisfied by the function $f$ depends on the kind of interaction
between the particles. If the interaction is through binary collisions, 
the resulting equation is the Boltzmann equation~\cite{Eh}. Here we assume 
that the particles are freely falling, i.e., they interact only through gravity. 
By the equivalence principle, the particles trajectories must coincide with the 
geodesics of the spacetime. Accordingly, $f$ has to be constant along the particles trajectories
and thus solves the Vlasov equation
\begin{equation}\label{vlasoveq}
\partial_t f +\frac{v^j}{v^0}\partial_{x^j}f-\frac{1}{v^0}\Gamma^j_{\mu\nu}v^\mu v^\nu\partial_{v^j}f=0\:,
\end{equation}
where $\Gamma^\sigma_{\mu\nu}$ are the Christoffel symbols of the 
metric and $v^0>0$ is determined in terms of the metric ${}^4\mathbf{g}_{\mu\nu}$ and the spatial coordinates $v^i$ 
of the momentum via the mass shell relation ${}^4\mathbf{g}_{\mu\nu} v^\mu v^\nu = -m^2$. (The characteristics of the Vlasov 
equation---which correspond to the particles trajectories and along which $f$ is constant---coincide with the lift of the 
geodesic flow on the tangent bundle, in accordance with the geometric interpretation of the free fall motion of the particles.)

Consider now a spacetime of Bianchi class~A with LRS Bianchi symmetry and define a time independent orthogonal frame 
such that the metric takes the form
\begin{equation}\label{metricLRS}
^4\mathbf{g}=-dt^2+g_{11}(t)\:\hat{\omega}_1\otimes\hat{\omega}_1+
g_{22}(t)(\hat{\omega}_2\otimes\hat{\omega}_2+\hat{\omega}_3\otimes\hat{\omega}_3)\:.
\end{equation} 
As proved in~\cite{MM}, the general solution of the Vlasov equation~\eqref{vlasoveq} on a background spacetime 
with the metric~\eqref{metricLRS} can be expressed as
\begin{equation}\label{generalf}
f=f_0(v_1,v_2^{\,2}+v_3^{\,2})\:,
\end{equation} 
where $f_0:\R\times\R_+\to\R_+$ is an arbitrary, sufficiently smooth function. The function $f_0$ can be 
interpreted as the `initial data' for $f$ at some time $t=t_0$. 

However, for the pair~\eqref{metricLRS} and~\eqref{generalf} to be a candidate for
a solution of the Einstein-Vlasov system, the energy momentum tensor must be compatible
with the structure of~\eqref{metricLRS}, i.e., it must be diagonal and satisfy $T^2_{\ 2}=T^3_{\ 3}$. 
This can be achieved by restricting to distribution functions~\eqref{generalf} that are invariant under the transformation $v_1\to -v_1$. 
These distribution functions are called {\it reflection symmetric}, see~\cite{RT}. 
We assume further that the support of $f_0$ 
does not intersect any of the axes ({\it split support} assumption), which ensures that $\rho>0$.
It is now clear that, with the hypotheses on $f_0$, 
Assumptions~\ref{assumptionT}--\ref{assumptiondiagonal} are satisfied. 

Assumption~\ref{assumptionwi} is only satisfied
in the case of massless particles. For $m=0$ we have 
\begin{subequations}\label{rhopi}
\begin{align}\label{rho}
\rho&=T^0_{\ 0}=\int f_0\: \big(g^{11}v_1^{\,2}+g^{22}(v_2^{\,2}+v_3^{\,2})\big)^{1/2}\:\frac{dv_1dv_2dv_3}{\sqrt{\det g}}\:,\\
\label{pi}
p_i&=T^i_{\ i}=\int f_0\:g^{ii}\:v_i^{\,2}\:\big(g^{11}v_1^{\,2}+g^{22}(v_2^{\,2}+v_3^{\,2})\big)^{-1/2}\:\frac{dv_1dv_2dv_3}{\sqrt{\det g}}\:;
\end{align}
\end{subequations}
note that $p_2=p_3$, because $f$ is invariant under the exchange of $v_2$ and $v_3$.  
It follows that 
\[ 
p =\textfrac{1}{3}\,\big(p_1+p_2+p_3\big)  =\textfrac{1}{3}\: \rho\:, 
\]
i.e., $w = \textfrac{1}{3}$. (For $m>0$, the equation of state 
between the isotropic pressure and the energy density is non-linear.)
The rescaled principal pressures are given by
\begin{subequations}\label{wiVlasov}
\begin{align}
w_1(s)&=(1-2s)\,
\frac{\int f_0\left[(1-2s)v_1^{\,2}+s(v_2^{\,2}+v_3^{\,2})\right]^{-1/2}v_1^{\,2}\,dv_1dv_2dv_3}
{\int f_0(v)\left[(1-2s)v_1^{\,2}+s(v_2^{\,2}+v_3^{\,2})\right]^{1/2}dv_1dv_2dv_3}\:,\\
u(s) = w_2(s) = w_3(s) &=s\,
\frac{\int f_0\left[(1-2s)v_1^{\,2}+s(v_2^{\,2}+v_3^{\,2})\right]^{-1/2}v_2^{\,2}\,dv_1dv_2dv_3}
{\int f_0(v)\left[(1-2s)v_1^{\,2}+s(v_2^{\,2}+v_3^{\,2})\right]^{1/2}dv_1dv_2dv_3}\:.
\end{align}
\end{subequations}
Therefore, using the terminology of Secs.~\ref{Sec:matter} and~\ref{LRSsols},
we obtain $u(0) = 0$ and $u(\textfrac{1}{2}) = \textfrac{1}{2}$; hence
Assumption~\ref{asswi} is satisfied with $v_- = 0$ and $v_+ = 1$.
Moreover, a direct 
calculation shows that $u(s)$ is monotonically increasing, see~\cite{RT}, 
so that Assumption~\ref{asspsi} is satisfied as well (i.e., there exists a unique
$\bar{s}$ such that $u(\bar{s}) = \textfrac{1}{3}$, whence $w_1 = w_2 = w_3 = \textfrac{1}{3}$
at $\bar{s}$).

In conclusion, Vlasov matter with massless particles falls into the class of 
matter models considered in this paper. 
Since $w=\textfrac{1}{3}$ and $v_-=0$, we have $\beta=1$, see~\eqref{betadef}, 
and therefore, Vlasov matter with massless particles is of type~\Bplus.  
The qualitative dynamics of solutions of the (Bianchi class~A LRS) Einstein-Vlasov equations 
for the various Bianchi types is the one depicted in Figs.~\ref{VlasovI},~\ref{VlasovII},~\ref{VlasovIX}.

\begin{Remark}
If we consider ensembles of collisionless particles with positive mass, $m >0$,
$w$ and the rescaled pressures $w_i$ are functions of $(s_1,s_2,s_3)$ and an additional scale variable, which
can be taken to be $z=m^2/(m^2+g^{11}+2g^{22})$. In~\cite{letter} we will show that the analysis of this paper
carries over to this more general situation.
\end{Remark}

%%%%%%%%%%%%%%%%%%%%%%%%%%%%%%%%%%%%%%%%%
\subsection{Magnetic fields}\label{magnetic}
%%%%%%%%%%%%%%%%%%%%%%%%%%%%%%%%%%%%%%%

For an electromagnetic field represented by the antisymmetric electromagnetic
field tensor $F_{\mu\nu}$,
the energy-momentum tensor is given by
\[
T^\mu_{\ \nu}=-\frac{1}{4\pi}\left(F^\mu_{\ \alpha}F^\alpha_{\ \nu}-\frac{1}{4}\delta^\mu_{\ \nu}F^\beta_{\ \alpha}F^\alpha_{\ \beta}\right)\:.
\]
The equations for the field are the Maxwell equations
\[
\nabla_\mu F^{\mu\nu}=0\:,\qquad \nabla_{[\sigma}F_{\mu\nu]}=0\:.
\] 
In an LRS Bianchi class~A spacetime with metric~\eqref{metricLRS}, consider specifically a purely magnetic field
that is aligned along the axis perpendicular to the plane of rotational symmetry, i.e., 
\[
F_{\mu\nu}=\left(\begin{matrix}0 & 0 & 0 & 0\\ 0 & 0 & 0 & 0\\ 0 & 0 & 0 & K\\ 0 & 0 & -K & 0 \end{matrix}\right);
\] 
the magnetic field is determined by $K$ through $B^1 = K (g^{11})^{1/2}g^{22}$, $B^2= B^3 = 0$.
The Maxwell equations imply that $K$ is a constant; this implies that the energy 
density $\rho$,
\[
\rho = \frac{1}{8 \pi} \: \big(g^{22}\big)^2\: K^2\;,
\]
is a function of the metric (which depends on the initial data for the magnetic field).
Furthermore, $T^\mu_{\ \nu}$ is diagonal and
\begin{equation}\label{Tmagnetic}
T^1_{\ 1}={-\rho}\:,\qquad T^2_{\ 2}=T^3_{\ 3}=\rho\:.
\end{equation}
Accordingly, Assumptions~\ref{assumptionT}--\ref{assumptiondiagonal}
are satisfied and diagonal models exist.
It follows from~\eqref{Tmagnetic} that $w_1 = -1$ and $w_2=w_3 = 1$, so that $w = \textfrac{1}{3}$. 
Thus Assumption~\ref{assumptionwi} is satisfied, but Assumptions~\ref{asswi} and~\ref{asspsi} are \textit{violated};
in particular, there exists no isotropic state of the matter. 
Despite this fact, the analysis for magnetic fields cosmologies 
can be carried out as straightforwardly as the analysis for other matter fields by 
using the `building blocks' discussed in this paper: 

{\it Bianchi type~I}. The reduced dynamical system for LRS Bianchi type~I aligned magnetic fields cosmologies is given by
setting $u =w_2 = w_3 \equiv 1$ and $w = \textfrac{1}{3}$ in~\eqref{dynsysbianchiI}, i.e.,
\begin{equation}\label{BianchiImag}
\Sigma_+'=(1-\Sigma_+^2)(2-\Sigma_+)\:,\qquad s'=-6s(1-2s)\Sigma_+\:.
\end{equation}
Due to the decoupling of these equations, the original state space can be reduced to the
one-dimensional state space $\Sigma_+ \in [-1,1]$.
The equation for $\Sigma_+$ is identical to~\eqref{eqcalIflat} with $w= \textfrac{1}{3}$ and $\beta = {-2}$
(or, equivalently, to~\eqref{eqcalIsharp} with $w= \textfrac{1}{3}$ and $\beta = 4$).
Since $\Sigma_+^\prime > 0$, we conclude that the Taub point
is the $\alpha$-limit and the non-flat LRS point is the $\omega$-limit of all orbits.

{\it Bianchi type~II}.  The reduced dynamical system for LRS Bianchi type~II aligned magnetic fields cosmologies is derived
from~\eqref{dynsysbianchiII} by
setting $u =w_2 = w_3 \equiv 1$ and $w = \textfrac{1}{3}$, i.e.,
\begin{subequations}
\begin{align}
\Sigma_+'&=(\textfrac{1}{6}M_1^2+\Omega)(2-\Sigma_+)\:,\qquad \qquad s'=-6s(1-2s)\Sigma_+\:,\\
M_1'&=M_1[2(1-2\Sigma_+)-\textfrac{1}{6}M_1^2-\Omega]\:,
\end{align}
\end{subequations}
where $\Omega=1-\Sigma_+^2-\textfrac{1}{12}M_1^2$. 
The equation for $s$ decouples which implies that we 
may restrict ourselves to the reduced state space 
$\big\{{-1}\leq \Sigma_+ \leq 1, \: 0\leq M_1 \leq [12(1-\Sigma_+^2)]^{1/2}\:\big\}$.
The system of equations on this state space coincides with 
the system~\eqref{dynsysAflat} on the set $\mathcal{S}_\flat$  when we set 
$w= \textfrac{1}{3}$ and $\beta = {-2}$. The phase portrait
thus corresponds to the phase portrait of type \Dminus\ on  $\mathcal{S}_\flat$,
see Fig.~\ref{magneticAflat}.
(The system~\eqref{dynsysAsharp} on the set $\mathcal{S}_\sharp$ with
$w= \textfrac{1}{3}$ and $\beta = 4$ is identical;
see the phase portrait of type \Dplus\ on $\mathcal{S}_\sharp$,
Fig.~\ref{magneticAsharp}.)
Accordingly, the Taub point 
is the $\alpha$-limit and the non-flat LRS point the $\omega$-limit of every LRS Bianchi type~II orbit.

{\it Bianchi type~IX}.  The reduced dynamical system for LRS Bianchi type~IX aligned magnetic fields cosmologies is 
obtained from~\eqref{dynsyspolar} on $\mathcal{Y}_{\mathrm{IX}}^+$ by
setting $u =w_2 = w_3 \equiv 1$ and $w = \textfrac{1}{3}$, i.e.,
\begin{subequations}\label{dynsyspolarmag}
\begin{align}
r^\prime & = 2\, r \left( H_D (q - H_D \Sigma_+) - 3 \Sigma_+ \sin^2 \vartheta \right), \\[0.5ex]
\vartheta^\prime & = - 3 \Sigma_+ \sin (2\vartheta)\:, \\[0.5ex]
\Sigma_+^\prime & = r \sin\vartheta -1 + (H_D -\Sigma_+)^2 + H_D \Sigma_+ (q- H_D\Sigma_+) + 2\Omega\:,
\end{align}
\end{subequations}
where 
$\Omega =  1 - \Sigma_+^2 - \textfrac{1}{4}\, r \sin\vartheta$, $q=2\Sigma_+^2+\Omega$ 
and $H_D=\sqrt{1-2r\cos\vartheta}$. 
The state space is $\mathcal{Y}_{\mathrm{IX}}^+$, see~\eqref{statespaceIXpolar}.
(By using the formulation of Subsec.~\ref{blowup} 
we restrict ourselves to the dynamics of solutions in their expanding phase $H_D > 0$;
by applying the discrete symmetries~\eqref{symIX} the entire dynamics is obtained.)
The analysis of the system~\eqref{dynsyspolarmag} is completely analogously to
the analysis of Secs.~\ref{B9sec} and~\ref{B9res}.
First, we analyze the flow induced on the four boundary subsets.
On the side $\mathcal{S}_\sharp$, the system induced by~\eqref{dynsyspolarmag}
corresponds to~\eqref{Sidesys} with $w = \textfrac{1}{3}$ and $\beta = 4$, i.e., to the case 
\Dplus\ of Fig.~\ref{magneticAsharp}.
The dynamical system induced on the base $\mathcal{B}_{\mathrm{IX}}$
coincides with~\eqref{Basesys} for $w = \textfrac{1}{3}$ and $\beta=-2$ (which is of type~\Dminus);
see Fig.~\ref{Bfigmag}.
The flow on the vacuum boundary is independent of the matter model.
On the Bianchi type~I boundary $\mathcal{X}_{\,\mathrm{I}}$
the flow is equivalent to~\eqref{BianchiImag}. 
Second, in analogy to the procedure of Sec.~\ref{B9res} we use
the monotone function~\eqref{delta} and glue together the pieces.
The final conclusion is that the fixed point 
$\mathrm{T}_\sharp$ is the $\alpha$-limit of every LRS Bianchi type~IX orbit.
Furthermore, every model is recollapsing and 
the $\omega$-limit set is the Taub point arising from $\mathrm{T}_\sharp$ by applying
the discrete symmetry.

\begin{Remark}
The qualitative behavior of magnetic field cosmologies considered in this section is 
exactly the same as that of vacuum models, see Sec.~\ref{perfectfluid}. 
\end{Remark}

\begin{Remark}
The conclusions are identical if we add an electric field parallel to the magnetic 
field. The Maxwell equations then 
show that $E^1 = L (g^{11})^{1/2}g^{22}$ 
and $E^2=E^3=0$, where $L = \mathrm{const}$.
The energy density becomes $8 \pi \rho = (g^{22})^2 (K^2 + L^2)$,
so that the energy-momentum tensor is a functional of
the metric; accordingly,~\eqref{Tmagnetic} and its consequences remain valid.
\end{Remark}

\begin{Remark}[Erratum]
The statement in~\cite{CH2} that magnetic fields are a matter model of 
type~\Dminus\ in the 
context of Bianchi type~I is erroneous.
The content of this subsection corrects this claim:
Magnetic fields do not directly fit into the classification of Table~\ref{tab2} and Fig.~\ref{mattermodelsfig},
but can be treated in an analogous manner.
\end{Remark}

%%%%%%%%%%%%%%%%%%%%%%%%%%%%%%%%%%%%%%%%%%%%%%%%%%%%%%%%
%%%%%%%%%%%%%%%%%%%%%%%%%%%%%%%%%%%%%%%%%%%%%%%%%%%%%%%%
\section{Discussion and open problems}\label{discussion}
%%%%%%%%%%%%%%%%%%%%%%%%%%%%%%%%%%%%%%%%%%%%%%%%%%%%%%%%
%%%%%%%%%%%%%%%%%%%%%%%%%%%%%%%%%%%%%%%%%%%%%%%%%%%%%%%%

In this paper we have studied in detail 
spatially homogeneous locally rotationally symmetric
solutions of the Einstein-matter equations.
Rather than explicitly specifying a particular matter source, 
we have imposed a set of assumptions on the matter that
characterize a large class of matter models
including classical examples like elastic matter (associated
with a large variety of equations of state), collisionless matter
(Vlasov matter), and magnetic fields;
in addition, the perfect fluid matter model with a linear equation of state, which 
is the most commonly used matter model in cosmology and astrophysics, 
is naturally embedded in the class of matter models we consider.
The main aim of our analysis has been to study the influence 
of the matter source on
the dynamics of the associated (SH LRS) cosmological models.
A natural focus has been to ask how robust
the qualitative dynamics of perfect fluid solutions is 
under changes of the matter model. Does matter `matter'?

The main result of our analysis is that, indeed, anisotropies
of the matter model `matter'. Two facts are of particular interest:
(i) There are anisotropic matter models that satisfy the
strong (and the dominant) energy condition such that the
associated cosmological solutions of Bianchi type~IX (i.e., closed
cosmological models) do not recollapse but expand forever;
this is in stark contrast to perfect fluid models and other anisotropic matter models;
see Theorem~\ref{equalfootingthm} and its interpretation in Sec.~\ref{B9res}.
(ii) For a class of matter models including collisionless matter
the asymptotic behavior of solutions toward the initial singularity
is completely different from that of vacuum and perfect fluid solutions
in that the approach to the singularity is not governed by the Taub solution
but oscillatory; see Theorem~\ref{BianchiIXtheo} and its interpretation in Sec.~\ref{B9res}.

There are several interesting open problems.
We have restricted our attention to three Bianchi types of class~A:
Bianchi type~I,~II and~IX. 
However, as already indicated in Secs.~\ref{reducedsystem} and~\ref{B89sec},
it is clear that the methods we have developed 
can be extended and/or modified to analyze other cosmological models of interest:
On the one hand, there are the Kantowski-Sachs models, which are interesting because
their spatial topology is $S^1 \times S^2$ and thus closed.
(Do models necessarily recollapse?)
On the other hand, there are the Bianchi type~III models (which are of class~B)
and the type~VIII models, the latter being as fundamental as
type~IX models for our understanding of generic spacelike singularities.
Note that in the analysis of type~III and type~VIII models an additional
challenge might arise, since in each case the state space for the reduced 
dynamical system describing these models is probably unbounded. 
We leave the analysis of the 
Kantowski-Sachs models and the LRS Bianchi type~III and type~VIII models as interesting open problems.   

We expect 
that there exists a number of ways to extend the analysis and the 
results presented in this paper.
Let us begin with Assumption~\ref{assumptionT}.
There exists an obvious generalization of this assumption, which is 
to require that the stress-energy tensor is a function
not only of the spatial metric but of the second fundamental form $k_{ij}$ as well. 
This new assumption would lead to a generalization of eq.~\eqref{Tij} and, in the LRS case, 
give rise to rescaled principal pressures that depend not only on $s$ but also on $\Sigma_+$. 

Let us next turn our attention to Assumption~\ref{assumptionwi}.
This assumption imposes a linear equation of state between the isotropic 
pressure and the energy density, which implies that $\rho$ depends on $n=1/\sqrt{\det g}$ 
as indicated in~\eqref{rhons}. However, based on the results derived in this paper,
it is relatively simple to treat nonlinear equations of state.
Let us elaborate. Assume that $\rho$ is given by a more general
function $\rho = \rho(n,s_1,s_2,s_3)$ instead of~\eqref{rhons}. An interesting subcase is 
$\rho(n,s_1,s_2,s_3) = \varphi(n) \psi(s_1,s_2,s_3)$. In the latter case we obtain that $w_i = w_i(n, s_1,s_2,s_3)$, $i=1,2,3$,
satisfy~\eqref{wis} where $w = w(n) = (\partial \log \varphi/\partial n) -1$. 
In this case, the Einstein evolution equations, 
written in terms of the dynamical systems variables of Sec.~\ref{reducedsystem},
decouple into an equation for $D$ and a reduced system of equations for the remaining variables,
which is given by~\eqref{domsys}, supplemented by an additional equation for $n$, i.e., $n^\prime = -3 n$.
Assume we have fixed the Bianchi type and let $\mathcal{X}$ 
denote the state space for the corresponding reduced dynamical system.
When we choose to compactify the variable $n$, i.e., when we replace $n$ by $N = n/(1+n)$,
the state space of this dynamical system is $\mathcal{X}\times (0,1)$.
If we assume an equation of state such
that $w_i(N,s_1,s_2,s_3)$ possess well-defined limits as $N\rightarrow 0$
and $N\rightarrow 1$, we can extend the dynamical system to the boundaries 
$\mathcal{X}_0 = \mathcal{X} \times \{0\}$ and
$\mathcal{X}_1 = \mathcal{X} \times \{1\}$ of the state space.
The dynamical system on each of these boundary subsets coincides with the system~\eqref{domsys}
that we have discussed so extensively in this paper. In particular, we may define a parameter $\beta_0$ 
using the rescaled principal pressures $w_i(0,s_1,s_2,s_3)$ in~\eqref{betadef1} and a parameter $\beta_1$ 
using $w_i(1,s_1,s_2,s_3)$ and apply our classification of matter models for $N=0$ and $N=1$ separately. 
Of course, in general, the matter model on the boundaries $N= 0$ and $N= 1$ will not be of the same type.
Since the variable $N$ is strictly monotone, the asymptotic dynamics of solutions of
the dynamical system is associated
with the limits $N\rightarrow 0$ and $N\rightarrow 1$.
Accordingly, asymptotically, the flow of the boundary subsets 
$\mathcal{X}_0$ and $\mathcal{X}_1$ (and thus the results of the present paper)
constitute the key to an understanding of 
the dynamics of the more general problem with nonlinear equations of state.

Finally, consider generalizations of Assumptions~\ref{asswi} and~\ref{asspsi}.
Concerning the latter, it is possible to permit the existence of more than one isotropic 
state or the case that isotropic states do not exist at all. 
The latter is satisfied by magnetic fields, see Subsec.~\ref{magnetic},
which thus seem more interesting from the physical point of view.
A natural way to 
generalize Assumption~\ref{asswi} is to require the existence of three constants $c_1$, $c_2$, $c_3$, such 
that $w_i(s_1,s_2,s_3)=c_i$ when $s_i=0$. A simple subcase is $w_i(s_1,s_2,s_3) \equiv c_i$, i.e., 
the rescaled pressures are constants, which holds for magnetic fields, see Subsec.~\ref{magnetic}. 
If we generalize Assumption~\ref{asswi} as indicated, there would appear 
three parameters, instead of $\beta$ alone, in the reduced dynamical system. In this case the number 
of possible dynamics and bifurcations would increase considerably, but no conceptual new difficulty would 
arise. Whether new phenomena occur under these relaxed assumptions is an open question.

%%%%%%%%%%%%%%%%%%%%%%%%%%%%%%%%%%%%%%%%%%%%%%%%%%%%%%%%%%%%%%%%%%%%
%%
%% APPENDIXES
%%
%%%%%%%%%%%%%%%%%%%%%%%%%%%%%%%%%%%%%%%%%%%%%%%%%%%%%%%%%%%%%%%%%%%%

\begin{appendix}

%%%%%%%%%%%%%%%%%%%%%%%%%%%%%%%%%%%%%%%%%%%%%%%%%%%%%%%%%%%%%%%%%%%%
\section{Exact solutions}
\label{exact}
%%%%%%%%%%%%%%%%%%%%%%%%%%%%%%%%%%%%%%%%%%%%%%%%%%%%%%%%%%%%%%%%%%%%

In the course of our analysis of the system of equations~\eqref{domsys} 
we have discovered a number of fixed points, see Table~\ref{newfixed}.
In this appendix we derive the exact solutions that correspond to the 
these fixed points.

\begin{table}[Ht!]
\begin{center}
\begin{tabular}{|c|c|c|c|c|c|c|}
\hline  & & & &  & & \\[-2ex]
Fixed  & \multirow{2}{*}{$s$} & \multirow{2}{*}{$\Sigma_+$}  &  \multirow{2}{*}{$M_1^2$} &   
\multirow{2}{*}{$H_D$} & \multirow{2}{*}{Solution} & Bianchi\\[-0.5ex]
point & & & & & & type \\ 
\hline  & & & & &  & \\[-2ex] 
$\mathrm{T}_\flat$ & $0$ & ${-1}$ & $0$ & $1$ & \eqref{taub} & I \\
$\mathrm{T}_\sharp$ & $\textfrac{1}{2}$ &  ${-1}$ & $0$ & $1$ & \eqref{taub} & I \\
$\mathrm{Q}_\flat$ & $0$ & $1$ & $0$ & $1$ & \eqref{solQ} & I \\
$\mathrm{Q}_\sharp$ & $\textfrac{1}{2}$ &  $1$ & $0$ & $1$ & \eqref{solQ} & I \\
$\mathrm{F}$ & $\bar{s}$ & $0$ & $0$ & $1$ & \eqref{FRWsol} & I \\
$\mathrm{R}_\flat$  &  $0$ & ${-\beta}$  & $0$ & $1$ &\eqref{rfsol} & I\\
$\mathrm{R}_\sharp$  &  $\textfrac{1}{2}$ & $\textfrac{\beta}{2}$ & $0$ & $1$ & \eqref{rssol}& I \\
$\mathrm{C}_\flat$ &  $0$ & $\textfrac{1+3w}{8+3\beta(1-w)}$  & 
$\textfrac{36(1-w)[3\beta^2(1-w)+8\beta+(1+3w)]}{[8+3\beta(1-w)]^2}$ & $1$ &\eqref{cfsol} & II\\
$\mathrm{C}_\sharp$  &  $\textfrac{1}{2}$ & $\textfrac{1 + 3 w}{16 -3 \beta (1-w)}$  & 
$\textfrac{36(1-w)[3\beta^2(1-w)-16\beta+4(1+3w)]}{[16-3\beta(1-w)]^2}$ & $1$ & \eqref{cssol} & II\\
$\mathrm{P}$  & $0$ & $\textfrac{1+3w}{\sqrt{(1-3 w)^2 + 6 \beta (1-w)}}$  & $0$ & 
$\textfrac{2 + 3\beta (1-w)}{\sqrt{(1-3 w)^2 + 6 \beta (1-w)}}$ & \eqref{Psol} & III/KS\\
\hline 
\end{tabular}
\caption{Table of the fixed points of the reduced dynamical system.
Each of the points $\mathrm{R}_\flat$, $\mathrm{R}_\sharp$, 
$\mathrm{C}_\flat$, $\mathrm{C}_\sharp$, and $\mathrm{P}$
exists only under certain restrictions on the anisotropy parameter $\beta$.}
\label{newfixed}
\end{center}
\end{table}

Consider an arbitrary fixed point of Table~\ref{newfixed} and
let $(s,\Sigma_+,M_1, H_D)$ be the coordinates of this fixed point;
the associated deceleration parameter is
\[
q =2 \Sigma_+^2+\textfrac{1}{2}(1+3w)\big(1-\Sigma_+^2-\textfrac{1}{12}\,\hat{n}_1^2 M_1^2 \big)\:,
\]
see~\eqref{qD} and~\eqref{gaussconsimple}.
Expressed in $t$, eq.~\eqref{Ddecoupled} reads
\begin{equation}\label{Hzeq}
D^{-2}\,\partial_tD=-[H_D (1+q)+ \Sigma_+ (1-H_D^2)]\:,
\end{equation}
which can be solved to obtain 
\begin{equation}\label{Hz}
D(t)=\big[H_D (1+q)+\Sigma_+(1- H_D^2)\big]^{-1}\:t^{-1}\:.
\end{equation}
where we have shifted the origin of time so that $D$ diverges as $t\rightarrow 0$.
Using~\eqref{Omegarho} and~\eqref{gaussconsimple} this leads to 
\begin{equation}\label{rhoz}
\rho(t)=\frac{3\big(1-\Sigma_+^2-\textfrac{1}{12}\, \hat{n}_1^2 M_1^2\big)}{\big[H_D (1+q)+\Sigma_+ (1-H_D^2)\big]^2}\;t^{-2}\:.
\end{equation}
The principal pressures $p_1(t)$, $p_2(t)$ are then given by
\begin{equation}\label{principalpressures}
p_1(t)=w_1\rho(t) = \big(3w-2u(s)\big)\rho\:,\qquad 
p_2(t)=w_2\rho=u(s)\rho\:,
\end{equation}
see~\eqref{udef}, where the constant $u(s)$ is known
(since $s \in \{0, \bar{s}, \textfrac{1}{2}\}$), see Sec.~\ref{LRSsols}.

To obtain the metric we note that
\begin{equation}\label{metricevolnew}
\partial_t g_{11}={-2} g_{11}D(t)\, (2 \Sigma_+-H_D )\:,\qquad
\partial_t g_{22}= 2g_{22} D(t) \,(\Sigma_+ +H_D )\:,
\end{equation}
which follows from~\eqref{evol1diag} in combination with~\eqref{Hubble} and~\eqref{normvars}.
Inserting~\eqref{Hz} into~\eqref{metricevolnew} and solving we obtain
\begin{subequations}\label{exactsol}
\begin{equation}
g_{11}=a\,t^{2\gamma_1}\:,\qquad g_{22}=b\,t^{2\gamma_2}\:,
\end{equation}
where $a$, $b$ are positive constants and 
\begin{equation}
\gamma_1=\frac{H_D-2 \Sigma_+}{H_D (1+q)+\Sigma_+ (1-H_D^2)}\:,\qquad
\gamma_2=\frac{H_D+ \Sigma_+}{H_D (1+q)+\Sigma_+(1-H_D^2)}\:.
\end{equation}
\end{subequations}
We remark that the constants $a$, $b$ are in general not arbitrary: 
Imposing the Hamiltonian constraint may give rise to a restriction on their values. 

Applying the 
above algorithm to 
the fixed points of Table~\ref{newfixed}
yields the exact solutions represented by these points.
For $\mathrm{T}_\flat$, $\mathrm{T}_\sharp$ we obtain
the Taub solution~\eqref{taub};
for $\mathrm{Q}_\flat$, $\mathrm{Q}_\sharp$ we obtain
the non-flat LRS Kasner solution~\eqref{solQ};
for $\mathrm{F}$ we get the Robertson-Walker metric~\eqref{FRWsol}.

The fixed points $\mathrm{R}_\sharp$ and $\mathrm{R}_\flat$
are of Bianchi type~I, i.e., $\hat{n}_1 = 0$ and
$\hat{n}_2 = \hat{n}_3 = 0$.

\textit{Fixed point $\mathrm{R}_\sharp$}: 
\begin{subequations}\label{rssol}
\begin{alignat}{2}
\gamma_1&=\frac{8(1-\beta)}{3[\beta^2(1-w)+4(1+w)]}\,,& \qquad \gamma_2 & =\frac{4(2+\beta)}{3[\beta^2(1-w)+4(1+w)]}\:,\\
\rho &=\frac{16(4-\beta^2)t^{-2}}{3[\beta^2(1-w)+4(1+w)]^2}\,,& \qquad 
p_1 & =\textfrac{1}{2}[w(2+\beta)-\beta]\rho\:,
\qquad p_2=\textfrac{1}{4}[w(4-\beta)+\beta]\rho\:.
\end{alignat}
\end{subequations}

\textit{Fixed point $\mathrm{R}_\flat$}: 
\begin{subequations}\label{rfsol}
\begin{alignat}{2}
\gamma_1&=\frac{2(1+2\beta)}{3[\beta^2(1-w)+1+w]}\,,& \qquad\gamma_2 & =\frac{2(1-\beta)}{3[\beta^2(1-w)+1+w]}\:,\\
\rho &=\frac{4(1-\beta^2)t^{-2}}{3[\beta^2(1-w)+1+w]^2}\:, & \qquad 
p_1& = [w(1-\beta)+\beta]\rho\:,
\qquad p_2 =\textfrac{1}{2}[w(2+\beta)-\beta]\rho\:.
\end{alignat}
\end{subequations}

Let us analyze the solution~\eqref{rfsol} associated with $\mathrm{R}_\flat$.
The values of $\gamma_1$ and $\gamma_2$ vary considerably depending on $w$ and on
the anisotropy parameter $\beta$ 
(where the existence of $\mathrm{R}_\flat$ requires ${-1}< \beta < 1$). 
In Fig.~\ref{gamma12} we depict $\gamma_1$ and $\gamma_2$ as functions
of $\beta$ where $w$ takes a fixed value. The topmost curves corresponds
to the value $w = -\textfrac{1}{3}$; 
the lowermost curves correspond
to $w = 1$.
An interesting observation is the possible occurrence of partial (directional) 
accelerated expansion for the metric given by~\eqref{rfsol}. 
It is straightforward to see that $\gamma_1 > 1$ is possible for some range of $\beta$
provided that $w < \textfrac{1}{3}$. (The range is $\beta \in (0,1)$ for $w = {-\textfrac{1}{3}}$
and decreases to $\beta \in (1-\epsilon,1)$ when $w$ approaches $w = \textfrac{1}{3}$.)
The maximal rate of acceleration is obtained by maximizing
$\gamma_1$ over the admissible domain of $\beta$ and $w$;
this yields the maximal value of $\gamma_1 = \textfrac{1}{2} (1+\sqrt{3}) \approx 1.366$,
which can only be attained in the limit $w \rightarrow -1/3$ and 
$\beta \rightarrow \textfrac{1}{2} ({-1}+\sqrt{3}) \approx 0.366$.

Analogously, a straightforward calculation reveals
that $\gamma_2 > 1$ is possible, which means that lengths in the plane
of local rotational symmetry expand at an accelerated rate, see Fig.~\ref{gamma12}.
The condition for $\gamma_2 > 1$ is $w < \frac{1}{3} (1-\sqrt{3})$ and
$\beta \in (\beta_-,\beta_+)$, cf.~\eqref{betadomain}, i.e.,
$\gamma_2 > 1$ in the anisotropy case \AminusP.
The maximal rate of acceleration is obtained by maximizing
$\gamma_2$ over the admissible domain~\eqref{domain}; this yields
the maximal value of $\gamma_2 \approx 1.112$.
This value can only be 
attained for $w$ close to $-1/3$; for
larger values of $w$ the maximal acceleration is lower
(e.g., the maximum of $\gamma_2$ is $1.067$ 
for $w=-0.3$ and $1.007$ for $w=-0.25$).
While expansion of areas might be accelerating, 
the standard volume expansion is decelerating; %as a matter of course;
we have $\sqrt{\det g} \propto t^{3/(1+q)}$
hence the length scale behaves like $t^{1/(1+q)}$.

\begin{figure}[Ht]
\begin{center}
\psfrag{b}[lc][lc][1][0]{$\beta$}
\psfrag{u}[tc][bc][0.7][0]{${-1}$}
\psfrag{v}[tc][bc][0.7][0]{${-0.5}$}
\psfrag{0.5}[tc][bc][0.7][0]{$0.5$}
\psfrag{1}[tc][bc][0.7][0]{$1$}
{\psfrag{g}[bc][bc][1][0]{$\gamma_1$}
\psfrag{a}[rc][cc][0.7][0]{$1$}
\psfrag{c}[rc][cc][0.7][0]{${-1/3}$}
\subfigure{\includegraphics[width=0.45\textwidth]{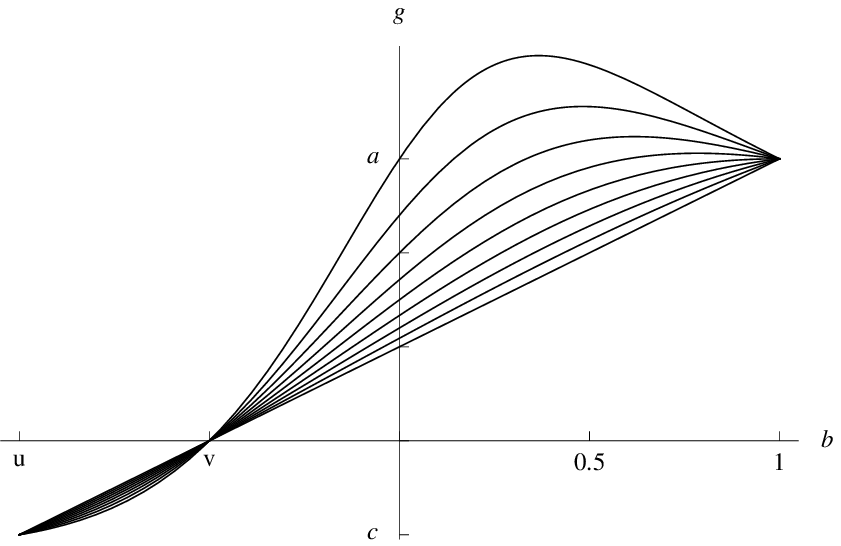}}}\qquad
{\psfrag{g}[bc][bc][1][0]{$\gamma_2$}
\psfrag{a}[lc][cc][0.7][0]{$1$}
\subfigure{\includegraphics[width=0.45\textwidth]{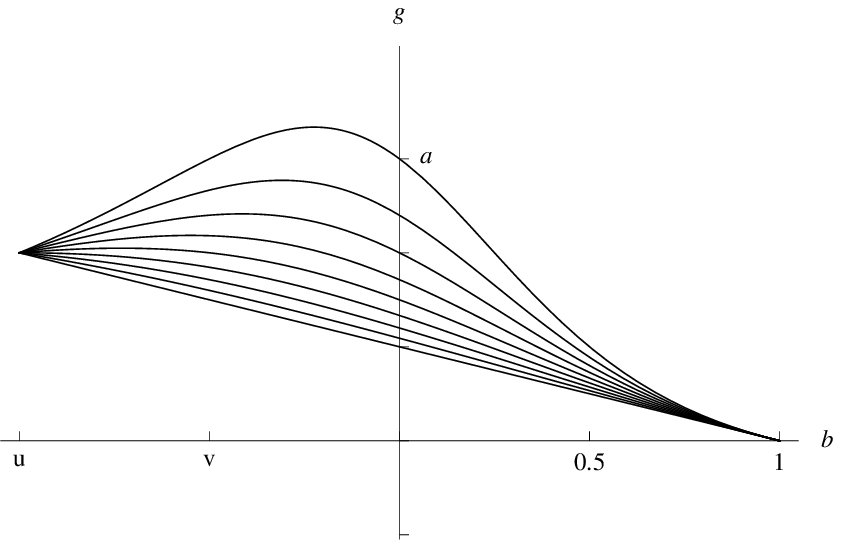}}}
\caption{The values of $\gamma_1$ and $\gamma_2$ for the exact solution~\eqref{rfsol} corresponding to $\mathrm{R}_\flat$.
Each curve corresponds to a given value of $w$, where the topmost curves correspond to $w = {-\textfrac{1}{3}}$
and the lowermost curves to $w = 1$.}
\label{gamma12}
\end{center}
\end{figure}

The fixed points $\mathrm{C}_\sharp$ and $\mathrm{C}_\flat$
are of Bianchi type~II, i.e., $\hat{n}_1 = 1$ and
$\hat{n}_2 = \hat{n}_3 = 0$.
These fixed points correspond to solutions that naturally
generalize the Collins-Stewart solution~\eqref{CSsol}.

\textit{Fixed point $\mathrm{C}_\sharp$}: 
\begin{subequations}\label{cssol}
\begin{alignat}{2}
\gamma_1& =\frac{(4-\beta)(1-w)}{8(1+w)-\beta(1-w)}\,& ,\qquad\gamma_2 & =\frac{2\beta(1+w)+6-\beta}{8w(1+\beta)+8-\beta}\:,\\
\rho &=\frac{16(5-w-\beta(1-w))t^{-2}}{[8(1+w)-\beta(1-w)]^2}\,,&\qquad 
p_1 &=\textfrac{1}{2}[w(2+\beta)-\beta]\rho\:,\qquad p_2=\textfrac{1}{4}[w(4-\beta)+\beta]\rho\:.
\end{alignat}
Furthermore, the constants $a$, $b$ of~\eqref{exactsol} are related 
by 
\begin{equation}
\frac{\sqrt{a}}{2\,b}=\frac{\sqrt{(1-w)[3\beta^2(1-w)-16\beta+4(1+3w)]}}{8-\beta+8w(1+\beta)}\:,
\end{equation}
\end{subequations}
where the function under the square root is positive for $\beta<\beta_\sharp$ (which is
the prerequisite for the existence of $\mathrm{C}_\sharp$).

\textit{Fixed point $\mathrm{C}_\flat$}: 
\begin{subequations}\label{cfsol}
\begin{alignat}{2}
\gamma_1 &=\frac{(2+\beta)(1-w)}{4(1+w)+\beta(1-w)}\,, & \qquad \gamma_2 & =\frac{3+\beta-2\beta(1+w)}{4w(1-2\beta)+\beta+4}\:,\\
\rho & =\frac{8(5-w+2\beta(1-w))t^{-2}}{[4(1+w)+\beta(1-w)]^2}\,,& \qquad 
p_1 & =[w(1-\beta)+\beta]\rho\:,\qquad p_2=\textfrac{1}{2}[w(2+\beta)-\beta]\rho\:.
\end{alignat}
Furthermore, $a$ and $b$ are related by 
\begin{equation}
\frac{\sqrt{a}}{2\,b}=\frac{\sqrt{(1-w)[3\beta^2(1-w)+8\beta+(1+3w)]}}{4+\beta+4w(1-2\beta)}\:,
\end{equation}
\end{subequations}
where the function under the square root is positive for $\beta>\beta_\flat$
(which is
the prerequisite for the existence of $\mathrm{C}_\flat$).
Note that passing from $\mathrm{R}_\sharp$ to $\mathrm{R}_\flat$ 
and from $\mathrm{C}_\sharp$ to $\mathrm{C}_\flat$ corresponds to
replacing $\beta$ by ${-2\beta}$.

\textit{Fixed point $\mathrm{P}$}:
\begin{subequations}\label{Psol}
\begin{alignat}{2}
\label{gamma1val}
\gamma_1 & =\frac{\beta(1-w)-2w}{1+w+\beta(1-w)}\,,& \qquad\gamma_2& =1\:,\\
\rho &=\frac{2\beta (1-w)-4w}{[1+w+\beta(1-w)]^2}\:t^{-2}\,,& \qquad 
p_1& =[w(1-\beta)+\beta]\rho\:,\qquad 
p_2 =\textfrac{1}{2}[w(2+\beta)-\beta]\rho\:.
\end{alignat}

The fixed point $\mathrm{P}$ is a fixed point on the Kantowski-Sachs
subset $\mathcal{B}_{\mathrm{IX}}$ of $\overline{\mathcal{X}}_{\mathrm{IX}}$;
its existence is restricted to the case \AminusP, which is defined through~\eqref{domain}.
The point $\mathrm{P}$ represents a solution of Kantowski-Sachs type, cf.~\eqref{KSmetric};
see~Sec.~\ref{KSIII} and the remark in Subsec.~\ref{subsect:BIXbound} for details.
The constant $b$ of~\eqref{exactsol} is given by 
\begin{equation}
b={-}\frac{(1-w)\big[3\beta^2(1-w)+2\beta+(1+3w)\big]}{\big(1+w+\beta(1-w)\big)^2} \:;
\end{equation}
since the numerator is negative, cf.~\eqref{Rtransvers}, $b$ is positive as required by~\eqref{exactsol}.
\end{subequations}

The value of $\gamma_1$, see~\eqref{gamma1val}, depends on $w$ and the anisotropic parameter $\beta$
(which are subject to~\eqref{domain}, since the anisotropy case is \AminusP).
We find that $\gamma_1 < 1$ irrespective of $w$ and $\beta$;
the possible values of $\gamma_1$ are depicted in Fig.~\ref{gammaminmax}.
Accordingly, the $S^1$ component of the Kantowski-Sachs metric~\eqref{KSmetric}
expands at a lesser rate than the $S^2$ component.

\begin{figure}[Ht]
\begin{center}
\psfrag{a}[tc][bc][0.7][0]{$-\textfrac{1}{3}$}
\psfrag{c}[tc][bc][0.7][0]{$-0.30$}
\psfrag{d}[tc][bc][0.7][0]{$-0.27$}
\psfrag{e}[tc][bc][0.7][0]{$\textfrac{1-\sqrt{3}}{3}$}
\psfrag{g}[bc][bc][1][0]{$\gamma_1$}
\psfrag{m}[lc][lc][0.7][0]{$0.366$}
\psfrag{n}[lc][lc][0.7][0]{$0.7$}
\psfrag{o}[lc][lc][0.7][0]{$1.0$}
\psfrag{w}[lc][lc][1][0]{$w$}
\includegraphics[width=0.55\textwidth]{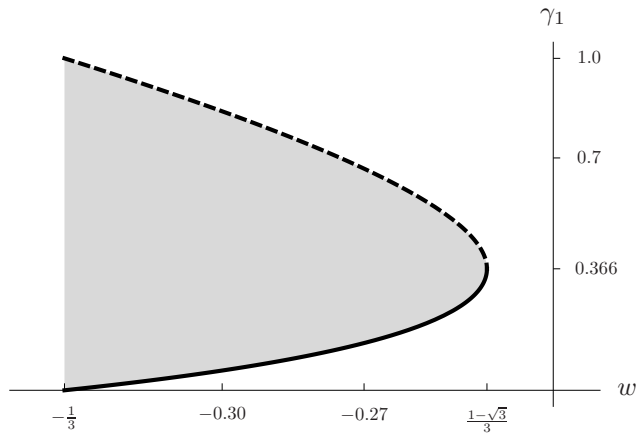}
\caption{The value of $\gamma_1$ for the exact solution~\eqref{Psol} corresponding to $\mathrm{P}$.
The anisotropy case is the Bianchi type~IX special case \AminusP, which is defined by~\eqref{domain}; see Fig.~\ref{betaminmax}.
For a given value of $w$, the value of $\gamma_1$ depends on $\beta$;
if $\beta = \beta_-$, then $\gamma_1$ takes the value on the solid black line;
if $\beta = \beta_+$, then $\gamma_1$ takes the value on the dashed black line;
if $\beta_- < \beta < \beta_+$, then the value of $\gamma_1$ is in
the gray region. The axes of the diagram are $\gamma_1 = 0$ and $w = -0.23$.}
\label{gammaminmax}
\end{center}
\end{figure}

\begin{Remark}
There exists a different range of the parameters $w$ and $\beta$
(different from \AminusP) 
such that $\mathrm{P}$ is located on the (type~III) boundary
of the Bianchi type~VIII state space.
In this case, the fixed point $\mathrm{P}$ represents a type~III solution
and the constant $b$ of~\eqref{exactsol} is given by 
\begin{equation*}
b=\frac{(1-w)\big[3\beta^2(1-w)+2\beta+(1+3w)\big]}{\big(1+w+\beta(1-w)\big)^2} \:,
\end{equation*}
which is then positive.
\end{Remark}

%%%%%%%%%%%%%%%%%%%%%%
\section{Local rotational symmetry}
\label{lrsexplained}
%%%%%%%%%%%%%%%%%%%%%

In 
this appendix we 
discuss in detail the concept of locally rotational symmetry (LRS). We analyze the LRS Bianchi class~A models,
as well as the Kantowski-Sachs models and the Bianchi type~III models.

Spatially homogeneous cosmological models that are locally rotationally symmetric (LRS) 
are spacetimes that admit a four-dimensional isometry group 
whose orbits are three-dimensional spacelike hypersurfaces.
We distinguish two cases: For the LRS Bianchi models the four-dimensional isometry group 
admits a three-dimensional subgroup $\mathcal{G}$
that acts simply transitively
on the three-dimensional orbits. For the Kantowski-Sachs models there does not
exist any three-dimensional subgroup with that property.

%---------------------------------------------------------------------------------
\subsection{LRS Bianchi class A models}
\label{LRSclassA}
%---------------------------------------------------------------------------------

For a Bianchi model of class~A 
there exists an adapted left-invariant frame
$\hx_i$ on $\mathcal{G}$ such that
\begin{equation*}
[\hx_i,\hx_j] = \varepsilon_{ijk} \hat{n}_k \hx_k
\end{equation*}
for any triple $(ijk)$, see Sec.~\ref{SH}. 
(Throughout this appendix we employ the Einstein summation convention.)
The structure constants are thus represented by
a triple $(\hat{n}_1,\hat{n}_2,\hat{n}_3)$ 
with $\hat{n}_i \in \{0,{-1},{+1}\}$ $\forall i$, 
see Table~\ref{tab1} of Sec.~\ref{SH}.
The spatial metric of a Bianchi model 
reads $g_{i j} \:\hat{\omega}^i \otimes \hat{\omega}^j$,
where the components $g_{i j}$ are
(spatial) constants (i.e., functions of the time variable alone)
and $\hat{\omega}^i$ is dual to $\hx_i$.

For an 
LRS Bianchi model the isometry group of the spatial metric is four-dimensional.
Therefore, there exists a one-dimensional isotropy group 
corresponding to a pointwise (left-invariant) symmetry
of the metric. Visualizing this symmetry as  
a symmetry of the `unit sphere' $\{ X^i \,|\, g_{i j} X^i X^j = 1\}$
we see that merely one continuous isometry comes into question: Axial symmetry.
An axially symmetric `unit sphere' is characterized by two degrees of freedom
(represented by the axes of the ellipsoid);
we conclude that the assumption of a four-dimensional isometry group will reduce
the number of degrees of freedom of the spatial metric to two.

In this section 
we analyze under which conditions LRS Bianchi models exist---note that there are Bianchi types
that are incompatible with the assumption of LRS symmetry---and 
we discuss their main properties.
We begin
with the simplest (and most illustrative) example.

The simplest of Bianchi models are those of type I.
Since 
\begin{equation}\label{BIalg}
[\hx_1,\hx_2] = 0\:,\qquad
[\hx_2,\hx_3] = 0\:,\qquad
[\hx_3,\hx_1] = 0\:,
\end{equation}
the symmetry-adapted frame is in fact a coordinate frame; 
we denote the coordinates by $x^i$. 
In the LRS case there exists a one-parameter isotropy group corresponding to a rotational symmetry---w.l.o.g.\
we assume that the axis is $\hx_1 = \partial_{x^1}$; 
we thus have $\hx_1 \mapsto \hx_1$ and $\hx_I \mapsto O^K_{\weg I} \hx_K$ 
where $O \in \mathrm{SO}(2)$; Latin capitals $I, J, \ldots$ run over $2$ and $3$ (and the Einstein summation convention
is understood).
The generator of this rotation in the $2$-$3$-plane is
\begin{equation*}
\heta = {-x^3}\, \frac{\partial}{\partial x^2} + x^2 \,\frac{\partial}{\partial x^3} = {-x^3}\,\hx_2 + x^2 \,\hx_3\:.
\end{equation*}
Accordingly, the four-dimensional Lie group that represents the isometry group of an LRS type~I
Bianchi model is generated by
the Lie algebra
\begin{equation}\label{BILRSalg}
[\hx_1,\hx_2] = 0\:,\quad
[\hx_2,\hx_3] = 0\:,\quad
[\hx_3,\hx_1] = 0\:,\quad
[\heta,\hx_1] = 0\:,\quad
[\heta,\hx_2] = {-\hx}_3\:,\quad
[\heta,\hx_3] = \hx_2\:.
\end{equation}
Any metric that is invariant under the action of the Lie group generated by~\eqref{BILRSalg} 
must satisfy 
\mbox{$g_{IK} \propto \delta_{IK}$},
i.e., 
\begin{equation*}
g_{22} = g_{33}\:.
\end{equation*}

\begin{Remark}
There exist several four-dimensional Lie algebras that contain~\eqref{BIalg} as a subalgebra, see, e.g.,~\cite{Petrov};
note, however, that~\eqref{BILRSalg} is the only Lie algebra that generates a Lie group that can act
as an isometry group on a Bianchi type~I spacetime.
\end{Remark}

Let $\epsilon_{I J}$ denote the standard permutation symbol (for the indices $I, J = 2,3$).
For an arbitrary Bianchi type (of class A) 
we may write the commutator relations as
\begin{equation}\label{commax}
[\hx_I, \hx_J] = \epsilon_{IJ}\: \hat{n}_1 \,\hx_1\:,\qquad 
[\hx_1, \hx_K] = B^I_{\weg K} \hx_I\:,  \qquad\text{where}\quad B = \begin{pmatrix} 0 & -\hat{n}_3 \\ \hat{n}_2 & 0 \end{pmatrix}\:.
\end{equation}
Consider 
a rotation in $\hx_2$, $\hx_3$ with axis $\hx_1$, i.e.,
let $O\in \mathrm{SO}(2)$, %i.e.,
and define $\hx_1^\prime = \hx_1$ and $\hx_I^\prime = O^K_{\weg I} \hx_K$.
The tensor $\epsilon_{I J}$ is invariant under this transformation, 
because $O^K_{\weg I} O^L_{\weg J} \epsilon_{K L} = \epsilon_{I J}$.
Hence the commutators~\eqref{commax} result in 
\begin{equation}\label{commaxprime}
[\hx_I^\prime, \hx_J^\prime] =  \epsilon_{IJ}\: \hat{n}_1 \,\hx_1^\prime\:,\qquad
[\hx_1^\prime, \hx_K^\prime] = (O^{-1} B \,O)^I_{\weg K} \,\hx_I^\prime\:.
\tag{\ref{commax}${}^\prime$}
\end{equation}
It is straightforward to see that $B$ is invariant, i.e., $O^{-1} B\, O = B$ for all $O \in \mathrm{SO}(2)$, 
if and only if
\begin{subequations}\label{lrsconds}
\begin{equation}
\hat{n}_2 = \hat{n}_3 \:.
\end{equation}
The assumption 
\begin{equation}
g_{22} = g_{33}
\end{equation}
\end{subequations}
then yields an invariant metric, i.e., a LRS Bianchi class A model.

Therefore, a spacetime of Bianchi class A admits a four-dimensional isometry group, if and only 
two structure constants are equal (which we assume to be $\hat{n}_2 = \hat{n}_3$) 
and if the corresponding metric components are equal as well (i.e., $g_{22} = g_{33}$).
The four-dimensional isometry group is generated by the Lie algebra
\begin{subequations}\label{lrslie}
\begin{alignat}{3}
\label{lrslie1}
[\hx_1,\hx_2] & = \hat{n}_3 \hx_3 \:,\qquad
& [\hx_2,\hx_3] & = \hat{n}_1 \hx_1 \:,\qquad
& [\hx_3,\hx_1] & = \hat{n}_2 \hx_2 \:, \\
\label{lrslie2}
[\heta,\hx_1] & = 0 \:,\qquad
& [\heta,\hx_2] & = -\hat{\epsilon}\, \hx_3 \:,\qquad
& [\heta,\hx_3] & = \hat{\epsilon}\, \hx_2 \:,
\end{alignat}
\end{subequations}
where $\hat{n}_2 = \hat{n}_3$ and 
$\hat{\epsilon} = 1$.
Let us discuss the Bianchi class~A models type by type.

In the Bianchi type~I case there exist three different 
but equivalent representations of LRS solutions; we make the choice
$g_{22} = g_{33}$.
In the Bianchi type~II case, for a fixed triple $(\hat{n}_1,\hat{n}_2,\hat{n}_3)= (1,0,0)$,
there exists a unique class of LRS models, which is characterized by $g_{22} = g_{33}$
(since $\hat{n}_2 = \hat{n}_3 = 0$).

The Bianchi type~$\mathrm{VI}_0$ case is exceptional: There do not exist
any LRS type~$\mathrm{VI}_0$ models.
The reason is that 
the structure constants are pairwise different,
which entails that the rotations $O \in \mathrm{SO}(2)$ do not act as isometries;
note that this is irrespective of the assumption $g_{22} = g_{33}$.
Alternatively, we note that~\eqref{lrslie} does not define a Lie algebra
in the type~$\mathrm{VI}_0$ case; 
setting $(\hn_1,\hn_2,\hn_3) = (0,1,{-1})$ in~\eqref{lrslie}
we find 
$[\heta,[\hx_1,\hx_2]]+ [\hx_1,[\hx_2,\heta]]+ [\hx_2,[\heta,\hx_1]] = 
-\hat{\epsilon} [\heta,\hx_3] + \hat{\epsilon} [\hx_1,\hx_3] = -2 \hat{\epsilon}\,\hx_2 \neq 0$.
However, while Bianchi type~$\mathrm{VI}_0$ is incompatible with a four-dimensional
isometry group, instead of a continuous symmetry there exists
a discrete symmetry. The reflection
\begin{equation}\label{reflection}
O = \begin{pmatrix} 0 & 1 \\ 1 & 0 \end{pmatrix} 
\end{equation}
leaves the matrix $B$ in~\eqref{commax} invariant, i.e., $O^{-1} B\, O = B$. 
The tensor $\epsilon_{I J}$ is not invariant, but $\hat{n}_1 = 0$; therefore,~\eqref{commax}
is invariant under~\eqref{reflection}.
The transformation~\eqref{reflection} is a discrete isometry if
\begin{equation*}
g_{22} = g_{33} \:.
\end{equation*}
Therefore, among Bianchi type~$\mathrm{VI}_0$ models there do not exist any
LRS models; however, there exist models that admit a discrete isometry
in addition to the three-dimensional group of isometries.

\begin{Remark}
In general, if $\hat{n}_2 = {-\hat{n}_3}$, we find that
the transformations $O \in \mathrm{SO}(1,1)$ define a continuous isometry
for a metric with $g_{22} = {-g_{33}}$.
Hence, if the spatial metric were not a Riemannian but a Lorentzian
metric, the models of Bianchi type~$\mathrm{VI}_0$ would admit
a natural subclass of models whose isometry group is four-dimensional;
one might call these models ``locally boost symmetric''.
For the Lie algebra of this isometry group we would have 
$[\heta,\hx_2] = \hat{\epsilon}\, \hx_3$ and $[\heta,\hx_3] = \hat{\epsilon}\, \hx_2$
in~\eqref{lrslie2}.
\end{Remark}

LRS models of Bianchi type~$\mathrm{VII}_0$ admit an isometry group generated by
\begin{equation}\label{VIIalg}
[\hx_1,\hx_2]  = \hx_3 \:,\quad
[\hx_2,\hx_3]  = 0 \:,\quad
[\hx_3,\hx_1] = \hx_2 \:, \quad
[\heta,\hx_1] = 0 \:,\quad
[\heta,\hx_2] = {-\hx_3} \:,\quad
[\heta,\hx_3]  = \hx_2 \:.
\end{equation}
Define
\begin{equation*}
\hx^\prime_1 = \textfrac{1}{2}\:\big(\hx_1 + \heta\big)\:,\quad
\hx^\prime_2 = \hx_2\:,\quad
\hx^\prime_3 = \hx_3\:,\quad
\heta^\prime = \textfrac{1}{2}\:\big({-\hx}_1 + \heta \big)\:.
\end{equation*}
We find that
\begin{equation}
[\hx^\prime_1,\hx^\prime_2]  = 0 \:,\quad
[\hx^\prime_2,\hx^\prime_3]  = 0 \:,\quad
[\hx^\prime_3,\hx^\prime_1]  = 0 \:, \quad
[\heta^\prime,\hx^\prime_1]  = 0 \:,\quad
[\heta^\prime,\hx^\prime_2]  = {-\hx^\prime_3} \:,\quad
[\heta^\prime,\hx^\prime_3]  = \hx^\prime_2 \:,
\tag{\ref{VIIalg}${}^\prime$}
\end{equation}
hence the Lie algebra~\eqref{VIIalg} coincides with the Lie algebra generating the
isometry group of LRS type~I models.
The action of the subgroup generated by $\{\hx^\prime_1,\hx^\prime_2,\hx^\prime_3\}$
must be simply transitive, see~\cite[Appendix B]{Collins}.
Therefore, LRS Bianchi type~$\mathrm{VII}_0$ models are in fact of type~I.

Finally, there exist LRS models of Bianchi type~VIII and~IX;
in the latter case there exist three different 
but equivalent representations of LRS solutions since $\hat{n}_1 = \hat{n}_2 = \hat{n}_3$; 
we make the choice $g_{22} = g_{33}$.

Under the assumption of local rotational symmetry, where we make the choice~\eqref{lrsconds},
i.e., $\hat{n}_2 = \hat{n}_3$, cf.~Table~\ref{tab1}, and $g_{22} = g_{33}$, we find
that the spatial Ricci curvature~\eqref{riccidiagonal} becomes
\begin{subequations}\label{lrsricciA}
\begin{equation}
R^1_{\weg 1} = \textfrac{1}{2}\: \hat{n}_1^2 \,m_1^2\:, \qquad
R^2_{\weg 2} = \hat{n}_1 \hat{n}_2 \, m_1 m_2 - \textfrac{1}{2}\: \hat{n}_1^2 \,m_1^2\:,
\end{equation}
where 
\begin{equation}\label{m1m2def}
m_1 = \frac{g^{22}}{\sqrt{g^{11}}} \,,\quad
m_2 = \sqrt{g^{11}} \qquad\Leftrightarrow\qquad
g^{11} = m_2^2 \:,\quad g^{22} = m_1 m_2 \:.
\end{equation}
\end{subequations}
The (spatial) curvature scalar~\eqref{curvscal} reads
\begin{equation}
R = {-\textfrac{1}{2}}\: \hat{n}_1^2 m_1^2 + 2 \,\hat{n}_1 \hat{n}_2 \,m_1 m_2 \:. 
\end{equation}

%---------------------------------------------------------------------------------
\subsection{Kantowski-Sachs models}
\label{KSmodels}
%---------------------------------------------------------------------------------

For Kantowski-Sachs models, by assumption, the orbits of 
each three-dimensional subgroup $\mathcal{G}_3$ of the isometry group $\mathcal{G}_4$
are merely two-dimensional. (Note that 
a four-parameter Lie group necessarily admits a three-dimensional 
subgroup, see, e.g.,~\cite[Appendix A]{Collins}.)
A two-dimensional manifold admitting a three-parameter isometry group 
is necessarily a manifold of constant curvature; this entails that
the generators $\hx_i$ of the group $\mathcal{G}_3$
form the Lie algebra
\begin{equation}\label{G3sub}
[\hx_1,\hx_2] = \hx_3\:,\qquad [\hx_2,\hx_3] = k \hx_1\:,\qquad [\hx_3,\hx_1] = \hx_2 \:,
\end{equation}
where $k \in \{0,{-1},{+1}\}$ denotes the normalized curvature of the constant curvature orbits.
(These Lie algebras are of Bianchi type $\mathrm{VII}_0$, $\mathrm{VIII}$, and $\mathrm{IX}$, respectively.)
However, in the cases $k=0$ and $k={-1}$, the existence of a three-dimensional
subgroup $\mathcal{G}_3$ according to~\eqref{G3sub} implies that the Lie group 
$\mathcal{G}_4$ possesses additional three-dimensional subgroups, whose orbits are in fact
three-dimensional, see~\cite[Appendix B]{Collins} and~\cite{Kantowski}; 
accordingly, $k=0$ and $k={-1}$ are LRS Bianchi models.

On the other hand, in the case $k={+1}$,
if~\eqref{G3sub} is a subalgebra of a four-dimensional Lie algebra, then 
that Lie algebra is necessarily represented by
\begin{subequations}\label{G4alg}
\begin{alignat}{3}
\label{G3sub2}
[\hx_1,\hx_2] & = \hx_3 \:,\qquad
& [\hx_2,\hx_3] & = \hx_1 \:,\qquad
& [\hx_3,\hx_1] & = \hx_2 \:, \\
[\heta,\hx_1] & = 0 \:,\qquad
& [\heta,\hx_2] & = 0\:,\qquad
& [\heta,\hx_3] & = 0 \:,
\end{alignat}
\end{subequations}
see~\cite[Appendix B]{Collins}.
The Lie group $\mathcal{G}_4$ generated by~\eqref{G4alg}
possesses a unique three-dimensional subgroup, the group $\mathcal{G}_3$ that
is generated by~\eqref{G3sub2}, cf.~\cite{Collins}.
Therefore, it is possible to assume that $\mathcal{G}_4$ 
acts as an isometry group of the spacetime 
in such a way that the orbits of $\mathcal{G}_3$ are two-dimensional 
spaces of positive constant curvature, i.e., $2$-spheres.
As a consequence, the metric reads
\begin{equation}\label{KSmetric}
{}^4\mathbf{g} = - d t^2 + g_{11}(t)\, d r^2 + g_{22}(t)\: \,{}^2\mathbf{g}_{[S^2]}\:,
\end{equation}
where ${}^2\mathbf{g}_{[S^2]}$ is the standard metric on the $2$-sphere.
These are the Kantowski-Sachs models.
Assuming that the coordinate $r$ ranges in $S^1$, the spatial topology is
$S^1 \times S^2$ and thus compact.

The Ricci curvature $R_{i j}$ of the spatial part of the Kantowski-Sachs metric~\eqref{KSmetric} 
has a rather simple structure. 
As suggested by~\eqref{KSmetric} let $x^1 = r$ and $x^I$, $I=2,3$, be coordinates on the 2-sphere.
Then $R_{11} = 0$, $R_{I K} = {}^2\mathbf{g}_{I K}$ (where $I,K=2,3$), and the remaining components of the Ricci tensor
vanish.
Accordingly, 
\begin{equation}\label{KSRic}
R^1_{\weg 1} = 0\:, \qquad
R^2_{\weg 2} = g^{2 2} = m_1 m_2 \:,
\end{equation}
where we resort to the quantities $m_1$, $m_2$ of~\eqref{m1m2def}.
The component $R^3_{\weg 3}$ is identical to $R^2_{\weg 2}$; the remaining components are zero.
The curvature scalar $R$ is simply $R = 2 m_1 m_2$.

\begin{Remark}
The Ricci tensor~\eqref{KSRic} of Kantowski-Sachs models is obtained from the Ricci tensor~\eqref{lrsricciA} 
of LRS Bianchi class~A models 
by formally setting $\hat{n}_1^2=0$ and $\hat{n}_1\hat{n}_2=1$. 
\end{Remark}

%---------------------------------------------------------------------------------
\subsection{Lie contractions}
\label{Liecontractions}
%---------------------------------------------------------------------------------

Contractions of a Lie algebra are obtained by considering 
sequences of basis transformations whose limit is singular 
`in a controlled way', i.e., in the limit 
one observes convergence of the structure constants~\cite{Saletan}.
In~\cite{Jantzen} the notion of Lie algebra contractions is applied
to the Bianchi models: For class~A models 
there exists a hierarchy of Lie algebra contractions, which 
corresponds to successively setting the structure constants
$\hat{n}_1$, $\hat{n}_2$, $\hat{n}_3$ to zero.
In this way, the type~IX algebra generates the algebras associated with
the `lower' Bianchi types~$\mathrm{VII}_0$,~II, and~I, and the type~VIII algebra generates
the algebras of the `lower' types~$\mathrm{VII}_0$,~$\mathrm{VI}_0$,~II, and~I.
This hierarchy of the Bianchi types is of fundamental importance in
the analysis of the dynamics of the `higher' Bianchi types --- the asymptotic
dynamics of `higher' Bianchi types are directly related to the dynamics
of the `lower' types, see~\cite{Jantzen} or, e.g.,~\cite{WE}.

Consider the Lie algebra~\eqref{lrslie} of the four-dimensional
isometry group of LRS type~IX models, i.e., $\hat{n}_1 = 1$, $\hat{n}_2 = \hat{n}_3 = 1$
and $\hat{\epsilon} = 1$.
There are three possible contractions: We may set $\hat{n}_1 =0$, or $\hat{n}_2 = \hat{n}_3 = 0$,
or $\hat{\epsilon} = 0$.
Setting $\hat{n}_1 = 0$ (while leaving the remaining constants unchanged)
we obtain the Lie algebra of LRS type~$\mathrm{VII}_0$ models;
this algebra coincides with the LRS type~I algebra, see Sec.~\ref{LRSclassA}.
Setting $\hat{n}_2 = \hat{n}_3 = 0$ we obtain 
the Lie algebra of LRS type~$\mathrm{II}$ models.
Finally, setting $\hat{\epsilon} = 0$ we obtain the Lie algebra~\eqref{G4alg}
representing the isometries of Kantowski-Sachs models.
For a detailed discussion of the contraction of~\eqref{lrslie} see~\cite{Saletan}.

\begin{Remark}
The type~IX Lie algebra~\eqref{lrslie}, where $\hat{n}_1 = 1$, $\hat{n}_2 = \hat{n}_3 = 1$
and $\hat{\epsilon} = 1$, is isomorphic to the Lie algebra~\eqref{G4alg}.
In other words, the Lie algebra contraction obtained by setting $\hat{\epsilon}$ to
zero does not yield anything new on the level of the Lie algebras (or the associated Lie groups).
However, when~\eqref{lrslie} is regarded as an isometry group (on an LRS type~IX space),
the Lie algebra contraction process involves a contraction of the action of the group. 
In the singular limit that corresponds to the Lie algebra contraction one obtains 
a different action of~\eqref{lrslie}; this action corresponds to the action of~\eqref{G4alg}
on Kantowski-Sachs spaces.
\end{Remark}

%---------------------------------------------------------------------------------
\subsection{LRS Bianchi type III}
\label{typeIIIsubsec}
%---------------------------------------------------------------------------------

The Lie algebra of the four-dimensional
isometry group of LRS type~VIII models is given by~\eqref{lrslie} with
$\hat{n}_1 = {-1}$, $\hat{n}_2 = \hat{n}_3 = 1$
and $\hat{\epsilon} = 1$.
There are three possible contractions: We may set $\hat{n}_1 =0$, or $\hat{n}_2 = \hat{n}_3 = 0$,
or $\hat{\epsilon} = 0$.
Setting $\hat{n}_1 = 0$ (while leaving the remaining constants unchanged)
we obtain the Lie algebra of LRS type~$\mathrm{VII}_0$ models;
this algebra coincides with the LRS type~I algebra, see Sec.~\ref{LRSclassA}.
Setting $\hat{n}_2 = \hat{n}_3 = 0$ we obtain 
the Lie algebra of LRS type~$\mathrm{II}$ models (in a different
representation where the sign of $\hat{n}_1$ is negative instead of positive). 
Finally, setting $\hat{\epsilon} = 0$ we obtain the Lie algebra~\eqref{G4alg}
\begin{subequations}\label{G4alg2}
\begin{alignat}{3}
\label{G3sub22}
[\hx_1,\hx_2] & = \hx_3 \:,\qquad
& [\hx_2,\hx_3] & = {-\hx_1} \:,\qquad
& [\hx_3,\hx_1] & = \hx_2 \:, \\
[\heta,\hx_1] & = 0 \:,\qquad
& [\heta,\hx_2] & = 0\:,\qquad
& [\heta,\hx_3] & = 0 \:.
\end{alignat}
\end{subequations}
The remarks of Sec.~\ref{Liecontractions} apply analogously.
In~\cite[Appendix B]{Collins} it is shown that this Lie algebra contains
not only the subalgebra~\eqref{G3sub22} but also 
a subalgebra of the type
\begin{equation}
[\hx^\prime_1,\hx^\prime_2] = 0 \:,\qquad
[\hx^\prime_2,\hx^\prime_3] = \hx^\prime_2 \:,\qquad
[\hx^\prime_3,\hx^\prime_1] = 0 \:.
\end{equation}
This is a Lie algebra of Bianchi type~III, see, e.g.,~\cite{WE}.

The action of the Lie group generated by~\eqref{G4alg2}
resembles the action of~\eqref{G4alg} on the Kantowski-Sachs models.
Consider the metric
\begin{equation}\label{LRSIII}
{}^4\mathbf{g} = - d t^2 + g_{11}(t)\, d r^2 + g_{22}(t)\: \,{}^2\mathbf{g}_{[H^2]}\:,
\end{equation}
where ${}^2\mathbf{g}_{[H^2]}$ is the metric of constant negative curvature on
hyperbolic space.
This space is an orbit of the subgroup generated by~\eqref{G3sub22}, i.e.,~\eqref{G3sub22} 
acts multiply transitively on ${}^2\mathbf{g}_{[H^2]}$.

In analogy to~\eqref{KSRic}, the Ricci curvature $R_{i j}$ of the spatial part of the LRS type~III metric~\eqref{LRSIII} 
has a rather simple structure. 
We obtain
\begin{equation}\label{typeIIIricci}
R^1_{\weg 1} = 0\:, \qquad
R^2_{\weg 2} = {-g^{2 2}} = -m_1 m_2 \:,
\end{equation}
where we resort to the quantities $m_1$, $m_2$ of~\eqref{m1m2def}.
The component $R^3_{\weg 3}$ is identical to $R^2_{\weg 2}$; the remaining components are zero.
The curvature scalar $R$ is simply $R = 2 m_1 m_2$.

\begin{Remark}
The Ricci tensor~\eqref{typeIIIricci} of LRS Bianchi type~III models 
is obtained from the Ricci tensor~\eqref{lrsricciA} 
of LRS Bianchi class~A models 
by formally setting $\hat{n}_1^2=0$ and $\hat{n}_1\hat{n}_2={-1}$. 
\end{Remark}

%%%%%%%%%%%%%%%%%%%%
\section{Dynamical systems toolkit}
\label{dynsysapp}
%%%%%%%%%%%%%%%%%%%%%

In this appendix we briefly present some concepts from dynamical systems theory that 
are used in this paper. For a thorough development of these concepts we refer to the 
book~\cite{Perko}; for a discussion of applications
in cosmology we refer to~\cite{Coley,WE}. 

Let $f:\R^n\to\R^n$ be a $\mathcal{C}^1$ vector field and consider the autonomous dynamical system
\begin{equation}\label{dynamicalsystem}
\dot{x}=f(x)
\end{equation}
for the function $x=x(t)\in\R^n$. We assume that for every initial data point $x_0$, 
the system~\eqref{dynamicalsystem} 
admits a unique global solution $x\in\mathcal{C}^1(\R)$, $x(0) = x_0$.
A point $x_*\in\R^n$ is said to be an $\alpha$-limit point of the solution $x(t)$ if 
there exists a sequence of times $t_n\to-\infty$ such that $x(t_n)\to x_*$ as $n\to\infty$. 
The concept of $\omega$-limit is defined analogously by considering a sequence of 
times $t_n\to +\infty$. The $\alpha$-limit set of a solution is the set of all its $\alpha$-limit points. 
By the autonomy property, an orbit $\gamma$ in the state space is 
associated with a one-parameter set of solutions; if the orbit reduces to a point, the 
associated solution of~\eqref{dynamicalsystem} is a stationary solution. Since 
solutions with the same orbit $\gamma$ have the same limit sets, we may speak of
$\alpha$- and $\omega$-limit sets of orbits and we shall denote these sets 
as $\alpha(\gamma)$ and $\omega(\gamma)$. 
The following lemma collects some properties of $\alpha(\gamma)$ that are used in the paper; analogous 
properties hold for $\omega(\gamma)$. 
\begin{Lemma}
Let $\gamma$ be an orbit of the dynamical system~\eqref{dynamicalsystem}. The following holds:
\begin{itemize}
\item[(i)] $\alpha(\gamma)$ is closed.
\item[(ii)] If the exists a compact set $E\subset\R^n$ and $t_0\in(-\infty,+\infty]$ 
such that $\gamma(t)\in E$, for all $t<t_0$, then $\alpha(\gamma)$ is a non-empty, 
connected, compact subset of $E$.
\item[(iii)] If\/ $\mathrm{P}\in\alpha(\gamma)$, then the whole orbit through the point\/ $\mathrm{P}$ 
belongs to $\alpha(\gamma)$. In particular, the $\alpha$-limit set is flow invariant.
\item[(iv)] If $\alpha(\gamma)$ consists of a point\/ $\mathrm{P}$ only, then\/ $\mathrm{P}$ is a fixed point. 
\item[(v)] A source $\mathrm{P}$ is 
the $\alpha$-limit of all orbits in a neighborhood of\/ $\mathrm{P}$. If\/ $\mathrm{P}\in\alpha(\gamma)$, 
then $\alpha(\gamma)=\{\mathrm{P}\}$.
\end{itemize}
\end{Lemma}
Let us recall the definition of a heteroclinic cycle. A connecting orbit $\gamma$ is an 
orbit for which both the $\alpha$- and $\omega$-limit sets consist of 
a point (which is then necessarily a fixed point). If $\alpha(\gamma)=\omega(\gamma)=\{\mathrm{P}\}$, 
the set $\{\gamma\}\cup\{\mathrm{P}\}$ is called a homoclinic cycle. If $\gamma_1$, $\gamma_2$ are 
two connecting orbits such that $\{\mathrm{P}_1\}=\omega(\gamma_1)=\alpha(\gamma_2)$ and 
$\mathrm{P}_2=\omega(\gamma_2)=\alpha(\gamma_1)$, then the set 
$\{\gamma_1\}\cup\{\gamma_2\}\cup\{\mathrm{P}_1\}\cup\{\mathrm{P}_2\}$ is called 
a heteroclinic cycle. Likewise one can define heteroclinic cycles consisting of an arbitrary number of 
connecting orbits. 

In the paper we often make use of the monotonicity principle.
It can be stated as follows (for an extension see~\cite{FHU}):

\begin{Theorem}\label{monopri}
Let $S\subseteq \R^n$ be an invariant subset of the dynamical system~\eqref{dynamicalsystem} and $Z:S\to\R$ be 
strictly monotonically decreasing along the flow; set $a=\inf\{Z(y),y\in S\}$ and \mbox{$b=\sup\{Z(y), y\in S\}$}. 
Let $\gamma$ be any orbit in $S$. Then
\begin{align*}
\alpha(\gamma)\subseteq\{s\in\partial S:\lim_{y\to s}Z(y)\neq a\}\:,\quad
\omega(\gamma)\subseteq\{s\in\partial S:\lim_{y\to s}Z(y)\neq b\}\:.
\end{align*}
\end{Theorem}

For a dynamical system over a two-dimensional state space, there exist further results that help studying 
the asymptotic behavior of orbits without the need for a monotone function (which is often hard to find). 
The most useful result is the Poincar\'e-Bendixson theorem, which states the following:

\begin{Theorem}
Consider the dynamical system~\eqref{dynamicalsystem} in the plane, i.e., $n=2$, and assume that there are at 
most a finite number of equilibrium points. Then, for any orbit, the $\alpha$-limit set ($\omega$-limit set)
can only be one of the following: A fixed point, a periodic orbit, or a heteroclinic cycle 
(or a heteroclinic network, i.e., the union of heteroclinic cycles). 
\end{Theorem}  
A proof of this theorem can be found in~\cite{Perko}.

\end{appendix}

\end{document}